\documentclass[11pt,a4paper]{article}

\usepackage[a4paper,hmargin=25mm,vmargin=25mm]{geometry}
\usepackage[english]{babel}
\usepackage[utf8]{inputenc}
\usepackage[T1]{fontenc}

\usepackage{multirow}

\usepackage{authblk}
\usepackage{amsmath}
\usepackage{amsbsy}
\usepackage{amsfonts}
\usepackage{amssymb}
\usepackage{amscd}
\usepackage{amsthm}
\usepackage{color}
\usepackage{appendix}

\usepackage{dsfont}
\usepackage{units}
\usepackage{pdfsync}
\usepackage{tikz}
\usetikzlibrary{calc}
\usepackage{caption}
\usepackage{subcaption}
\usepackage{hyperref}
\hypersetup{colorlinks=true, linkcolor=black,citecolor=black}

\newtheorem{de}{Definition}
\newtheorem{remark}{Remark}
\newtheorem{lemma}{Lemma}
\newtheorem{prop}{Proposition}
\newtheorem{thm}{Theorem}
\newtheorem{corollary}{Corollary}

\DeclareMathOperator{\mdiv}{div}
\DeclareMathOperator{\trace}{trace}
\DeclareMathOperator{\var}{Var}

\def\R{\mathbb{R}}

\def\MeasX{\mathcal{M}^X}
\def\MeasXC{\mathcal{M}^{X,C}}
\def\one{\mathbf{1}}
\def\eJ{\tilde{J}}
\def\X{{\mathcal X_0}}

\def\Haus{\mathcal{H}}
\def\sgn{\text{sgn}}

\def\df{\delta\!f}
\def\FB{\mathcal F}
\def\Cat{\text{Cat}}

\providecommand{\f}{\ensuremath{\boldsymbol{f}} }
\providecommand{\x}{\ensuremath{\boldsymbol{x}} }
\providecommand{\y}{\ensuremath{\boldsymbol{y}} }
\providecommand{\g}{\ensuremath{\boldsymbol{g}} }
\providecommand{\h}{\ensuremath{\boldsymbol{h}} }
\providecommand{\bzeta}{\ensuremath{\boldsymbol{\zeta}} }
\providecommand{\p}{\ensuremath{\boldsymbol{p}} }
\providecommand{\bJ}{\ensuremath{\boldsymbol{J}} }
\providecommand{\bu}{\ensuremath{\boldsymbol{u}} }

\def\Xtemp{X}
\def\ftemp{f}

\providecommand{\ie}{{\it i.e.} }
\providecommand{\dT}{d}
\providecommand{\nT}{T}
\providecommand{\dP}{n}
\providecommand{\nP}{P}
\providecommand{\abs}[1]{\left\lvert#1\right\rvert}
\providecommand{\sabs}[1]{\lvert#1\rvert}

\providecommand{\norm}[1]{\left\lVert#1\right\rVert}

\providecommand{\snorm}[1]{\lVert#1\rVert}
\providecommand{\prs}[1]{\left\langle #1\right\rangle}

\newcommand\rst[2]{{#1}_{\restriction_{#2}}}

\usepackage{algorithm}
\usepackage{algpseudocode}

\algblock[corp]{Begin}{End}

\DeclareUnicodeCharacter{00A0}{~}

\author[1,2]{B. Charlier}
\author[1,3]{N. Charon}
\author[1]{A. Trouvé}
\affil[1]{\small CMLA, UMR 8536, \'Ecole normale supérieure de Cachan, France}
\affil[2]{\small I3M, UMR 5149, Université Montpellier II, France}
\affil[3]{\small DIKU, University of Copenhagen, Denmark}

\title{The fshape framework for the variability analysis of functional shapes}

\begin{document}

\maketitle

\begin{abstract}
This article introduces a full mathematical and numerical framework for treating functional shapes (or fshapes) following the landmarks of shape spaces and shape analysis. Functional shapes can be described as signal functions supported on varying geometrical supports. Analysing variability of fshapes' ensembles require the modelling and quantification of joint variations in geometry and signal, which have been treated separately in previous approaches. Instead, building on the ideas of shape spaces for purely geometrical objects, we propose the extended concept of fshape bundles and define Riemannian metrics for fshape metamorphoses to model geometrico-functional transformations within these bundles. We also generalize previous works on data attachment terms based on the notion of varifolds and demonstrate the utility of these distances. Based on these, we propose variational formulations of the atlas estimation problem on populations of fshapes and prove existence of solutions for the different models. The second part of the article examines thoroughly the numerical implementation of the tangential simplified metamorphosis model by detailing discrete expressions for the metrics and gradients and proposing an optimization scheme for the atlas estimation problem. We present a few results of the methodology on a synthetic dataset as well as on a population of retinal membranes with thickness maps. 
\end{abstract}

\tableofcontents

\section{Introduction}
Shape spaces have emerged as a natural mathematical setting to think about shapes as a structured space, usually a differential manifold or even more a Riemannian manifold. In that setting, group actions of diffeomorphisms \cite{trouve98:_diffeom,dupuis1998variational,Miller2002,Trouve2,Miller2006,Trouve2011,Younes,michor2013zoo}  that should be rooted to the seminal work of V. Arnold \cite{Arnold1966} are powerful vehicles to build a full theoretical as well as computational framework for a comprehensive quantitative analysis of shape variability  recently coined as \emph{diffeomorphometry} and successfully applied in computational anatomy  \cite{doi:10.1142/S2339547814500010}.

The main purpose of the paper is to develop a framework embedding the situation of geometrical shapes carrying functional information that we call here fshapes for \emph{functional shapes}. The idea of functional shape is quite natural since in many scientific settings, a geometrical shape $X$ (basically a submanifold) has associated with it a scalar field $f:X\to\mathbb{R}$ attached to every geometrical point of the support $X$. What is much more unusual however is to consider a fshape { i.e.} the pair $(X,f)$ as a single object that should live in some well defined and well structured ensemble $\mathcal{F}$ of fshapes from which further processing can be derived. We believe this approach of putting the effort on a well grounded definition of shape spaces that play a similar role to functional spaces in modern analysis should be pushed forward to the setting of fshapes. This global point of view leading to the notion of geometry of shape spaces is now well established with fruitful development in many regards \cite{micheli2013sobolev} for purely geometrical shapes but as far as we know the extension of such approach to functional shapes has only started very recently in \cite{Charon_thesis}.

A core issue in working with fshapes is their mixed geometrical and functional nature so that smooth infinitesimal transformations of a given fshape $(X,f)$ should combine a geometrical and functional component simultaneously. The geometrical part should transport the supporting manifold $X$, and the functional one should modify the functional signal $f$. This notion of combined geometrical and functional infinitesimal transformation has been introduced previously in \cite{Trouve1995a} in the setting of images on a fixed support $X$ (where $X$ in the unit square or a flat torus). In that case the infinitesimal signal evolution $\delta f$ is the combination of two factors: on the one hand, the variation of the signal due to the geometric transport of pixel values, 
on the other hand a purely additive perturbation of the signal. This has been further studied in \cite{trouve05:_local_geomet_defor_templ} and conceptualize as the  \emph{metamorphosis} framework in \cite{Trouve1}. The situation we are looking at here is more general since now the support is a submanifold that can freely evolve during the metamorphosis (as in the example of Figure \ref{fig:statues_metam}). A first global outcome is that the associated space of fshapes is not only an orbit $G.X_0$ under the action of a subgroup of smooth diffeomorphisms $G$ of the ambient space on a geometrical template $X_0$: the orbit $G.X_0$ is now the base space of a vector bundle and above each manifold $X$ in the orbit we consider the full vector space $L^2(X)$ of square integrable signals so that we end up with a vector bundle $\mathcal{F}$ over the orbit (see Figure \ref{fig:fshape_bundle}). The precise construction of the fshape bundle $\mathcal{F}$ is described in Section 2 as well as the associated metamorphosis based Riemannian metric for which we can prove two  key results: the existence of geodesic between any two fshapes (Theorem \ref{thm:geo}) and the existence of Karcher means for a population of such fshapes in $\mathcal{F}$ (Theorem \ref{thm:Kar}). 

Obviously there is no hope to build a fshape bundle $\mathcal{F}$ that can contain any possible pair $(X,f)$ in particular since any two fshapes within $\mathcal{F}$ will have diffeomorphic supports. Every fshape in $\mathcal{F}$ should be understood as an ideal model for truly noisy observed fshapes that generically do not belong to $\mathcal{F}$. As a consequence, a second core issue is to build proper smooth data attachment terms or dissimilarity measures that can be defined between any arbitrary pair of fshapes with possibly non diffeomorphic supports. The powerful setting of mathematical currents equipped with dual norms for purely geometrical shapes \cite{Glaunes2} has been successfully extended to the situation of fshapes in \cite{Charon1}. However, we develop here (Section \ref{sec:atlas_estimation_formulation}) a different approach encoding non oriented tangential space information that is based on the concept of varifold along the lines of \cite{Charon2} and extended here to the new situation of functional varifolds. We believe that this new type of dissimilarity measure, which is able to compare fshapes on smooth manifolds as well as on polyhedral meshes, is well suited in combination with metamorphoses distances to define a theoretical and computational framework for fshapes. In particular, we establish that useful functional varifolds metrics can be induced by smooth embeddings of functional varifolds into Reproducible Kernel Hilbert Spaces (RKHS) associated with computable kernels. We then establish that the resulting metric has several key regularity properties with respect to smooth geometrical and functional variations of fshapes (Theorem \ref{theo:variation_formula}). Such regularity results open the way to various smooth and computationally tractable variational problems on fshapes involving observed fshapes. 

\begin{figure}[H]
\centering
  \includegraphics[width=.24\textwidth]{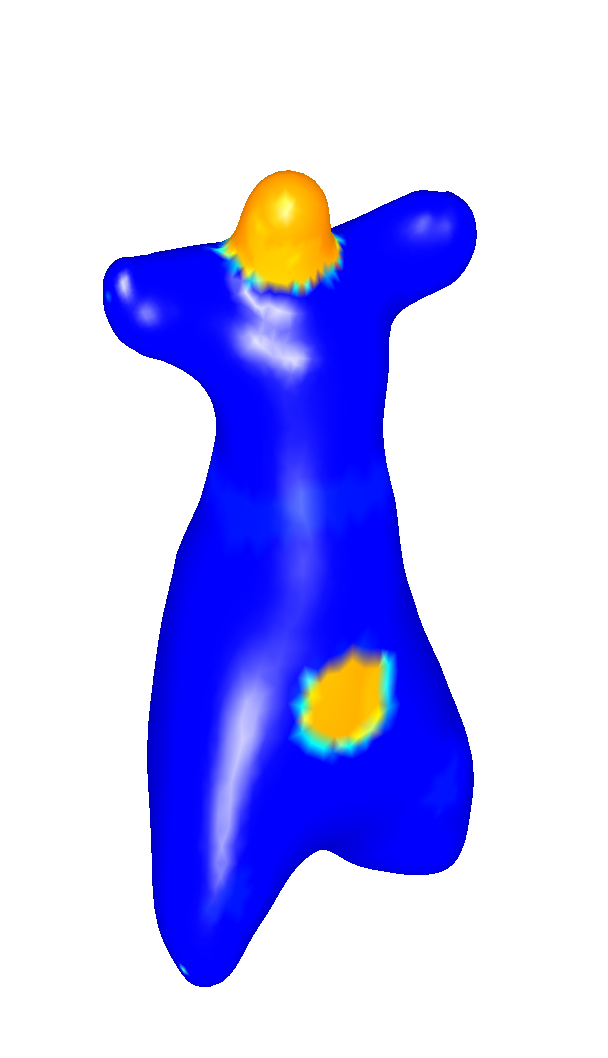}  \includegraphics[width=.24\textwidth]{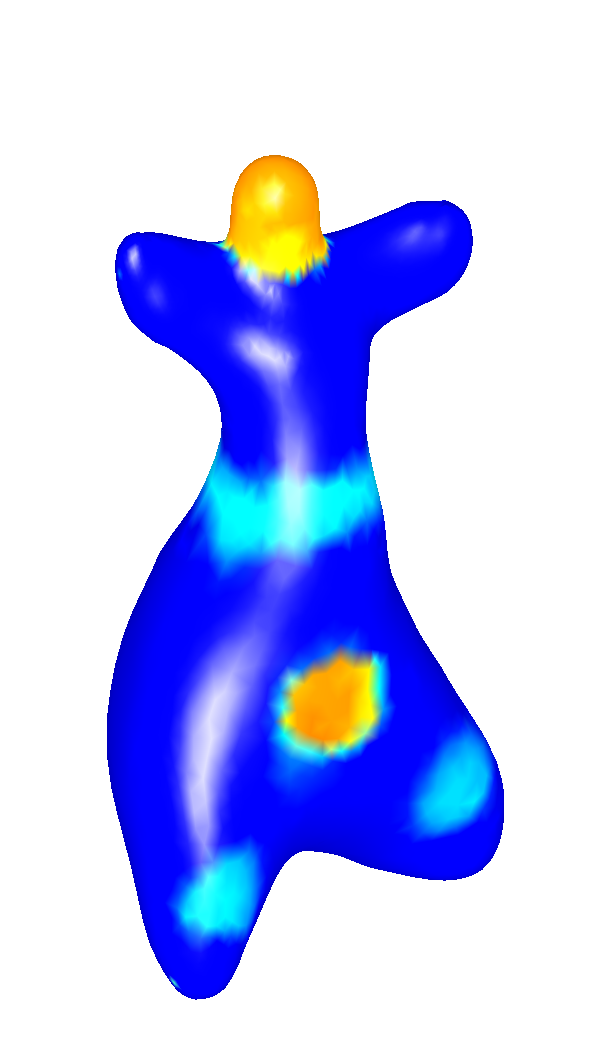}  \includegraphics[width=.24\textwidth]{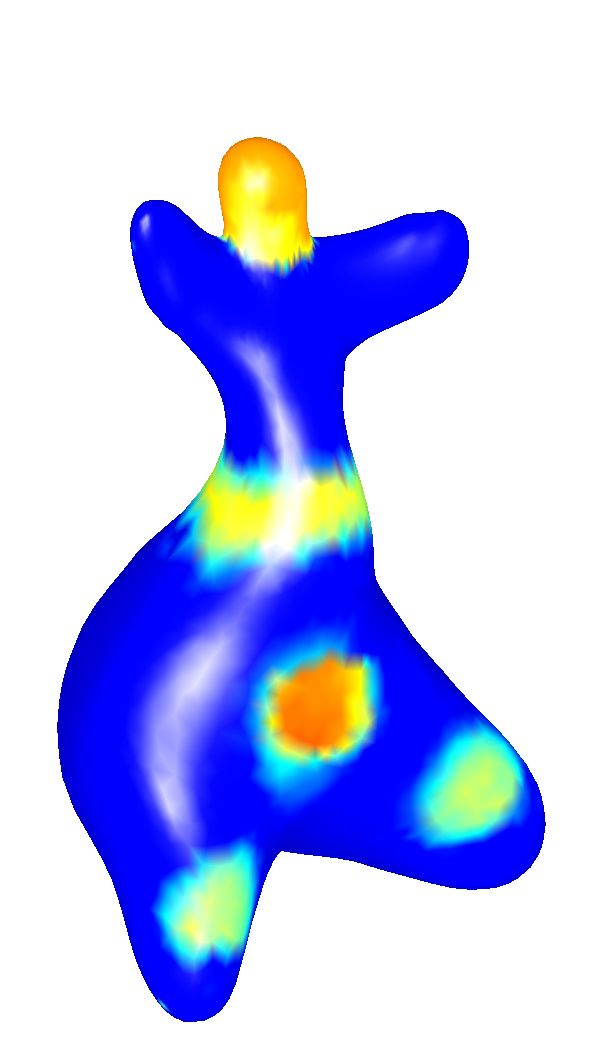}  \includegraphics[width=.24\textwidth]{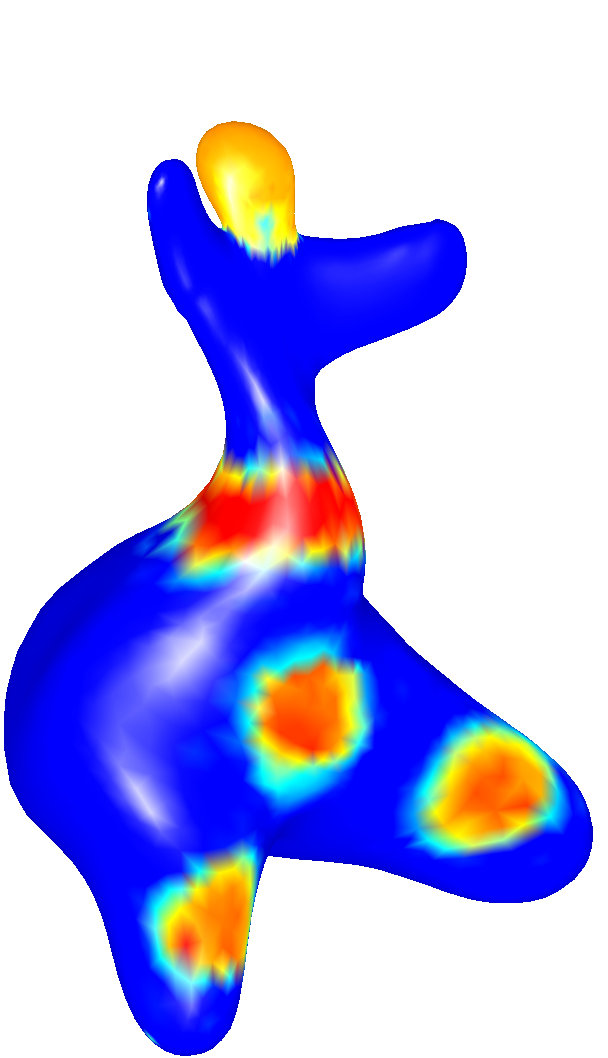}
\caption{Example of metamorphosis transformation of a fshape.}
\label{fig:statues_metam}
\end{figure}

We concentrate on the central problem of template (or atlas in the context of computational anatomy) estimation which is the extension of the Karcher mean problem to the much broader situation where the observed fshapes do not belong to any \emph{a priori} defined fshape bundle $\mathcal{F}$. A much more studied situation is the case of purely image data that has been pushed forward on practical ground in the framework of actions of diffeomorphisms in \cite{joshi2004unbiased}. Here the goal is to obtain, via a coherent fshape framework, a mean template $(X,f)$ encoding simultaneously the mean geometrical and functional information from a dataset $(X^i,f^i)$ of fshapes as the minimizer of a well posed variational problem. This variational problem is basically the sum of the square of the metamorphosis distances between the template $(X,f)$ and the approximating models $(\tilde{X}^i,\tilde{f}^i)$ of the observations $(X^i,f^i)$ within the fshape bundle ${\mathcal F}$ (associated to $(X,f)$)  with the addition of the square of the fvarifold discrepancy measure for the difference between the models $(\tilde{X}^i,\tilde{f}^i)$ and the true observation $(X^i,f^i)$. We prove that introducing a hypertemplate in the spirit of \cite{Ma2008,ma2010bayesian} and restricting the template to belong to the fshape bundle ${\mathcal F}_0$ of a hypertemplate $(X_0,f_0)$ is sufficient to establish a well posed variational problem (Section \ref{sec:atlas_estimation_formulation}). The proof of existence of an optimal template $(X,f)$ as the minimizer within the fshape bundle $\mathcal{F}_0$ is established on fairly general assumptions that cover most of the practical situations (Theorem \ref{theo:existence_atlas_fvar_tang}). Note that even restricted to the situation of purely geometrical shapes, as far as we know, there is no available results in the literature providing a rigorous proof of existence of an average shape from a population of noisy geometrical shapes. Available rigorous results have usually focused on the situation where the dataset is sampled from a Riemannian manifold (usually finite dimensional). When the functional part is also involved, the situation is more complex due to the interactions between function and geometry inside the Riemannian metamorphosis distance but also the fvarifold metric defining the dissimilarity term. We establish similar existence results in a slightly simpler situation, called the tangential model for the metamorphosis metric, where the infinitesimal signal variations along a path are computed with respect of a frozen $L^2(X)$ metric on the geometrical support $X$ of the initial fshape. 

Since the paper aims at presenting both a theoretical and a computational framework, we extensively explore the material needed to bridge the gap between the theoretical framework and its derivation into a computational and algorithmic one in the latter part of this paper. We believe that this part will give also more practical insights into what goals can be achieved with the fshape framework. Moreover, we provide the details necessary to conform to the minimal standard of reproducible research. Just as the absence of a well posed theoretical framework may often produce inconsistent algorithms, we believe that the lack of precise descriptions of the underlying numerical schemes and hidden ``tricks'' may also hinder the development of incrementally better solutions and fair comparison by other researchers. We start in Section \ref{part:discreteScheme} with the derivation of a discrete framework based on polyhedral representation of fshapes leading to particle based approximations of the fvarifold dissimilarity metric and of the Hamiltonian underlying the optimal control formulation of the geodesic trajectories. The actual implemented framework is based on the tangential model for the metamorphosis metric that leads to a slightly more straightforward variational problem. The optimization scheme is carefully described in Section 7 and numerical illustrations are provided in Section \ref{part:numerical} on real and synthetic datasets. A special attention is given in Section \ref{part.traps} on various  numerical issues and pitfalls that may affect the computation process. In particular we discuss the very important problem of open surfaces that introduce free boundaries evolutions during the optimisation process. The tuning of free parameters are also discussed. 

As discussed in this introduction, this paper describes on the one hand a new theoretical shape space framework to work with functional shapes with an emphasis on precise mathematical statements and proofs. The more technical aspects of the novel framework have been collected and presented in Appendix to ease the reading of the paper. The reader more interested in applying the framework could read Section 2 together with the first two Subsections of Section 3, skip Section 4 and 5 and jump over directly to Section 6 to 9 dealing with the numerical part of the paper.	

\section{Riemannian metamorphosis framework for fshapes}
In the classical Grenander's setting, shape spaces are modelled as sets of shapes homogeneous under the action of a group of space transformations. Metrics between shapes are then induced from right-invariant Riemannian metrics on the group itself. The goal of this section is to propose a similar but extended framework for the case of functional shapes. We show, in the first place, that sets of functional shapes can be structured naturally as vector bundles and we then define a Riemannian setting to model and quantify transformations within those spaces.  

\subsection{Fshape bundles}
Let's introduce a finite dimensional vector space $E$ that shall be the embedding space of all shapes. In the large deformations' model (LDDMM, \textit{cf} \cite{Trouve2}), diffeomorphisms are constructed as flows of time-varying vector fields. The basic ingredient is a \textit{reproducing kernel Hilbert space} (RKHS) of vector fields on $E$ that is denoted $V$ and which is continuously embedded in $C_0^1(E,E)$ the space of continuously differentiable functions from $E$ to $E$ vanishing at the infinity. Let $G=G_V$ be the associated group of diffeomorphisms obtained by flowing vector fields in $L^2([0,1],V)$ (see \cite{Trouve2011,Younes,arguillere14:_shape}). Let $\X$ be a homogeneous space generated by a finite volume $d$-dimensional rectifiable compact subset $X_0$ (see \cite{Morgan}) {\it i.e.} $\X=G.X_0=\{\phi(X_0) \ | \ \phi \in G\}$. 

Now we consider the space 
\begin{equation}
  \label{eq:FB1}
  \FB\doteq \{\ (X,f)\ |\ X\in\X \text{ and }f\in L^2(X)\ \}\,,
\end{equation}
where $L^2(X)$ is the space of square integrable functions on $X$, i.e the set of functions $f: X \rightarrow \R$ such that 
$$ \int_X f^{2}(x) d\Haus^d(x) < +\infty$$
for $\Haus^d$ the $d$-dimensional Hausdorff (or volume) measure. $\FB$ can be considered as a vector bundle with fiber $L^2(X_0)$ (here we will not try to define any explicit differentiable structure on it so that the last statement should be considered as formal). Each element of $\FB$ is thus given as a couple of a geometrical shape $X$ and a signal function living on $X$, that we call a \textit{functional shape} (see \cite{Charon1}) or fshape.  

\subsection{Metamorphosis distance on fshapes bundles}
\label{ssec:metamorphosis_bundles}
For a purely geometrical deformation $\phi \in G_V$, a natural transport of a fshape $(X,f)$ would be given by $\phi.(X,f)=(\phi(X), f\circ \phi^{-1})$, which corresponds to deforming the support $X$ and transporting the values of the signal $f$ onto the deformed shape, which is the exact generalization of image deformation. Such a model, originally presented in \cite{Charon1}, remains insufficient to account for variations of signals within fibers themselves. This justifies the following efforts to propose a model of joint geometric and functional transformations. 

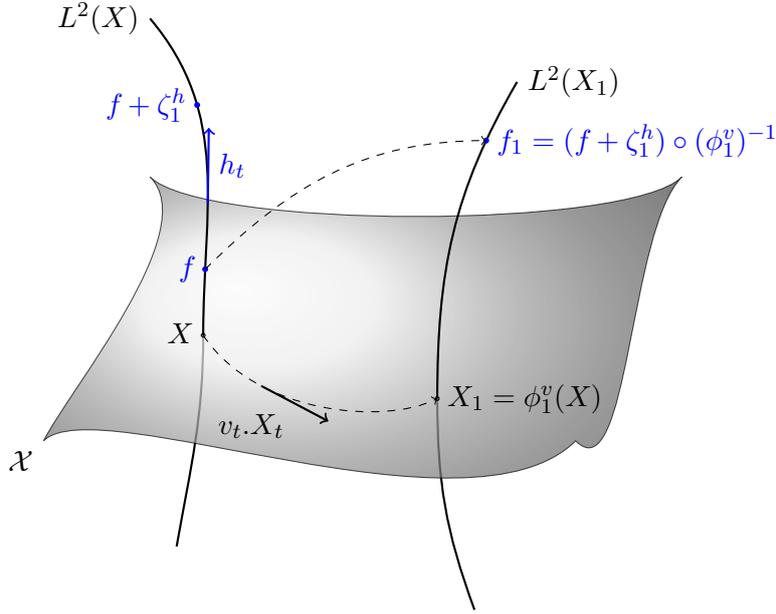
\begin{figure}[h]
\centering
\usetikzlibrary{calc}\usetikzlibrary{decorations.markings}
	\begin{tikzpicture}[
			tangent/.style={
				decoration={
					markings,
					mark=
					at position #1
					with
					{
						\coordinate (tangent point-\pgfkeysvalueof{/pgf/decoration/mark info/sequence number}) at (0pt,0pt);
						\coordinate (tangent unit vector-\pgfkeysvalueof{/pgf/decoration/mark info/sequence number}) at (1,0pt);
						\coordinate (tangent orthogonal unit vector-\pgfkeysvalueof{/pgf/decoration/mark info/sequence number}) at (0pt,1);
					}
				},
				postaction=decorate
			},
			use tangent/.style={
				shift=(tangent point-#1),
				x=(tangent unit vector-#1),
				y=(tangent orthogonal unit vector-#1)
			},
			use tangent/.default=1
		,scale=.7]

		\coordinate (X0) at (-3,-1); \coordinate (X1) at (1.4,-2.2);

		\path[draw,thick] (X0) to[out=-90, in=80] ($(X0) + (-.5,-4)$);
		\path[draw,thick] (X1) to[out=-90,in=110] ($(X1) + (+.7,-4)$);

		\coordinate (a) at (-6,-3);\coordinate (b) at (4,-3); \coordinate (c) at (6,2); \coordinate (d) at (-4,2);
		\shadedraw[shading=ball,ball color=gray!20, opacity=.7] (a).. controls ($(a) +(1,1)$) and ($(b) -(2,2)$) .. (b) .. controls ($(b) +(1,-1)$) and ($(c) + (-1,-1)$).. (c) .. controls ($(c) +(-1,-1)$) and ($(d) +(1,-1)$).. (d) .. controls ($(d) +(1,-1)$) and ($(a) +(.5,1)$) .. (a);

		\node[below left] (XX) at (a) {$\mathcal{X}$};
		\draw[->, dashed, tangent=0.3] (X0) .. controls ($(X0) + (1,-1.5)$) and ($(X1) - (1,0.5)$).. (X1) node[pos=0.7, auto=left, below] {};
		\draw[use tangent,->, thick] (0,0) -- node[below left]{$v_t.X_t$} (1,0);

		\draw[] (X0)node[left]{$X_{}$} circle (1pt);
		\draw[] (X1)node[right]{$X_1 = \phi^v_1(X_{})$} circle (1pt);

		\path[draw, thick,tangent=.8,tangent=.6, tangent=.3] ($(X0) + (-1,6)$)node[left]{$L^2(X)$} to[out=-50, in=90] (X0) ;

		\coordinate[use tangent=1] (F0) at (0,0);  

		\draw [blue, thick, fill] (F0) node[left]{$f_{}$} circle (1pt);
		\draw [blue, thick, use tangent=2, ->] (0,0) --node[right]{$h_t$} (-1,0);
		\draw [blue, thick, use tangent=3,fill] (0,0)node[left]{$f_{}+\zeta^h_1$} circle (1pt);

		\path[draw, thick,tangent=.2] ($(X1) + (1.5,6)$)node[right]{$L^2(X_1)$} to[out=-120,in=90] (X1);
		\draw [blue, thick, use tangent,fill] (0,0) node[right]{$f_1 = (f_{}+\zeta^h_1) \circ (\phi^v_1)^{-1}$} circle (1pt);

		\coordinate[use tangent=1] (F1) at (0,0); 
		\path[draw,->,dashed] (F0) to[out=45, in=-180] (F1) ;
	\end{tikzpicture}
\caption{Fshape bundle and metamorphosis.}
\label{fig:fshape_bundle}
\end{figure}

Let's consider $(X,f)\in\FB$ and instantaneous velocities $(v,h)\in L^2([0,1],V\times L^2(X))$. We can define from $(v,h)$ a path $(\phi^v,\zeta)_{t\in [0,1]}$ where $\phi^v_t$ is the usual flow of $v$ starting from the identity and $t\mapsto \zeta_t=\int_0^th_sds$ is the path in $L^2(X)$ with instantaneous speed given by $h$. From that path, we get a path $t\mapsto (X_t,f_t)$ in $\FB$ defined by
\begin{equation}
  \label{eq:FB36}
  (X_t,f_t)\doteq (\phi^v_t.X,(f+\zeta^h_t)\circ(\phi^v_t)^{-1})\,.
\end{equation}
where $\phi^v_t.X\doteq \phi^v_t(X)$ is the natural diffeomorphisms transport action on rectifiable subsets of $E$. 
We denote $(\phi^{v}_1,\zeta^{h}_1)$ the end point value of the path. Now we denote for $\gamma_V$, $\gamma_f>0$
\begin{equation}
  \label{eq:FB3}
  E_{X}(v,h)\doteq \frac{\gamma_V}{2}\int_0^1|v_t|_V^2+\frac{\gamma_f}{2}\int_0^1\int_{X}|h_t|^2(x)|\rst{d_x\phi^v_t}{T_{x}X}|d\Haus^d(x)
\end{equation}
where $T_{x}X$ is the tangent space to $X$ at point $x$ (defined $\Haus^d$-almost everywhere on $X$ if $X$ is rectifiable, \textit{cf} \cite{Federer} 3.2.19) and, for any vector space $U$, $|\rst{d_x\phi^v_t}{U}|$ denotes the Jacobian of $\phi^v_{t}$ at $x$ restricted to $U$, which equals $|d_{x}\phi^v_{t}(u_{1})\wedge...\wedge d_{x}\phi^v_{t}(u_{d})|=(\text{det}(\langle d_{x}\phi^v_{t}(u_{i}),d_{x}\phi^v_{t}(u_{j})\rangle)_{i,j})^{1/2}$ if $(u_{1},...,u_{d})$ is an orthonormal basis of $U$.

It is clear that $E_X(v,h)<\infty$ since by definition $\int_0^1|v_t|^2_Vdt<\infty$ and since the usual controls on $\phi_t$ along finite energy paths give a uniform control in $(x,t)$ of $d_x\phi_t$ and the existence of an increasing function $C:\R_+\to\R_+$ (independent of $X$) such that
\begin{equation}
  \label{eq:FB4}
  E_X(v,h)\leq \frac{\gamma_V}{2}\int_0^1|v_t|_V^2dt+C(\int_0^1|v_t|_V^2dt)\frac{\gamma_f}{2}\int_0^1|h_t|^2_{L^2(X)}dt<\infty\,.
\end{equation}

Now for any $(X,f)$ and $(X',f')\in\FB$ we define
\begin{equation}
d_{\FB}((X,f),(X',f'))\doteq (\inf \{ E_X(v,h)\ |\ \phi^v_1.X=X',\ f'=(f+\zeta^h_1)\circ(\phi^v_1)^{-1}\ \})^{1/2}\label{eq:FB8}
\end{equation}
which is a true distance on $\FB$ as stated in the following Theorem:
\begin{thm}
  The function $d_\FB$ given by (\ref{eq:FB8}) defines a distance on $\FB$ ie. is symmetric, satisfies the triangle inequality and $d_{\FB}((X,f),(X',f'))=0$ iff $X=X'$ and $f=f'\in L^2(X)$.

Moreover, the distance is finite everywhere on $\FB\times \FB$.
\end{thm}
\begin{proof}
The symmetry comes from a usual process of \emph{time reversal}. 

Let $(X,f),(X',f')\in \FB$ and $(v,h)\in L^2([0,1],V\times L^2(X))$ such that 
\begin{equation}
X'=\phi^v_1.X,\ f'=(f+\zeta^h_1)\circ(\phi^v_1)^{-1}\,.\label{eq:FB10}
\end{equation}

If we define
$$\tilde{v}_t\doteq-v_{1-t}\text{ and }\tilde{h}_t\doteq -h_{1-t}\circ(\phi^v_1)^{-1}$$ 
for any $t\in [0,1]$, then for $X_t\doteq \phi^v_t.X$ (so that $X_0=X$ and $X_1=X'$) we have 
$$\phi^v_t.X_0=X_t=\phi^{\tilde{v}}_{1-t}.X_1$$
and 
$(\tilde{v},\tilde{h})\in L^2([0,1], V\times L^2(X'))$ with 
\begin{equation}
E_{X}(v,h)=E_{X'}(\tilde{v},\tilde{h})\,.\label{eq:FB9}
\end{equation}
Since one easily checks that $\zeta^{\tilde{h}}_1=-\zeta^{h}_1\circ(\phi^v_1)^{-1}$ so that if $f'=(f+\zeta^h_1)\circ(\phi^v_1)^{-1}$ we have $f=(f'+\zeta^{\tilde{h}}_1)\circ(\phi^{\tilde{v}}_1)^{-1}$ with $X=\phi^{\tilde{v}}_1.X'$  This gives immediately the symmetry.\medskip
 
Concerning the triangle inequality, is comes from a usual process of\emph{ path concatenation}. 
Let $(X,f)$, $(X,f')$ and $(X'',f'')$ be three fshapes in $\FB$ such that $d_\FB((X,f),(X',f'))>0$ and $d_{\FB}((X',f'),(X'',f''))>0$. One easily checks that for $(v,h)\in L^2([0,1],V\times L^2(X))$ and $(v',h')\in L^2([0,1],V\times L^2(X'))$ with 
$$(X',f')=(\phi^{v}_1.X,(f+\zeta^h_1)\circ(\phi^v_1)^{-1})\text{ and }(X'',f'')=(\phi^{v'}_1.X',(f'+\zeta^{h'}_1)\circ(\phi^{v'}_1)^{-1})$$ 
then denoting for $\alpha,\beta>1$ such that $1/\alpha+1/\beta=1$ and $s\in[0,1]$
\begin{equation}
\Cat_{\alpha}((v',h'),(v,h))_s\doteq\beta(v'_{\beta(s-1/\alpha)},h'_{\beta(s-1/\alpha)}\circ\phi)\one_{s\geq 1/\alpha}+\alpha (v_{\alpha s},h_{\alpha s})\one_{0 \leq s<1/\alpha}\label{eq:FB5}
\end{equation}
with $\phi=\phi^v_1$, we have for $(\tilde{v},\tilde{h})\doteq \Cat_{\alpha}((v',h'),(v,h))$ that $(\tilde{v},\tilde{h})\in L^2([0,1],V\times L^2(X))$ and 
\begin{equation}
  \label{eq:FB11}
  (X'',f'')=(\phi^{\tilde{v}}_1.X, (f+\zeta^{\tilde{h}}_1)\circ(\phi^{\tilde{v}}_1)^{-1})
\end{equation}
so that 
\begin{equation}
  \label{eq:FB12}
  d_{\FB}((X,f),(X'',f''))\leq E_{X}(\tilde{v},\tilde{h})^{1/2}\,.
\end{equation}
However, for $\alpha_*=(E_X(v,h)^{1/2}+E_{X'}(v',h')^{1/2})/E_X(v,h)^{1/2}$, we check easily that 
\begin{equation}
E_X(\tilde{v},\tilde{h})^{1/2}=E_X(v,h)^{1/2}+E_{}(v',h')^{1/2}\label{eq:FB7}
\end{equation}
and the triangle inequality follows immediately.

Now, if $d_\FB((X,f),(X',f'))=0$ then there exists a sequence $\phi_n$ such that $\phi_n.X=\phi_n(X)=X'$ and $\phi_n\to\text{Id}$ on $X$. In particular, $X'$ is dense in $X$. As $X'$ is compact (since $X$ is), $X\subset X'$. By symmetry we get $X=X'$. Moreover, there exists also a sequence $\zeta_n\in L^2(X)$ such that $\zeta\to 0$ in $L^2(X)$ and $f'=(f+\zeta_n)\circ \phi_n^{-1}$ with $\phi_n(X)=X$. We get 
$$\int_X(f'-f)^2\leq 2\int_X (f\circ\phi_n^{-1}-f)^2+2\int_X(\zeta_n\circ\phi_n^{-1})^2\to 0$$
and the result is proved.

A last point to be verified is that the distance $d_\FB$ is finite between any  pair of fshapes. Indeed, by definition there exists between any $X$ and $X'\in \X$ a path $t\mapsto\phi_t^v.X$ with $v\in L^2([0,1],V)$. Now, if $h_t=(f'\circ \phi_1^v-f)$, we get $(v,h)\in L^2([0,1],V\times L^2(X))$ and $f'=(f+\zeta^v_1)\circ(\phi^v_1)^{-1}$ so that $d_\FB((X,f),(X',f'))\leq E_X(v,h)^{1/2}<\infty$.
\end{proof}
As stated in the introduction, the framework extends metamorphosis on images \cite{Trouve1} which corresponds to the case where $X_0$ is $E$ itself or is the unit cube of $E$ (with the extra condition that $\phi.X_0=X_0$ for any $\phi\in G$). Here the support can be $d$-dimensional and is not fixed. 

\subsection{Existence of geodesics}
A natural question is the existence of a minimizing geodesic between fshapes in $\FB$. We have the following Theorem:
\begin{thm}\label{thm:geo}
  For any $(X,f)$, and $(X',f')$ in $\FB$, there exists $(v,h)\in L^2([0,1],V\times L^2(X))$ such that $E_X(v,h)^{1/2}=d_\FB((X,f),(X',f'))$.

In particular, if 
\begin{equation}
(X_t,f_t)\doteq(\phi^v_t.X,(f+\int_0^th_sds)\circ(\phi^v_t)^{-1})\label{eq:FB13}
\end{equation}
the path $t\mapsto (X_t,f_t)$ can be considered as a minimizing geodesic between $(X,f)$ and $(X',f')$. 
\end{thm}

\begin{proof}
  The proof extends the results for metamorphosis. It is sufficient to show that $(v,h)\to E_X(v,h)$ is lower semi-continuous for the weak convergence of the space $L^2([0,1],V\times L^2(X))$. 

Indeed, if this is the case, then from any minimizing sequence $(v_n,h_n)$ such that $E_X(v_n,h_n)\to d_\FB((X,f),(X',f'))^2$ we deduce that $v_n$ is bounded on $L^2([0,1],V)$ and using the inequality 

\begin{equation}
  \begin{split}
    \int_0^1\int_X|h_{n,t}(x)|^2d\Haus^d(x)dt &\leq \int_0^1(\sup_{x\in
    X}|\rst{d_x\phi^{v_n}_{t}}{T_{x}X}|^{-1})\int_X
    |h_{n,t}(x)|^2|\rst{d_x\phi^{v_n}_t}{T_{x}X}|d\Haus^d(x)\\ & \leq
    C(\int_0^1|v_{n,t}|_V^2dt)E_X(v_n,h_n)
  \end{split}
\label{eq:FB14}
\end{equation}
where $C$ is a increasing function depending only on $V$, we get that the sequence $(h_n)$ is bounded in $L^2([0,1],L^2(X))$. Hence, by weak compactness of strong balls in $L^2([0,1],V\times L^2(X))$, we can assume that, up to the extraction of a sub-sequence, that $(v_n,h_n)$ weakly converges towards $(v_\infty,h_\infty)\in L^2([0,1],V\times L^2(X))$ and by lower semi-continuity of $E_X$, we deduce that $E_X(v_\infty,h_\infty)\leq d_{\FB}((X,f),(X',f'))^2$. We only need to check that $(X',f')=(\phi^{v_\infty}_1.X,(f+\zeta^{h_\infty}_1)\circ(\phi^{v_\infty}_1)^{-1})$. This last result follows from the fact that, if $(v_n,h_n)$ weakly converges to $(v_\infty,h_\infty)$, then $\phi^{v_n}_1\to \phi^{v_\infty}_1$ uniformly on any compact sets (which is a well known result) and $(f+\zeta^{h_n}_1)\circ(\phi^{v_n}_1)^{-1}$ weakly converges to $(f+\zeta^{h_\infty}_1)\circ(\phi^{v_\infty}_1)^{-1}$. The last weak convergence is straightforward.

The proof of the lower semi-continuity itself is done now. We know from a classical result on weak convergence that $\int_0^1|v_{\infty,t}|_V^2dt\leq \liminf_{n\to\infty}\int_0^1|v_{n,t}|_V^2dt$. Moreover, we have
\begin{equation}
	\int_0^1\int_X|h_{\infty,t}(x)|^2|\rst{d_x\phi^{v_\infty}_{t}}{T_{x}X}|d\Haus^d(x)=\lim_{n\to\infty}\int_0^1 \int_X h_{n,t}(x)h_{\infty,t}(x)|\rst{d_x\phi^{v_\infty}_{t}}{T_{x}X}|d\Haus^d(x)\label{eq:FB17}
\end{equation}
by definition of the weak convergence of $(h_n)$. Now, since $d_x\phi^{v_n}_t\to d_x\phi^{v_\infty}_t$ uniformly on $t\in  [0,1]$ and $x\in X$, we get
\begin{equation}
  \label{eq:FB15}
  \begin{split}
	  \int_0^1&\int_X |h_{\infty,t}(x)|^2|\rst{d_x\phi^{v_\infty}_{t}}{T_{x}X}|d\Haus^d(x) \\
	  &=\lim_{n\to\infty}\int_0^1\int_Xh_{n,t}(x)h_{\infty,t}(x)(|\rst{d_x\phi^{v_n}_{t}}{T_{x}X}|)^{1/2}|\rst{d_x\phi^{v_\infty}_{t}}{T_{x}X}|^{1/2}d\Haus^d(x)\\
	  &\leq \liminf_{n\to\infty}\left(\int_0^1\int_X|h_{n,t}(x)|^2|\rst{d_x\phi^{v_n}_{t}}{T_{x}X}|d\Haus^d(x)\right)^{1/2}\left(\int_0^1\int_X|h_{\infty,t}(x)|^2|\rst{d_x\phi^{v_\infty}_{t}}{T_{x}X}|d\Haus^d(x)\right)^{1/2}
  \end{split}
\end{equation}
so that
\begin{equation}
  \label{eq:FB16}
  \int_0^1\int_X |h_{\infty,t}(x)|^2|\rst{d_x\phi^{v_\infty}_{t}}{T_{x}X}|d\Haus^d(x)\leq \liminf_{n\to\infty}\int_0^1\int_X|h_{n,t}(x)|^2|\rst{d_x\phi^{v_n}_{t}}{T_{x}X}|d\Haus^d(x)\,.
\end{equation}
\end{proof}

\subsection{Karcher means on fshape bundles}
\label{ssection:Karcher}
We consider here the existence problem of Karcher mean on the fshape bundle $\FB$. The problem can be stated as follows: let $(X^i,f^i)_{1\leq i\leq N}$ be a family of fshapes in $\FB$. Does there exist $(X_*,f_*)\in \FB$ minimizing the sum of the square distances to each $(X^i,f^i)$, {\it i.e.} solution of the minimizing problem:
\begin{equation}
  \label{eq:KM1}
  \min_{(X,f)\in \FB} \sum_{i=1}^N d_{\FB}((X,f),(X^i,f^i))^2\,.
\end{equation}
Note that this problem does not make sense on practical ground since usually there is no reason for a family of observed fshapes to belong to a pre-defined fshape bundle $\FB$. However, on a theoretical perspective the existence of such Karcher mean is a quite important point in addition to the existence of geodesic between any two fshapes in $\FB$. Such a result shows that $\FB$ has basic important properties for further statistical analysis.

Our main result is the following:
\begin{thm}\label{thm:Kar}
	For any family $(X^i,f^i)_{1\leq i\leq N}$ of fshapes in $\FB$, there exists at least a Karcher mean $(X_*,f_*)$ in $\FB$ {\it i.e.} a solution of problem (\ref{eq:KM1}).
\end{thm}
\begin{proof}
The first step of the proof is to reparametrize problem (\ref{eq:KM1}) in the Hilbert space $$W\doteq \prod_{i=1}^NL^2([0,1],V\times L^2(X^i))\,.$$
Indeed, using the symmetry of the problem, we can exchange the role of $(X,f)$ and $(X^i,f^i)$ and look for a family of minimizing paths $t\mapsto (X^i_t,f^i_t)$, starting from every fshape $(X^i,f^i)$ and ending on a \emph{common} fshape $(X,f)\in\FB$, each path being parametrized by $(v^i,h^i)\in L^2([0,1],V\times L^2(X^i))$. Hence, introducing 
$$W_0\doteq\{ \alpha=(v^i,h^i)_{1\leq i\leq N}\in W\ |\ (X^i_1,f^i_1)=(X^1_1,f^1_1),\ \forall 1\leq i\leq N\ \}$$
(where $(X^i_1,f^i_1)=(\phi^{v^i}_1.X^i,(f^i+\zeta^{h^i}_1)\circ (\phi^{v^i}_1)^{-1}$)
we need to prove that if we define
\begin{equation}
  \label{eq:KM2}
  J(\alpha) \doteq \sum_{i=1}^N E_{X^i}(v^i,h^i)
\end{equation}
then $W_0$ is closed and  $J$ is lower semi continuous (l.s.c.) for the weak topology on $W$.  

Since the fact that $J$ is l.s.c. was established inside the proof of Theorem \ref{thm:geo}, we just need to check that $W_0$ is closed for the weak topology. So, let $(\alpha_n)_{n\geq 0}$ be a sequence in $W_0$ weakly converging to $\alpha_\infty\doteq(v^i_\infty,h^i_\infty)_{1\leq i\leq N}\in W$. We get that $\phi^{v^i_n}_1\to \phi^{v^i_\infty}_1$ and $(\phi^{v^i_n}_1)^{-1}\to (\phi^{v^i_\infty}_1)^{-1}$ for the uniform $C^1$ topology on compact sets so that in particular, if $X_*=\phi^{v^1_\infty}_1.X^1$ then since by construction $\phi^{v^i_n}_1.X^i=\phi^{v^1_n}_1.X^1$ we get that $d_H(\phi^{v^i_n}_1.X^i,X_*)\to 0$ where $d_H$ denotes the Hausdorff distance\footnote{We are using here that for $X$ compact, the mapping $v\mapsto \phi^v_1.X$ is continuous for the weak convergence on $v\in L^2([0,1],V)$ and the convergence for the Hausdorff metric on the set of all compact subsets of $E$} and
\begin{equation}
\phi^{v_\infty^i}_1.X^i=X_*\,.\label{eq:KM4}
\end{equation}
Now, for any $1\leq i\leq N$, if we denote $\tilde{f}^i_{1,n}\doteq (f^i+\zeta^{h^i_n}_1)$ and $\psi^i_n\doteq (\phi^{v_n^i}_1)^{-1}\circ\phi^{v_n^1}_1:X^1\to X^i$, we have (since $\alpha_n\in W_0$ and $\alpha_n\rightharpoonup \alpha_\infty$) 
\begin{equation}
  \label{eq:KM3}
  \tilde{f}^i_{1,n}\circ\psi^i_n=\tilde{f}^1_{1,n}\text{ and }\tilde{f}^1_{1,n}\rightharpoonup \tilde{f}^1_{1,\infty}
\end{equation}
where the last statement is straightforward. Let us check that $ \tilde{f}^i_{1,n}\circ\psi^i_n\rightharpoonup \tilde{f}^i_{1,\infty}\circ\psi^i_\infty$ so that we will get $\tilde{f}^i_{1,\infty}\circ\psi^i_\infty=\tilde{f}^1_{1,\infty}$ and since we have (\ref{eq:KM4}) we will get $\alpha_\infty\in W_0$ and $W_0$ weakly closed.

For that, let $g:X^1\to\R$ be a Lipschitz mapping on $X^1$ and denote now $\langle f,f'\rangle_{X^1}\doteq \int_{X^1}f(x)f'(x)d\Haus^d(x)$ the usual dot product on $L^2(X^1)$. Since $X^1$ is compact and $\Haus^d(X^1)<\infty$, we have that $\rst{\Haus^d}{X^1}$ is a Radon measure on the compact metric space $X^1$ so that such Lipschitz mapping $g$ are dense in $L^2(X^1)$ and we just need to check that $\langle \tilde{f}^i_{1,n}\circ \psi^{i}_n-\tilde{f}^i_{1,\infty}\circ\psi^{i}_\infty,g\rangle_{X^1}\to 0$ to prove the weak convergence. We have
\begin{equation}
  \label{eq:KM5}
  \begin{split}
    \langle \tilde{f}^i_{1,n}&\circ 
    \psi^{i}_n-\tilde{f}^i_{1,\infty}\circ\psi^{i}_\infty,g\rangle_{X^1}=\\
 &\underbrace{\langle \tilde{f}^i_{1,n}\circ 
    \psi^{i}_n-\tilde{f}^i_{1,n}\circ\psi^{i}_\infty,g\rangle_{X^1}}_{A_n}+
 \underbrace{\langle \tilde{f}^i_{1,n}\circ 
    \psi^{i}_\infty-\tilde{f}^i_{1,\infty}\circ\psi^{i}_\infty,g\rangle_{X^1}}_{B_n}\,.
  \end{split}
\end{equation}
Concerning the $B$ term, by change of variable and using the fact that the Jacobian $d(\psi^i_\infty)^{-1}$ is bounded on $X^1$ we get
\begin{equation}
  \label{eq:KM6}
  B_n=\langle \tilde{f}^i_{1,n}-\tilde{f}^i_{1,\infty},g\circ(\psi^i_{\infty})^{-1}|d(\psi^i_\infty)^{-1}.\xi|\rangle_{X^i}\to 0\,,
\end{equation}
where $\xi$ is the unit $d$-vector representing the tangent space defined $\rst{\Haus^d}{X^i}$-a.e. on $X^i$. Concerning the $A$ term, we have again by change of variable
\begin{equation}
  \label{eq:KM7}
  \begin{split}
    A_n =\langle
    \tilde{f}^i_{1,n},g\circ(\psi^i_{n})^{-1}(|d(\psi^i_n)^{-1}.\xi|&-|d(\psi^i_\infty)^{-1}.\xi|)\rangle_{X^i}\\
&+\langle
    \tilde{f}^i_{1,n},(g\circ(\psi^i_{n})^{-1}-g\circ(\psi^i_{\infty})^{-1})|d(\psi^i_\infty)^{-1}.\xi|\rangle_{X^i}\,.
  \end{split}
\end{equation}
The first term of (\ref{eq:KM7}) is bounded by
\begin{equation}
  \label{eq:KM8}
  |\tilde{f}^i_{1,n}|_{L^2(X^i)}|g|_{L^\infty(X^1)}\left | |d(\psi^i_n)^{-1}.\xi|-|d(\psi^i_\infty)^{-1}.\xi| \right |_{L^\infty(X^i)}\to 0
\end{equation}
since the weak convergence of $\tilde{f}^i_{1,n}$ implies that $|\tilde{f}^i_{1,n}|_{L^2(X^i)}$ stays bounded. The second term of (\ref{eq:KM7}) is bounded by
\begin{equation}
  \label{eq:KM9}
   k_g|\tilde{f}^i_{1,n}|_{L^2(X^i)} \left | (\psi^i_n)^{-1}-(\psi^i_\infty)^{-1} \right |_{L^\infty(X^i)}|d(\psi^i_\infty)^{-1}|_{L^\infty(X^i)}\to 0
\end{equation}
where $k_g$ is the Lipschitz constant of $g$ on $X^i$.
\end{proof}

\subsection{The tangential model}
\label{ssection:tang_model}
The previous metamorphosis metrics can be also approximated in a simpler setting, which becomes closer to a natural extension of the classical image framework in the context of fshapes. We shall call it the \emph{tangential model}. Instead of computing the cost of the functional evolution $t\mapsto \zeta_t$ along a time dependent $L^2(X_t)$ metric as
\begin{equation}
  \label{eq:FB25}
  \int_0^1\int_X|\dot{\zeta}_t(x)|^2|\rst{d_x\phi^v_t}{T_{x}X}|d\Haus^d(x)dt=\int_0^1\int_{X_t}|\dot{\zeta}_t\circ(\phi_t^v)^{-1}|^2(y)d\Haus^d(y)dt
\end{equation}
one can freeze the metric to its initial value $X_0=X$ neglecting the change of measure weight:
\begin{equation}
  \label{eq:FB26}
  \int_0^1\int_X|\dot{\zeta}_t(x)|^2|\rst{d_x\phi^v_t}{T_{x}X}|d\Haus^d(x)dt\approx \int_0^1\int_{X}|\dot{\zeta}_t(x)|^2d\Haus^d(x)dt
\end{equation}
which gives after optimization with fixed end points condition $\zeta_1$ the usual $L^2(X)$ cost
\begin{equation}
  \label{eq:FB27}
  |\zeta_1|_X^2\doteq \int_X|\zeta_1(x)|^2d\Haus^d(x)
\end{equation}
which can be interpreted naturally as the log-likelihood of a Gaussian noise in the statistical framework. 

Note as we say in introduction of this section, that this is the usual noise term that appears in the classical situation of image matching and this can be considered as some kind of tangential metric for the functional part. Concerning the geometrical part, we can keep the usual $L^2([0,1],V)$ penalization to define
\begin{equation}
  \label{eq:FB28}
  \tilde{E}_X(v,\zeta)\doteq\frac{\gamma_V}{2}\int_0^1|v_t|_V^2dt+\frac{\gamma_f}{2}|\zeta|^2_X
\end{equation}
and consider 
\begin{equation}
\tilde{d}_\FB((X,f),(X',f'))\doteq (\inf\{\ \tilde{E}_X(v,\zeta)\ |\  X'=\phi^v_1.X,\ f'=(f+\zeta)\circ(\phi^v_1)^{-1}\ \})^{1/2}\,.\label{eq:FB29}
\end{equation}
With that definition, $\tilde{d}_{\FB}$ is no more a distance satisfying the symmetry and triangle inequality as previously, but the minimization of $\tilde{E}_{X}$ can be seen as the minimization of the log-likelihood of $(v,\zeta)$ given $(X,f)$ and $(X',f')$ for a quite natural statistical model. 

In this framework, one can still obviously consider the problem of existence for two given fshapes $(X,f)$ and $(X',f')$ in $\FB$ of an optimal $(v_*,\zeta_*)$ such that 
\begin{equation}
  \label{eq:FB37}
\left\{
  \begin{array}[h]{lcl}
  (X',f') & = &(\phi^{v_*}_1.X,(f+\zeta_*)\circ(\phi^{v_*}_1)^{-1})\\
\text{and} &&\\
  \tilde{E}_X(v_*;\zeta_*) & = &\min\{ \tilde{E}_X(v,\zeta)\ |\ (v,\zeta)\in F_{(X',f')}\}\\
\text{where}&&\\
F_{(X',f')} & \doteq  & \{ (v,\zeta) \in L^2([0,1],V)\times L^2(X)\ |\ (X',f') = (\phi^{v}_1.X,(f+\zeta)\circ(\phi^{v}_1)^{-1})\}
\end{array}\right.
\end{equation}

The proof existence of $(v_*,\zeta_*)$  can be done along the very same lines that the proof of existence of geodesics and is omitted. The proof is even more simpler since the functional metric is frozen and does not change along the path. We get eventually the following result:
\begin{thm}
	For $(X,f)$ and $(X',f')$ in $\FB$, there exists an optimal solution of the exact matching problem between $(X,f)$ and $(X',f')$ {\it i.e.} there exists $(v_*,\zeta_*)$ in $L^2([0,1],V)\times L^2(X)$ such that (\ref{eq:FB37}) holds. 
\end{thm}

 \section{Dissimilarity measure between fshapes}
 \subsection{Dissimilarity measures : state of the art}
 The previous framework has been focusing so far on the comparison of fshapes within one given common bundle $\mathcal{F}$. It is clear that it becomes inoperant in most practical situations where the elements of a dataset cannot be assumed to belong to the same bundle. The reason is that the geometrical supports of two subjects need not be obtained from one the other by a deformation belonging to the group $G$. In more usual cases of shape spaces such as sets of images, landmarks, curves or surfaces, the general approach has been to introduce {dissimilarity measures} between the subjects, which provides data attachment terms to perform inexact registration and atlas estimation in practice. 
 
 For images and landmarks, $L^{2}$ distances are natural and have been commonly used to compare such objects. The case of curves, surfaces and submanifolds is theoretically more involved and has drawn consistent attention in several fields of mathematics. In computational anatomy, one possible and very powerful setting was proposed through the adaptation of the concept of {currents} \cite{Glaunes,Durrleman}. The dissimilarity measures between oriented submanifolds of given dimension defined in these frameworks derive from Hilbert metrics on spaces of differential forms and combine the advantages of not relying on parametrizations or point to point matchings, of being easily computable for discrete shapes and robust to shape sampling. One potential drawback is the issue of consistent orientation that is required in the currents' representations. More recently, an alternative methodology based on {varifolds} was introduced and implemented in \cite{Charon2}, which has the interest of being close to the idea of currents while overcoming the problem of orientation. 
 
 Now, fshapes do pose additional difficulties with respect to the definition of dissimilarity measures. This problem has been only addressed very recently in \cite{Charon1} where the authors define the extended notion of {functional currents} that, similarly to usual currents, embeds fshapes in some dual of spaces of differential forms. The Hilbert metrics that are defined on functional currents can be then used again as dissimilarity measures to perform inexact registration between fshapes, as exposed in the article. Functional currents could be thus very well used in the generalized problem of atlas estimation that we focus on in this paper. This has been touched upon, essentially from the numerical point of view, in \cite{Charon_thesis} (chapter 4). Although the results and algorithms presented in this paper could be transposed quite easily to fshape dissimilarities provided by functional currents, we shall work instead with the slightly modified approach of varifolds. As aforementioned, this is a way to get rid of shape orientation which is usually quite desirable in applied situations. In the following section, we briefly present the {functional varifolds}' mathematical setting, that stands for a natural generalization of varifolds to fshapes. The more technical and discrete computations shall be detailed in Section \ref{part:discreteScheme} dedicated to the algorithmic part itself.

 \subsection{Functional varifolds}
 \label{sec:fvarifolds}
 The mathematical concept of varifold goes back to geometric measure theory and the original exposition of F. Almgren \cite{Almgren} which was further developed considerably by W. Allard in \cite{Allard} subsequently. The link to computational anatomy's problems is fairly more recent and is in most part presented in \cite{Charon2} or \cite{Charon_thesis}. We shall frequently refer the reader to these references for additional mathematical details. What we present in this section is an extension of varifolds to represent functional shapes, which is much related to the functional current idea of \cite{Charon1}.

\subsubsection{Representation of fshapes as functional varifolds} 
As previously, we call $E$ the $n$-dimensional vector space embedding all geometrical supports of fshapes. For any integer $1\leq d \leq n$, we will write $G_{d}(E)$ for the \textit{Grassmann manifold} of all $d$-dimensional (non-oriented) subspaces of $E$. $G_{d}(E)$ is a compact manifold that can be embedded trivially in the space $\mathcal{L}(E)$ by identifying any $V \in G_{d}(E)$ with the orthogonal projector $p_{V}$ on $V$. The usual space of $d$-dimensional varifolds is defined intuitively as the space of spatially spread distributions of Grassmannians or, more rigorously, as finite Borel measures on the product $E \times G_{d}(E)$ (\textit{cf} \cite{Charon2}). Now, to account for the existence of signals at each point, one can simply augment the varifold representation with a signal component, which we formalize by the following definition:
\begin{de}
 We say that $\mu$ is a $d$-dimensional functional varifold (fvarifold in short) if $\mu$ is a Borel finite measure on the space $E \times G_{d}(E) \times \mathbb{R}$ or equivalently if $\mu \in C_{0}(E \times G_{d}(E) \times \mathbb{R})'$. 
\end{de}
Note that we consider only the case of real-valued signals here, but this could be extended to different situations, as vector or tensor-valued signals (\textit{cf} \cite{Charon_thesis}). Now, any functional shape $(X,f)$, with $X$ a $d$-dimensional rectifiable subset, can be represented by a functional varifold $\mu_{(X,f)}$ defined by:
\begin{equation}
\label{eq:fshape_fvar}
 \mu_{(X,f)}(\omega)=\int_{X} \omega(x,T_{x}X,f(x)) d \mathcal{H}^{d}(x)
\end{equation}
where $\mathcal{H}^{d}$ is the $d$-dimensional Hausdorff measure on $E$. 

Some particularly simple functional varifolds are the Diracs that, in this context, can be written in the form $\delta_{(x,V,f)}$ with $x \in E$, $V \in G_{d}(E)$ and $f \in \mathbb{R}$, and act on any function $\omega \in C_{0}(E \times G_{d}(E) \times \mathbb{R})$ by the relation:
\begin{equation} \label{eq:fshape_dirac}
 \delta_{(x,V,f)}(\omega)=\omega(x,V,f) \,.
\end{equation}
We shall detail, in Section \ref{part:discreteScheme}, how to approximate polyhedral fshapes by a finite sums of Dirac in order to perform practical computations. 

Now, an important point is to express the way that geometrico-functional transformations act on functional varifolds, in such a way that this action is consistent with the transport of fshapes that we have been considering in section \ref{ssec:metamorphosis_bundles}. In a very similar fashion as with regular varifolds (\textit{cf} \cite{Charon2}), one can express it by usual pull-back and push-forward operations. Let $\phi \in \text{Diff}(E)$ and $\zeta$ a measurable function on $E$, this is given by the following set of equations: 
\begin{equation}
\label{eq:push_pull_var}
\left \{
\begin{array}[h]{l}
 \forall \omega \in C_{0}(E\times G_{d}(E)\times \mathbb{R}), \ \left ( (\phi,\zeta)_{\ast} \mu \right )(\omega)= \mu \left ( (\phi,\zeta)^{\ast}\omega \right ) \\
 \text{where} \\
 \left ( (\phi,\zeta)^{\ast}\omega \right )(x,V,f)=|\rst{d_{x}\phi}{V}|\,\omega(\phi(x),d_{x}\phi(V),f+\zeta) 
\end{array} 
 \right.
\end{equation}
where for $V \in G_{d}(E)$, $|\rst{d_{x}\phi}{V}|$ denotes as previously the Jacobian of $\phi$ along subspace $V$ ({\it i.e.} the volume change along $V$ at point $x$) and $d_{x}\phi(V)$ is the image of $V$ by the invertible linear application $d_{x}\phi$. It is then a simple verification that one has the following consistency property:
\begin{prop}
\label{prop:commutation_fvar}
Let $(X,f)$ be a fshape where $X$ is a $d$-dimensional rectifiable subset and $f$ a $L^{2}$ function on $X$. If $\phi \in \text{Diff}(E)$ and $\zeta\in L^{2}(X)$ then:
\begin{equation*}
 (\phi,\zeta)_{\ast} \mu_{(X,f)} = \mu_{(\phi(X),(f+\zeta)\circ\phi^{-1})}\,.
\end{equation*}
\end{prop}
The proof does not involve any additional difficulty than in the case of usual varifolds, for which we refer the reader to \cite{Charon2}.

Thus, functional varifolds enable all possible fshapes of given dimension to be embedded into a common space of distributions. The next step then is to equip such a space with a metric that shall induce a dissimilarity measure on the set of fshapes. This can be efficiently addressed by introducing reproducing kernels on the product $E \times G_{d}(E) \times \mathbb{R}$, as we explain in the following text.

\subsubsection{RKHS of functional varifolds} 
The use of reproducing kernel in the context of currents and varifold spaces has been argued and implemented many times, for instance in \cite{Glaunes,Durrleman,Charon_thesis}. The general scheme is to consider Hilbert spaces embedded in $C_{0}(E \times G_{d}(E) \times \mathbb{R})$ which are built as the reproducing kernel Hilbert space (RKHS) associated to a certain positive kernel on $E \times G_{d}(E) \times \mathbb{R}$. A natural and convenient (but not exhaustive) way to define kernels on product spaces is to consider tensor products of kernels, the classical result from kernel theory being that:
\begin{lemma}
 Let $A$ and $B$ be two sets and $k_{A}$, $k_{B}$ positive kernels respectively on $A$ and $B$. Then $k_{A} \otimes k_{B}$  defined for all $a_{1},a_{2} \in A$ and $b_{1},b_{2} \in B$ by:
 $$ k_{A} \otimes k_{B}((a_{1},b_{1}),(a_{2},b_{2})) \doteq k_{A}(a_{1},a_{2})k_{B}(b_{1},b_{2}) $$
 is a positive kernel on $A \times B$.
 \label{lemma:tensorprod_kernel}
\end{lemma}
Now going back to the case of functional varifolds itself, one has the following property:
\begin{prop}
  Let $k_{e}$ be a positive kernel on the space $E$ such that $k_{e}$ is continuous, bounded and for all $x \in E$, the function $k_{e}(x,.)$ vanishes at infinity, $k_{t}$ a kernel on the manifold $G_{d}(E)$ that is also continuous, and $k_{f}$ a kernel on $\mathbb{R}$ continuous, bounded and such that $k_{f}(f,.)$ vanishes at infinity for all $f \in \mathbb{R}$. Then $W$, the RKHS associated to the positive kernel $k \doteq k_{e} \otimes k_{t} \otimes k_{f}$, is continuously embedded into the space $C_{0}(E \times G_{d}(E) \times \mathbb{R})$. 
 \label{prop:kernel_fvar}
\end{prop}
\begin{proof}
The proof is very similar to the one of \cite{Charon2} in the case of varifolds. By lemma \ref{lemma:tensorprod_kernel}, $k$ is indeed a positive kernel on $E \times G_{d}(E) \times \mathbb{R}$ and by definition:
  \begin{equation}
\label{eq_def_newkernel}
  k \left ((x,V,f),(\tilde{x},\tilde{V},\tilde{f}) \right ) = k_{e}(x,\tilde{x})\,k_{t}(V,\tilde{V})\, k_{f}(f,\tilde{f})
 \end{equation} 
 and so, thanks to the assumptions on the kernels, $k((x,V,f),.)$ is continuous on $E \times G_{d}(E) \times \mathbb{R}$ and belongs to $C_{0}(E\times G_{d}(E)\times \mathbb{R})$. The vector space $W_{0}$ generated by these functions is thus included in $C_{0}(E\times G_{d}(E) \times \mathbb{R})$. Moreover, if $\omega \in W_{0}$, by the reproducing kernel property, we have that:
 \begin{equation*}
  \omega(x,V,f)=\delta_{(x,V,f)}(\omega)=\langle k((x,V,f),.), \omega \rangle_{W} \,.
 \end{equation*}
With Cauchy-Schwarz inequality: $|\omega(x,V,f)| \leq \|k((x,V,f),.) \|_{W}.\|\omega\|_{W}$. In addition, 
$$\|k((x,V,f),.) \|_{W}=\sqrt{k((x,V,f),(x,V,f))}$$ 
and all three kernels $k_{e}$, $k_{t}$ and $k_{f}$ are bounded so that $k$ is also bounded. We conclude that $|\omega|_{\infty} \leq \sqrt{|k|_{\infty}}.\|\omega\|_{W}$. Thus Cauchy sequences in $W_{0}$ for the $W$-norm are also Cauchy sequences for the infinite norm. It results that all their limits belong to $C_{0}(E\times G_{d}(E) \times \mathbb{R})$ and therefore $W$ is included in $C_{0}(E\times G_{d}(E) \times \mathbb{R})$. The previous inequality then holds for all $\omega \in W$, which shows that the inclusion embedding $\imath: \ W\hookrightarrow C_{0}(E\times G_{d}(E))$ is indeed continuous.
\end{proof}
Consequently, there exists a continuous mapping $i^*$ of the space of fvarifolds $C_{0}(E\times G_{d}(E)\times \mathbb{R})'$ into the dual of $W$. This induces a \textit{pseudo-distance} on fvarifolds resulting from the Hilbert structure of the RKHS. If we introduce the isometry $K_{W}: \ W' \rightarrow W$ defined by $\langle K_{W} \mu, \omega \rangle_{W} = \mu(\omega)$ for all $\mu \in W'$ and $ \omega \in W$, by the reproducing kernel property, we know that $K_{W} \delta_{(x,V,f)} = k((x,V,f),.)$. Then for all $x_{1},x_{2} \in E$, $V_{1},V_{2} \in G_{d}(E)$ and $f_{1}, f_{2} \in \mathbb{R}$, 
\begin{align*}
 \langle \delta_{(x_{1},V_{1},f_{1})}, \delta_{(x_{2},V_{2},f_{2})} \rangle_{W'} &= \langle K_{W}\delta_{(x_{1},V_{1},f_{1})}, K_{W}\delta_{(x_{2},V_{2},f_{2})} \rangle_{W} \\
 &= K_{W}\delta_{(x_{2},V_{2},f_{2})}(x_{1},V_{1},f_{1})
\end{align*}
and thus we have the following expression for the inner product between two Diracs:
 \begin{equation}
\label{eq_dot_product_diracs}
  \langle \delta_{(x_{1},V_{1},f_{1})}, \delta_{(x_{2},V_{2},f_{2})} \rangle_{W'}=k_{e}(x_{1},x_{2})\,k_{t}(V_{1},V_{2})\, k_{f}(f_{1},f_{2}) \,.
 \end{equation}
Now, if $(X,f)$ and $(Y,g)$ are two fshapes (of dimension $d$) and $\mu_{(X,f)}$ and $\mu_{(Y,g)}$ the associated fvarifolds (as defined previously), one can prove easily the following formula:
 \begin{equation}
\label{eq:dot_product_rect_var}
  \langle \mu_{(X,f)}, \mu_{(Y,g)} \rangle_{W'}=\int_{X} \int_{Y} k_{e}(x,y)\,k_{t}(T_{x}X,T_{y}Y) \, k_{f}(f(x),g(y)) d\mathcal{H}^{d}(x) d \mathcal{H}^{d}(y) \,.
 \end{equation}
 One can thus propose a dissimilarity measure between any two fshapes that is simply the norm of the difference in the space of fvarifolds {\it i.e.} $\|\mu_{(X,f)}-\mu_{(Y,g)} \|_{W'}^{2}$, which can be computed easily using \eqref{eq:dot_product_rect_var}. We shall detail more specifically the discrete expressions in Section \ref{part:discreteScheme}.

Yet, as we mentioned earlier, this does not necessarily result from a real distance on the space of varifolds because the dual application $\imath^{\ast}$ does not need to be an embedding. This is actually the case if and only if the RKHS $W$ is dense in $C_{0}(E \times G_{d}(E) \times \mathbb{R})$, in which case the kernel $k$ is said to be \textbf{$C_{0}$-universal}. In our previous construction, this holds in particular if all kernels $k_{e}, k_{t}$ and $k_{f}$ are themselves $C_{0}$-universal. In \cite{Carmeli}, authors study thoroughly construction of $C_{0}$-universal kernels on vector spaces. Notably, it is proven that all Gaussian kernels satisfy this property. Thus, such kernels are easy to provide in the case of kernels $k_{e}$ and $k_{f}$ defined on the spaces $E$ and $\mathbb{R}$. As for kernel $k_{t}$ on the Grassmann manifold, kernels can be defined by using the embedding of $G_{d}(E)$ in $\mathcal{L}(E)$ (we refer to \cite{Charon2} for more details), and therefore one can obtain a similar notion of Gaussian kernels that can be also shown to verify the $C_{0}$-universality property.

To conclude this section, let us mention a generalization of proposition \ref{prop:kernel_fvar} to higher-order regularities that shall be useful in the rest of the paper. 
\begin{prop}
  Let $k$ be a positive kernel on the product space $E \times G_{d}(E) \times \mathbb{R}$ such that $k$ is continuously differentiable of order $2p$ and such that all the derivatives of $k$ up to the order $p$ are bounded. Assume in addition that for any $(x,V,f) \in E \times G_{d}(E) \times \mathbb{R}$, the function $k((x,V,f),.)$ and all its derivatives up to order $p$ vanish at infinity. Then, the RKHS associated to $k$ is continuously embedded into $C_{0}^{p}(E\times G_{d}(E) \times \mathbb{R})$.
 \label{prop:kernel_fvar_reg}
\end{prop}
The proof follows the same pattern as the one of proposition \ref{prop:kernel_fvar} and can be also adapted from the one given in \cite{Glaunes} (chapter 2). 

\subsection{Properties of the metrics}
\subsubsection{Control results}
We now show a few control results on fvarifolds norm that shall be useful for the following. As a measure on $E \times G_{d}(E) \times \R$, we will call the \textit{total variation norm} of a given fvarifold $\mu$ the quantity $\mu(E \times G_{d}(E) \times \R)$. In particular, one can easily check that the total variation norm of a rectifiable fvarifold $\mu_{(X,f)}$ equals $\Haus^{d}(X)$. We have the following control of $W'$-norms:
\begin{prop}
\label{prop:control_RKHSnorm_totalvar}
 For a RKHS $W$ continuously embedded into $C_{0}(E \times G_{d}(E) \times \R)$, there exists a constant $c_{W}>0$ such that for all $\mu \in C_{0}(E \times G_{d}(E) \times \R)'$:
 $$ \|\mu\|_{W'} \leq c_{W}.\mu(E \times G_{d}(E) \times \R) $$
\end{prop}
\begin{proof}
 From the continuous embedding property, we have the existence of $c_{W}>0$ such that for all $\omega \in W'$, $|\omega|_{\infty} \leq c_{W} \|\omega\|_{W}$. Now, by definition of the kernel metric $\|\mu\|_{W'}^{2} = \mu(K_{W}\mu)$ for $K_{W}$ the Riesz isometry between $W'$and $W$ and thus:
 \begin{align*}
  \|\mu\|_{W'}^{2} = \mu(K_{W}\mu) &= \int K_{W}\mu(x,V,f) d\mu(x,V,F) \leq |K_{W}\mu|_{\infty} \mu(E \times G_{d}(E) \times \R) \\
  &\leq c_{W} \|K_{W}\mu\|_{W}\mu(E \times G_{d}(E) \times \R) \\
  &\leq c_{W} \|\mu\|_{W'} \mu(E \times G_{d}(E) \times \R)
 \end{align*}
 so that dividing on both sides by $\|\mu\|_{W'}$, we obtain the result.
\end{proof}
\noindent A direct corollary is that the RKHS norm of a rectifiable fvarifold is controlled by the $d$-volume of its geometrical support.

A second important property to establish is the continuity of RKHS norms with respect to deformations of geometrical supports. We first consider general $C^{1}$-diffeomorphisms of $E$ such that $\phi$ and $d\phi$ tends to $\text{Id}$ at infinity and we will denote $\|\phi-\text{Id}\|_{1,\infty}=|\phi-\text{Id}|_{\infty}+|d\phi - \text{Id}|_{\infty}$. The result we show is the following:
\begin{prop}
 \label{chap_fcurrent:theo_controlnormdefor}
Let $X$ be a $d$-dimensional rectifiable subset of $E$ of finite volume and $f \in L^{2}(X)$. Assume that $W$ is constructed as in proposition \ref{prop:kernel_fvar} and is continuously embedded into $C_{0}^{1}(E \times G_{d}(E) \times \R)$. Then: 
\begin{equation*}
 \|\mu_{(X,f)} - \mu_{(\phi(X),f \circ \phi^{-1})}\|_{W'} \rightarrow 0
\end{equation*}
as $\|\phi-\text{Id}\|_{1,\infty}\rightarrow 0$.
\end{prop}
\begin{proof}
 We start by writing:
 $$ \mu_{(X,f)} = \int_{X} \delta_{(x,T_{x}X,f(x))} d\Haus^{d}(x) $$
 and, in the same way, 
 $$ \mu_{(\phi(X),f \circ \phi^{-1})} = \int_{\phi(X)} \delta_{(y,T_{y}\phi(X),f\circ \phi^{-1}(y))} d\Haus^{d}(x) \,.$$
 Now, applying the area formula (corollary 3.2.20 in \cite{Federer}), we have:
 $$ \mu_{(\phi(X),f \circ \phi^{-1})} = \int_{X} \delta_{(\phi(x),d_{x}\phi(T_{x}X),f(x))} |\rst{d_{x} \phi}{T_{x}X}| d\Haus^{d}(x) \,.$$
 where $|\rst{d_{x}\phi}{T_{x}X}|$ is the local volume change at $x$ along $T_{x}X$. We obtain therefore: 
 \begin{equation}
  \label{eq:control1}
  \|\mu_{(X,f)} - \mu_{(\phi(X),f \circ \phi^{-1})}\|_{W'} \leq \int_{X} \left \| |\rst{d_{x} \phi}{T_{x}X}|.\delta_{(\phi(x),d_{x}\phi(T_{x}X),f(x))} - \delta_{(x,T_{x}X,f(x))} \right \|_{W'} d\Haus^{d}(x) \,.
 \end{equation}
 Now, we assumed that the kernel $k$ of $W$ is built as a tensor product $k= k_{e}\otimes k_{t} \otimes k_{f}$. We can then write $W=W_{g} \otimes W_{f}$ where $W_{g}$ is the RKHS of kernel $k_{e} \otimes k_{t}$ and $W_{f}$ the one of kernel $k_{f}$. It is then straightforward to check that:
 \begin{align*}
	 & \left \| |\rst{d_{x} \phi}{T_{x}X}|.\delta_{(\phi(x),d_{x}\phi(T_{x}X),f(x))} - \delta_{(x,T_{x}X,f(x))} \right \|_{W'} \\
	 &= \|\delta_{f(x)}\|_{W_{f}'}. \| |\rst{d_{x} \phi}{T_{x}X}|.\delta_{(\phi(x),d_{x}\phi(T_{x}X))} - \delta_{(x,T_{x}X)} \|_{W_{g}'} 
 \end{align*}
Since $\|\delta_{f(x)}\|_{W_{f}'} = \sqrt{k_{f}(f(x),f(x))}$ and $k_{f}$ is bounded, we have $\|\delta_{f(x)}\|_{W_{f}'} \leq \text{Cte}$. Moreover:
\begin{align}
	&\| |\rst{d_{x} \phi}{T_{x}X}|.\delta_{(\phi(x),d_{x}\phi(T_{x}X))} - \delta_{(x,T_{x}X)} \|_{W_{g}'} \nonumber\\ 
	&\leq \left | |\rst{d_{x} \phi}{T_{x}X}|-1 \right |. \|\delta_{(\phi(x),d_{x}\phi(T_{x}X))}\|_{W_{g}'} + \|\delta_{(\phi(x),d_{x}\phi(T_{x}X))}-\delta_{(x,T_{x}X)}\|_{W_{g}'} \,.
 \label{eq:control2}
\end{align}
Focusing on the first term, we have again, since $k_{e}$ and $k_{t}$ are bounded, that $\|\delta_{(\phi(x),d_{x}\phi(T_{x}X))}\|_{W_{g}'}$ is uniformly bounded on $X$. In addition, the volume variation $|\rst{d_{x} \phi}{T_{x}X}|-1$ converges to $0$ uniformly on $X$ whenever $\|\phi -\text{Id}\|_{1,\infty}\rightarrow 0$. As for the second term in the sum of \eqref{eq:control2}, we know that $\|\delta_{(\phi(x),d_{x}\phi(T_{x}X))}-\delta_{(x,T_{x}X)}\|_{W_{g}'} = \sup_{\|\omega\|_{W_{g}}=1} \left | (\delta_{(\phi(x),d_{x}\phi(T_{x}X))}-\delta_{(x,T_{x}X)} )(\omega) \right |$ and for any $\omega \in W_{g}$, we have:
\begin{align*}
 \left | (\delta_{(\phi(x),d_{x}\phi(T_{x}X))}-\delta_{(x,T_{x}X)} )(\omega) \right | &= \left | \omega(\phi(x),d_{x}\phi(T_{x}X))-\omega(x,T_{x}X) \right | \\
 &\leq |d\omega|_{\infty}.d_{E \times G_{d}(E)}\left ( (\phi(x),d_{x}\phi(T_{x}X)),(x,T_{x}X) \right ) 
\end{align*}
and $d_{E \times G_{d}(E)}\left ( (\phi(x),d_{x}\phi(T_{x}X)),(x,T_{x}X) \right )$ is obviously uniformly upper bounded for $x \in X$ by $\text{Cte}.\|\phi-\text{Id}\|_{1,\infty}$. In addition, with the assumptions on the kernels, we also have that $|d\omega|_{\infty} \leq \text{Cte} \|\omega\|_{W_{g}}$. It results eventually that $\|\delta_{(\phi(x),d_{x}\phi(T_{x}X))}-\delta_{(x,T_{x}X)}\|_{W_{g}'}$ is bounded above by $\text{Cte}.\|\phi-\text{Id}\|_{1,\infty}$. Now, since both terms in \eqref{eq:control2} converge to $0$ uniformly for $x \in X$, it results the same behavior for the integral of \eqref{eq:control1}, which completes the proof.
\end{proof}
Now, going back to the previous model of deformations, we consider diffeomorphisms $\phi=\phi_{1}^{v}$ obtained as the flow at time $1$ of a time varying vector fields $v\in L^{2}([0,1],V)$, where $V$ is a given Hilbert space of vector fields on $E$. The result of proposition \ref{chap_fcurrent:theo_controlnormdefor} implies the following important corollary:
\begin{corollary}
\label{cor:weak_continuity_W'}
 If $V$ is continuously embedded into $C_{0}^{2}(E,E)$ then, for any fixed and bounded rectifiable set $X$ of finite volume and $L^{2}$ signal $f$ on $X$, the application $v\mapsto \mu_{(\phi_{1}^{v}(X),f\circ (\phi_{1}^{v})^{-1})}$ is weakly continuous from $L^{2}([0,1],V)$ to $W'$.
\end{corollary}
\begin{proof}
 This essentially relies on classical results on differential equations and flows that can be found in \cite{Glaunes,Younes}. It is shown in particular that, since $V\hookrightarrow C_{0}^{2}(E,E)$, diffeomorphisms $\phi_{1}^{v}$ and their differentials tend to $\text{Id}$ at infinity. In addition, if $v_{n}$ is a sequence that weakly converges to $v$ in $L^{2}([0,1],V)$, then $\phi_{1}^{v_{n}}$ and $d\phi_{1}^{v_{n}}$ converge respectively to $\phi_{1}^{v}$ and $d\phi_{1}^{v}$ uniformly on every compact subset of $E$. Then, thanks to proposition \ref{chap_fcurrent:theo_controlnormdefor}, we are allowed to conclude that $\mu_{(\phi_{1}^{v_{n}}(X),f\circ (\phi_{1}^{v_{n}})^{-1})} \xrightarrow{W'} \mu_{(\phi_{1}^{v}(X),f\circ (\phi_{1}^{v})^{-1})}$, which proves the weak continuity.   
\end{proof}

\subsubsection{Variation formula for fvarifold metrics} \label{part:varForm}
We now extend the variation formula shown in \cite{Charon2} to the present setting of functional varifolds. The aim is to have a theoretical description of the variations of fvarifold metrics with respect to variations of a functional shape in both its geometrical support and signal. As we shall see, the behaviour is qualitatively similar, except for the orientation, to the setting of fcurrents that was examined in \cite{Charon_thesis}. Fixing a functional shape $(X,f)$ with $X$ a compact submanifold and $f: \ X \rightarrow \mathbb{R}$ a $C^{1}$ signal on $X$, we wish to compute variations of terms like $\langle \mu_{(X,f)}, \mu \rangle_{W'}$ for any $\mu \in W'$. Using the previous isometry $K_{W}$, we know that $\langle \mu_{(X,f)}, \mu \rangle_{W'} = \langle K_{W} \mu_{(X,f)}, K_{W} \mu \rangle_{W}= \mu_{(X,f)}(\omega)$ with $\omega = K_{W} \mu$ thanks to the reproducing kernel property and thus one is led to consider variation of terms $\mu_{(X,f)}(\omega)$ for $\omega \in W$. Rigorously, such variations can be expressed with respect to the action of infinitesimal deformations in the geometrico-functional domain. One can define such infinitesimal geometrico-functional deformations by considering, exactly as in the previous tangential model, a $C^{1}$ compactly supported vector field $v$ on $E$ and a $L^{2}$ function $h$ on $X$. Now we shall denote by $\phi_{t}$ the flow of $v$ at time $t$ and by $\psi_{t}$ the application defined on $X \times \mathbb{R}$ by $\psi_{t}(x,m)=(\phi_{t}(x),m+th(x))$. In addition, $\psi$ can be extended straightforwardly to the whole space $E \times \mathbb{R}$ by setting $\psi_{t}(x,m)=(\phi_{t}(x),m+t\tilde{h}(x))$ where $\rst{\tilde{h}}{X}=h$ and $\tilde{h}=0$ elsewhere. The function $\tilde{h}$ is then measurable on $E$ and the push-forward action $\psi_{t}^{\ast}\omega$ is well-defined by equation \eqref{eq:push_pull_var}. The variation we wish to compute can be expressed by:   
\begin{align*}
	\rst{\dfrac{d}{dt}}{t=0}\mu_{(\phi_{t}(X),(f+th)\circ \phi^{-1})}(\omega)&= \rst{ \dfrac{d}{dt}}{t=0} \int_{X} \psi_{t}^{\ast}\omega(x,T_{x}X,f(x)) d\mathcal{H}^{d}(x) \\
	&= \int_{X} \rst{ \dfrac{d}{dt}}{t=0} \psi_{t}^{\ast}\omega(x,T_{x}X,f(x)) d\mathcal{H}^{d}(x) \\
 &= \int_{X} (\pounds_{(v,h)}\omega)(x,T_{x}X,f(x)) d\mathcal{H}^{d}(x)
\end{align*}
where $\pounds_{(v,h)}\omega \doteq \rst{ \dfrac{d}{dt}}{t=0} \psi_{t}^{\ast}\omega$ is a notation for the derivative of function $\omega$ in the direction of $(v,h)$. In the following, we shall adopt the shortcut notation $\int_{X} g$ to denote the integral $\int_{X} g(x) d\mathcal{H}^{d}(x)$. Now, the result is the following:
\begin{thm}
 \label{theo:variation_formula}
Let $X$ be a compact submanifold and $f$ a $C^{1}$ function on $X$. Let $v$ be any $C^{1}$ vector field with compact support defined on $E$ and $h$ a $L^{2}$ function on $X$. We denote by $\psi_{t}(x,m)\doteq (\phi_{t}(x),m+th(x))$. Then, we have for any $C^{1}$ function $(x,V,m)\mapsto \omega(x,V,m)$ on $E \times G_{d}(E)\times \mathbb{R}$:
\begin{equation}
		\label{eq:variation_formula}
	\begin{split}
		\int_{X} (\pounds_{(v,h)}\omega)(x,T_{x}X,f(x)) d\mathcal{H}^{d}(x) = & \int_{X} \left ( \dfrac{\partial \omega}{\partial x} - \mdiv_{X} \left (\dfrac{\partial \omega}{\partial V} \right ) - \omega H_{X} | v^{\bot} \right ) + \dfrac{\partial \omega}{\partial m}.(h-\langle \nabla f,v^{\top}\rangle) \\
		&+ \int_{\partial X} \left \langle \nu, \omega v^{\top} + \left( \dfrac{\partial \omega}{\partial V} | v^\bot \right ) \right \rangle
	\end{split}
\end{equation}
where $v^\top$ and $v^\bot$ denote the tangential and normal part of $v$ along $X$, $\nu$ is the unit outward normal along $\partial X$, and $H_{X}$ the mean curvature vector to $X$.
 \end{thm}
The proof is given in appendix \ref{appendix:varition_fvar}. This formula is interesting at several levels because it gives qualitative information on the gradients of the kernel metrics with respect to points and signal values of the fshape $(X,f)$. In particular, as already pointed out in \cite{Charon2}, for a constant signal $f$, we see that, for points in the interior of $X$, the variation with respect to the geometrical support only involves the orthogonal component of $v$ which means that the gradient of the metric is orthogonal to the shape. This is somehow quite natural since tangential components of $v$ do not change the shape itself in that situation. However, we see that for non-constant signals, a term involving the tangential component $v^{\top}$ appears in the variation and is concentrated on regions of important gradient of the signal $f$. The second important consequence to mention is the presence of special terms on the boundary of $X$. On points located on $\partial X$, the gradient is orthogonal to the boundary (but not necessarily to $X$). The presence of these boundary singularities have important consequences on the numerical behaviour of the gradients that we shall address with more details in section \ref{part.traps}.

\section{Mathematical formulation of atlas estimation}
\label{sec:atlas_estimation_formulation}
To go beyond Karcher means within a single bundle (as exposed in Section \ref{ssection:Karcher}) and address the problem of atlas computation, we have to consider now a family of (possibly noisy) observed fshapes $(X^i,f^i)_{1\leq i\leq N}$. The main point here is the optimization of a template $(\Xtemp,\ftemp)$ given the observations. We propose, in this section, a way to formalize the atlas estimation as a variational problem. The existence of solutions to such problems will be examined in the next section.  

\subsection{Template space}
Ideally, the template should be optimized among all the possible templates to avoid any bias effect towards a particular configuration but this looks as a quite badly posed optimization problem. We follow here a more simple and secure route by introducing \emph{a priori} a large but restricted space of possible templates as another vector bundle $\FB_0$ generated by a hypertemplate $(X_0,f_0)$. As presented previously, $\FB_0$ is defined from the $(X_0,f_0)$ as the set of every fshape defined from the mapping 
$(\phi,\zeta)\mapsto (\phi.X_0,(f_0+\zeta)\circ\phi^{-1})$ for $(\phi,\zeta)\in G_0\times L^2(X_0)$ with $G_0\doteq G_{V_0}$ is the group of diffeomorphisms generated by an additional RKHS space $V_0$ of vector fields on $E$:
\begin{equation}
  \label{eq:FB40}
  \FB_0\doteq\{\ (X,f)\ |\ (X,f)=(\phi.X_0,(f_0+\zeta)\circ\phi^{-1}),\ (\phi,\zeta)\in G_0\times L^2(X_0)\}\,.
\end{equation}
 One could take $V_0=V$ but we may want less regularity on $V_0$ than on $V$ to generate a larger space and get closer to the ideal case.

Now, given any $(\Xtemp,\ftemp)\in \FB_0$, we should consider the associated fshape bundle $\FB_{\Xtemp}$ 
\begin{equation}
  \label{eq:FB40b}
  \FB_{\Xtemp}\doteq\{\ (X',f')\ |\ (X',f')=(\phi.\Xtemp,(\ftemp+\zeta)\circ\phi^{-1}),\ (\phi,\zeta)\in G\times L^2(\Xtemp)\}
\end{equation}
corresponding to the previously introduced space $\FB$. 

\subsection{Variational formulation (metamorphosis riemannian setting)}
In the Riemannian setting, the problem of atlas estimation can be addressed as the minimization of  
\begin{equation}
  \label{eq:FB38}
  \begin{split}
    J^{riem}_1(&(\Xtemp,\ftemp),(X^i,f^i)_i) \doteq \\
&  d_{\FB_0}((X_0,f_0);(\Xtemp,\ftemp))^2   +\sum_{i=1}^N\left(d_{\FB_X}((\Xtemp,\ftemp);(\tilde X^i,\tilde f^i))^2+g(\tilde{X}^i,\tilde{f}^i;X^i,f^i)\right)
  \end{split}
\end{equation}
on 
\begin{equation}
  \label{eq:FB39}
  (\Xtemp,\ftemp)\in\FB_0\text{ and }(\tilde{X}^i,\tilde{f}^i)_{1\leq i\leq N}\in\FB_{\Xtemp}^N
\end{equation}
where $g$ is the ``out-of-orbit'' or previous data attachment term, which, in that case, can be defined, as we explained in section \ref{sec:fvarifolds} by fvarifold Hilbertian metrics. 

The problem can be then reformulated equivalently as a minimization problem on 

\begin{equation}
\label{eq:FB41}
\left\{
\begin{array}[h]{l}
  (v^0,h^0)\in L^2([0,1],V_0\times L^2(X_0))\\
  (v^i,h^i)\in L^2([0,1],V\times L^2(\Xtemp))\\
\text{with }\\
\Xtemp=\phi^{v^0}_1.X_0
\end{array}\right.
\end{equation}
of
\begin{equation}
  \label{eq:FB42}
  \begin{split}
	  &J^{riem}_2((v^0,h^0),(v^i,h^i)_i)\doteq  \frac{\gamma_{V_0}}{2}\int_0^1|v^0_t|_{V_0}^2dt +\frac{\gamma_f^0}{2}\int_0^1\int_{X_0}|h^0_t|^2|\rst{d_x\phi^{v^0}_t}{T_{x}X}|d\Haus^d(x)\\
    & +\sum_{i=1}^N\left(
    \frac{\gamma_{V}}{2}\int_0^1|v^i_t|_{V}^2dt+\frac{\gamma_f}{2}\int_0^1\int_{\Xtemp}|h^i_t|^2|\rst{d_x\phi^{v^i}_t}{T_{x}X}|d\Haus^d(x)+\frac{\gamma_W}{2}g((\tilde{X}^i,\tilde{f}^i),(X^i,f^i))\right)
    \end{split}
\end{equation}
where
\begin{equation}
  \label{eq:FB42b}
  \left\{
    \begin{array}[h]{l}
      \tilde{f}^i=(\ftemp+\zeta^{h^i}_1)\circ(\phi^{v^i}_1)^{-1}\\
      \text{and}\\
     \ftemp=(f_0+\zeta^{h^0}_1)\circ(\phi^{v^0}_1)^{-1}\,.
    \end{array}
\right.
\end{equation}
\subsection{Variational formulation (tangential setting)}
\label{sec:atlas_tangential_setting}
The atlas estimation problem can also be formulated using the tangential metric and with a statistical flavor. It seems quite natural to choose the template as the pivotal point for the linearization towards the hypertemplate $(X_0,f_0)$ and the observations $(X^i_{},f^i_{})$ that is to freeze the functional metric on the $L^2(\Xtemp)$ space. 

The derivation is then quite straightforward and we jump directly to the associated minimization problem which is the minimization of the functional
 \begin{equation}
   \label{eq:FB44}
   \begin{split}
    J^{tan}((v^0,\zeta^0),(v^i,\zeta^i)_i)&\doteq  \frac{\gamma_{V_0}}{2}\int_0^1|v^0_t|_{V_0}^2dt +\frac{\gamma_f^0}{2}|\zeta^0|_{\Xtemp}^2\\
 & +\sum_{i=1}^N\left(
 \frac{\gamma_{V}}{2}\int_0^1|v^i_t|_{V}^2dt+\frac{\gamma_f}{2}|\zeta^i|^2_{\Xtemp}+ \frac{\gamma_W}{2} g((\tilde{X}^i,\tilde{f}^i),(X^i,f^i))\right)
    \end{split}
 \end{equation}
in the variables
\begin{equation}
 \label{eq:FB41b}
 \left\{
 \begin{array}[h]{l}
   (v^0,\zeta^0)\in L^2([0,1],V_0)\times L^2(\Xtemp)\\
   (v^i,\zeta^i)\in L^2([0,1],V)\times L^2(\Xtemp)\\
 \text{with }\\
 \Xtemp=\phi^{v^0}_1.X_0
 \end{array}\right.
 \end{equation}
where we have
\begin{equation}
  \label{eq:FB43}
  \left\{
    \begin{array}[h]{l}
      \tilde{f}^i=(\ftemp+\zeta^i)\circ(\phi^{v^i}_1)^{-1}\\
      \text{and}\\
     \ftemp=f_0\circ(\phi^{v^0}_1)^{-1}+\zeta^0\,.
    \end{array}
\right.
\end{equation}
The signals $(\zeta^i)_{i=1,..,N}$ will be named the \textbf{residual functions} of the atlas. From this general setting, we can consider an even simpler situation when $f_0\equiv 0$. This new setting is equivalent to the minimization of 
\begin{equation}
   \label{eq:FB45}
   \begin{split}
    J^{tan}_0((v^0,f),(v^i,\zeta^i)_i)&\doteq  \frac{\gamma_{V_0}}{2}\int_0^1|v^0_t|_{V_0}^2dt +\frac{\gamma_f^0}{2}|\ftemp|_{\Xtemp}^2\\
 & +\sum_{i=1}^N\left(
       \frac{\gamma_{V}}{2}\int_0^1|v^i_t|_{V}^2dt+\frac{\gamma_f}{2}|\zeta^i|^2_{\Xtemp} + \frac{\gamma_W}{2} g((\tilde{X}^i,\tilde{f}^i),(X^i,f^i))\right)
    \end{split}
\end{equation}
in the variables
\begin{equation}
 \label{eq:FB45b}
 \left\{
 \begin{array}[h]{l}
   (v^0,\ftemp)\in L^2([0,1],V_0)\times L^2(\Xtemp)\\
   (v^i,\zeta^i)\in L^2([0,1],V) \times L^2(\Xtemp)\\
 \text{with }\\
 \Xtemp=\phi^{v^0}_1.X_0
 \end{array}\right.
 \end{equation}
where we have
\begin{equation}
  \label{eq:FB43b}
 \tilde{f}^i=(\ftemp+\zeta^i)\circ(\phi^{v^i}_1)^{-1}\,.
\end{equation}

\section{Existence results for fshape atlases}
\subsection{Introduction}
The existence of an atlas for a population of fshapes with fcurrents or fvarifold based data term does not follow from the same arguments than in the more well-known pure geometrical case. The optimization of a signal $f$ on the geometrical template and for the more sophisticated model of additional residuals $\zeta^i$ to match each observations $(X^i,f^i)$ introduces new difficulties. 

The main one is that the fvarifold data term $\|\mu_{(X,f^i)}-\mu_{(X,f)}\|_W^2$ is not \emph{lower semi-continuous} for the weak convergence in $L^2(X)$ as a function of $f$ for $X$ fixed. This comes from the \emph{non-linearities} appearing in the various kernels defining the RKHS norm on $W$. Yet, interestingly, it remains continuous for the weak convergence on the space of measures. We prove that minimizing measure sequences are tight and that existence results can be established under quite general assumption for the extended functional. 

However, the existence of strong solution (\ie a true fshape template $(X,f)$) does not seem to be guaranteed in full generality but is true as soon as the regularization on the $L^2$ penalization on $f$ (and on the residuals $\zeta^i$ when introduced) is strong enough. This condition, although imposed theoretically for the existence of a fshape solution, is not clearly required for practical fshape atlases that we have computed so far. We show also that the computed template function and residuals are smooth proving a regularization effect of fvarifold attachment terms. The result holds in both the tangential and metamorphosis setting for fshapes, which is summed up by Theorem \ref{theo:existence_atlas_fvar_tang} (tangential model) and Theorem \ref{theo:existence_atlas_fvar_meta} (metamorphosis model), whose proofs are the core issues of this section. 

The most technical part of the proof is the existence of the fshape template $(X,f)$ on which we will focus in the first place. The proof we present follows basically the so-called \textit{direct method} of geometric measure theory where we first show the existence of a minimizer in a larger space of varifolds and then show that this solution does indeed result from a true fshape. 

\subsection{Existence of the template fshape}
We shall focus first on the tangential setting of \ref{sec:atlas_tangential_setting}. In this part, we will also consider the simpler situation where all variables $v^{i}$ and $\zeta^{i}$ in equation \eqref{eq:FB44} are frozen and show that a minimum of the functional exists in the variables $X$ and $f$. We will start by assuming that, in addition, $X$ is also a fixed bounded and finite volume $d$-dimensional rectifiable subset of $E$.

\subsubsection{Existence with $X$ fixed}
\label{ssec:atlas_existence_Xfixed}
This subsection is dedicated to the proof of the following proposition: 
\begin{prop}
\label{prop:JP1}
  Assume that $W$ is continuously embedded in $C_0^2(E\times G_{d}(E) \times \R)$, that $X$ and $(X^i)_{1\leq i\leq N}$ are finite volume bounded $d$-dimensional rectifiable subsets and that $f^i\in L^2(X^i)$ for $1\leq i\leq N$. Assume that $\gamma_f/\gamma_W$ is large enough with $\gamma_f,\gamma_W\geq 0$. Then  
$$J_X(f)\doteq \frac{\gamma_f}{2}\int_X |f(x)|^2d\Haus^d(x)+\frac{1}{2}\sum_{i=1}^N \gamma_W\|\mu_{(X^i,f^i)}-\mu_{(X,f)}\|_{W'}^2$$
achieves its minimum on $L^2(X)$ and any minimizer $f_*$ is such that $f_*\in L^\infty(X)$. Moreover, if $X$ is also a $C^p$ submanifold and $W\hookrightarrow C_0^m(E\times G_{d}(E) \times \R)$ with $m\geq \max \{p,2\}$, then $f_*\in C^{p-1}(X)$.
\end{prop}
The proof proceeds in two steps.

The first one is to show the existence of a minimizer in the larger space of fvarifolds: namely, for $X$ fixed as previously, we consider 
the set $\MeasX$ of all Borel finite measures $\nu$ on $E \times G_{d}(E) \times \mathbb{R}$ such that:
\begin{equation}
  \label{eq:JP8}
  \int h(x,V)d\nu(x,V,f)=\int_X h(x,T_{x}X)d\Haus^d(x)\,.
\end{equation}
for all continuous and bounded applications $h$ on $E \times G_{d}(E)$. These are fvarifolds that have a particular marginal on $E \times G_{d}(E)$. Note that any varifold $\mu_{(X,f)}$ for $f$ a $L^{2}$ function on $X$ belongs to $\MeasX$. In addition, $\MeasX$ is a closed subset for the usual weak convergence of measures defined by $\nu_n \rightharpoonup \nu_\infty$ if for any $\omega \in C_b(E \times G_{d}(E) \times \R)$
\begin{equation}
  \label{eq:JP9}
  \nu_n(\omega)\to \nu_\infty(\omega)\,.
\end{equation}  
Then, we consider the extended functional $\eJ$ defined on measures $\nu\in\MeasX$ by 
\begin{equation}
  \label{eq:JP1}
  \eJ(\nu)\doteq \frac{\gamma_f}{2}\nu(|f|^2)+\frac{\gamma_W}{2}\sum_{i=1}^N\|\nu-\mu_{(X^i,f^i)}\|_{W'}^2
\end{equation}
where $\nu(|f|^2)$ is the notation we shall use for $\int |f|^{2} d\nu$. Then, 
\begin{lemma}
 \label{lemma:JP1}
 There exists $\nu_{\ast} \in \MeasX$ that minimizes the functional $\eJ$. 
\end{lemma}
The second step consists in proving that $\nu_{\ast}$ can be actually expressed as a fvarifold associated to a true fshape $(X,f_{\ast})$ with $f_{\ast}\in L^\infty(X)$. We have detailed the full proofs of Lemma \ref{lemma:JP1} and Proposition \ref{prop:JP1} in Appendix \ref{appendix:proof_propJP1}.

\subsubsection{Existence with non-fixed $X$}
\label{subsection:atlas_existence_Xnonfixed}
We now consider the existence of a template when $X$ is no more fixed and is estimated as well. In this case, as we mentioned earlier, it is enough to introduce a RKHS Hilbert space $V_0$ continuously embedded in $C^2_0(E,E)$, an initial hypertemplate $X_0$ and consider also an optimization of the template $X$ in the orbit of $X_0$ under the action of $\phi^0\in  G_0$, the group of diffeomorphisms associated with $V_0$ {\it i.e.} $X=\phi^0.X_0$. To prove existence result, we will need to introduce a penalization depending on the distance between $X_0$ and $X$ {\it i.e.} on $d_{G_0}(\text{Id},\phi^0)$. A typical functional would be, if $\X$ is the orbit of $X_0$ under $G_0$ and $d_\X(X,X')\doteq \inf_{\phi\in G_0,\phi.X=X'}d_{G_0}(\text{Id},\phi)$ is the induced distance between two templates in the orbit, 
\begin{equation}
  \begin{split}
    J_1(X,f) & =\frac{\gamma_{V_0}}{2}d_\X(X_0,X)^2\\
&+\frac{\gamma_f}{2}\int_X|f(x)|^2d\Haus^d(x)+\frac{\gamma_{W}}{2}\sum_{i=1}^N\|\mu_{(X^i,f^i)}-\mu_{(X,f)}\|_{W'}^2\label{eq:JP27}
  \end{split}
\end{equation}
Since a diffeomorphism $\phi^{0}\in G_{0}$ is obtained as the flow of a time-varying vector field $v^{0} \in L^2([0,1],V_0)$, we consider the minimization of the following functional: 
\begin{equation*}
  \begin{split}
    J_2(v^0,f) & =\frac{\gamma_{V_0}}{2}\|v^0\|^2_{L^2([0,1],V_0)}\\
&+\frac{\gamma_f}{2}\int_X|f(x)|^2 d\Haus^d(x)+\frac{\gamma_{W}}{2}\sum_{i=1}^N\|\mu_{(X^i,f^i)}-\mu_{(X,f)}\|_{W'}^2
  \end{split}
\end{equation*}
for $v^0\in {L^2([0,1],V_0)}$, $f\in L^2(X)$ with $X=\phi^0.X_0$. However, it is more convenient for $X=\phi^0.X_0$ to consider the change of variable $f\mapsto f_0=f\circ\phi^0$ from $L^2(X)\to L^2(X_0)$ so that we keep working in a fixed space $L^2(X_0)$. Hence we end up with functional
\begin{equation}
\begin{split}
  J_3(v^0,f_0)& =\underbrace{\frac{\gamma_{V_0}}{2}\|v^0\|_{{L^2([0,1],V_0)}}^2}_{\text{penalization on $X$}}+\frac{\gamma_f}{2}\underbrace{\int_{X_0}
  |f_0(x)|^2|\rst{d_{x}\phi^0}{T_{x}X}|d\Haus^d(x)}_{=\int_X|f(x)|^2d\Haus^d(x)\text{ for }f=f_0\circ(\phi^0)^{-1}\in L^2(X)}\\
&+\frac{\gamma_{W}}{2}\sum_{i=1}^N \|\mu_{(X^i,f^i)}-\phi^{0}.\mu_{(X_0,f_0)}\|_{W'}^2
\end{split}\label{eq:JP26}
\end{equation}
with $f_0\in L^2(X_0)$ and $\phi.\mu\in W'$ denoting for a fvarifold $\mu$ its diffeormorphic transport by a $C^1$ diffeormorphism $\phi$ defined by $(\phi.\mu)(\omega)\doteq \int |\rst{d_x\phi}{V}|\omega(\phi(x),d_x\phi(V),f)d\mu(x,V,f)$ for $\omega\in W$. The existence of a minimizer $(v^0_{*},f_{0,*})$ for $J_3$ gives immediately the existence of a minimizer $(X_*,f_*)$ for $J_1$ with $X_*\doteq \phi^{v^0_{*}}_1.X_0$, $f_*\doteq f_{0,*}\circ (\phi^{v^0_{*}}_1)^{-1}\in L^2(X_*)$.

The existence result then becomes: 
\begin{prop}
\label{prop:JP2}
  Assume that $W$ is continuously embedded in $C_0^2(E\times G_{d}(E) \times \R)$, that $X_0$ and $(X^i)_{1\leq i\leq N}$ are finite volume and bounded $d$-dimensional rectifiable subsets and that $f^i\in L^2(X^i)$ for $1\leq i\leq N$. Assume $\gamma_{V_0}>0$ and $\gamma_f/\gamma_W$ is large enough with $\gamma_f,\gamma_W\geq 0$. Then 
  \begin{itemize}
	  \item $J_1$ given by equation \eqref{eq:JP27} achieves its minimum on $\{ (X,f)\ |\ X\in\X,\ f\in L^2(X)\}$;	  
          \item any minimizer $(X_*,f_*)$ is such that $f_*\in L^\infty(X_*)$;
	  \item if $X_0$ is also a $C^p$ submanifold and $W\hookrightarrow C_0^m(E\times G_{d}(E) \times \R)$ with $m\geq \max \{p,2\}$, $f_*\in C^{p-1}(X_*)$.
  \end{itemize}
\end{prop}
The proof of Proposition \ref{prop:JP1} can in fact be easily adapted to this new situation. We can consider the formulation of equation \eqref{eq:JP26} with $J_{3}$ a functional on the vector field $v^{0}$ and the function $f \in L^{2}(X_0)$. With respect to $v^{0}$, thanks to the penalization in \eqref{eq:JP26}, we can restrict the search of a minimum on a closed ball $B_0$ of given radius $b$ in $L^2([0,1],V_0)$, which guarantees at the same time that the Jacobians $|\rst{d_{x}\phi^0}{T_{x}X}|$ are uniformly lower bounded. This closed ball is also compact for the weak topology in $L^2([0,1],V_0)$. In addition, it follows from Corollary \ref{cor:weak_continuity_W'} that $v^{0}\mapsto \sum_{i=1}^N \|\mu_{(X^i,f^i)}-\phi^{0}_*\mu_{(X_0,f_0)}\|_{W'}^2$ is weakly continuous on $L^2([0,1],V_0)$ and it is also classical that $v^{0} \mapsto \|v^{0}\|_{L^{2}([0,1],V_0)}$ is lower semicontinuous for the weak convergence topology. Therefore, for all fixed $f_{0} \in L^{2}(X_{0})$, $v^{0}\mapsto J_{3}(v^{0},f_{0})$ is weakly lower semicontinuous on $L^2([0,1],V_0)$. It results that we obtain existence of a minimizing vector field $v^{0}$ and, reasoning as in the previous subsection and Appendix \ref{appendix:proof_propJP1}, one deduces easily the claim of Proposition \ref{prop:JP2}.  

\subsection{Existence of full fshape atlases (tangential setting)}
We now generalize the previous results to the existence of complete atlases of fshapes' datasets. In addition to the template $(X,f)$, one wants to simultaneously estimate  transformations from the template to each subject. In the tangential model of Section \ref{sec:atlas_tangential_setting}, these consist in deformations $(\phi^{i})_{i=1,..,N}$ obtained as flows of time varying vector fields $(v^{i})$ and residual signals $(\zeta^{i})$ that are $L^{2}$ functions on $X$. In that situation, following Section \ref{sec:atlas_tangential_setting}, the optimization functional for atlas estimation writes:
\begin{equation}
\label{eq:A1}
  \begin{split}
    &J(X,f,(\zeta^{i}),(v^i)) \doteq \frac{\gamma_{V_0}}{2}d_\X(X_0,X)^2 +\frac{\gamma_f}{2}\int_X|f(x)|^2d\Haus^d(x)\\
&+\frac{1}{2}\sum_{i=1}^N\left(\gamma_V\|v^i\|^2_{L^{2}([0,1],V)}+\gamma_\zeta\int_X|\zeta^i(x)|^2d\Haus^d(x)+\gamma_W\|\mu_{(X^i,f^i)}-\mu_{(\phi^{v^i}_{1}(X),(f+\zeta^i)\circ(\phi^{v^i}_{1})^{-1})}\|^2_{W'}\right)
\\
  \end{split}
\end{equation}

The main existence result is the following:
\begin{thm}
 \label{theo:existence_atlas_fvar_tang}
  Assume that $W$ is continuously embedded in $C_0^2(E\times G_d(E)\times \R)$, that $X_0$ and $(X^i)_{1\leq i\leq N}$ are finite volume bounded $d$-dimensional rectifiable subsets and that $f^i\in L^2(X^i)$ for $1\leq i\leq N$. Assume $\gamma_{V_0},\gamma_V>0$ and $\gamma_f/\gamma_W$ and $\gamma_\zeta/\gamma_W$ are large enough with $\gamma_f,\gamma_W,\gamma_\zeta\geq 0$. Then
  \begin{itemize}
	  \item $J$ given by equation \eqref{eq:A1} achieves its minimum on $\{ (X,f,\zeta^i, (v^i))\ |\ X\in\X,\ f\in L^2(X),\ \zeta=(\zeta^i)\in L^2(X)^N,\ (v^i)\in L^{2}([0,1],V)^N\}$;
	  \item any minimizer $(X_*,f_*,(\zeta^i_*),(v_{*}^i))$ is such that $f_*$ and $\zeta^i_*$ for $1\leq i\leq N$ are in $L^\infty(X_*)$;
	  \item if $X_0$ is also a $C^p$ submanifold and $W\hookrightarrow C_0^m(E\times G_{d}(E) \times \R)$ with $m\geq \max \{p,2\}$, $f_*$ and the $\zeta^i_*$'s are in $C^{p-1}(X_*)$.
  \end{itemize}
\end{thm}
The proof leans essentially on the same arguments as detailed in the previous subsections: we refer the reader to Appendix \ref{appendix:proof_existence_atlas_fvar_tang}. This result is of fundamental importance for the rest of the paper since it ensures that the atlas estimation problems that we shall study numerically in the next sections do have at least a solution, and we see that this holds with only $L^{2}$ regularity assumptions on the signals.

\subsection{Existence in the metamorphosis framework}
Interestingly, the previous existence of solutions in the simplified tangential setting can be also used to show existence of solutions to the corresponding variational problem in the metamorphosis framework given by \eqref{eq:FB38}. In that model, the subjects are obtained approximately as metamorphoses of the template $(X,f)$, which is itself a metamorphosis of an hypertemplate $(X_{0},f_{0})$ where $X_{0}$ is a bounded finite-volume rectifiable subset of $E$ and $f_{0} \in L^{2}(X_{0})$. The important lemma that bridges both approaches is the following:

\begin{lemma}
\label{lemma_meta_atlas}
  For $v$ fixed and $\zeta_1\in L^2(X)$ fixed, the infimum over $h\in L^2([0,1],L^2(X))$ of 
  $$ \frac{\gamma_f}{2}\int_0^1\int_X|h_{t}(x)|^2|\rst{d_x\phi^{v}_{t}}{T_{x}X}|d\Haus^d(x) $$
under the constraint that $\int_0^1h_tdt=\zeta_1$ is reached on a unique point $h^*\in L^2([0,1],L^2(X))$ given by
\begin{equation}
  \label{eq:F1}
  h^*_t(x)\doteq C(x) \frac{\zeta_1(x)}{|\rst{d_x\phi^{v}_{t}}{T_{x}X}|}
\end{equation}
where $C(x)\doteq(\int_0^1\frac{1}{|\rst{d_x\phi^{v}_{s}}{T_{x}X}|}ds)^{-1}$.
\end{lemma}
\begin{proof} The proof could be deduced from general optimization results but for sake of completeness we give here a short proof. Let us first notice that over $(x,t)\in X\times [0,1]$ we have   $0<\inf |\rst{d_x\phi^{v}_{t}}{T_{x}X}|\leq \sup |\rst{d_x\phi^{v}_{t}}{T_{x}X}|<+\infty$.

	Now, if we denote  $\alpha_t(x)\doteq |\rst{d_x\phi^{v}_{t}}{T_{x}X}|^{-1/2}$ and consider the change of variable $g_t(x)=h_t(x)/\alpha_t(x)$, the problem is equivalent (after the use of Fubini-Tonelli Theorem) to the minimization on $g\in L^2([0,1],L^2(X))$ of 
$$\frac{\gamma_f}{2}\int_X\left(\int_0^1|g_{t}(x)|^2dt\right)d\Haus^d(x)$$
under the constraint $\int_0^1\alpha_s(x)g_s(x)ds=\zeta_1(x)$ for a.e. $x\in  X$.

However, this is a separable problem and for any fixed  $x\in X$ we have to consider on $H\doteq L^2([0,1])$ the very simple quadratic minimization problem: $\inf |u|^2_H$ under the constraint $\langle \alpha_.(x),u\rangle_H=\zeta_1(x)$ that has the unique solution $u=\lambda(x)\alpha_.(x)$ with $\lambda(x)= \zeta_1(x)/|\alpha_.(x)|^2_H$.

Hence, denoting $h_t^*(x)=\zeta_1(x)\frac{|\alpha_t(x)|^2}{|\alpha_.(x)|_{H}^2}$, we get the result.
\end{proof}

Now, introducing, as in \ref{ssec:metamorphosis_bundles}, the Riemannian energies 
$$E_{X_{0}}(v^{0},h^{0}) \doteq \frac{\gamma_{V_{0}}}{2} \int_{0}^{1} |v_{t}|_{V_{0}}^{2} dt + \frac{\gamma_{f_{0}}}{2} \int_0^1\int_{X}|h^{0}_t|^2(x)|\rst{d_x\phi^{v^{0}}_t}{T_{x}X}|d\Haus^d(x)$$ 
and  
$$E_{X}(v,h) \doteq \frac{\gamma_{V}}{2} \int_{0}^{1} |v_{t}|_{V}^{2} dt + \frac{\gamma_{f}}{2} \int_0^1\int_{X}|h_t|^2(x)|\rst{d_x\phi^{v}_t}{T_{x}X}|d\Haus^d(x)$$ 
the atlas estimation functional becomes:
\begin{equation}
\label{eq:M1}
  \begin{split}
    J((v^{0},h^{0}),(v^{i},h^{i})) & \doteq E_{X_0}(v^{0},h^{0})\\
      & + \sum_{i=1}^N\left(E_{X}(v^{i},h^{i}) +\frac{\gamma_W}{2}\|\mu_{(X^i,f^i)}-\mu_{(\phi^{v^i}_{1}(X),(f+\zeta^{h^{i}}_{1})\circ(\phi^{v^i}_{1})^{-1})} \|^2_{W'}\right)
  \end{split}
\end{equation}
where $(X,f)=(\phi_{1}^{v^0}(X_0),(f_0+\zeta_{1}^{h^{0}})\circ(\phi_{1}^{v^0})^{-1})$.

\begin{thm}
 \label{theo:existence_atlas_fvar_meta}
 Let $W$ be continuously embedded in $C_0^2(E\times G_d(E)\times \R)$, $X_0$ and $(X^i)_{1\leq i\leq N}$ finite volume bounded $d$-dimensional rectifiable subsets and $f^i\in L^2(X^i)$ for $1\leq i\leq N$. Assume that $\gamma_{V_0},\gamma_V>0$ and that $\gamma_f/\gamma_W$ and $\gamma_\zeta/\gamma_W$ are large enough with $\gamma_f,\gamma_W,\gamma_\zeta\geq 0$. Under these assumptions,
 \begin{itemize}
  \item $J$ defined by equation \eqref{eq:M1} achieves its minimum on $v^{0} \in L^{2}([0,1],V_{0})$, $h^{0} \in L^{2}([0,1],L^{2}(X_{0}))$, $(v^i)\in L^{2}([0,1],V)^N$ and $(h^i)\in L^{2}([0,1],L^{2}(X))^N$;
  \item any minimizer $(X_*,f_*,(\zeta^i_*),(v_{*}^i))$ with $f_*=(f_0+\zeta_{1}^{h^{0}_*})\circ (\phi_{1}^{v^0_*})^{-1}$ and $\zeta^i_* =\zeta_{1}^{h^{i}_*} \circ (\phi_{1}^{v^{i}_*})^{-1}$ is such that $f_*$ and $\zeta^i_*$ for $1\leq i\leq N$ are in $L^2(X_*)$;
  \item if $X_0$ is also a $C^p$ submanifold, $f_{0}$ a $C^{p-1}$ function on $X$ and $W\hookrightarrow C_0^m(E\times G_{d}(E) \times \R)$ with $m\geq \max \{p,2\}$, then $f_*$ and the $\zeta^i_*$'s are in $C^{p-1}(X_*)$.  
 \end{itemize}
\end{thm}
The proof can be found in appendix \ref{appendix:proof_existence_atlas_fvar_meta}.


\section{A discrete framework for fshape}
\label{part:discreteScheme}

The previous sections mainly dealt with the setting of a well posed theoretical framework to work with fshapes. Here we would like to tackle the problem of the actual processing of fshapes. We believe that the description of an explicit numerical scheme is of the utmost importance for the fshape setting to be of practical utility. Our main motivations come from medical imaging where acquisition is done in digital forms. It means that data are discrete and in very high dimensions (tens of thousands points). Our goal is to demonstrate that the fshape framework may be used to handle real data and not just low dimensional examples.


\subsection{Fvarifold norm of polyhedral meshes} 
\label{part.discreteNotation}

We assume hereafter that data at hand are finite functional polyhedral meshes composed by $\nT$ cells of dimension $\dT$ immersed in $E=\R^\dP$. Most common examples are piecewise linear planar fcurves corresponding to the case $\dP=2$ and $\dT=1$, piecewise linear curves in space corresponding to $\dP=3$ and $\dT=1$, and piecewise triangular surfaces in the space corresponding to $\dP=3$ and $\dT=2$. A finite polyhedral fshape $(X,f)$ is fully described by three matrices: a first $\nP\times\dP$ matrix $\x$ contains the coordinates of the $P$ vertices $x_k$ in $E$ (\ie $\x=(x_k)_{1\leq k\leq P}$), a second $\nP\times1$ column vector $\f$ contains the $P$ values $f_{k} \doteq f(x_k)\in\R$ of the  signal (\ie $\f=(f_k)_{1\leq k\leq P}$) and a third $T\times(\dT+1)$ matrix contains the list of edges for each cell so that each row contains $(\dT+1)$ integers corresponding to the indices of the vertices of the cell. 

In the sequel, it will be convenient to introduce the corresponding geometrical diffeomorphic action in the discrete setting: for a diffeomorphism $\phi\in G$ and a discretized fshape $(\x,\f)=(x_k,f_k)_{1\leq k\leq P}$ we denote
\begin{equation}
  \label{eq:dicret_action}
  \phi.(\x,\f)\doteq (\phi.\x,\f) \text{ where }\phi.\x\doteq (\phi(x_k))_{1\leq k\leq P}
\end{equation}
Note that the discrete action is more straightforward than its continuous version $\phi.(X,f)=(\phi.X,f\circ\phi^{-1})$ since along a flowing particle $x_k$ we have $f\circ\phi^{-1}(\phi(x_k))=f(x_k)=f_k$ (Lagrangian point of view). Note also that the connectivity matrix (\ie list of edges) remains unchanged through the deformation process.

There is a rather natural way to approximately represent polyhedral meshes in the fvarifold space. The method is very similar to the one described in \cite{Charon2} in the case of fcurrents space. Using formula \eqref{eq:fshape_dirac}, each cell is coded by a Dirac $\delta_{(\hat x,V,\hat f)}$ so that a discrete polyhedral fshape is viewed as a distribution of Dirac spread in the space $\R^\dP$. Therefore, the measure $\mu_{(X,f)}$ of equation \eqref{eq:fshape_fvar} which is associated to $(X,f)$ is approximated by the following finite sum of Dirac
	\begin{equation}\label{eq.discrete}
		\mu_{(X,f)} \approx \mu_{(\x,\f)} \doteq \sum_{\ell=1}^{\nT} r_{\ell} \delta_{(\hat x_{\ell},V_{\ell},\hat f_{\ell})}
	\end{equation}
	where  $\hat x_{\ell}\in\R^{\dP}$ is the center of the $\ell$-th polyhedral cell, $V_{\ell}\in G_d(\R^{\dP})$ gives the direction of the cell, $\hat f_{\ell} \in \R^{}$ is the average value of the signal on this cell and finally $r_\ell\in\R$ is equal to the $d$-volume of the cell. Note that for the case of surfaces (resp. curves), $V_\ell$ will be simply a unit normal vector (resp. unit tangent) and the orientation will be removed directly ``by the kernel'': the formula of the kernel $k_t$ used to compute the scalar product in the fvarifold space will be invariant with respect to changes in the orientation as in equation \eqref{eq.GaussKern}.

	As an illustration, let us precisely describe how to get the representation of the $\ell$-th triangle $T_\ell$ belonging to a polyhedral functional surface $(X,f)$. If we assume that the vertices of $T_\ell$ (read in the connectivity matrix) are $(x_{k_1},x_{k_2},x_{k_3})\in(\R^3)^3$ for some $k_1,k_2,k_3 = 1,\cdots,\nP$, we have  
\begin{equation}\label{eq.fvarRepresentation}
	\begin{cases}	\hat x_\ell = \nicefrac{(x_{k_1}+x_{k_2}+x_{k_3})}{3},\\ V_{\ell} = \nicefrac{(x_{k_2}-x_{k_1}) \wedge (x_{k_3}-x_{k_1})}{ \norm{(x_{k_2}-x_{k_1}) \wedge (x_{k_3}-x_{k_1})}},\\ \hat f_{\ell} = \nicefrac{(f_{k_1}+f_{k_2}+f_{k_3})}{3},\\ r_{\ell} = \norm{(x_{k_2}-x_{k_1}) \wedge (x_{k_3}-x_{k_1})}.\end{cases}
\end{equation}
where $\|.\|$ is the standard Euclidean norm in $\R^3$ and $\wedge$ correspond to the cross product. To compute the fvarifold norm of a functional surface we use in our numerical experiments Gaussian kernels satisfying assumptions of Proposition \ref{prop:kernel_fvar}. Let us define:   
\begin{equation}\label{eq.GaussKern}
	\begin{cases}
		k_e(\hat x_1,\hat x_2)  =  \exp\big( -\nicefrac{\|\hat x_{1} - \hat x_{2}\|^2}{\sigma^2_e} \big),\  \\
		k_t(V_1,V_2)  =  \exp\big( -\frac{2}{\sigma^2_t}( 1 - \prs{V_1,V_2}^2 )\big),\   \\
		k_f(\hat f_1,\hat f_2)  =  \exp\big( -\nicefrac{|\hat f_{1} - \hat f_{2}|^2}{\sigma^2_f} \big).  
	\end{cases}
\end{equation}
where $\sigma_e,\sigma_t,\sigma_f >0$ are scale parameters fixed by the practitioner. Thence the scalar product of two fvarifolds $\mu_{(X,f)}$ and $\mu_{(Y,g)}$ associated to two polyhedral triangular functional surfaces $(X,f)$ and $(Y,g)$, and approximated respectively by $\sum_{\ell=1}^{\nT_{\x}}r_{\ell}\delta_{\hat x_{\ell},V_{\ell},\hat f_{\ell}}$ and $\sum_{\ell'=1}^{\nT_{\y}} s_{\ell'}\delta_{\hat y_{\ell'},W_{\ell'},\hat g_{\ell'}}$,  may be written:
\begin{equation}
	\label{eq.discreteprs}
	\prs{\mu_{(X,f)},\mu_{(Y,g)}}_{W'} \approx \prs{\mu_{(\x,\f)},\mu_{(\y,\g)}}_{W'} =\sum_{\ell=1}^{\nT_{\x}}\sum_{\ell'=1}^{\nT_{\y}} r_{\ell}s_{\ell'}k_e(\hat x_\ell,\hat y_{\ell'})k_t(V_\ell,W_{\ell'})k_f(\hat f_\ell,\hat g_{\ell'}).
\end{equation}
We discuss in Section \ref{part:bd} some numerical aspects of the fvarifolds.

\subsection{Discrete approximations of continuous fshapes. }

A continuous fshape $(X_c,f_c)$ can be approximated by a finite polyhedral fshape in the fvarifold space as follows. First, extract $\nP$ points and their corresponding signal from $(X_c,f_c)$. Then, compute a mesh with the extracted points to define a polyhedral fshape $(X,f)$. Finally, use \eqref{eq.discrete} to approximate $\mu_{(X,f)}$ by a finite sum $\mu^{}_{(\x,\f)}$ of Diracs. 

Although this method of approximation seems reasonable in many practical cases, we do not provide  explicit conditions to ensure the convergence of $\mu^{}_{(\x,\f)}$ toward $\mu_{(X_c,f_c)}$ when the number $\nP$ of extracted vertices tends to infinity. 
To the best of our knowledge, this is still an open problem and a famous illustration of pathological cases is the Schwarz polyhedron. Nevertheless, there exists various theoretical results concerning the polyhedral approximation of continuous surfaces for the currents norms and the varifold norms, see \textit{e.g.} \cite{Simon}. Moreover, in \cite{theseThibert}, the author studies sufficient conditions to ensure the convergence of the area of triangular meshes toward the area of the continuous surface. A general result extending the aforementioned works to our geometrico-functional framework involves some materials beyond the scope of this paper.

\subsection{Shooting equations for the metamorphosis Riemannian framework}
\label{Dyn_tangential}

In the large deformation setting for usual shapes, it was shown (see for instance \cite{Miller2006}) that the dynamic of optimal vector fields giving geodesics in groups of diffeomorphisms can be described through Hamiltonian systems of equations, called the \textbf{forward equations}. In the discrete situation of a finite set of points, these actually reduce to a coupled evolution of the position of the particles and extra variables called the \textit{momenta} attached to every particle. In addition, the Hamiltonian structure implies that all geodesic trajectories are eventually parametrized only by the initial positions and momenta, which is the principle underlying \textit{geodesic shooting algorithms} for shape matching (see \cite{arguillere14:_shape} for an extensive presentation of the Hamiltonian setting for shape deformation analysis).      

We now describe how to obtain similar shooting frameworks in our more general situation of fshapes, first for the fshape metamorphosis model. As discussed above, any element $(X,f)$ in $\FB$ will be discretized as a family $(\x,\f) = (x_k, f_k)_{1\leq k\leq \nP}$ of points $x$ with signal value $f$. In that setting, for a proper weighting matrix (the precise definition of this matrix depends on actual choices for the approximation of the $L^2$ norm), we can discretize the $L^2(X)$ dot product as $(D(\x)\h|\h)$ for any $\h=(h_k)\in\R^\nP$. The continuous problem \eqref{eq:FB3} is then approximated by a simple discrete control problem
\begin{equation}
 \label{eq:optimal_control_discrete}
\min \frac{\gamma_V}{2}\int_0^1|v_t|_V^2dt+\frac{\gamma_f}{2}\int_0^1 (D(\x_t)\h_t|\h_t)dt
\end{equation}
for fixed end point conditions $\rst{(\x_t,\f_t)}{t=0}$ and $\rst{(\x_t,\f_t)}{t=1}$ and controlled dynamic in $V\times \R^{\nP}$ given by
\begin{equation}
\label{eq:alain1}
  \left\{
  \begin{array}[h]{rcl}
    \dot{\x}_t & = & v_t.\x_t\\
    \dot{\f}_t & = & \h_t
  \end{array}\right.,
\end{equation}
where $v.\x=(v(x_k))_{1\leq k\leq P}$.
\begin{remark}
  An important point is to note that the discrete evolution is based on a Lagrangian particle based representation and to understand where the equations (\ref{eq:alain1}) are coming from. In the discrete setting, $t\mapsto \x_t=(x_{t,k})_{1\leq k\leq P}$ is the evolution the vertices of a polyhedral mesh i.e. $x_{t,k}=\phi_t^v(x_{k})$ where $\phi^v_t$ is the flow of the time dependent vector field $t\mapsto v_t$ and $\x=(x_k)$ are the vertices on the initial mesh $X$ so that $\dot{x}_{t,k}=v_t(\phi^v_t(x_k))=v_t(x_{t,k})$ which gives the first equation of (\ref{eq:alain1}). Moreover, $\f_t=(f_{t,k})_{1\leq k\leq P}$ are the signal values attached to $X_t$ at the vertices positions $\x_t=(x_{t,k})_{1\leq k\leq P}$ so that $f_{t,k}$ is the discretization of the continuous signal $f_t:X_t\to\mathbb{R}$ in \emph{Lagrangian coordinates} i.e. $f_{t,k}=f_t(\phi^v_t(x_k))=f_t(x_{t,k})$. In particular, one have $\dot{f}_{t,k}=\frac{d f_t(\phi^v_t(x_k))}{dt}=\frac{d\int_0^t h_s(x_k)ds}{dt}=h_t(x_k)=h_{t,k}$ where $\h_t=(h_{t,k})=(h_t(x_k))$ is the discretization of $h_t:X\to\mathbb{R}$ on the vertices $(x_k)$ which gives the second equation of  (\ref{eq:alain1}).  
\end{remark}
In that situation, introducing co-states or \textit{momenta} $(\p,\p^f) = (p_k,p^f_k)_{1\leq k\leq \nP}$ we get the associated Hamiltonian 
\begin{equation*}
  H((\x,\f),(\p,\p^f),(v,\h))=(\p|v.\x)+(\p^f|\h)-\frac{\gamma_V}{2}|v|_V^2-\frac{\gamma_f}{2}(D(\x)\h|\h)
\end{equation*}
so that we deduce from the Pontryagin's Maximum Principle that the optimal controls satisfy
\begin{equation}
  \label{eq:FB21a}
	v(\cdot)=\frac{1}{\gamma_V}\sum_{k=1}^{\nP} K_V(\cdot,x_k)p_{k}\text{ and } \h=\frac{1}{\gamma_f}D^{-1}(\x)\p^f
\end{equation}
where $K_V$ is the kernel associated with the RKHS space $V$. Then, plugging the optimal control into the Hamiltonian $H$ we get the reduced Hamiltonian given as
\begin{equation}
  \label{eq:FB21}
  H_r((\x,\f),(\p,\p^f))=\frac{1}{2\gamma_V}(K_{\x,\x}\p|\p)+\frac{1}{2\gamma_f}(D^{-1}(\x)\p^f|\p^f)
\end{equation}
with $K_{\x,\x}\doteq (K_V(x_k,x_{k'}))_{1\leq k,k'\leq \nP}$.
Note that the weighting matrix $D(\x)$ may be chosen to be diagonal (``mass lumping'') so that the computation of $D^{-1}(\x)$ is straightforward. From $H_r$, we can derive the forward equation given by the Hamiltonian dynamic:
\begin{equation}
\label{eq:forward_discrete_tang_a}
  \left\{
  \begin{array}[h]{rcl}
	  \dot{\x} & = & \partial_{\p} H_r(\x,\p) \\
    \dot{\p} & = & -\partial_{\x} H_r(\x,\p)\,. 
  \end{array} \right.
\end{equation}
As usual, since the Hamiltonian is not depending on time, it is a conserved quantity during the geodesic evolution so that we get from (\ref{eq:FB21a}) and (\ref{eq:FB21}) that
\begin{equation}
  \label{eq:8.4.1}
  \frac{\gamma_V}{2}\int_0^1|v_t|_V^2dt+\frac{\gamma_f}{2}\int_0^1 (D(\x_t)\h_t|\h_t)dt= \rst{H_r((\x,\f),(\p,\p^f))}{t=0}\,.
\end{equation}
We see also that the reduced Hamiltonian is a perturbation of the more familiar Hamiltonian for the pure geometrical case mentioned earlier, $H^{geo}_r(\x,\p)=\frac{1}{2\gamma_V}(K_{\x,\x}\p|\p)=\frac{\gamma_V}{2}|v|_V^2$ with a geometrico-functional term
\begin{equation}
H^{geo-fun}_r(\x,\f,\p^f)\doteq\frac{1}{2\gamma_f}(D^{-1}(\x)\p^f|\p^f)\,.\label{eq:FB35}
\end{equation}
Hence the shooting equations contains a new source term which is $\partial_{\x}H_r^{geo-fun}$ in the dynamic of the momenta $\p$:
\begin{equation}
  \label{eq:FB23}
  \dot{\p}=-\partial_{\x}H^{geo}_r(\x,\p)-\underbrace{\partial_{\x}H^{geo-fun}_r(\x,\f,\p^f)}_{\text{new source term}}
\end{equation}
and since the $H^{geo-fun}_r$ does not depend on $\f$, we have 
\begin{equation}
  \label{eq:FB24}
  \dot{\p^f}=-\partial_{\f}H^{geo-fun}_r(\x,\f,\p^f)=0\,.
\end{equation}
In particular, the functional speed $\dot{f}_k$ is only modulated by the evolution of the local weight $D(\x)_{k,k}$ which depends in the continuous limit both on the divergence of the ``tangential'' part of $v$ to the manifold spanned by the $x_k$'s and on its normal part when the mean curvature is non vanishing.

\subsection{Shooting equations for the tangential setting}
\label{Dyn_tangential2}
In the tangential model of section \ref{ssection:tang_model}, equations are simplified once again. Indeed, the $L^2$ metric being frozen to the initial position of $X$, the second term in \eqref{eq:optimal_control_discrete} becomes $\frac{\gamma_f}{2}\int_0^1 (D(\x_0)\h_t|\h_t)dt$ where $D(\x_0)$ is now a constant weighting matrix and the optimal functions $\h_t$ for fixed time 1 end point condition $\f_1$ are simply given by $\h_t\equiv\bzeta=\f_1-\f$, \ie $\f_t=\f+t\bzeta=t\f_1+(1-t)\f$ where $\bzeta = (\zeta(x_k))_{1\leq k\leq \nP}$ and $\zeta$ is the residual introduced in Section \ref{sec:atlas_tangential_setting}. In that situation, the new source term disappears in \eqref{eq:FB23} and the dynamic of $(\x,\p)$ is described by the usual Hamiltonian equations in the purely geometrical case, which gives:
\begin{equation}
\label{eq:forward_discrete_tang}
  \left\{
  \begin{array}[h]{rcl}
	  \dot{\x} & = & \phantom{-}\partial_{\p} H_r^{geo}(\x,\p) \\
	  \dot{\p} & = & -\partial_{\x} H_r^{geo}(\x,\p) \\
    \f_t & = & \f+t\bzeta
  \end{array} \right.
\end{equation}
These are the forward equations in the tangential model. Compared to metamorphoses, the dynamic on the signal part can be expressed in closed form while the evolution in spatial position is described by the usual Hamiltonian system without extra term. This makes it particularly simple for implementation, which shall be exploited in Section \ref{part:algo}. Optimizing over trajectories can be reduced to optimization with respect to initial momentum $p_0$ and signal $\f_1$ (or the residual $\bzeta$) since again $v(\cdot)=\frac{1}{\gamma_V}\sum_{k=1}^{\nP} K_V(\cdot,x_{k})p_{k}$, $\h=\bzeta$ and
\begin{equation}
  \label{eq:8.4.2}
  \frac{\gamma_V}{2}\int_0^1|v_t|_V^2dt+\frac{\gamma_f}{2}\int_0^1 (D(\x_0)\h_t|\h_t)dt=\rst{H_r^{geo}((\x,\p))}{t=0}+ \frac{\gamma_f}{2}(D(\x_0)\bzeta|\bzeta)\,.
\end{equation}
Variation of functionals with respect to $\p_0$ can be obtained by the usual \textbf{backward} integration of the adjoint Hamiltonian system ({\it c.f.} \cite{Charon_thesis} Chapter 1 or \cite{arguillere14:_shape} Section 4).    

\section{Algorithms to compute mean template of fshapes}\label{part:algo}

We present two different algorithms to compute a mean template from a sample of $N$ discrete fshapes. Both methods consist in solving a variational problem via an adaptive gradient descent algorithm presented in Section \ref{part:optim}. As the optimization method is the same in both cases, we just need to describe how to compute the functional and its gradient. These methods are implemented in Matlab, Cuda and C and some numerical experiments are shown Section \ref{part:numerical} and \ref{part.traps}. 

\subsection{Hypertemplate and tangential setting}
\label{part:HT}

Algorithms \ref{algo:j0tan}, \ref{algo:dj0tan} and \ref{algo:optim} give a way to implement the atlas estimation in the tangential setting described in Section \ref{sec:atlas_tangential_setting} when $f_0=0$. The principle of the method is illustrated by Figure \ref{fig:HT}. The method consists in computing a \emph{reduced} version based on geodesic shooting and initial momenta of the functional $J^{tan}_0$ of equation \eqref{eq:FB45}
 with Algorithm \ref{algo:j0tan} and its gradient with Algorithm \ref{algo:dj0tan} and then plugging this in the optimization box given by Algorithm \ref{algo:optim}. 
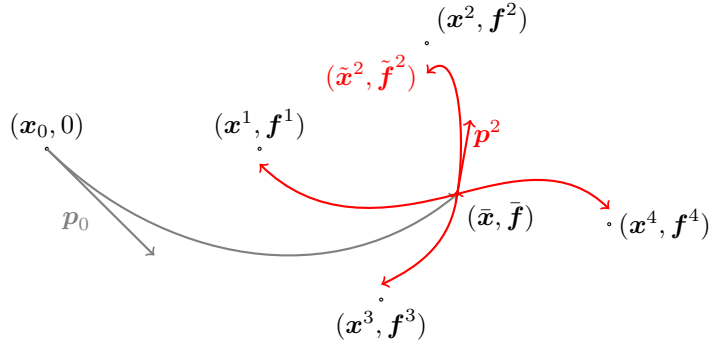
\begin{figure}[H]
	\begin{center}
\begin{tikzpicture}[scale=2]
	{\small
	\coordinate (ht) at (0,0)       ; 
	\coordinate (mt) at (2.7,-.3)     ; 
	\coordinate (dat1) at (1.4,0)   ;\coordinate (dat1p) at (1.4,-0.1)   ;
	\coordinate (dat2) at(2.5,0.7)  ;\coordinate (dat2p) at(2.5,0.5) ;
	\coordinate (dat3) at (2.2,-1)  ;\coordinate (dat3p) at (2.2,-.9) ;
	\coordinate (dat4) at (3.7,-0.5);\coordinate (dat4p) at (3.7,-0.4);

	\draw  (ht)  node[above]{$(\x_0,0)$}circle (.01);
	\draw  (mt)  node[below right]{$(\bar \x,\bar \f)$}circle (.01);
	\draw  (dat1)node[above]{$(\x^1_{},\f^1_{})$}circle (.01);
	\draw (dat2) node[above right]{$(\x^2_{},\f^2_{})$}circle (.01);
	\draw[white ] (dat2p) node[left, red]{$(\tilde \x_{}^2,\tilde \f_{}^2)$}circle (.01);
	\draw (dat3) node[below]{$(\x^3_{},\f^3_{})$}circle (.01);
	\draw  (dat4)node[right]{$(\x^4_{},\f^4_{})$}circle (.01);

	\draw (ht) edge[out=-45,in=-140,gray,->,thick] (mt) ;
	\draw (mt) edge[out=-97,in=25,red,->,thick] (dat3p);
	\draw (mt) edge[out=195,in=-45,red,->,thick] (dat1p);
	\draw[->,red,thick] (mt) edge[out=80,in=45](dat2p) ;
	\draw (mt) edge[out=15,in=135,red,->,thick] (dat4p);

	\draw[->,gray,thick] (ht) -- node[below left]{$\p_0$} +(-45:1);
	\draw[->,red,thick] (mt) -- node[above right]{$\p^2$} +(80:.5);

}
\end{tikzpicture}
	\end{center}
	\caption{The hypertemplate method: the estimated mean template $(\bar \x,\bar \f) = (\phi^{v^{\p_0}}_1(\x_0),\bar \f)$ is a deformed version of the hypertemplate $(\x_0,0)$.} \label{fig:HT}
\end{figure}

More precisely, the inputs are $N$ observations $(\x^{i}_{},\f^{i}_{})$, a hypertemplate $(\x_0,\f_0)$ with $\f_0=0$, the momenta $\p_0$ and $\p^1,\cdots,\p^N$ and the functionals $\f$ and $\bzeta^1,\cdots,\bzeta^N$. All the fshapes should be provided with their corresponding connectivity matrix. The optimization is conducted on the reduced version of $J^{tan}_0$ defined as
\begin{equation}
   \label{eq:FB45d}
   \begin{split}
	   \bJ^{tan}_{0}(\p_0,\f,(\p^i,\f^i)_{1\leq i\leq N})&\doteq  \frac{\gamma_{V_0}}{2}|v^{\p_0}|_{V_0}^2 +\frac{\gamma_f^0}{2}\sabs{\f}_{\x}^2\\
 & +\sum_{i=1}^N\left(
 \frac{\gamma_{V}}{2}|v^{\p^i}|_{V}^2+\frac{\gamma_f}{2}\sabs{\bzeta^i}^2_{\x} + \frac{\gamma_W}{2} g((\tilde{\x}^i,\tilde{\f}^i),(\x^i,\f^i))\right)
    \end{split}
\end{equation}
in the variables
\begin{equation}
 \label{eq:FB45bd}
 \begin{cases}
   (\p_0,\f)\in E^P\times \mathbb{R}^P\\
   (\p^i,\bzeta^i)\in E^P\times \mathbb{R}^P\\
 \text{with }\\
 \x=\phi^{v^{\p_0}}_1.\x_0,\,\tilde{\x}^i=\phi^{v^{\p^i}}_1.\x,\ \tilde{\f}^i_k=\f_{}+\bzeta^i_{}
 \end{cases}
 \end{equation}
 where $|\boldsymbol{l}|_ {\x}^2=(D(\x)\boldsymbol{l}|\boldsymbol{l})$ is the discretization of the $L^2(X)$ norm for $\boldsymbol{l}=(\boldsymbol{l}_k) \in \R^{\nP}$ defined on the vertices $\x \in E^{\nP}$ as introduced in Section \ref{Dyn_tangential}. 

 The hypertemplate may be difficult to define and the choice of this initial guess may obviously affect the quality of the atlas estimation. Nevertheless, the outputs of the algorithms are found to be stable even for the non-trivial experiments shown in Figure \ref{fig:statues_templates}. We recommend generating a simple and smooth fshape with a meshing program. This hypertemplate should be topologically equivalent to the observations as much as possible. Using one of the   observations as the hypertemplate could be tempting but it may induce some bias in the estimated atlas as the final mean template may contain the same specific features as the chosen observation (recall that the template is a diffeomorphic deformation of the hypertemplate). 

 At the end of the minimization procedure, the outputs are: an estimation of the mean template $(\bar{\x}, \bar{\f}) = (\phi^{v^{\p_0}}_1.\x_0,\bar{\f})$, the momenta $(\p^i)_{1\leq i\leq N}$ and the functional residuals $(\bzeta^i)_{1\leq i \leq N}$  so that $(\tilde \x^i,\tilde \f^i) =( \phi^{v^{\p^i}}_1.\bar{\x},\bar{\f}+\bzeta^{i})$ is close to $(\x^i,\f^i)$ for any $i=1,\cdots,N$.  

\begin{algorithm}[H]
	\caption{Computation of the energy $\bJ^{tan}_{0}$ with formula \eqref{eq:FB45d} \label{algo:j0tan}}

	\begin{algorithmic}[1]
		\State {\bf Require:} A hypertemplate $(\x_{0},\f_0=0)$ and $N$ fshapes $(\x^{i},\f^{i})$.
		\State {\bf Inputs:} A momentum $\p_{0}$, a template signal $\f$ and $N$ momenta $\p^{i}_{}$ and functional residuals $\bzeta^{i}$.
	\Begin
		\State Deformation: compute $\x$ by forward integration of $(\x_0,\p_{0})$. 			
		\For {$i=1$ to N}
			\State Deformation: compute $ \tilde \x^{i}$ by forward integration of $(\x,\p^{i})$; compute the signal $\tilde \f^{i} \gets \f+\bzeta^{i}$.
			\State Fvarifold norm: compute the fvarifold representation $ \mu_{(\tilde \x^{i}, \tilde \f^{i})}$ and $\mu_{(\x^{i},\f^{i})}$; compute $g_i \gets \snorm{\mu_{(\tilde \x^{i},\tilde \f^{i})}- \mu_{(\x^{i},\f^{i})}}^2_{W'}$.
			\State Penalty terms: compute $\sabs{v^{\p^i}}^2_{V} $ and $\sabs{\bzeta^{i}}^2_{\x}$.
		\EndFor
		\State Penalty terms: compute $\sabs{v^{\p_0}}^2_{V_0} $ and $\sabs{\f^{}}^2_{\x}$.
	\End
	\State {\bf Outputs:} $\frac{\gamma_{V_0}}{2}\sabs{v^{\p_0}}^2_{V_0}+ \frac{\gamma_{f}^0}{2} \sabs{\f^{}}^2_{\x} + \sum_i \limits \frac{\gamma_{f}}{2} \sabs{\bzeta^{i}}^2_{\x} +  \frac{\gamma_{V}}{2} \sabs{v^{\p^i}}^2_{V}+ \frac{\gamma_W}{2} g_i $.
	\end{algorithmic}
\end{algorithm}

\begin{algorithm}[H]
	\caption{Computation of the gradient $\nabla \bJ^{tan}_{0}$ \label{algo:dj0tan}}
	\begin{algorithmic}[1]
		\State {\bf Require:} A hypertemplate $(\x_{0}\f_{0})$ and $N$ fshapes $(\x^{i},\f^{i})$.
		\State {\bf Inputs:} A momentum $\p_{0}$, a template signal $\f_{}$ and $N$ momenta $\p^{i}$ and functional residuals $\bzeta^{i}$.
		\Begin
			\State Deformation: compute $\x_{0}$ by forward integration of $(\x_0,\p_{0})$.
			\For {$i=1$ to N}
				\State Deformation: compute $\x$ by forward integration of $(\x,\p^{i})$; compute the signal $\tilde \f^{i} \gets \f+\bzeta^i$.
				\State Gradient of $g_i$ wrt $\f,\bzeta^i,\p^i$ and $\x$: compute directly $\nabla_{\f} g_i$ and $\nabla_{\bzeta^i} g_i$; compute $(\nabla_{\x} g_i,\nabla_{\p^i} g_i)$ by backward integration of $(\nabla_{\tilde \x^i} g_i,0)$.
				\State Gradient of penalty terms: compute directly $\nabla_{\x} \sabs{v^{\p^i}}^2_{V}$, $\nabla_{\p^i} \sabs{v^{\p^i}}^2_{V}$, $\nabla_{\x} \sabs{\bzeta^{i}}^2_{\x}$ and $\nabla_{\bzeta^i} \sabs{\bzeta^{i}}^2_{\x}$.
			\EndFor
			\State Gradient of $\bJ^{tan}_{0}$ wrt $\p_0$: compute directly $\nabla_{\p_0} \sabs{v^{\p_0}}^2_{V_0}$; compute $\nabla_{\p_0} \sum_i \limits (\frac{\gamma_f}{2}\sabs{\bzeta^{i}}^2_{\x}+ \frac{\gamma_V}{2}\sabs{v^{\p^i}}^2_{V}+\frac{\gamma_W}{2}g_i)$ by backward integration of $\left( \sum_i \limits \nabla_{\x} (\frac{\gamma_f}{2}\sabs{\bzeta^{i}}^2_{\x}+ \frac{\gamma_V}{2}\sabs{v^{\p^i}}^2_{V})+ \frac{\gamma_W}{2}g_i,0 \right)$.
		\State Gradient of penalty term: compute $\nabla_{\f} \sabs{f}^2_{\x}$.
\End
\State {\bf Outputs:} 
$\nabla_{\p_0} \bJ^{tan}_{0} = \frac{\gamma_{V_0}}{2}\nabla_{\p_0} \sabs{v^{\p_0}}^2_{V_0} +\nabla_{\p_0} \sum_i \limits (\frac{\gamma_f}{2}\sabs{\bzeta^{i}}^2_{\x}+ \frac{\gamma_V}{2}\sabs{v^{\p^i}}^2_{V}+ \frac{\gamma_W}{2}g_i) $; 
$\nabla_{\f} \bJ^{tan}_{0}=\frac{\gamma_{f}^0}{2}\nabla_{\f} \sabs{\f}^2_{\x}+\frac{\gamma_W}{2}\nabla_{\f} \sum_i \limits g_i$; 
$\nabla_{\p^i} \bJ^{tan}_{0} = \frac{\gamma_V}{2}\nabla_{\p^i} \sabs{v^{\p^i}}^2_{V}+ \frac{\gamma_W}{2}\nabla_{\p^i} g_i $ for $i=1,..N$; 
$\nabla_{\bzeta^i} \bJ^{tan}_{0} =\frac{\gamma_f}{2} \nabla_{\bzeta^i} \sabs{\bzeta^{i}}^2_{\x} +\frac{\gamma_W}{2}\nabla_{\bzeta^i} g_i$ for $i=1,..N$.
	\end{algorithmic}
\end{algorithm}

\subsection{``Free'' mean fshape and tangential model}
\label{part:AlgoFree}

Algorithms \ref{algo:Jfreetan}, \ref{algo:dJfreetan} and \ref{algo:optim} use a more direct approach to compute a mean fshape compared to the hypertemplate method of Section \ref{part:HT}. The general framework is similar but the mean template is no longer defined as a deformed version of an hypertemplate. The coordinates of the points composing the mean template will be directly updated along the gradient descent. Notice that there is no extra cost as, in Algorithm \ref{algo:dj0tan}, we compute the gradient $\nabla_{\x} g_i$ in the backward integration step (see also Section \ref{part:bd}). The aim of this method is simply to use this gradient directly to update the mean fshape thereby making the evolution ``free'' along the optimization process. Figure \ref{fig:Free} gives an illustration of the method.

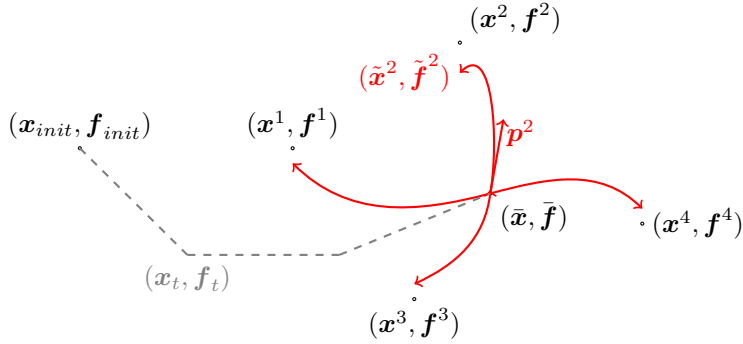
\begin{figure}[H]
	\begin{center}
\begin{tikzpicture}[scale=2]
	{\small
	\coordinate (ht) at (0,0); 
	\coordinate (mt) at (2.7,-.3)     ; 
	\coordinate (dat1) at (1.4,0)   ;\coordinate (dat1p) at (1.4,-0.1)   ;
	\coordinate (dat2) at(2.5,0.7)  ;\coordinate (dat2p) at(2.5,0.5) ;
	\coordinate (dat3) at (2.2,-1)  ;\coordinate (dat3p) at (2.2,-.9) ;
	\coordinate (dat4) at (3.7,-0.5);\coordinate (dat4p) at (3.7,-0.4);

	\draw  (ht)  node[above]{$(\x_{init},\f_{init})$}circle (.01);
	\draw  (mt)  node[below right]{$(\bar \x,\bar \f)$}circle (.01);
	\draw  (dat1)node[above]{$(\x^1_{},\f^1_{})$}circle (.01);
	\draw (dat2) node[above right]{$(\x^2_{},\f^2_{})$}circle (.01);
	\draw[white ] (dat2p) node[left, red]{$(\tilde \x_{}^2,\tilde \f_{}^2)$}circle (.01);
	\draw (dat3) node[below]{$(\x^3_{},\f^3_{})$}circle (.01);
	\draw  (dat4)node[right]{$(\x^4_{},\f^4_{})$}circle (.01);


	\draw (mt) edge[out=-97,in=25,red,->,thick] (dat3p);
	\draw (mt) edge[out=195,in=-45,red,->,thick] (dat1p);
	\draw[->,red,thick] (mt) edge[out=80,in=45](dat2p) ;
	\draw (mt) edge[out=15,in=135,red,->,thick] (dat4p);

	\draw[gray,thick,dashed] (ht) -- (0.707106781,-0.707106781) node[below]{$(\x_t,\f_t)$};
	\draw[gray,thick,dashed] (0.707106781,-0.707106781) --  (1.707106781,-0.707106781);
	\draw[gray,thick,dashed] (1.707106781,-0.707106781) --  (mt);
	
	\draw[->,red,thick] (mt) -- node[above right]{$\p^2$} +(80:.5);

}
\end{tikzpicture}
	\end{center}
	\caption{The ``free'' mean template method: the estimated mean template $(\bar \x, \bar \f)$ is a gradient-descent-based update of an initial fshape $(\x_{init},\f_{init})$. The notation $(\x_t,\f_t)$ symbolizes the state of the mean template at iteration $t$.} \label{fig:Free}
\end{figure}

The free mean fshape method is a gradient descent on the following functional:
\begin{equation}
	\bJ^{tan}_{free}((\x,\f),(\p^i,\f^i)_{1\leq i\leq N})\doteq  \frac{\gamma_f^0}{2}\sabs{\f}_{\x}^2 +\sum_{i=1}^N\left(\frac{\gamma_{V}}{2} \sabs{v^{\p^i}}_{V}^2+\frac{\gamma_f}{2}|\bzeta^i|^2_{\x} + \frac{\gamma_W}{2} g((\tilde{\x}^i,\tilde{\f}^i),(\x^i,\f^i))\right)
	\label{eq:free}
\end{equation}
in the variables
\begin{equation}
	\begin{cases}
	(\x,\f)\in E^{\nP} \times \R^{\nP},\\
	(\p^i,\bzeta^i)\in E^{\nP} \times  \R^{\nP}\\
	\text{with}\\
 \tilde{\x}^i=\phi^{v^{\p^i}}_1.\x,\ \tilde{\f}^i=\f+\bzeta^i
	\end{cases}
	\label{eq:freeVar}
\end{equation}
This algorithm does not have yet the same theoretical warranties as those given by Theorem \ref{theo:existence_atlas_fvar_tang} for the hypertemplate version. We are not able to show, with the same arguments, that there exists proper minimizers of the continuous counterpart of \eqref{eq:free} as the fshape $(X,f)$ is not constrained to belong to a single fshape bundle. However, we use in practice a regularized gradient descent to keep the evolution of the template smooth along the minimization. Therefore, if the size of the updates are small enough, the output of the free mean fshape algorithm may be considered as a discretization of a smooth deformation of $(X_{init},f_{init})$ as Property 4.2.1 of \cite{Charon_thesis} shows. This trick also prevents odd behaviours of the algorithm as discussed in Section \ref{part:bd}. Compared to the hypertemplate algorithm of previous section, this numerical scheme still reduces the constraint imposed on the template evolution which is likely to provide less dependency in the choice of initialization.    

The inputs are $N$ observations $(\x^{i}_{},\f^{i}_{})$, an initial fshape $(\x_{init},\f_{init})$, the momenta $\p^1,\cdots,\p^N$ and the functionals $\bzeta^1,\cdots,\bzeta^N$. All the fshapes should be provided with their respective connectivity matrix. The momenta $\p^1,\cdots,\p^N$ are usually initialized to 0 and $\bzeta^1,\cdots,\bzeta^N$ are usually initialized to a constant.

At the end of the minimization procedure, the outputs are: an estimation of the mean template $(\bar{\x},\bar{\f})$, the momenta $(\p^i)_{1\leq i\leq N}$ and the functional residuals $(\bzeta^i)_{1\leq i \leq N}$  so that $(\tilde \x^i,\tilde \f^i) = (\phi^{v^{\p^i}}_1.\bar{\x},\bar{\f}+\bzeta^{i})$ is close to $(\x^i,\f^i)$ for any $i=1,\cdots,N$.

\begin{algorithm}[H]
	\caption{Computation of the energy $\bJ^{tan}_{free}$ \label{algo:Jfreetan} of formula \eqref{eq:free}}

	\begin{algorithmic}[1]
		\State {\bf Require:} $N$ fshapes $(\x^{i},\f^{i})$.
		\State {\bf Inputs:}  A template fshape $(\x,\f)$, $N$ momenta $\p^{i}_{}$ and functional residuals $\bzeta^{i}$.
		\Begin
			\For {$i=1$ to N}
			\State Deformation: compute $\tilde \x^i$ by forward integration of $(\x,\p^{i})$; compute the signal $\tilde \f^{i}\gets \f+\bzeta^{i}$.
			\State Fvarifold norm: compute the fvarifold representation $\mu_{(\tilde \x^i,\tilde \f^i)}$ and $\mu_{(\x^i,\f^i)}$; compute $g_i \gets \snorm{\mu_{(\tilde \x^{i},\tilde \f^{i})}- \mu_{(\x^{i},\f^{i})}}^2_{W'}$.
			\State Penalty terms: compute $\sabs{v^{\p^i}}^2_{V} $ and $\sabs{\bzeta^{i}}^2_{\x}$.
		\EndFor
	\State Penalty term: compute $\sabs{\f^{}}^2_{\x}$.
	\End
	\State {\bf Outputs:} $\frac{\gamma_{f}^0}{2} \sabs{\f^{}}^2_{\x} + \sum_i \limits \frac{\gamma_{V}}{2} \sabs{v^{\p^i}}^2_{V}  + \frac{\gamma_{f}}{2} \sabs{\bzeta^{i}}^2_{\x} + \frac{\gamma_W}{2} g_i$.
	\end{algorithmic}
\end{algorithm}

\begin{algorithm}[H]
	\caption{Computation of the gradient $\nabla \bJ^{tan}_{free}$ \label{algo:dJfreetan}}
	\begin{algorithmic}[1]
		\State {\bf Require:} $N$ fshapes $(\x^{i},\f^{i})$.
		\State {\bf Inputs:} A template $(\x,\f)$, $N$ momenta $\p^{i}$ and functional residuals $\bzeta^{i}$.
		\Begin
		\For {$i=1$ to N}
		\State Deformation: compute $\x^i$ by forward integration of $(\x,\p^{i})$; compute the signal $\tilde \f^{i} \gets \f+\bzeta^i$
		\State Gradient of $g_i$ wrt $\f,\bzeta^i,\p^i$ and $\x$: Compute directly $\nabla_{\f} g_i$ and $\nabla_{\bzeta^i} g_i$; compute $(\nabla_{\x} g_i,\nabla_{\p^i} g_i)$ by backward integration of $(\nabla_{\tilde \x^i} g_i,0)$.
		\State Gradient of penalty terms: compute directly $\nabla_{\x} \sabs{v^{\p^i}}^2_{V}$, $\nabla_{\p^i} \sabs{v^{\p^i}}^2_{V}$, $\nabla_{\x} \sabs{\bzeta^{i}}^2_{\x}$ and $\nabla_{\bzeta^i} \sabs{\bzeta^{i}}^2_{\x}$.
	\EndFor
	\State Gradient of penalty term: compute directly $\nabla_{\f} \sabs{\f}^2_{\x}$ and $\nabla_{\x} \sabs{\f}^2_{\x}$.
\End
\State {\bf Outputs:} 
$\nabla_{\x} \bJ^{tan}_{free} = \frac{\gamma^0_f}{2}\nabla_{\x} \sabs{\f}^2_{\x} + \sum_i \limits \nabla_{\x} (\frac{\gamma_f}{2}\sabs{\bzeta^{i}}^2_{\x}+ \frac{\gamma_V}{2}\sabs{v^{\p^i}}^2_{V}+\frac{\gamma_W}{2}g_i) $; 
$\nabla_{\f} \bJ^{tan}_{free}=\frac{\gamma_{f}^0}{2}\nabla_{\f} \sabs{\f}^2_{\x}+\frac{\gamma_W}{2}\nabla_{\f} \sum_i \limits g_i  $; 
$\nabla_{\p^i} \bJ^{tan}_{free} = \frac{\gamma_V}{2}\nabla_{\p^i} \sabs{v^{\p^i}}^2_{V}+\frac{\gamma_W}{2} \nabla_{\p^i} g_i $ for $i=1,\cdots,N$; 
$\nabla_{\bzeta^i} \bJ^{tan}_{free} =\frac{\gamma_f}{2} \nabla_{\bzeta^i} \sabs{\bzeta^{i}}^2_{\x} +\frac{\gamma_W}{2}\nabla_{\bzeta^i} g_i $ for $i=1,\cdots,N$.
	\end{algorithmic}
\end{algorithm}

\subsection{Optimization} \label{part:optim}

An atlas estimation is a smooth, non-convex and high dimensional global optimization problem. We use an adaptive gradient descent algorithm to solve it numerically. This method is popular in geometric atlas estimation as it is quite robust and performs relatively well on various real examples, see \cite{Glaunes2004,Durrleman}. In our geometrico-functional framework, one main issue comes from the fact that the functionals to minimize depend on various types of variables living in different spaces: momenta, functional values or points coordinates. Adding the functional part in the optimization process makes it harder, compared to a pure geometrical approach, as some non desirable phenomenon may hold. For a thorough discussion, please see Section \ref{part:massCancel}.

Let us assume hereafter that we want to minimize a functional $\bJ(\bu_1,\cdots,\bu_K)$ depending on $K\geq 1$ types of variables. For instance, in an atlas estimation of fshapes with a hypertemplate (see Section \ref{part:HT}), we have $\bJ=\bJ^{tan}_{0}$ and $K=4$ (namely the variables $\f,\p,\f^i$ and $\p^i)$). 

If the $K$ types of variables are non homogeneous, we note that in practice, the different types of variables should not be updated at the same speed along the optimization process. It means that the gradient is not the proper descent direction to follow in order to reach a reasonable solution. As Algorithm \ref{algo:optim} shows, we use $K$ different steps sizes denoted $(\delta_i)_{1\leq i\leq K}$ and we adapt them separately. The quality of the final result will depend on a fine balance between the $\delta_i$'s. 

\begin{algorithm}[H]
	\caption{Adaptive gradient descent \label{algo:optim}}
	\begin{algorithmic}[1]
		\State {\bf Require:} Coefficients used to adapt the step sizes $0<s^-<1< s^+$.
		\State{\bf Inputs:} A functional $\bJ$ and its gradient $\nabla \bJ = \left( \nabla_{\bu_{i}} \bJ\right)_{i=1}^K$. Points $(\bu_{1}^{init},\cdots,\bu_{K}^{init})$ and step sizes $(\delta_{1}^{init},\cdots,\delta_{K}^{init})$.
		
	\Begin
			\State Initialize : $(\bu_1,\cdots,\bu_K) \gets (\bu_{1}^{init},\cdots,\bu_{K}^{init})$ and $(\delta_{1},\cdots,\delta_{K}) \gets (\delta_{1}^{init},\cdots,\delta_{K}^{init})$. 
		\Repeat
		\State Compute $J^{cur} \gets \bJ(\bu_{1},\cdots,\bu_{K})$ and $\nabla_{\bu_i} J^{cur} \gets \nabla_{\bu_i} \bJ^{}(\bu_{1},\cdots,\bu_{K})$.
		\State Update all the variables simultaneously : compute $J^{new} \gets \bJ^{}(\bu_{1}^{new},\cdots,\bu_{K}^{new})$ where $\bu_{i}^{new} \gets  \bu_{i}+\delta_{i} \nabla_{\bu_i} J^{cur}$ for $i=1,\cdots,K$.
		\State Adapt steps : {\bf if } {$J^{new}< J^{cur}$} {\bf then } $s \gets s^+$ {\bf else } $s \gets s^-$ {\bf end if}
		     \State Declare a boolean to break the following loop: breakLoop $\gets$ False.
		     \Repeat
		     \State Update each variable separately: for each $i=1,\ldots,K$ compute $\hat J^{new}_i \gets \bJ(\bu_{1}^{new}, \cdots, \hat{\bu}_{i}^{new},\cdots,\bu_K^{new})$ where $\hat \bu_i^{new} \gets \bu_{i} + s\delta_i \nabla_{\bu_i} J^{cur}$.
		\State Update every variable at the same time: $\hat J^{new}_{} \gets \bJ(\hat{\bu}_{1}^{new},\cdots,\hat{\bu}_K^{new})$.
		\State Keep the best configuration: let $(\bu_{1}^{*},\cdots,\bu_{K}^{*})$ and $(\delta_1^*,\cdots,\delta^*_K)$ be so that $(\bu_{1}^{*},\cdots,\bu_{K}^{*}) = (\bu_{1}^{} + \delta^*_1 \nabla_{\bu_1} J^{cur},\cdots,\bu_{K}^{} + \delta^*_K \nabla_{\bu_K}J^{cur})$ satisfies $\bJ(\bu_1^*,\cdots,\bu_K^*) =  \min \{J^{new},\hat{J}^{new}_1,\cdots,\hat{J}^{new}_K,\hat{J}_{}^{new}\}$.	     

		\If{$\bJ(\bu_1^*,\cdots,\bu_K^*) < J^{cur}$}
                	\State Update points and steps sizes: $(\bu_1,\cdots,\bu_K) \gets (\bu_1^*,\cdots,\bu_K^*) $ and $(\delta_1,\cdots,\delta_K) \gets (\delta_1^*,\cdots,\delta_K^*)$.
			\State breakLoop $\gets$ True.
		\Else 
			\State Decrease all steps sizes: $(\delta_1,\cdots,\delta_K) \gets (s\delta_1,\cdots,s\delta_K) $.
		\EndIf
		
		\Until{ ``breakLoop $=$ True'' or ``steps sizes are too small''.}

	\Until{``Maximum iteration is reached'' or ``descent is to small'' or ``steps sizes are too small''.}
	\End

	\State {\bf Output:} Points $(\bu_1,\cdots,\bu_K)$.
\end{algorithmic}
\end{algorithm}


Some tricks may be also used to improve the optimization strategy. For instance, it may be convenient to regularize some part of the gradient as discussed in Section \ref{part.traps}. Another efficient method is to change the fvarifold kernel widths $\sigma_e,\sigma_f,\sigma_t$ along the optimization process. In that case, the optimization involves several runs: a first gradient descent is performed at coarse scale by choosing large $\sigma_e,\sigma_f,\sigma_t$ and then a second gradient descent is performed with smaller $\sigma_e,\sigma_f,\sigma_t$ (with starting point equals to the end point of the first run) and so on.


\section{Numerical experiments}
\label{part:numerical}

\subsection{Synthetic dataset}

We now present a first set of results of the previous atlas estimation algorithms on a synthetic dataset of six textured statues constructed using Sculptris software. The set of subjects is shown on the last row of Figure \ref{fig:statues_matching+subjects}. The atlas is estimated using the free mean fshape tangential evolution algorithm of Section \ref{part:AlgoFree} after the dataset is preregistered with respect to translations. The template is initialized with a prototype fshape with zero-valued signal everywhere in order to avoid as much as possible bias toward one particular individual. Figure \ref{fig:statues_template_evol} shows the current template at several intermediate steps of the estimation. We use all Gaussians for the kernels $k_{e}$, $k_{t}$ and $k_{f}$ defining the data attachment terms as in equation \eqref{eq.GaussKern}. The template is obtained by refining twice the scales $\sigma_e$ and $\sigma_f$ throughout the process in order to have a coarse estimation of the atlas in the first 
place before being able to retrieve finer details from the dataset. The resulting template captures the average shape and signal behavior of the set of subjects: in particular, \textit{the algorithm tends to recover the signal patterns that are the most shared among the population}, as the uniform coloring of the head for instance.

\begin{figure}[H]
\centering
\begin{tabular}{ccccc}
  \includegraphics[height=4cm]{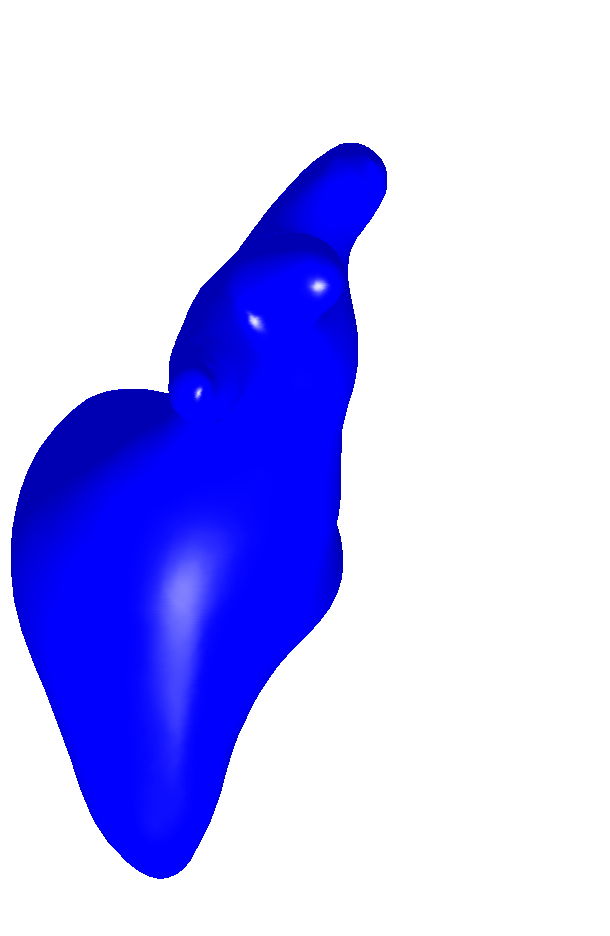} & \includegraphics[height=4cm]{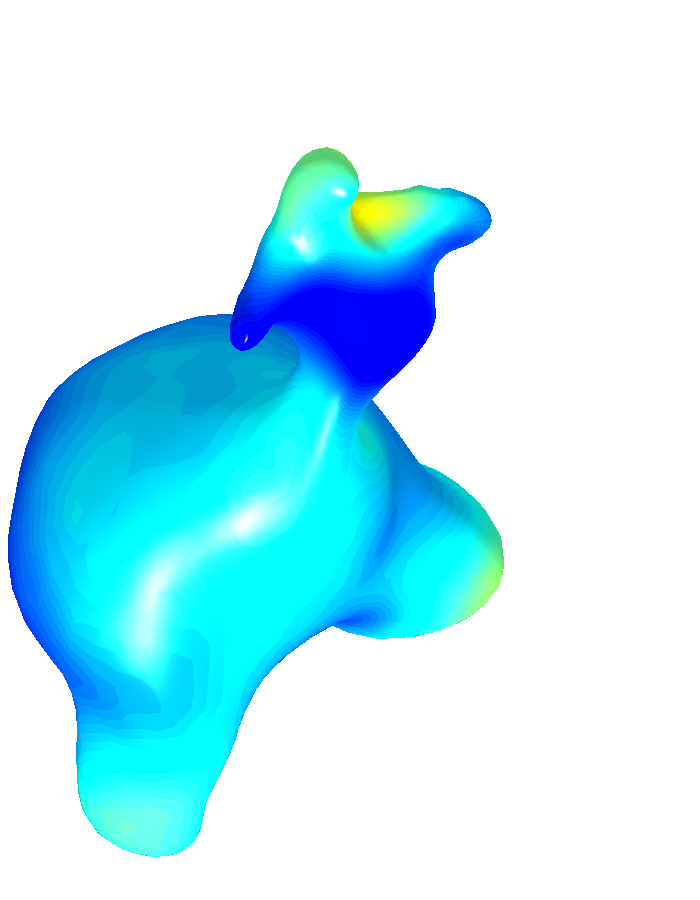} & \includegraphics[height=4cm]{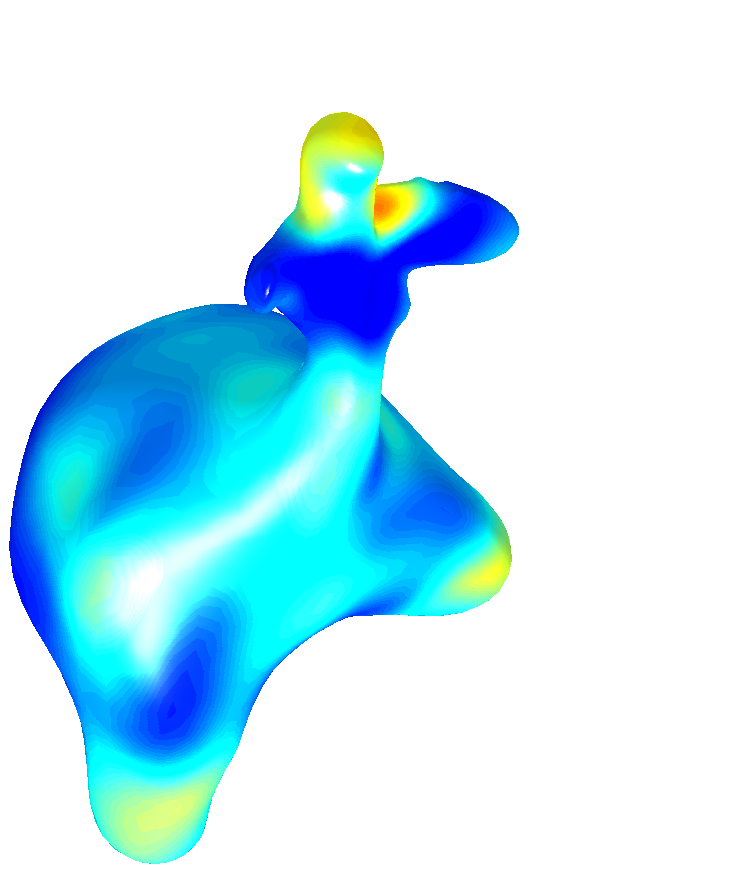} & \includegraphics[height=4cm]{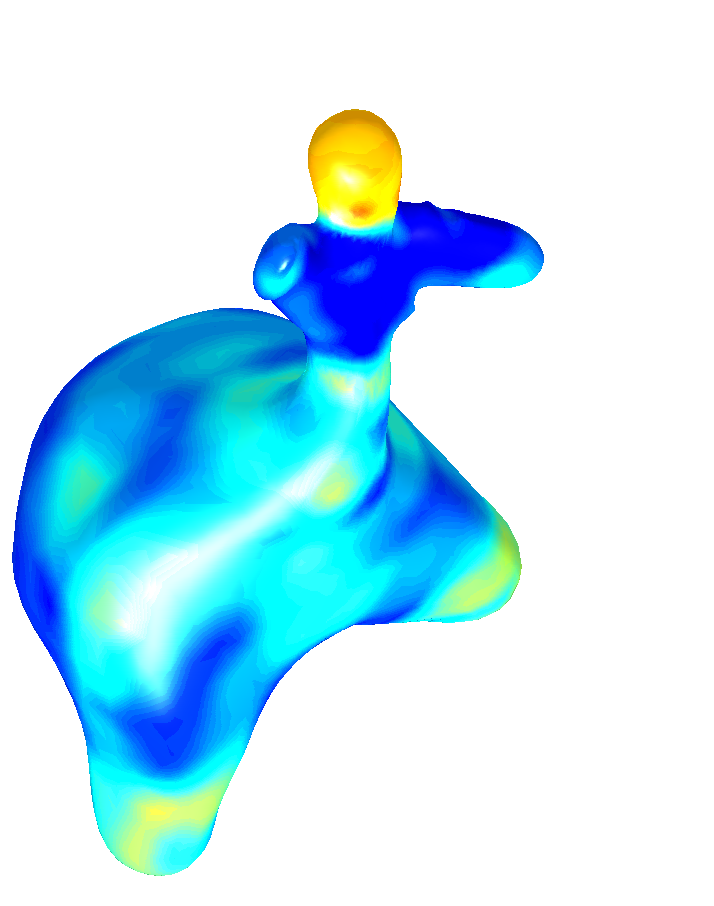} & \includegraphics[height=4cm]{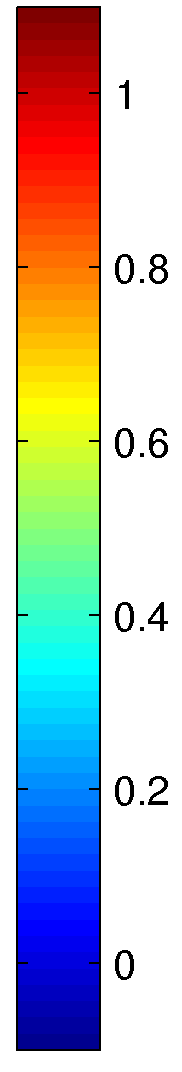} 
\end{tabular}
\caption{Template evolution along the gradient descent steps of the algorithm. On the left is the prototype fshape which serves as the initialization for the template. We use a multiscale approach that decreases, after a certain number of iterations, the characteristic sizes of the kernels on geometry and signal.}
\label{fig:statues_template_evol}
\end{figure}

In addition to the template, we obtain, at convergence, a set of residual signals as well as deformation momenta that map the template on the different subjects both in terms of shape and texture. The kernel for the space $V$ of deformation fields is taken as a sum of two Gaussian kernels, following the approach of \cite{Vialard2012} which allows to introduce multiscale deformations. We show two of such mappings in Figure \ref{fig:statues_match_2}. The whole set of shape and texture matchings compared to the original subjects is finally summed up in Figure \ref{fig:statues_matching+subjects}. In terms of numerics, the subjects of the dataset are unequally sampled between 2800 and 6500 points while the template of Figure \ref{fig:statues_template_evol} has 7000 points. Using our GPU implementation for kernel computations, each iteration of the algorithm takes approximately 38 secs for a total atlas estimation time of 3 hours (300 iterations) on a server equipped with a Nvidia GTX 555.

The dependency in the initialization of the template is also an important issue. Section \ref{part.invariances} explains how to make the previous procedures independent of rescaling or sampling changes of the fshapes. But obviously, in both the hypertemplate and 'free' template evolution algorithm, the shape of the template is still constrained to live in the diffeomorphic orbit of the initial one. It results that one cannot expect to remove completely the bias resulting from the initialization's choice. In Figure \ref{fig:statues_templates}, we show the template obtained after convergence for different initializations. Although these results do demonstrate some variations for the estimated template, they still show a quite remarkable stability in most of the important geometric and functional features of the dataset.

\begin{figure}[H]
\centering
\begin{tabular}{cccc}
  \includegraphics[width=3cm]{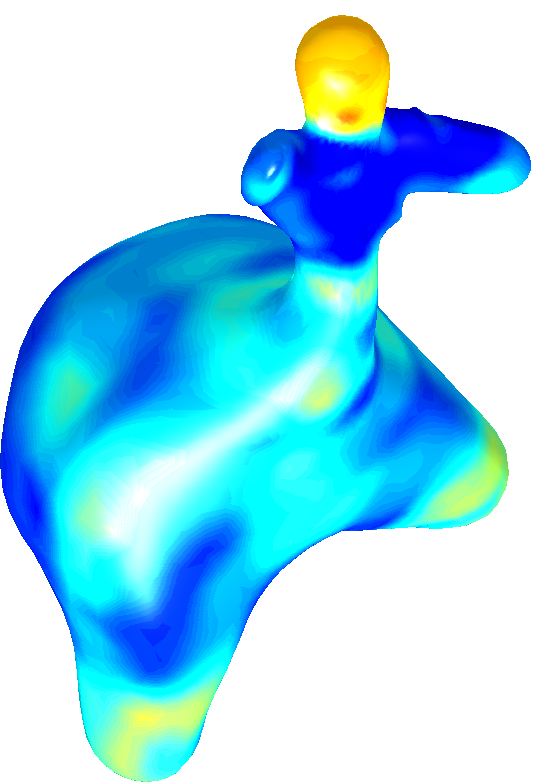} & \includegraphics[width=3cm]{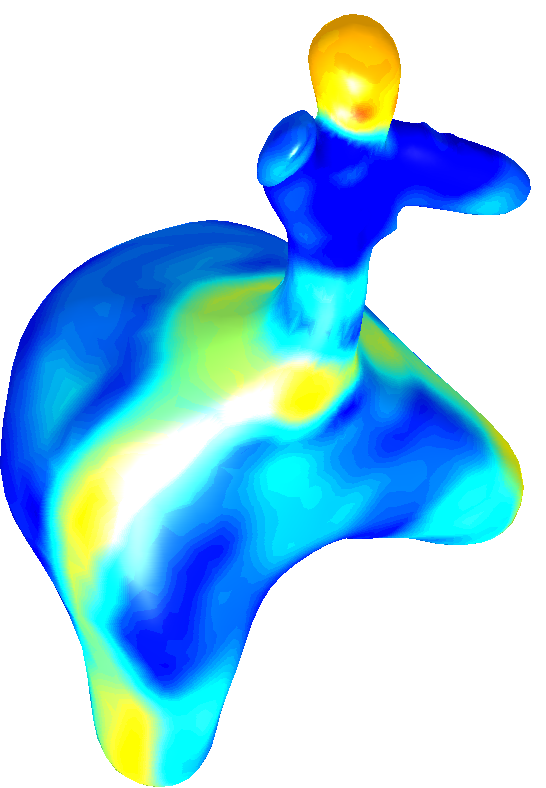} & \includegraphics[width=3cm]{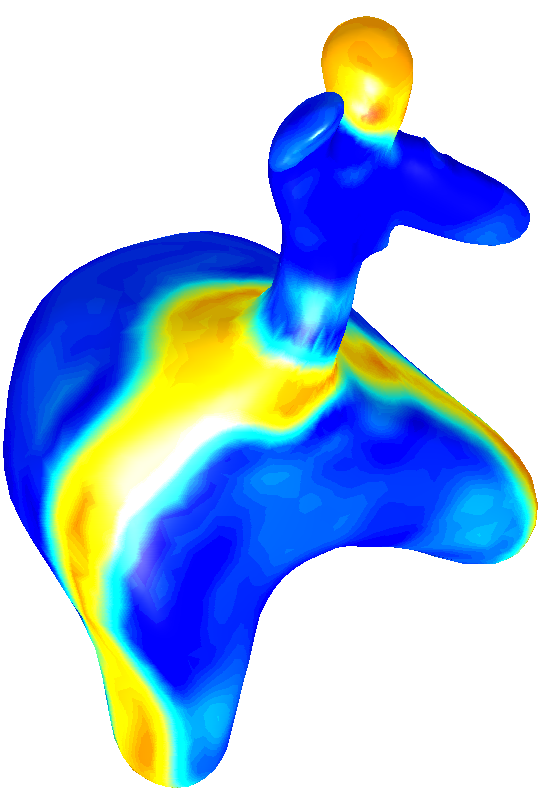} & \includegraphics[width=3cm]{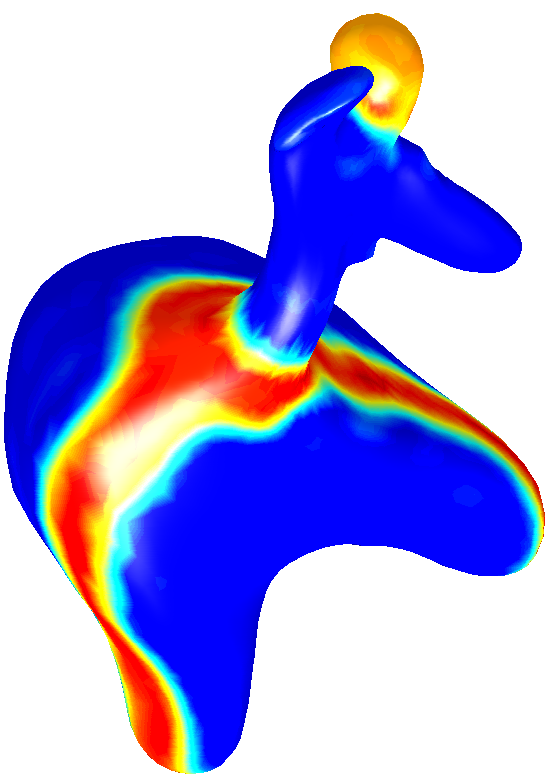} \\
  \includegraphics[width=3cm]{./figures/statues_template_final.png} & \includegraphics[width=3cm]{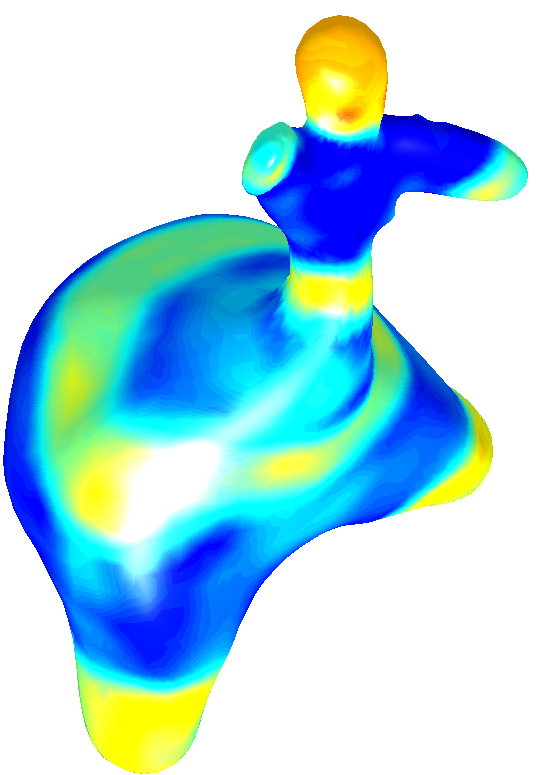} & \includegraphics[width=3cm]{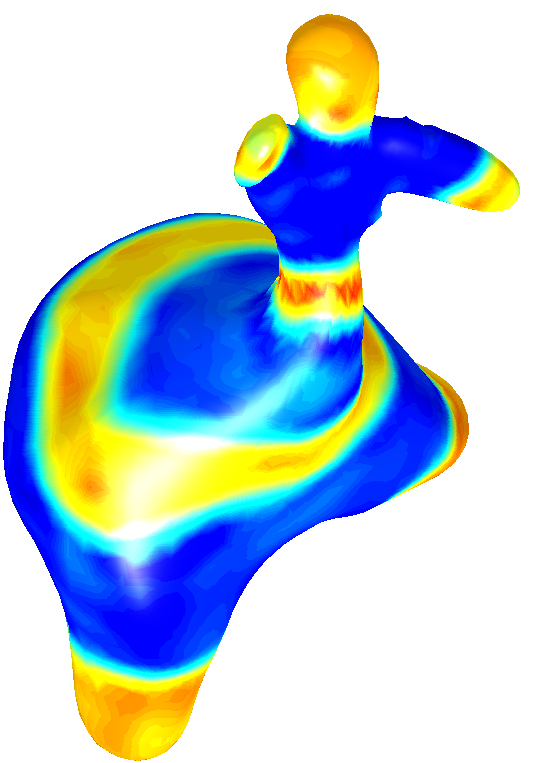} & \includegraphics[width=3cm]{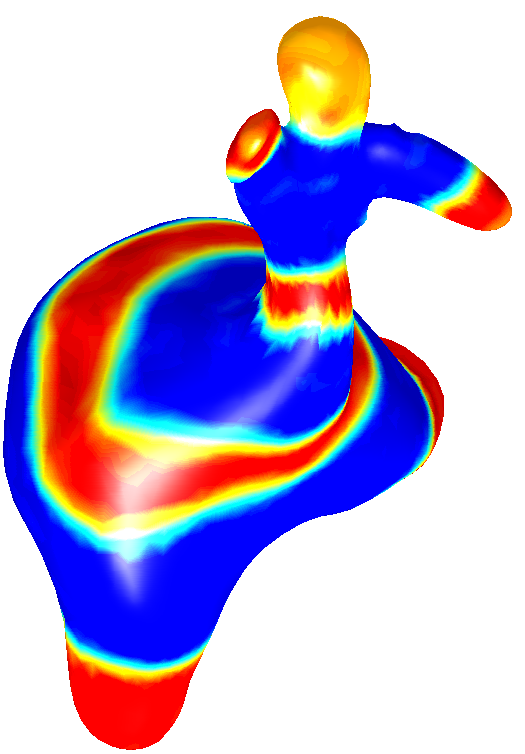} \\
  Template & t=1/3 & t=2/3 & t=1
\end{tabular}
\caption{Mapping of the template on subjects 2 and 5. From left to right, we display the template and several intermediate time steps of the geometrico-functional transformations, \ie $(\phi^{i}_{t}(\bar{\x}),\bar{\f}+\bzeta^{i}_t)$ where $\bzeta^{i}_{t}=t\bzeta^{i}$ as in formula \eqref{eq:forward_discrete_tang}.}
\label{fig:statues_match_2}
\end{figure}

\begin{figure}[H]
\centering
   \includegraphics[width=\textwidth]{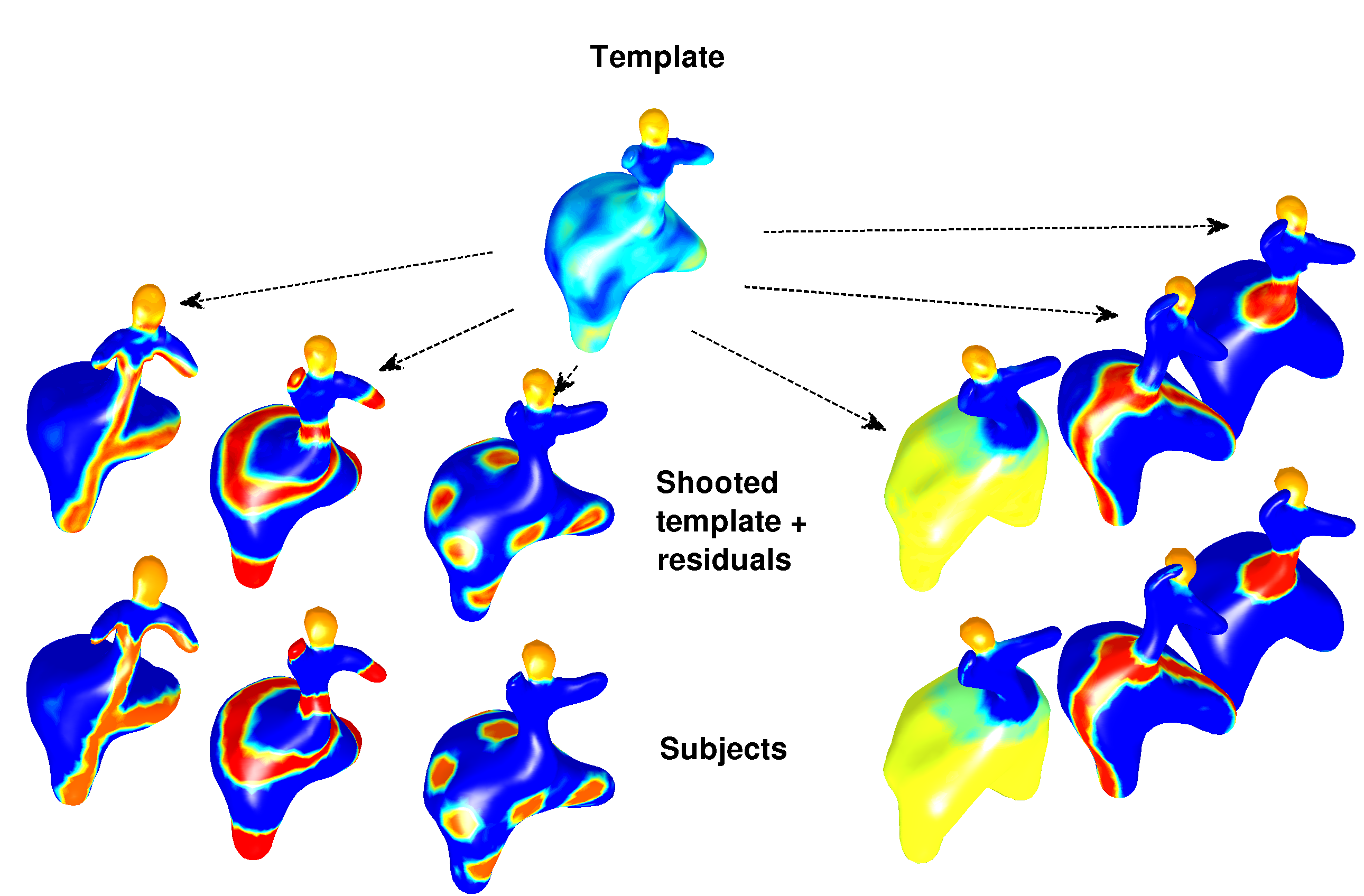} 
\caption{Estimated template and the geometrico-functional matchings to each subject.}
\label{fig:statues_matching+subjects}
\end{figure}

\begin{figure}
\centering
\begin{tabular}{cccc}
	Initialization & \parbox[c]{3.2cm}{\includegraphics[width=3cm]{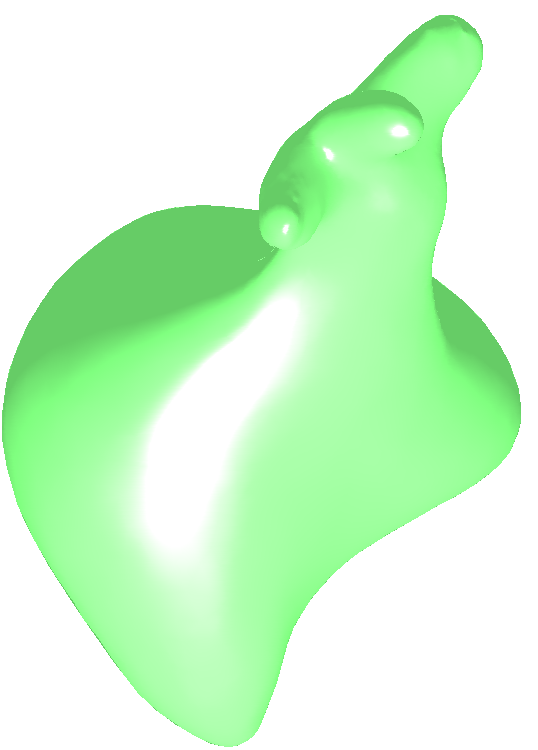}} &\parbox[c]{3.2cm}{ \includegraphics[width=3cm]{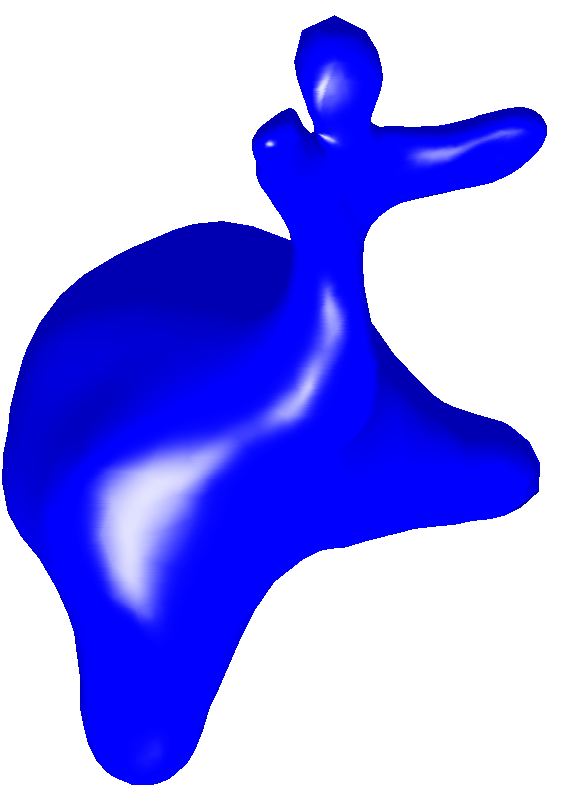}} &\parbox[c]{3.2cm}{ \includegraphics[width=3cm]{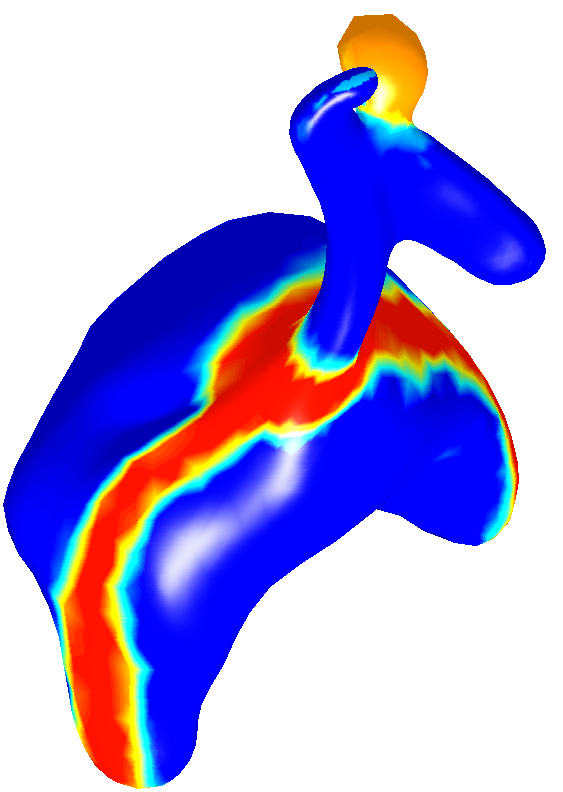}} \\
	Estimated template & \parbox[c]{3.2cm}{ \includegraphics[width=3cm]{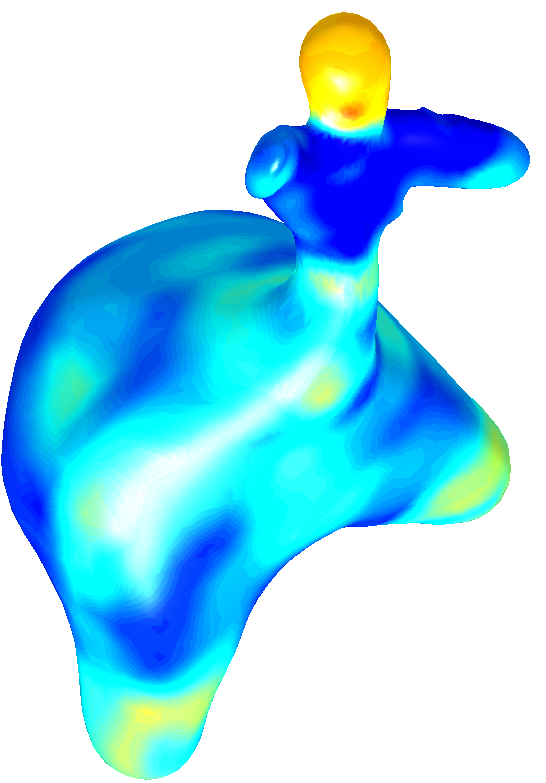}}  &\parbox[c]{3.2cm}{\includegraphics[width=3cm]{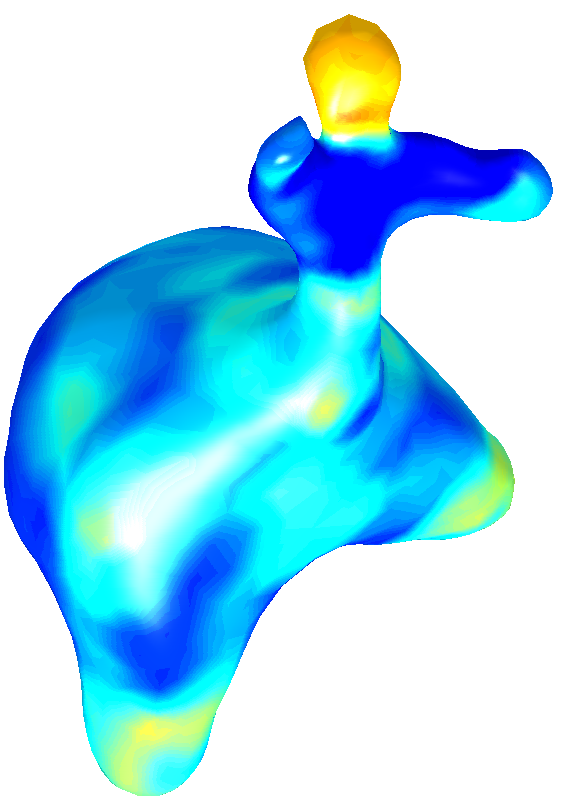}} & \parbox[c]{3.2cm}{\includegraphics[width=3cm]{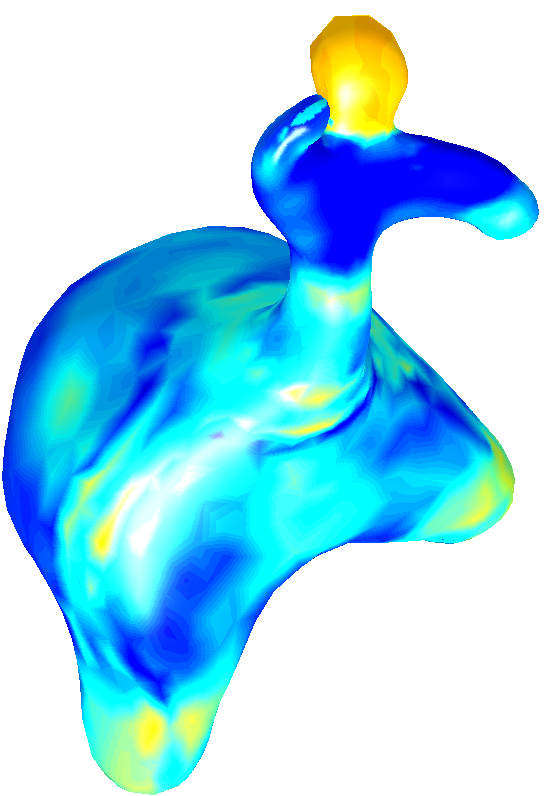}}
\end{tabular}
\caption{Variations of the final estimated template with respect to the initialization, respectively (from left to right) with the same prototype shape as in figure \ref{fig:statues_template_evol} but with constant 0.5 signal, with the shape of subject 4 and with subject 2.}
\label{fig:statues_templates}
\end{figure}

\subsection{OCT dataset}\label{part:OCT}

We now present a template estimation performed on a real dataset. The dataset contains surfaces segmented from volumetric images of the retina acquired by optical coherence tomography (OCT) as described in \cite{OCTseg,OCT}. The aim of these measurements is to detect early glaucoma by analysing changes of conformations of the inner layers of the retina around the optic disc (where the optic nerve slots into the retina). Data are fshapes:  the surfaces represent the lower boundary of the nerves fiber layer (NFL) and the signals represent the thickness of the NFL. A loss of thickness may be an indicator of glaucoma. We depict below two typical observations in two different views: Data 1 is a normal subject in Figures \ref{figData11}-\ref{figData12} and Data 2 is a glaucoma subject in Figure \ref{figData21}-\ref{figData22}. The ``view 1'' has the same scale for the 3 axis and ``view 2'' is a flipped version of ``view 1'' where the depth axis has been scaled ($\times 4$) to better represent fine reliefs. 
The typical size of square boundary ranges from $5$ to $7$mm and the signal ranges from $0$ to $0.3$mm. The color scale for the signal is the same for all the pictures of this Section. The overall geometry of the NFL boundary does not seem to be challenging but these surfaces contains many boundaries making the problem hard to handle in practice.  We discuss in Section \ref{part:bd} how we manage problems arising from boundary effects. Moreover, the difficulty is increased by the fact that some observations are not centered and the opening may be closed to the boundary of the acquisition area as in Figure \ref{fig:ResGlauc}. 

We ran our code, using the hypertemplate method described Section \ref{part:HT} on a dataset containing 51 observations (19 normal, 25 glaucomatous and 7 suspects). The observations were preregistered by hand with respect to translations so that the center of the opening is at the origin. We use downsampled data (raw data contains more than 130000 points and 270000 triangles each) and each observation contains about $5000$ points and $9850$ triangles. The hypertemplate is the flat rectangle with a hole and with a null signal depicted Figure \ref{figHT}. It contains $5700$ points and $11100$ triangles. Computations of the mean template $(\bar \x, \bar \f)$ (Figure \ref{figMT1} and \ref{figMT2}) and the deformations took 7 hours (120 iterations) using a server equipped with a Nvidia GTX 555 graphical processor unit.
\begin{figure}[H]
	\centering
	\begin{tabular}{cccc}
		\parbox[b]{.28\textwidth}{\subcaptionbox[.28\textwidth]{Relative positions of the hypertemplate and Data 1 and 2 (Figures \ref{fig:ResNormal} and \ref{fig:ResGlauc}) (view 1) \label{figHT}}{\includegraphics[width=.28\textwidth]{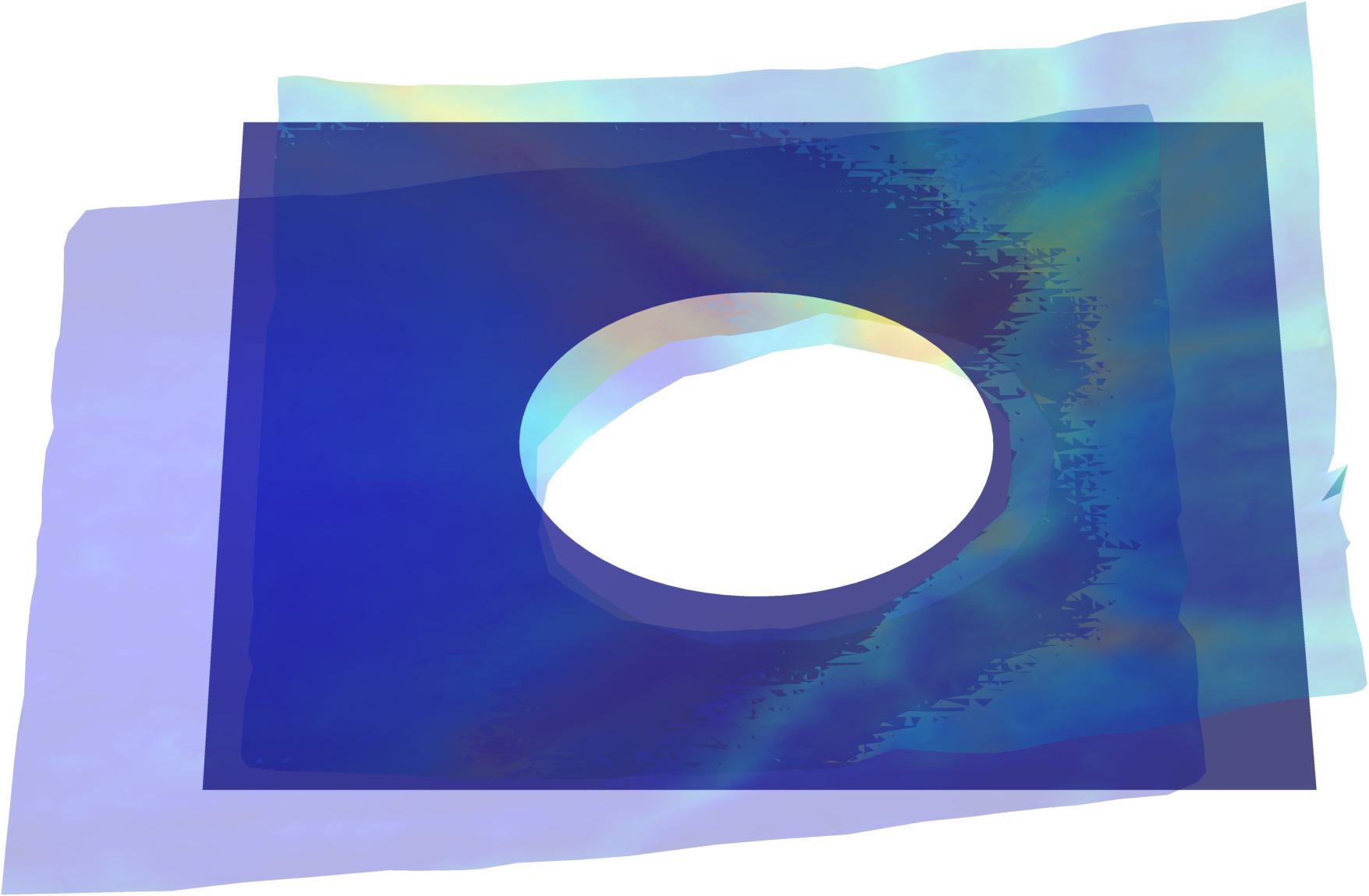}}}  &
		\parbox[b]{.28\textwidth}{\subcaptionbox[.28\textwidth]{Mean template $(\bar \x, \bar \f)$ (view 1)\label{figMT1}}{\includegraphics[width=.28\textwidth]{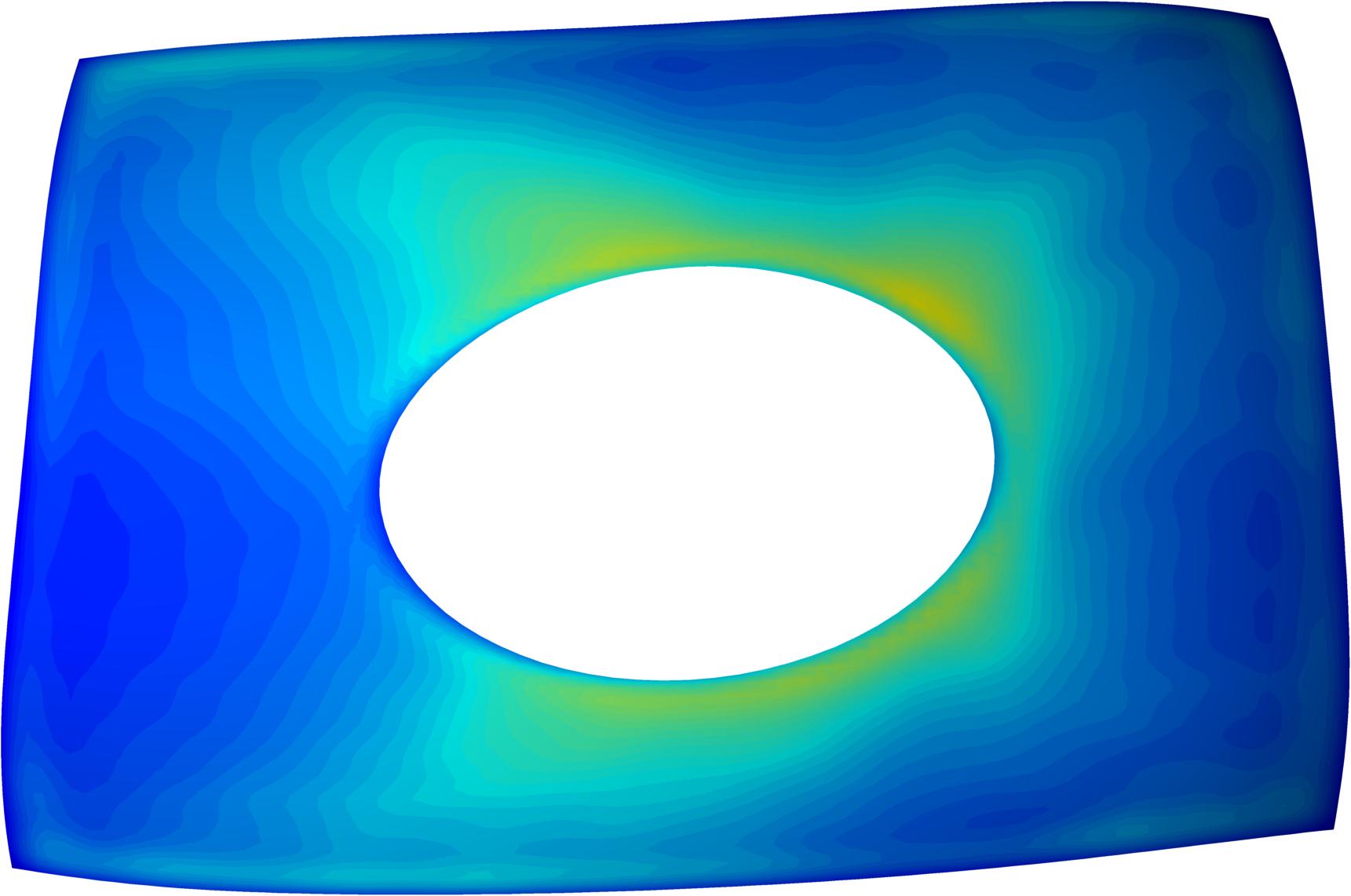}} } &
		\parbox[b]{.28\textwidth}{\subcaptionbox[.28\textwidth]{Mean template $(\bar \x, \bar \f)$ (view 2)\label{figMT2}}{\includegraphics[width=.28\textwidth]{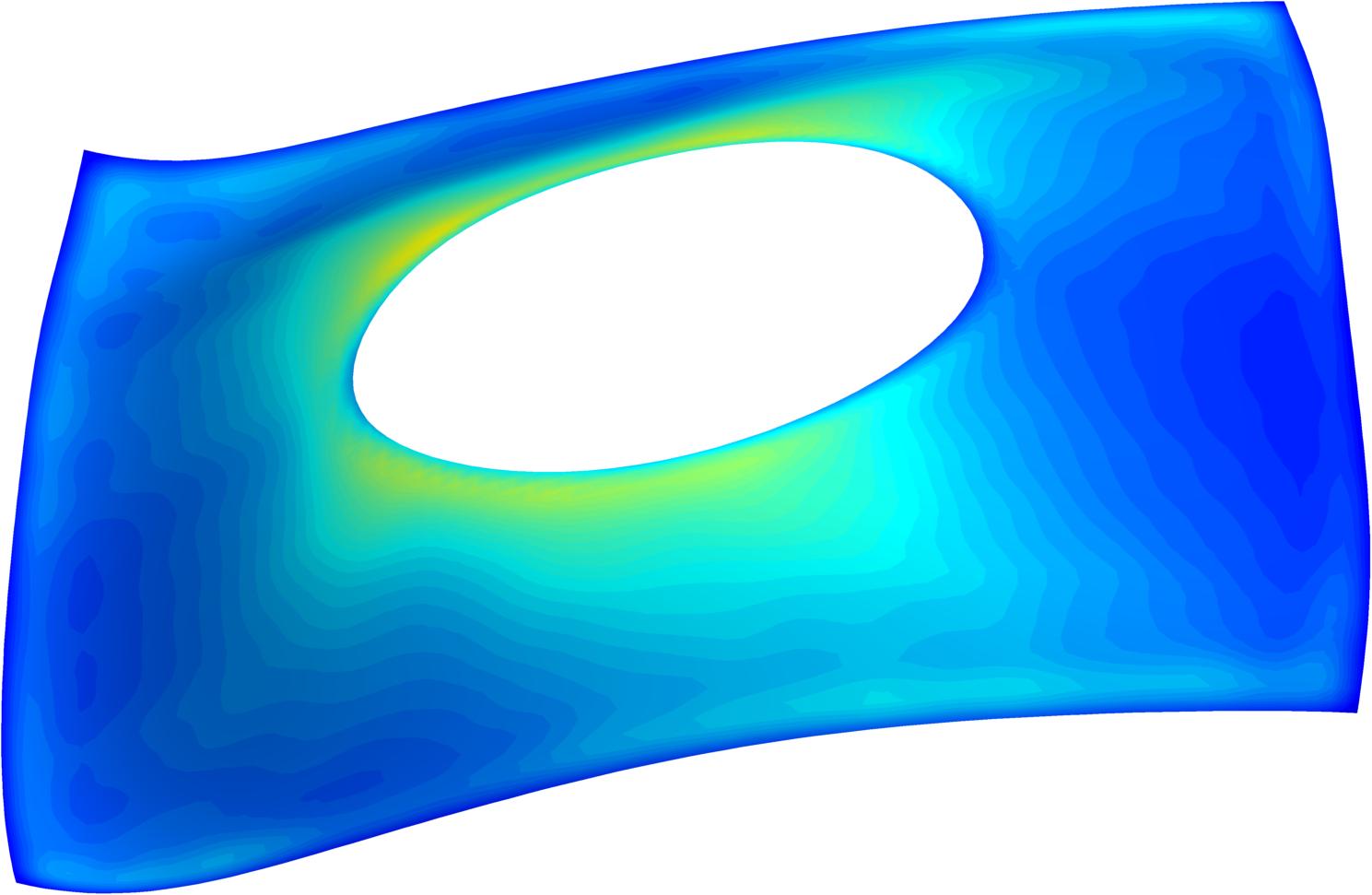}\hspace{1em}}} &
		\hspace{1em}\raisebox{-.3\height}{\includegraphics[height=5cm]{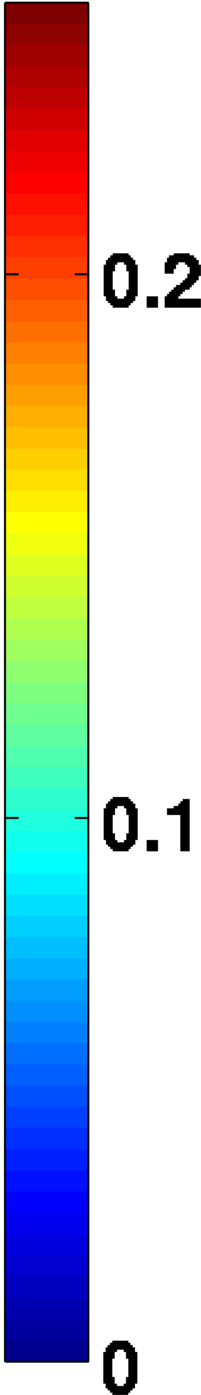}}
	\end{tabular}
		\caption{The hypertemplate and the estimated mean (computed with the full dataset containing 51 observations).\label{fig:OCTdataset}}
\end{figure}


The estimated template is depicted Figures \ref{figMT1} and  \ref{figMT2}. The average shape is a bended version of the flat hypertemplate which is a reasonable guess. The baseline signal $f$ captures a light but typical ``C''-shaped signal. This blurred signal is due to the rather high variability of signals across the dataset. The results of the deformations and functional residuals estimations for Data 1 and 2 are given by Figures \ref{fig:ResNormal} and \ref{fig:ResGlauc} respectively.  The functional part of all the 51 observations is well reconstructed and we are now able to compare these signals as they are all defined on the mean template, see Figure \ref{figSub10} and \ref{figSub20}. The deformations are also satisfying although most of the energy of the deformations is spent to match the (non informative) outer boundaries as data are misaligned, compare Figures \ref{figData10} and \ref{figData20}.
		\begin{figure}[H]
			\centering
			\subcaptionbox[.33\textwidth]{Relative position of data 1 and the mean template (view 1)\label{figData10}}{\includegraphics[width=.28\textwidth]{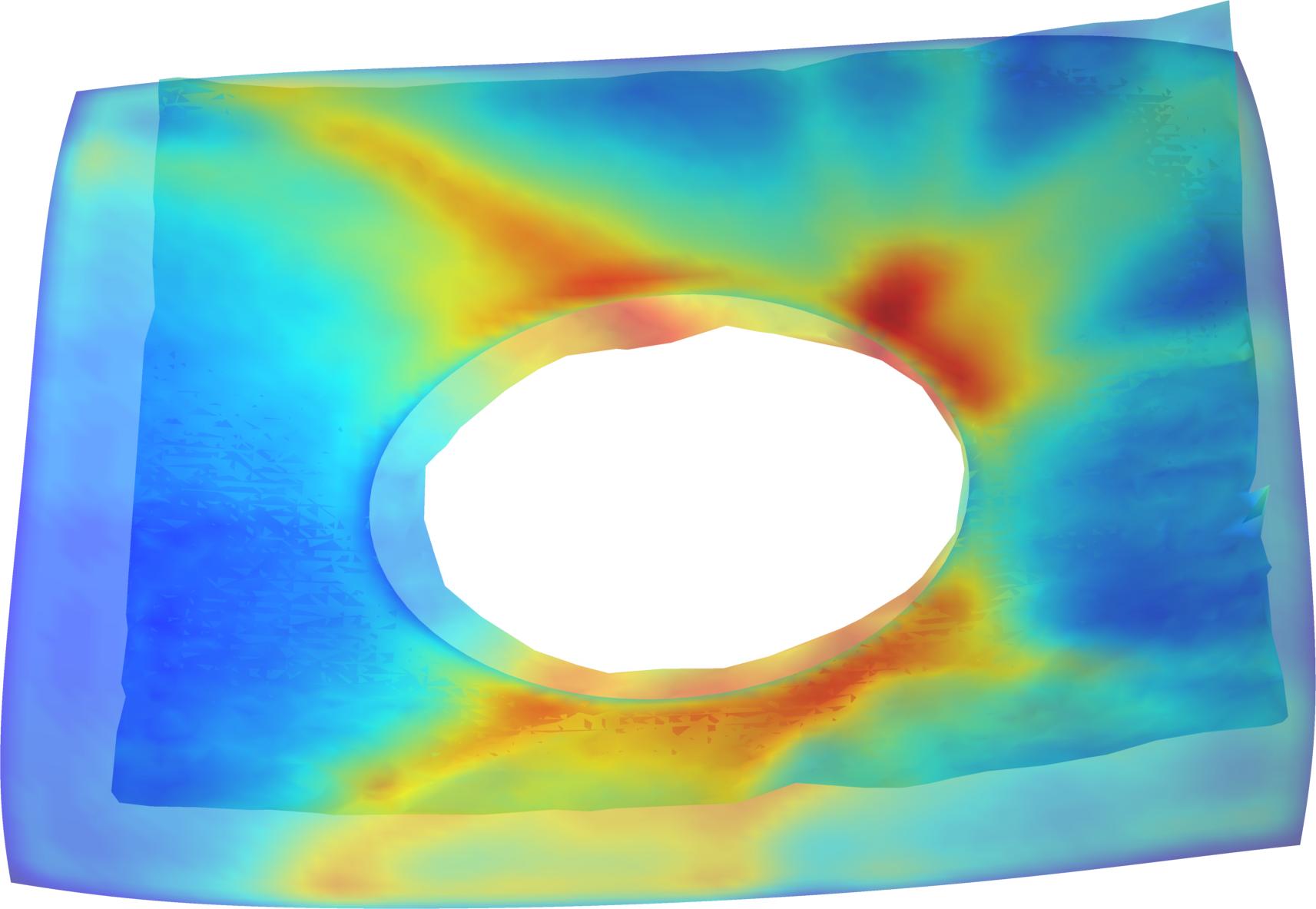}}\hfill
			\subcaptionbox[.33\textwidth]{Data 1 (view 1)\label{figData11}}{\includegraphics[width=.28\textwidth]{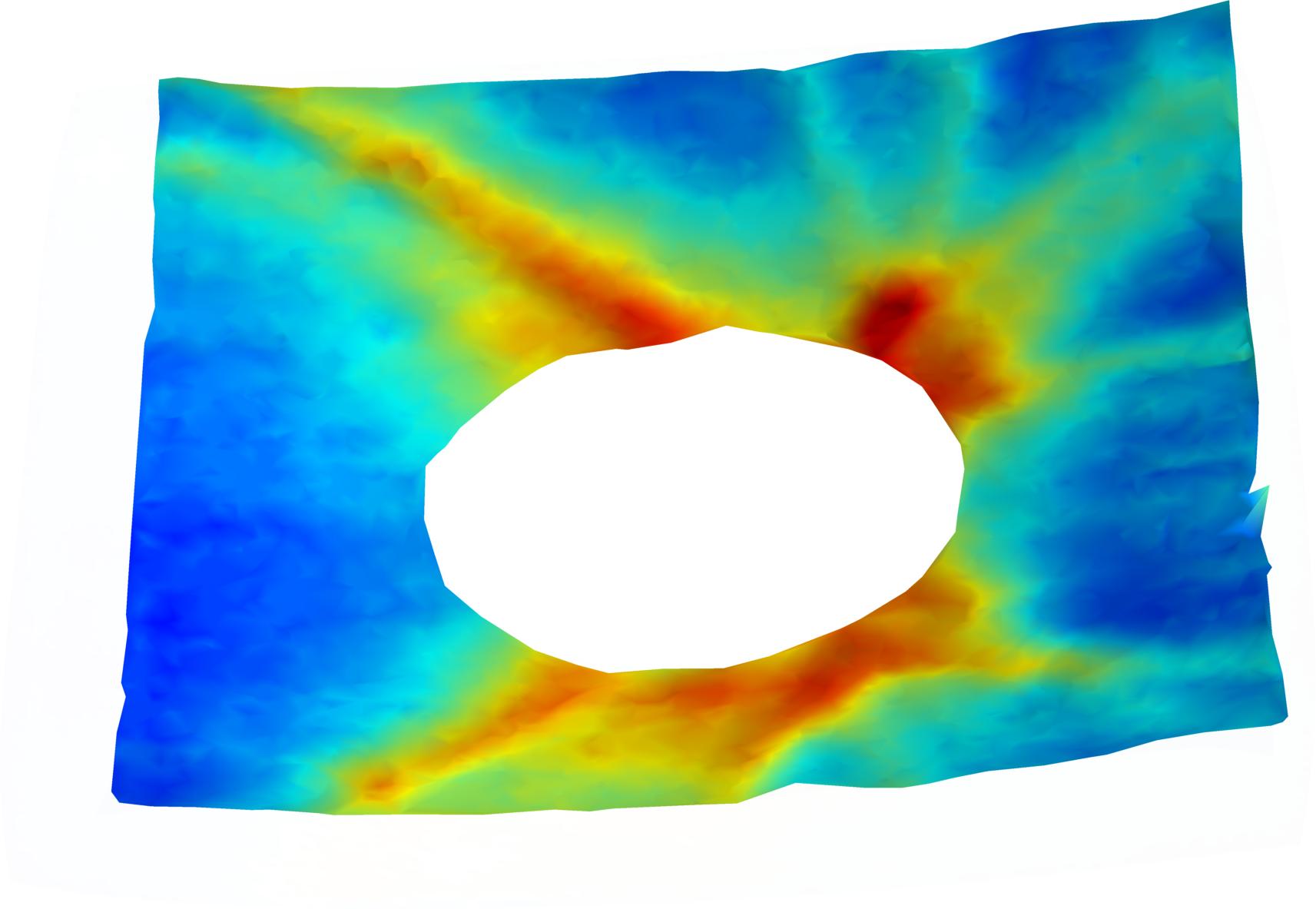}}\hfill
			\subcaptionbox[.33\textwidth]{Data 1 (view 2)\label{figData12}}{\includegraphics[width=.28\textwidth]{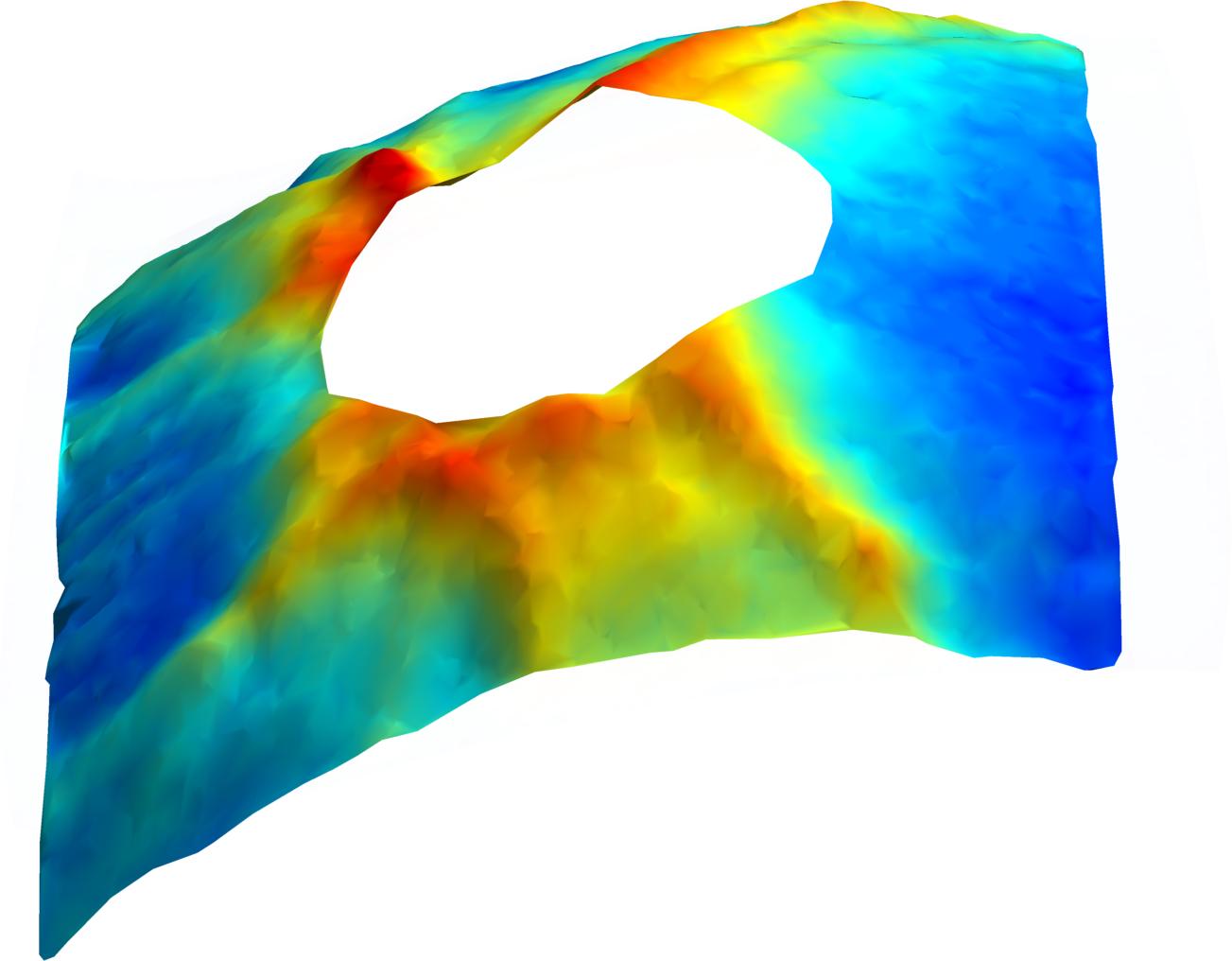}}\\
			\subcaptionbox[.33\textwidth]{Mean template and the residuals $(\bar \x, \bar \f + \bzeta^1)$ (view 1)\label{figSub10}}{\includegraphics[width=.28\textwidth]{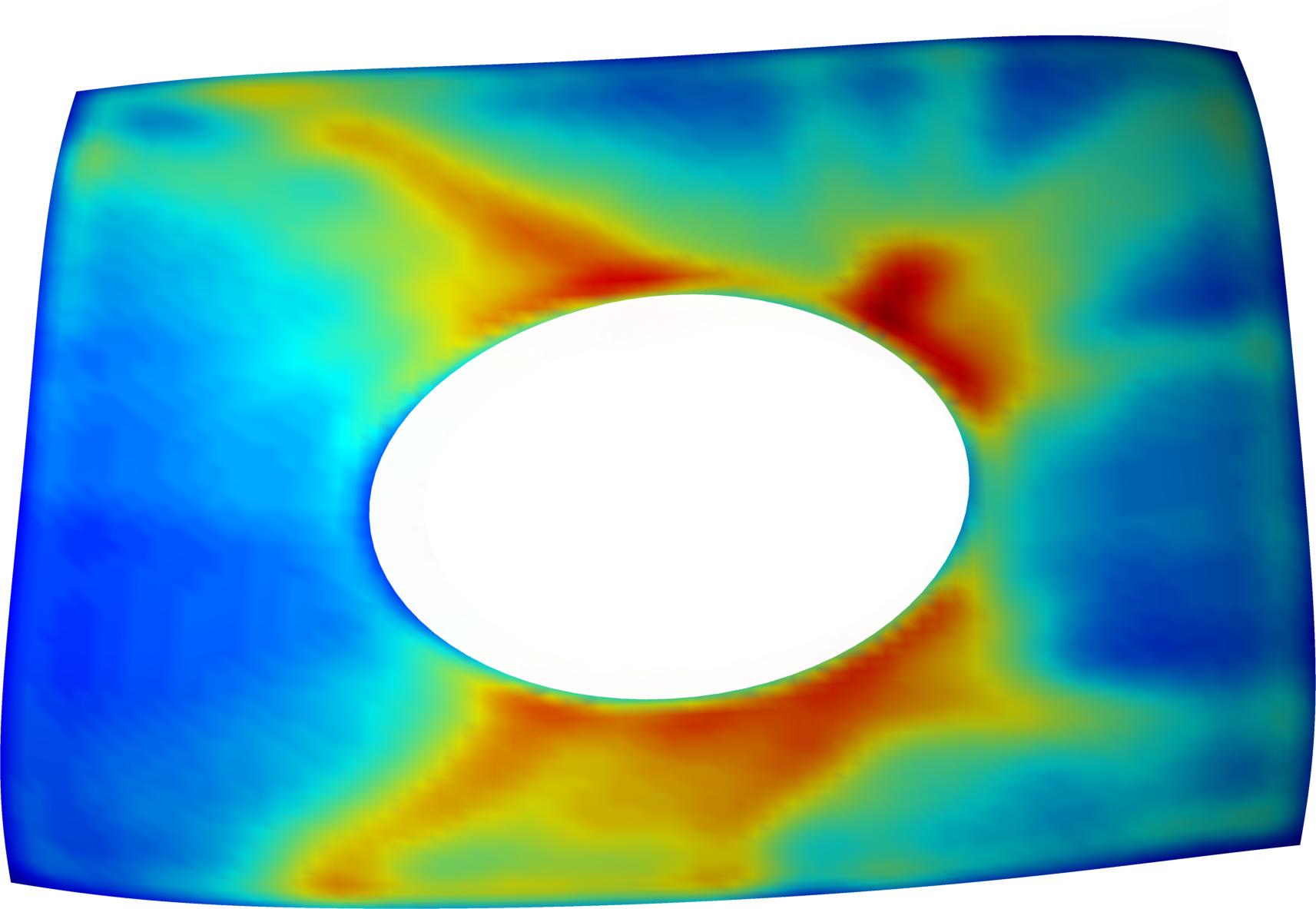}}\hfill
			\subcaptionbox[.33\textwidth]{Deformed mean template and residuals $\phi^{v^{\p^1}}.(\bar \x, \bar \f + \bzeta^1)$ (view 1)\label{figSub11}}{\includegraphics[width=.28\textwidth]{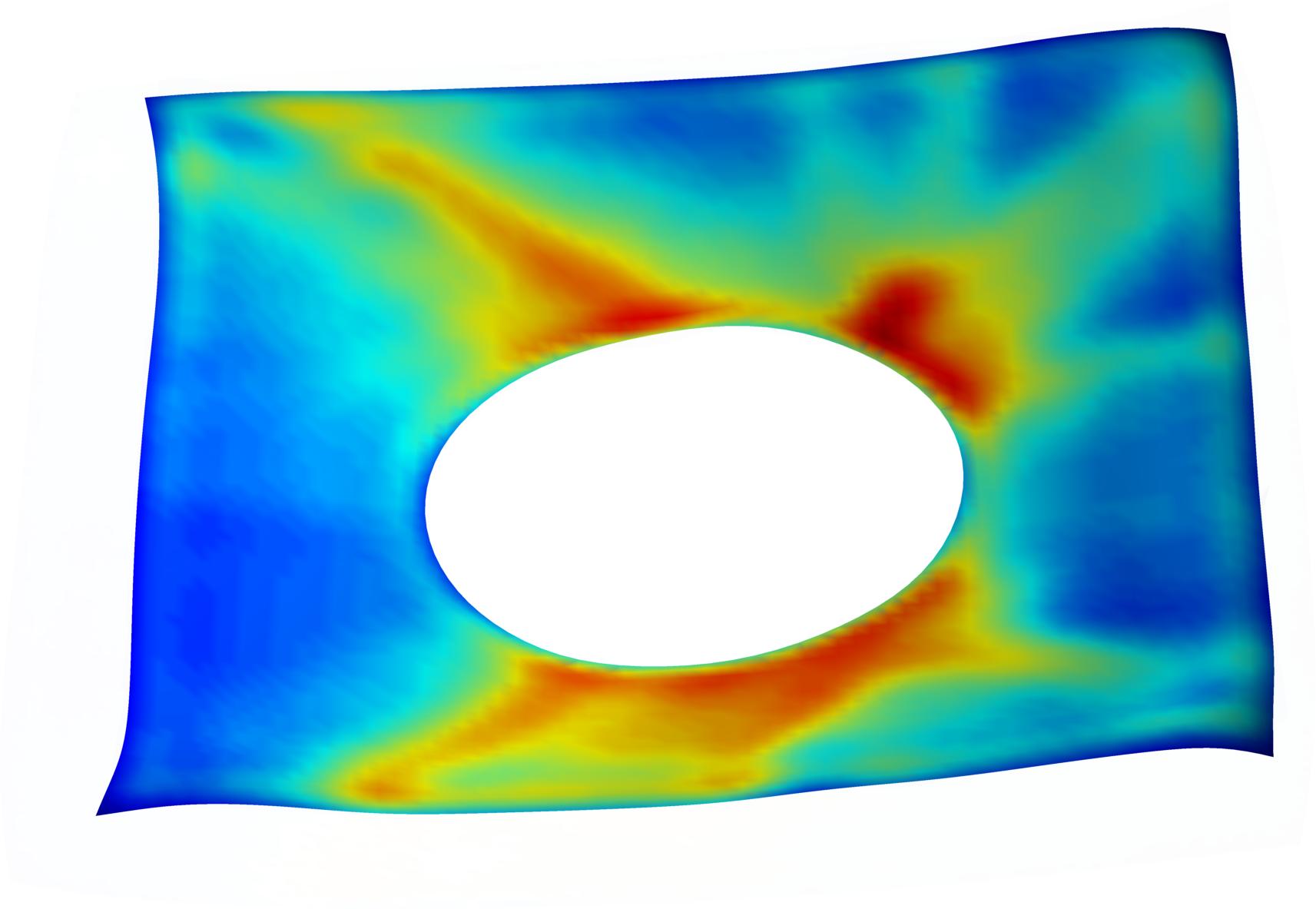}}\hfill
			\subcaptionbox[.33\textwidth]{Deformed mean template and residuals $\phi^{v^{\p^1}}.(\bar \x, \bar \f + \bzeta^1)$ (view 2)\label{figSub12}}{\includegraphics[width=.28\textwidth]{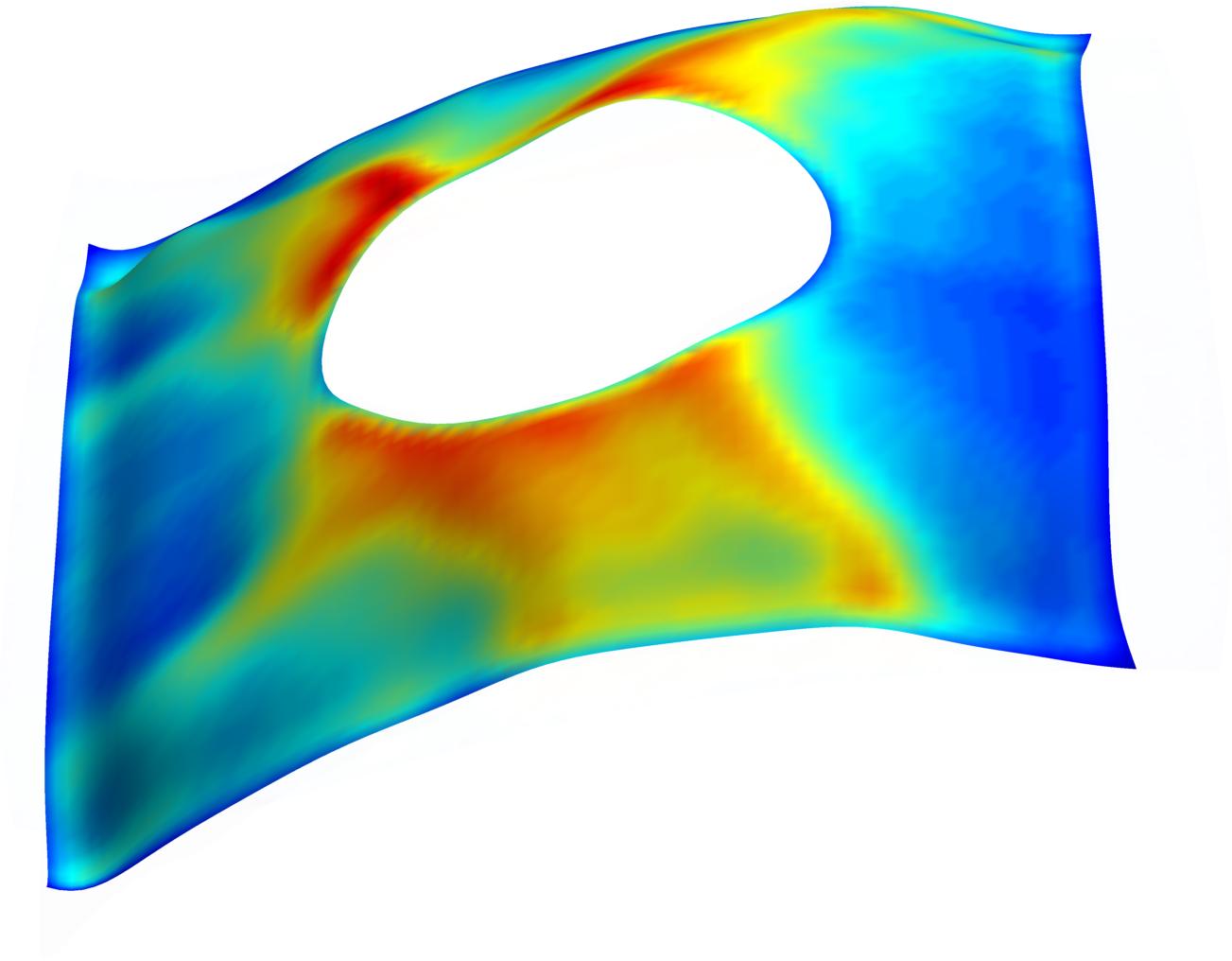}}
			\caption{Results for Data 1 (control dataset). Figure \ref{figSub11} and \ref{figSub12} should be compared with Figure \ref{figData11} and \ref{figData12} respectively.\label{fig:ResNormal}}
		\end{figure}
		\begin{figure}[H]
			\centering
			\subcaptionbox[.33\textwidth]{Relative position of data 2 and the mean template (view 1)\label{figData20}}{\includegraphics[width=.28\textwidth]{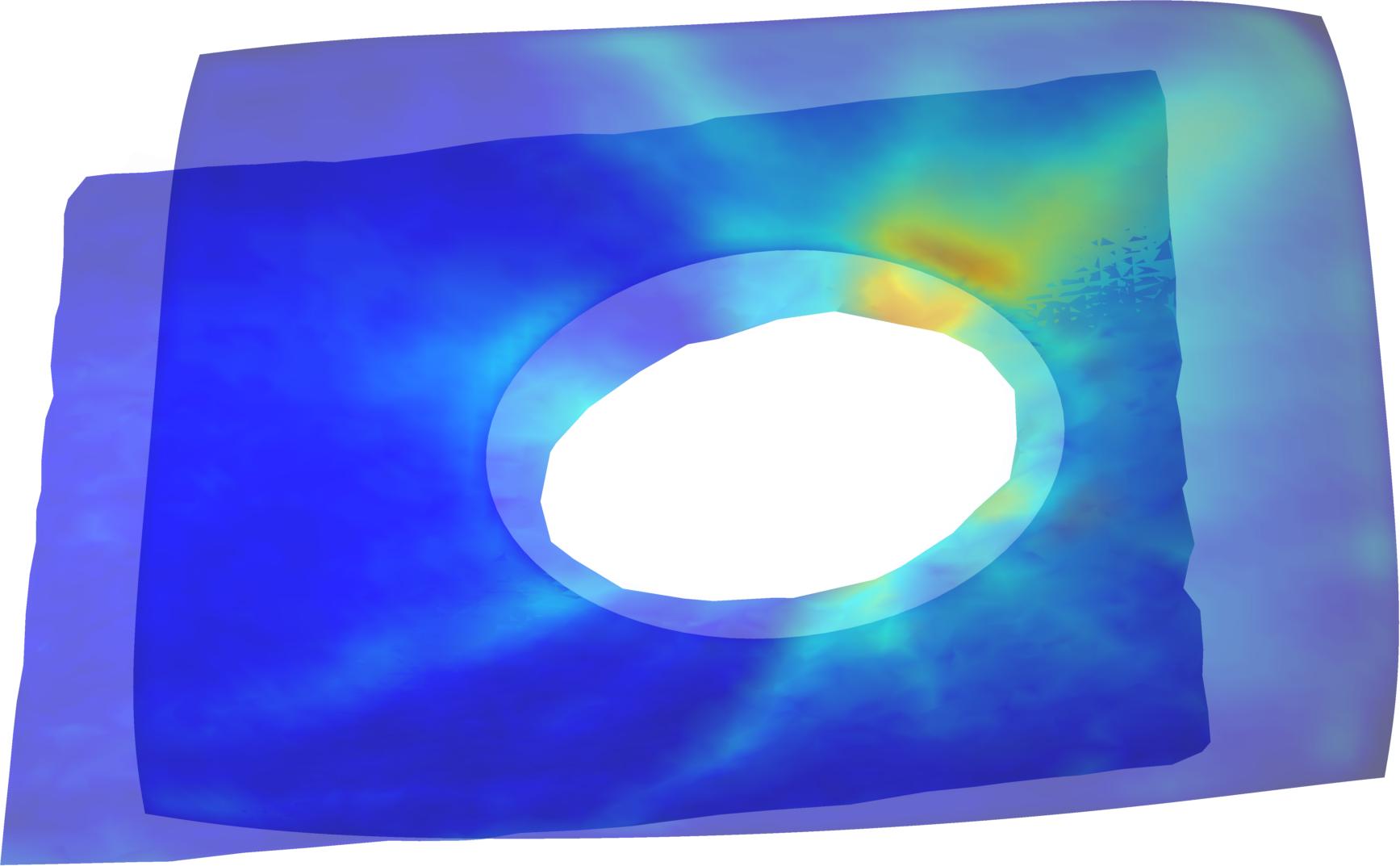}}\hfill
			\subcaptionbox[.33\textwidth]{Data 2 (view 1)\label{figData21}}{\includegraphics[width=.28\textwidth]{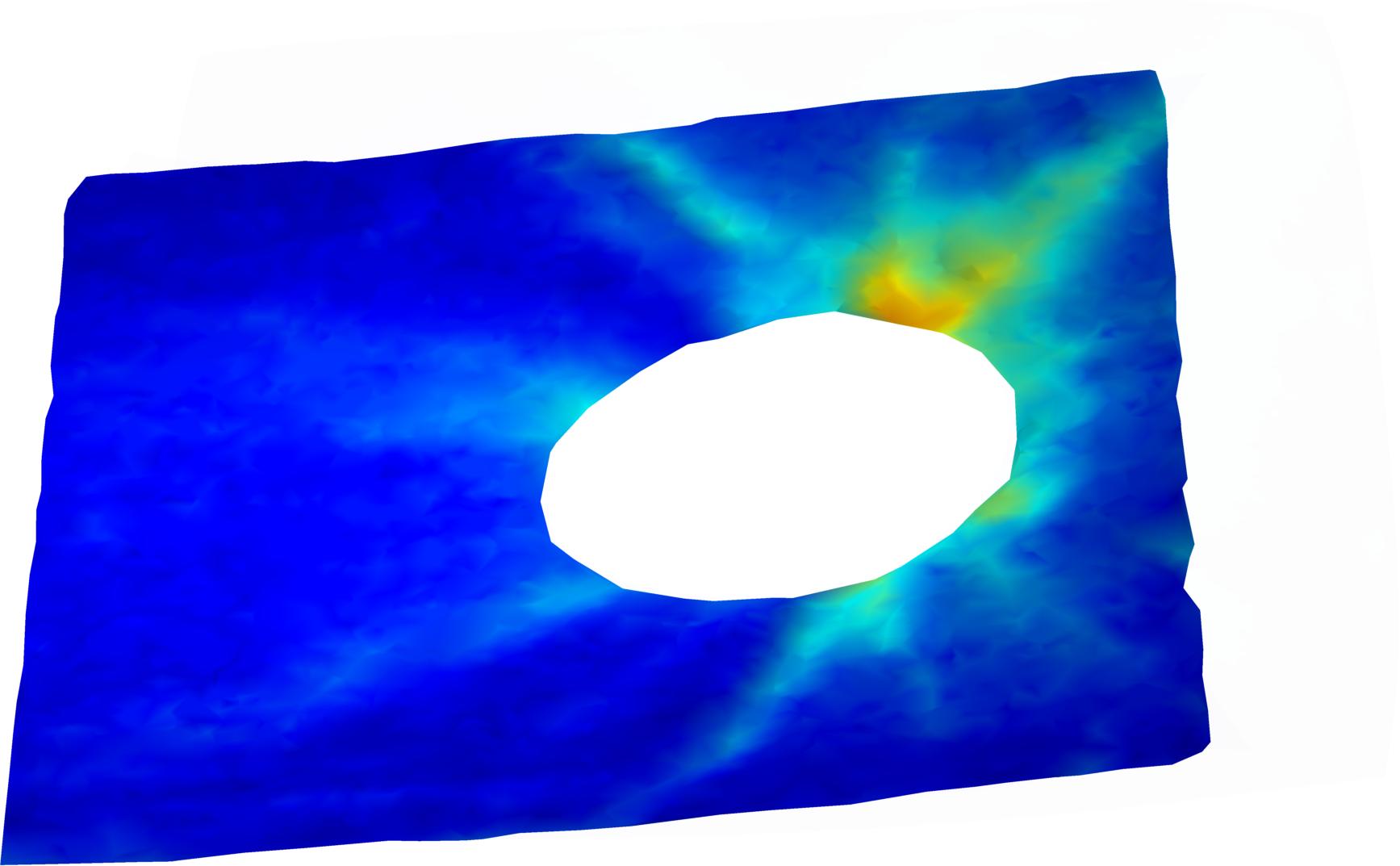}}\hfill
			\subcaptionbox[.33\textwidth]{Data 2 (view 2)\label{figData22}}{\includegraphics[width=.28\textwidth]{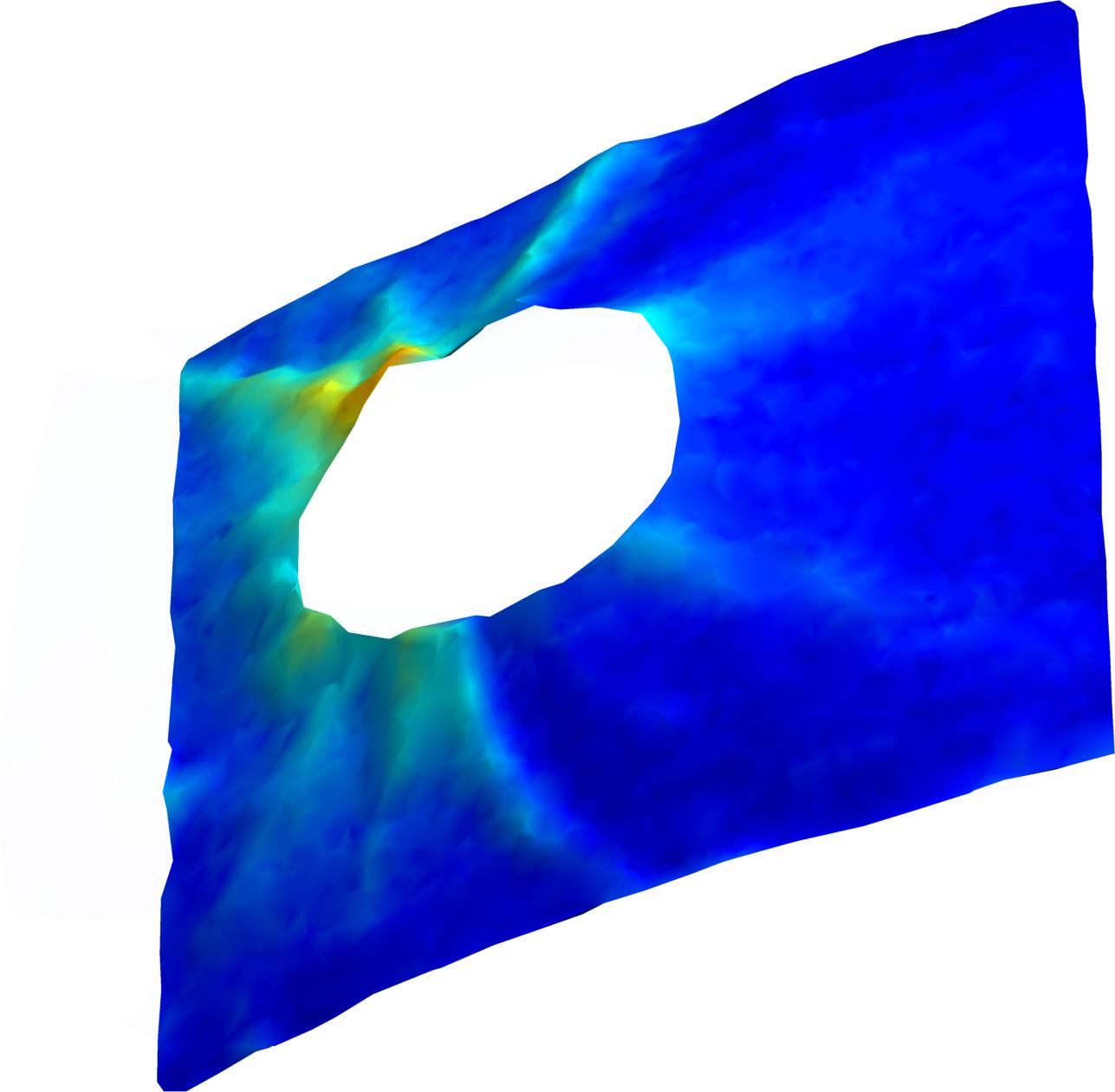}}\\
			\subcaptionbox[.33\textwidth]{Mean template and the residuals $(\bar \x, \bar \f + \bzeta^2)$ (view 1)\label{figSub20}}{\includegraphics[width=.28\textwidth]{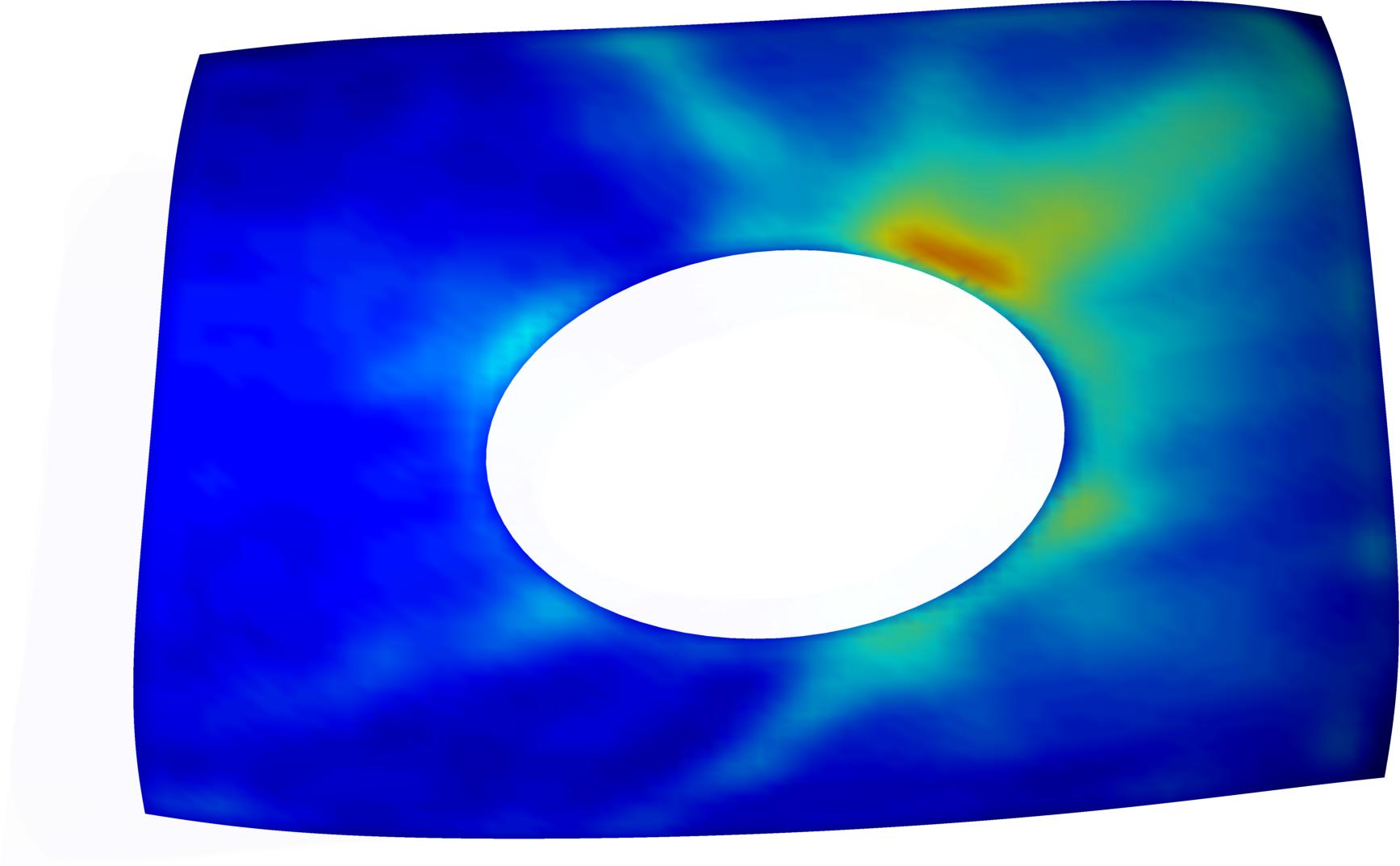}} \hfill
			\subcaptionbox[.33\textwidth]{Deformed mean template and residuals $\phi^{v^{\p^2}}.(\bar \x, \bar \f + \bzeta^2)$ (view 1)\label{figSub21}}{\includegraphics[width=.28\textwidth]{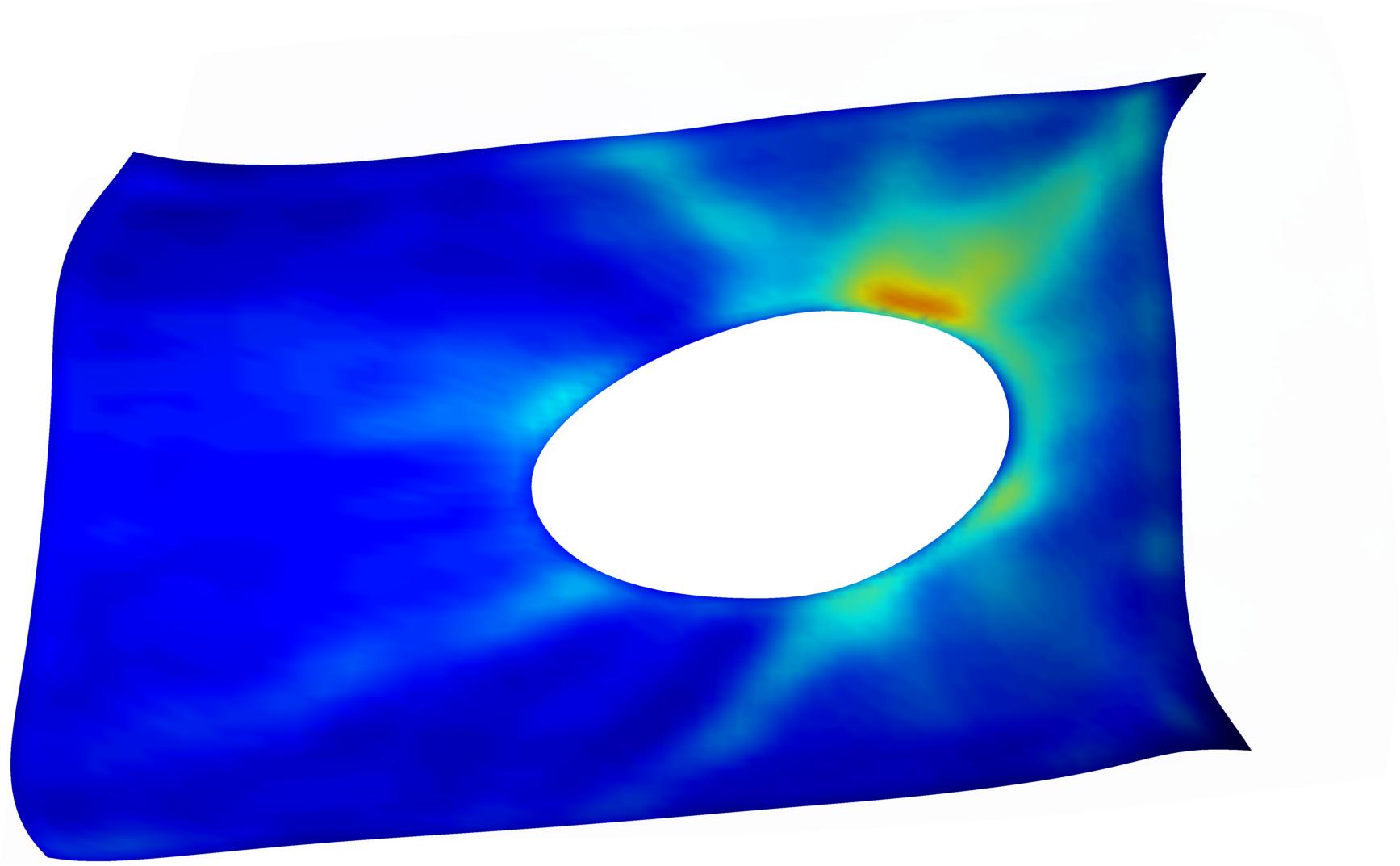}} \hfill
			\subcaptionbox[.33\textwidth]{Deformed mean template and residuals $\phi^{v^{\p^2}}.(\bar \x, \bar \f + \bzeta^2)$ (view 2)\label{figSub22}}{\includegraphics[width=.28\textwidth]{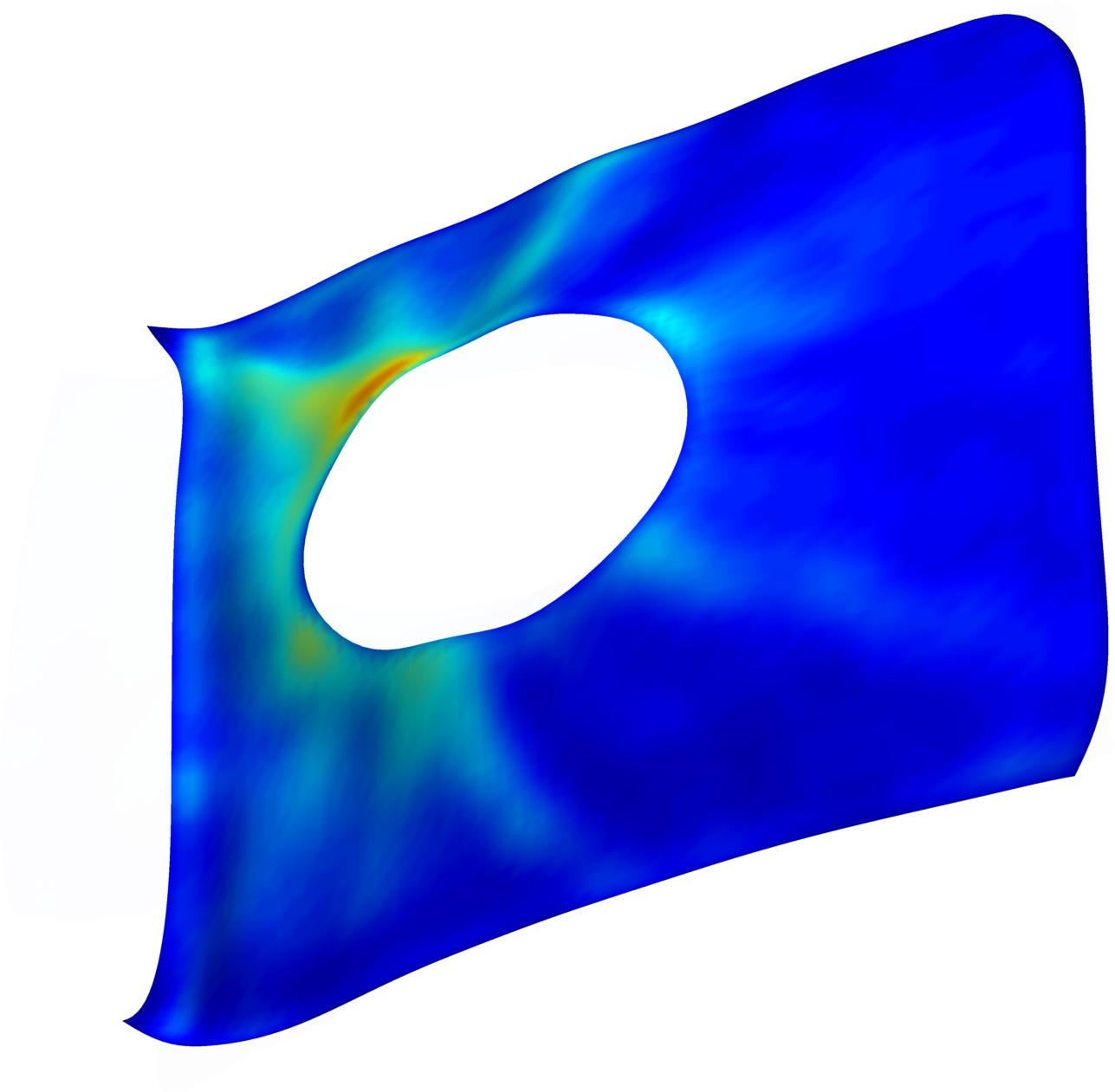}}
			\caption{Results for Data 2 (glaucoma dataset). Figure \ref{figSub21} and \ref{figSub22} should be compared with Figure \ref{figData21} and \ref{figData22} respectively.\label{fig:ResGlauc}}
		\end{figure}

\section{Numerical pitfalls}\label{part.traps}

In this section, we discuss some issues affecting the quality of the mean template estimation. We illustrate these problems by numerical experiments and propose some solutions to fix them.

\subsection{Mass cancellation}\label{part:massCancel}

The varifold norm was first introduced in the pure geometrical setting to avoid mass cancellation appearing with the classical current norms. As orientation matters in the currents' setting, a surface with positive area sufficiently crumpled may have an arbitrarily small current norm. This phenomenon is common in practical applications and particularly during an atlas estimation. This is described in detail in the introduction of \cite{Charon2} or in the Chapter 3 of \cite{Charon_thesis}. In the (pure geometrical) varifold setting, the norm of a surface cannot decrease too much when two pieces of surface are folded: Theorem $3.4.1$ in \cite{Charon_thesis} is in some sense a reciprocal inequality to Proposition \ref{prop:control_RKHSnorm_totalvar}. Unfortunately, there is no such a control for the fvarifold norm as one can exhibit a functional surface of positive area and non-zero signal with small fvarifold norm.

We follow here the notations introduced in Section \ref{part:discreteScheme}. Recall that, in formula \eqref{eq.discreteprs}, the inner product of two fvarifolds was approximated by a double sum of inner products of Diracs. We assume hereafter that the kernel fvarifold inner product satisfies the assumptions of Proposition \ref{prop:kernel_fvar}. Therefore, by formula \eqref{eq_dot_product_diracs} we have, for any Diracs $\delta_{(x_{1},V_{1},f_{1})}$ and $\delta_{(x_{2},V_{2},f_{2})}$:
\begin{equation}\label{eq:prsDiracs}
	\abs{\langle \delta_{(x_{1},V_{1},f_{1})}, \delta_{(x_{2},V_{2},f_{2})} \rangle_{W'}} \leq  \sabs{k_{e}(\cdot,\cdot)}_\infty \sabs{k_{t}(\cdot,\cdot)}_\infty  k_{f}(f_{1},f_{2}) \,.
\end{equation}
The right hand side of \eqref{eq:prsDiracs} may be arbitrarily small when $|f_1 - f_2|$ is large as $k_f(f_1,\cdot)$ is continuous and vanishes at infinity. It means, in particular, that sufficiently high variations in the signal make any two Diracs orthogonal with respect to the fvarifold inner product. Now, let $(\x,\f)$ be a finite polygonal mesh approximated by a finite sum of Diracs $\sum_{\ell=1}^{\nT} r_\ell \delta_{(\hat x_{\ell},V_{\ell},\hat f_{\ell})}$. Equations \eqref{eq:dot_product_rect_var} and \eqref{eq:prsDiracs} yield  
\begin{equation}\label{eq.normX}
	\snorm{\mu_{(\x,\f)}}_{W'}^2 \leq \sabs{k_e(\cdot,\cdot)}_\infty \sabs{k_t(\cdot,\cdot)}_\infty \Big( \sum_{\ell=1}^{\nT} r_\ell^2 k_f(\hat f_{\ell},\hat f_{\ell}^{})  +  \sum_{\ell,\ell'=1,\ell \neq \ell'}^{\nT_{}} r_{\ell}r_{\ell'} k_f(\hat f_{\ell},\hat f_{\ell'}^{}) \Big).
\end{equation}
From here on, we consider the case where $k_e$, $k_t$ and $k_f$ are the Gaussian kernels of equation \eqref{eq.GaussKern}. Let $0<\varepsilon <1$ and assume that $\hat f_\ell = K \sigma_f \ell$ for any $\ell=1,\ldots,\nT$ and $K= \sqrt{-\ln \varepsilon}$. Inequality \eqref{eq.normX} becomes, in that case~:
\begin{align*}
	\snorm{\mu_{(\x,\f)}}_{W'}^2 & \leq \sum_{\ell=1}^{\nT} r_{\ell}^2 + \sum_{\ell,\ell'=1, \ell\neq\ell'}^{\nT} r_{\ell}r_{\ell'} \exp(-K^2 |\ell -\ell'|^2) \\
	&\leq C_1(\nT) + \varepsilon  C_2(T), 
\end{align*}
where $C_1(\nT) =\sum_{\ell=1}^{\nT} r_{\ell}^2 $ and $C_2(\nT) = \sum_{\ell,\ell'=1, \ell\neq\ell'}^{\nT} r_{\ell}r_{\ell'}$. Thence, if $\nT$ remains fixed and $K \to \infty$ the fvarifold norm of $(\x,\f)$ is less of equal to $C_1(\nT)$. Assuming that the fshape $(\x,\f)$ is regularly discretized with $\nP$ points, we may consider that $r_\ell \approx \nP^{-1}$ and $C_1(\nT) \approx \nP^{-1}$. When the discretization of a fshape becomes finer, the number $\nT$ of triangles increases and the fvarifold norm of $(\x,\f)$ may be arbitrarily small. 

To illustrate this, we consider the discrete version of the variational problem of Proposition \ref{prop:JP1} in the simple case where $N=1$ and $\gamma_W=2$. It gives:
\begin{equation}\label{eq.Jfun}
	\begin{cases}
		\min_{\f\in\R^{\nP_{\x}}}\limits \bJ_{fun}(\f) \\
		\text{where } \\
		\bJ_{fun}(\f) =  \frac{\gamma_f}{2} \sabs{\f}^2_{\x} + \snorm{\mu_{(\x,\f)} - \mu_{(\y,\g)}}^2_{W'} 
	\end{cases}
\end{equation}
where the source fshape $(\x,\f)\in E^{\nP_{\x}} \times \R^{\nP_{\x}}$ and the target fshape $(\y,\g)\in E^{\nP_{\y}} \times \R^{\nP_{\y}} $ are two flat overlapping squares belonging to the same plane, see  Figure \ref{fig:sourceTarget}. The source $(\x,\f)$ contains $\nP_{\x} = 6400$ vertices and the target $(\y,\g)$ contains $\nP_{\y} = 400$ vertices distinct from the vertices of $\x$. 
The meaning of \eqref{eq.Jfun} is the following: we are trying to register $(\x,\f)$ onto the fixed target $(\y,\g)$ by tuning the signal $\f\in \R^{\nP_{\x}}$ of the source only. This is not obvious to figure out what a good solution of this problem should be. Note that Proposition \ref{prop:JP1} ensures the existence of a proper solution of the continuous version of \eqref{eq.Jfun} if $\gamma_f$ is large enough. 

\begin{figure}[H]
	\centering
	\begin{tabular}{cccc}
		\parbox[t]{.28\textwidth}{\begin{subfigure}{.28\textwidth}
			\centering
			\includegraphics[width=3cm]{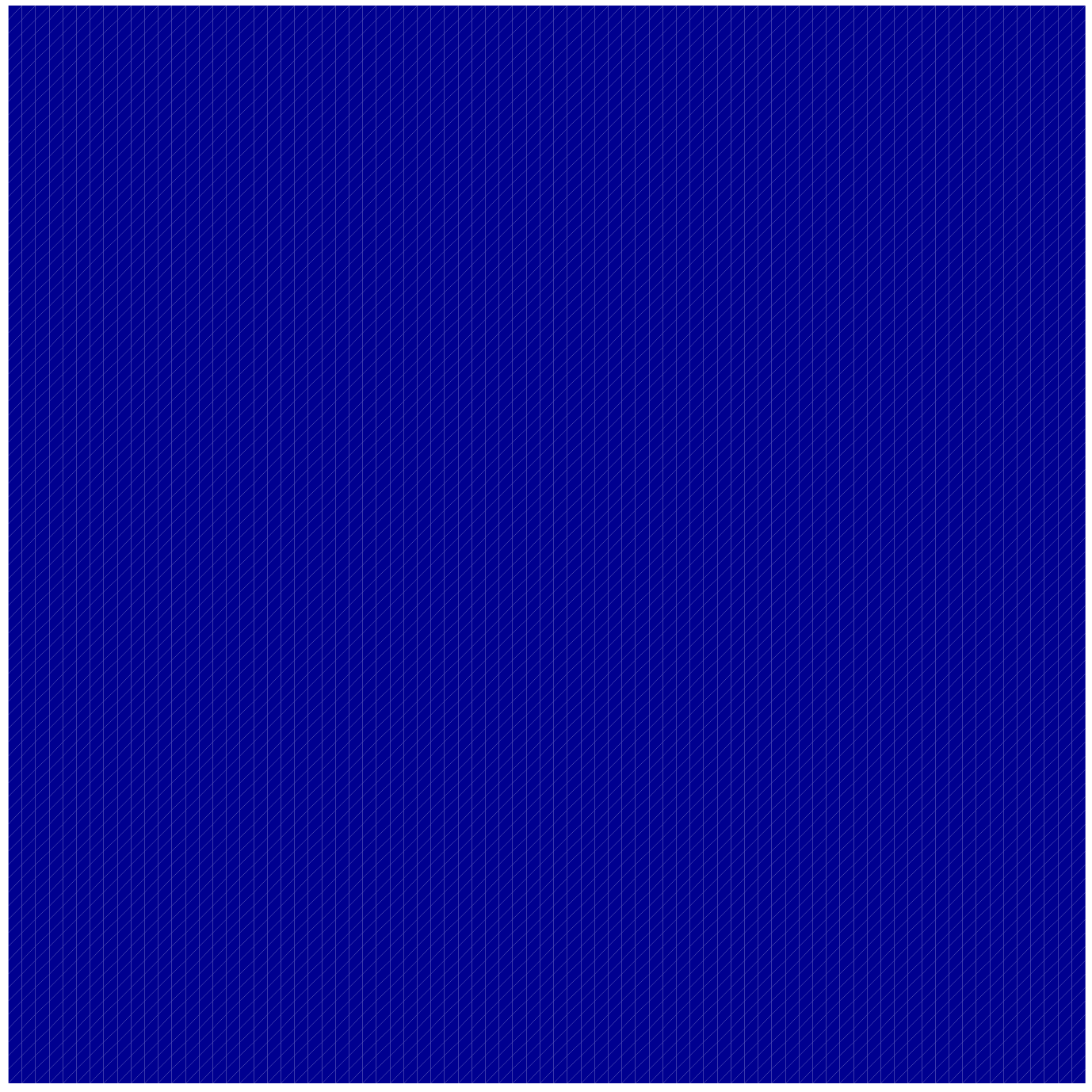}
			\caption{The source $\left(\x,\f \right)$}\label{figa}
		\end{subfigure}}&
		\parbox[t]{.28\textwidth}{\begin{subfigure}{.28\textwidth}
			\centering
			\includegraphics[width=3cm]{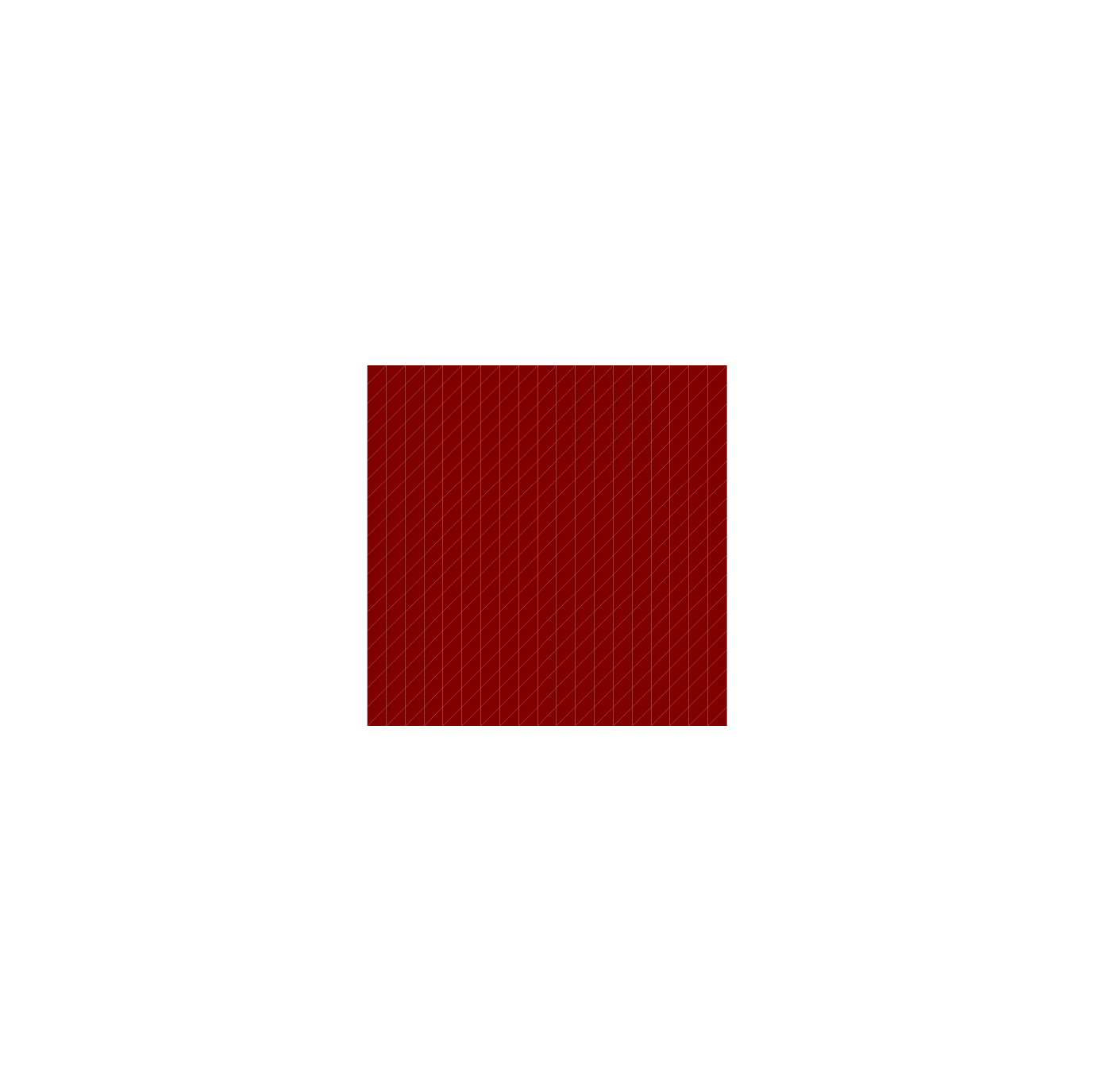}
			\caption{The target $\left( \y,\g \right)$}\label{figb}
		\end{subfigure}}&
		\parbox[t]{.28\textwidth}{\begin{subfigure}{.28\textwidth}
			\centering
			\includegraphics[width=3cm]{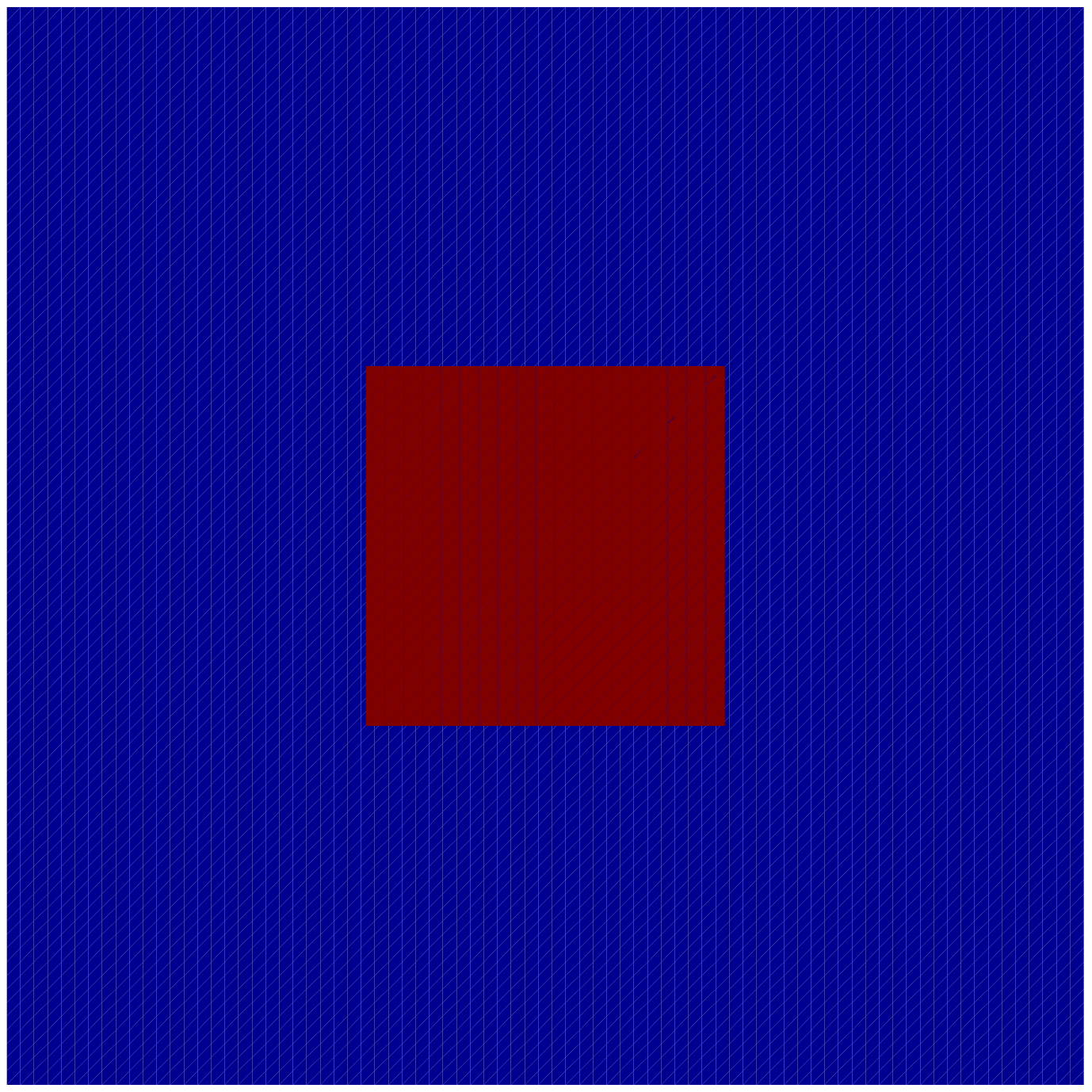}
	\caption{$\left( \x,\f \right)$ and $\left( \y,\g \right) $} \label{figc}
		\end{subfigure}
		}&
		\raisebox{-.4\height}{\includegraphics[height=3.8cm]{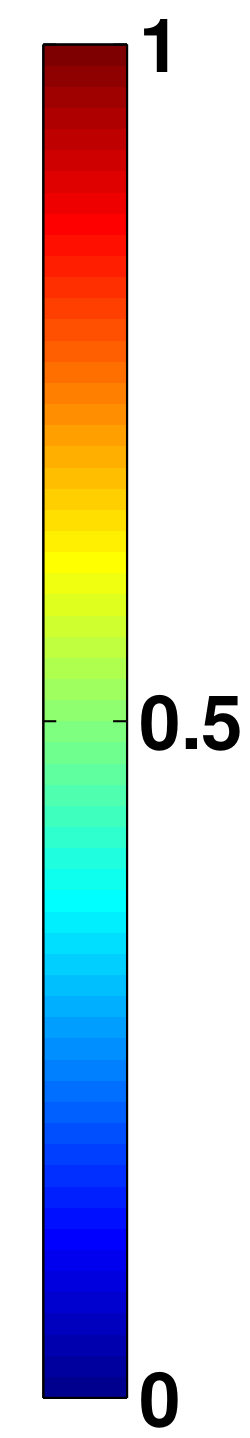}}
	\end{tabular}
		\caption{The source and target fshapes are respectively a big and a small flat square both lying in the $yz$-plane.  The initial source signal is 0 and the target signal is 1. Figure \ref{figc} simply depicts the two fshapes together.}
	\label{fig:sourceTarget}
\end{figure}

In Figure \ref{fig:oscill}, we show three outputs of an adaptive gradient descent on $\bJ_{fun}$ with various choices of penalty parameter $\gamma_f$. We denote $\f_* \in \R^{\nP_{\x}}$ the solution found after 800 iterations. Figure \ref{fig.Resa} shows the result with no penalty term, {\it i.e} $\gamma_f=0$. The signal $\f_*$ is almost equals to 1 on the overlapping square whereas $\f_*$ oscillates dramatically on the complementary part to decrease the fvarifold norm of the remaining part of $\x$. Figure \ref{fig.Resb} shows that with a small penalty, we are able to recover a signal close to 1 on the central square but oscillations are still present though weaker compared to Figure \ref{fig.Resa}. Finally, Figure \ref{fig.Resc} shows the solution $\f_*$ found with a larger $\gamma_f$. The oscillations have almost disappeared and the signal remains small in the non-overlapping part (in blue). The price to pay is a lower intensity in the central square (around $0.25$).
	
\begin{figure}[H]
	\centering
	\begin{subfigure}{.32\textwidth}
	\includegraphics[width=\textwidth]{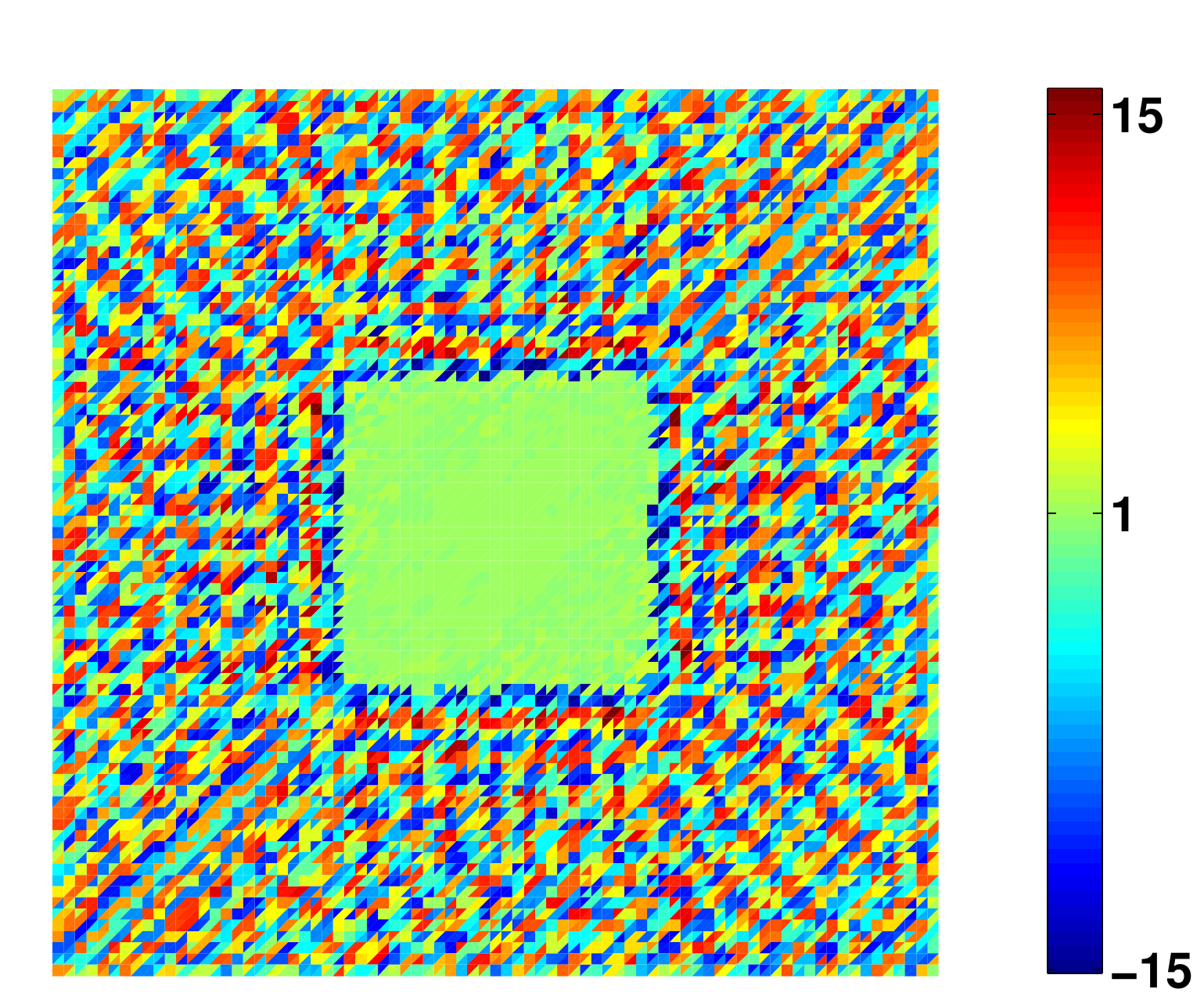}
	\caption{$\gamma_f = 0$\label{fig.Resa}}
		\end{subfigure}
		\begin{subfigure}{.32\textwidth}
	\includegraphics[width=\textwidth]{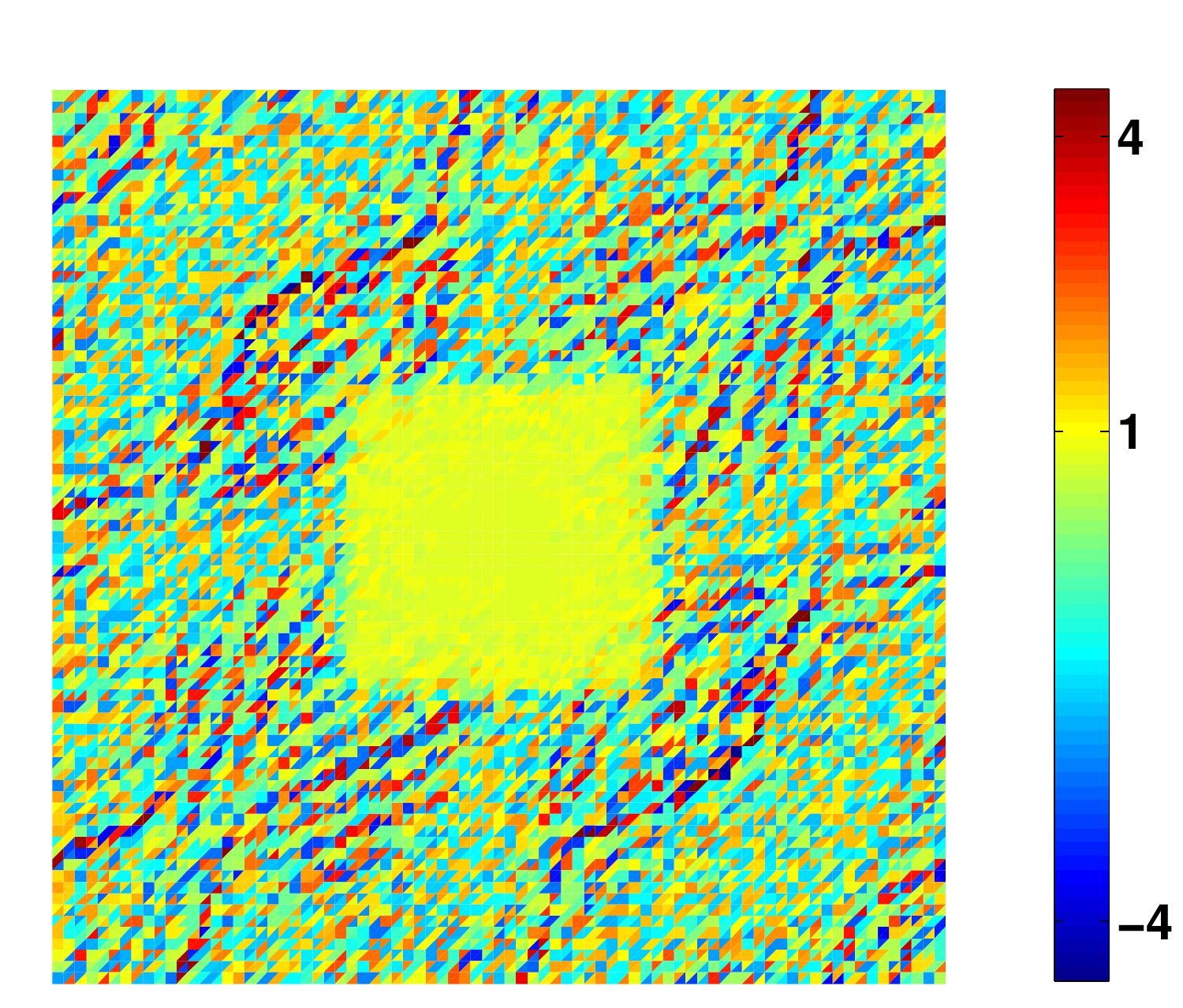}
	\caption{$\gamma_f = 4$\label{fig.Resb}}
		\end{subfigure}
		\begin{subfigure}{.32\textwidth}
	\includegraphics[width=\textwidth]{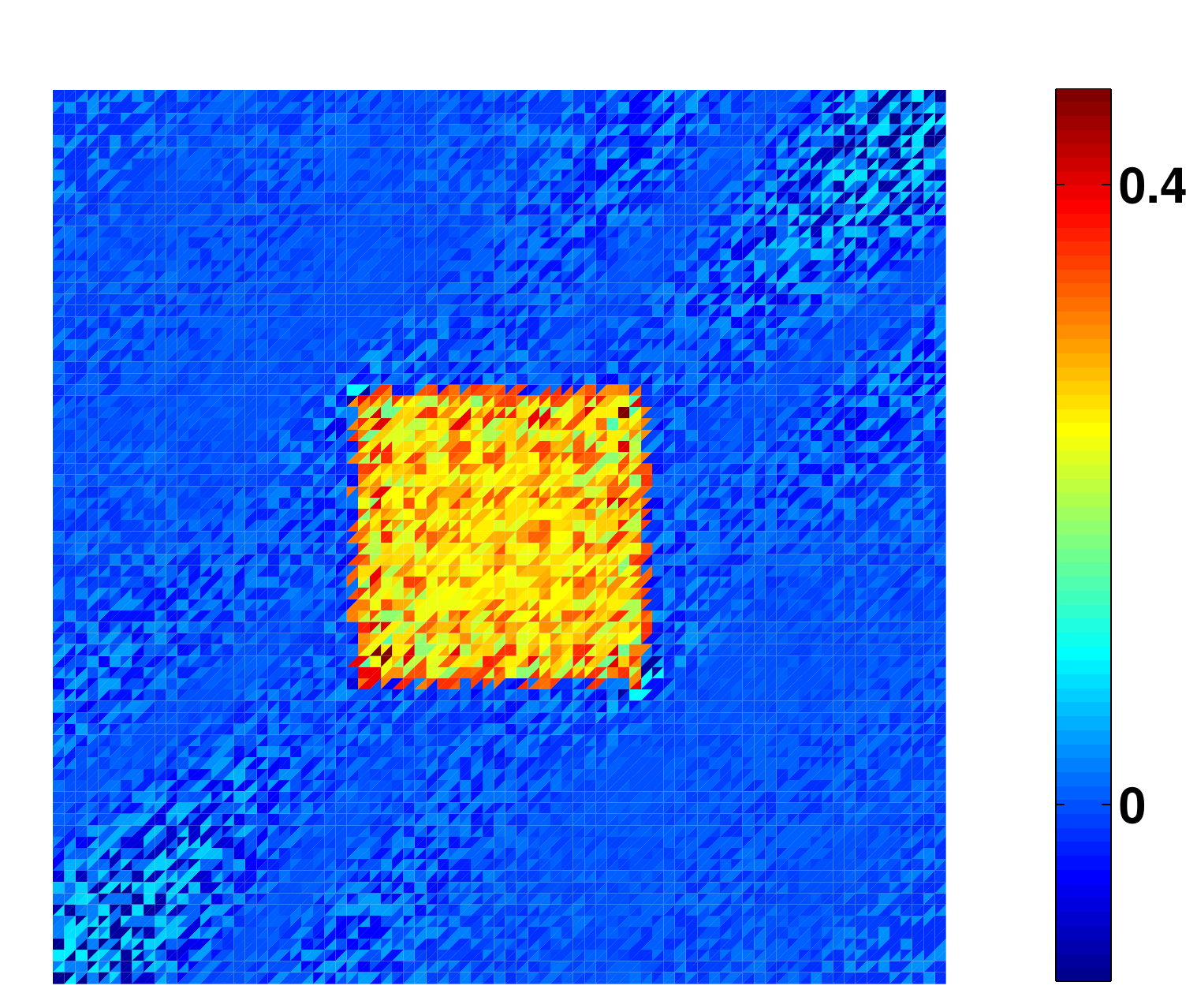}
	\caption{$\gamma_f = 20$\label{fig.Resc}}
		\end{subfigure}
		\caption{The fshape $(\x,\f_*)$ where $\f_*$ is the output of a gradient descent on $\bJ_{fun}$ of equation \eqref{eq.Jfun} with various value of $\gamma_f$.}
	\label{fig:oscill}
\end{figure}

Note that this example is an extreme case but such a problem arises in real datasets. In medical imaging, it is common to get functional data acquired on a subpart of a surface only. Even with a rigid registration during the preprocessing, subparts with a well defined signal may be highly non-overlapping on two different subjects. This may yield to oscillating solutions if the gradient descent is used naively. We note that in practice, regularizing the gradient with respect to $\f$ may fix this issue in some cases. The method is similar to the one described at the end of Section \ref{part:bd} for the geometric part of the gradient.  

\subsection{Boundary problems} \label{part:bd}

In this section, we discuss some problems arising when the ``free'' mean template method of Section \ref{part:AlgoFree} is used with fshapes containing boundaries. An issue concerns the gradient (with respect to $\x$) of the data attachment terms $g_i$ (see Algorithm \ref{algo:Jfreetan}) - note that the following discussion is relevant both for the purely geometric and functional varifolds or currents framework. The values of the signal has no influence here and we assume without loss of generality that the signals are 0. Typical examples of data with boundaries (inspired by the dataset presented in Section \ref{part:OCT}) are depicted in Figure \ref{fig:gradS}.

\begin{figure}[H]
	\centering
	\subcaptionbox[.48\textwidth]{Initialization $(\x_{init},0)$ \label{fig:source}}{\includegraphics[width=5cm]{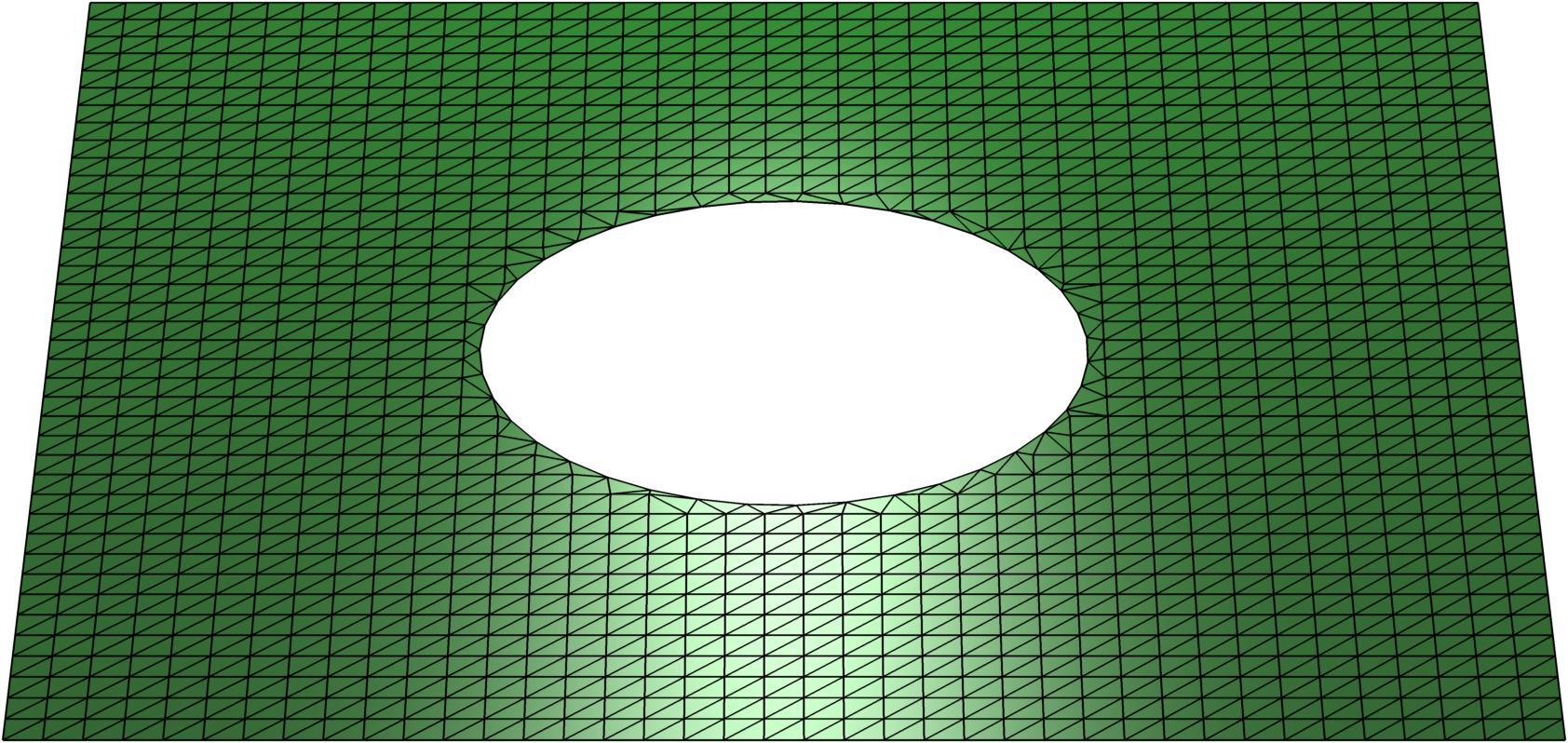}} \hspace{5em}
\subcaptionbox[.48\textwidth]{Target $(\y,0)$ \label{fig:target}}{\includegraphics[width=5cm]{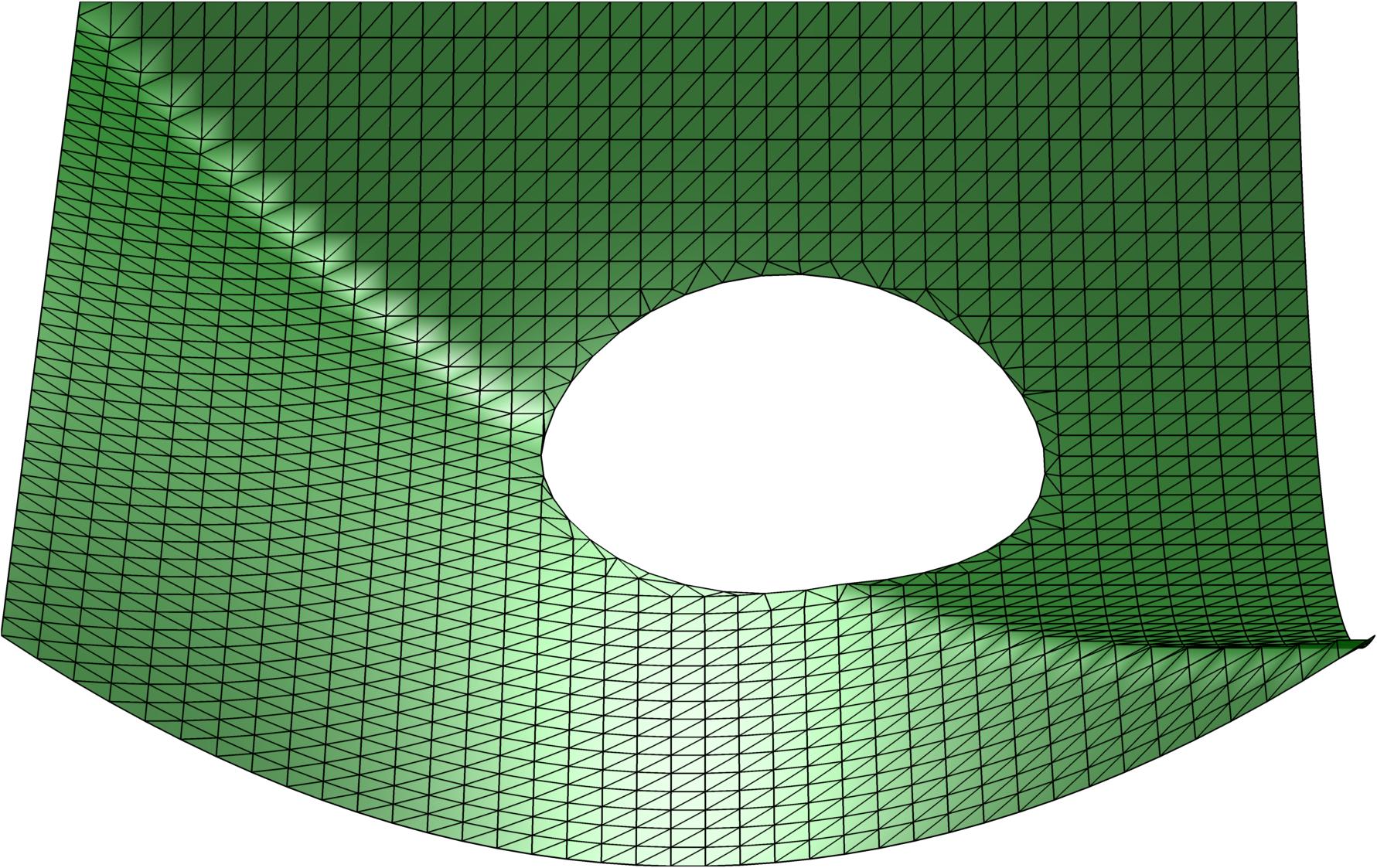}}
\caption{The initial fshape $(\x_{init},0)$ is a flat square with a centered opening. The target fshape $(\y,0)$ is a (non smooth) bended version of $(\x_{init},0)$. All the signals are zero.\label{fig:gradS}}	
\end{figure}

To illustrate the problem that one can face, let $E=\R^3$ and $(\y,0) \in E^{\nP_{\y}}\times \R^{\nP_{\y}}$ be a fixed discrete target fshape as in Figure \ref{fig:target}. We consider the following minimization problem,
\begin{equation}\label{eq:minJg}
	\begin{cases}
		\min_{\x\in E^{\nP}}\limits \bJ_{geo}(\x) \\
		\text{ where } \\
		\bJ_{geo}(\x) = \snorm{\mu_{(\x,0)} - \mu_{(\y,0)} }_{W'}^2,
	\end{cases}
\end{equation}
Our goal is to study the behaviour of a gradient descent in the variable $\x=(x_k)_{k=1}^{\nP} \in E^{\nP}$ on $\bJ_{geo}$ starting from an initial discrete fshape  $(\x_{init},0)$. In numerical experiments, the gradient descent on $\bJ_{geo}$ does not perform well if the fshape $(\x_{init},0)$ contains boundaries as in Figure \ref{fig:source}. The boundary of a discrete fshape is the set of vertices that are an end of an edge belonging to a single triangle (see Figure \ref{fig:bd}).

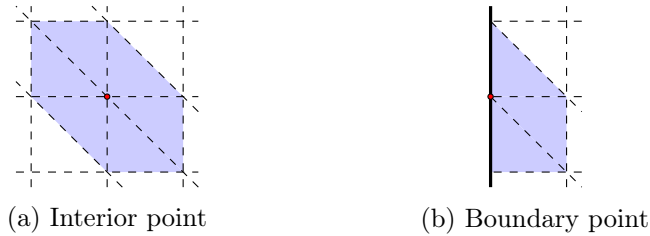
\begin{figure}[H]
	\begin{center}
		\begin{subfigure}{7cm}
			\centering
						\begin{tikzpicture}[scale=1.0]

				\fill[blue!20] (0,0) -- (-1,1) -- (0,1);
				\fill[blue!20] (0,0) -- (1,-1) -- (1,0);
				\fill[blue!20] (0,0) -- (0,-1) -- (-1,0);
				\fill[blue!20] (0,0) -- (1,-1) -- (0,-1);
				\fill[blue!20] (0,0) -- (1,0) -- (0,1);
				\fill[blue!20] (0,0) -- (-1,0) -- (-1,1);

				\draw[dashed] (-1.2,0) -- (1.2,0);
				\draw[dashed] (-1.2,1) -- (1.2,1);
				\draw[dashed] (-1.2,-1) -- (1.2,-1);

				\draw[dashed] (0,-1.2) -- (0,1.2);
				\draw[dashed] (-1,-1.2) -- (-1,1.2);
				\draw[dashed] (1,-1.2) -- (1,1.2);

				\draw[dashed] (1.2,-1.2) -- (-1.2,1.2);
				\draw[dashed] (1.2,-0.2) -- (-0.2,1.2);
				\draw[dashed] (.2,-1.2) -- (-1.2,0.2);

				\draw[fill=red] (0,0) circle (1pt);

			\end{tikzpicture}\caption{Interior point}
		\end{subfigure}
		\begin{subfigure}{4cm}
			\centering
						\begin{tikzpicture}[scale=1.0]
				\fill[blue!20] (0,0) -- (1,-1) -- (1,0);
				\fill[blue!20] (0,0) -- (1,-1) -- (0,-1);
				\fill[blue!20] (0,0) -- (1,0) -- (0,1);



				\draw[dashed] (1.2,0) -- (0,0);
				\draw[dashed] (1.2,1) -- (0,1);
				\draw[dashed] (1.2,-1) -- (0,-1);

				\draw[very thick] (0,-1.2) -- (0,1.2);
				\draw[dashed] (1,-1.2) -- (1,1.2);

				\draw[dashed] (-0,0) -- (1.2,-1.2);
				\draw[dashed] (-0,1) -- (1.2,-0.2);

				\draw[fill=red] (0,0) circle (1pt);
			\end{tikzpicture}\caption{Boundary point}
		\end{subfigure}
		\end{center}
		\caption{An example of a typical Delaunay triangulation and the characterization of a boundary point \label{fig:bd}}
	\end{figure}

The main issue on the gradient of $\bJ_{geo}$ is the following. The norm of $\frac{\partial}{\partial x_k} \bJ_{geo}$ has different orders of magnitude depending on the location of the points $x_k$ in the fshape: gradient of boundary points may be much larger than gradient of interior points, see Figure \ref{fig:gradUnbalance}. These unbalanced values between interior points and boundary points in the gradient induce undesirable effects such as self-crossings or changes in topology during the gradient descent as in Figure \ref{fig:targetAA}. Note also that the situation becomes worse when the number $\nP$ of points in $\x$ increases.

\begin{figure}[H]
	\centering
	\subcaptionbox[.45\textwidth]{The arrows represent the initial gradient $\nabla_{x} \bJ_{geo}(\x_{init})$ on $(\x_{init},0)$. The color represent the norm of this gradient (logarithmic scale).\label{fig:gradUnbalance}}{\includegraphics[width=6cm]{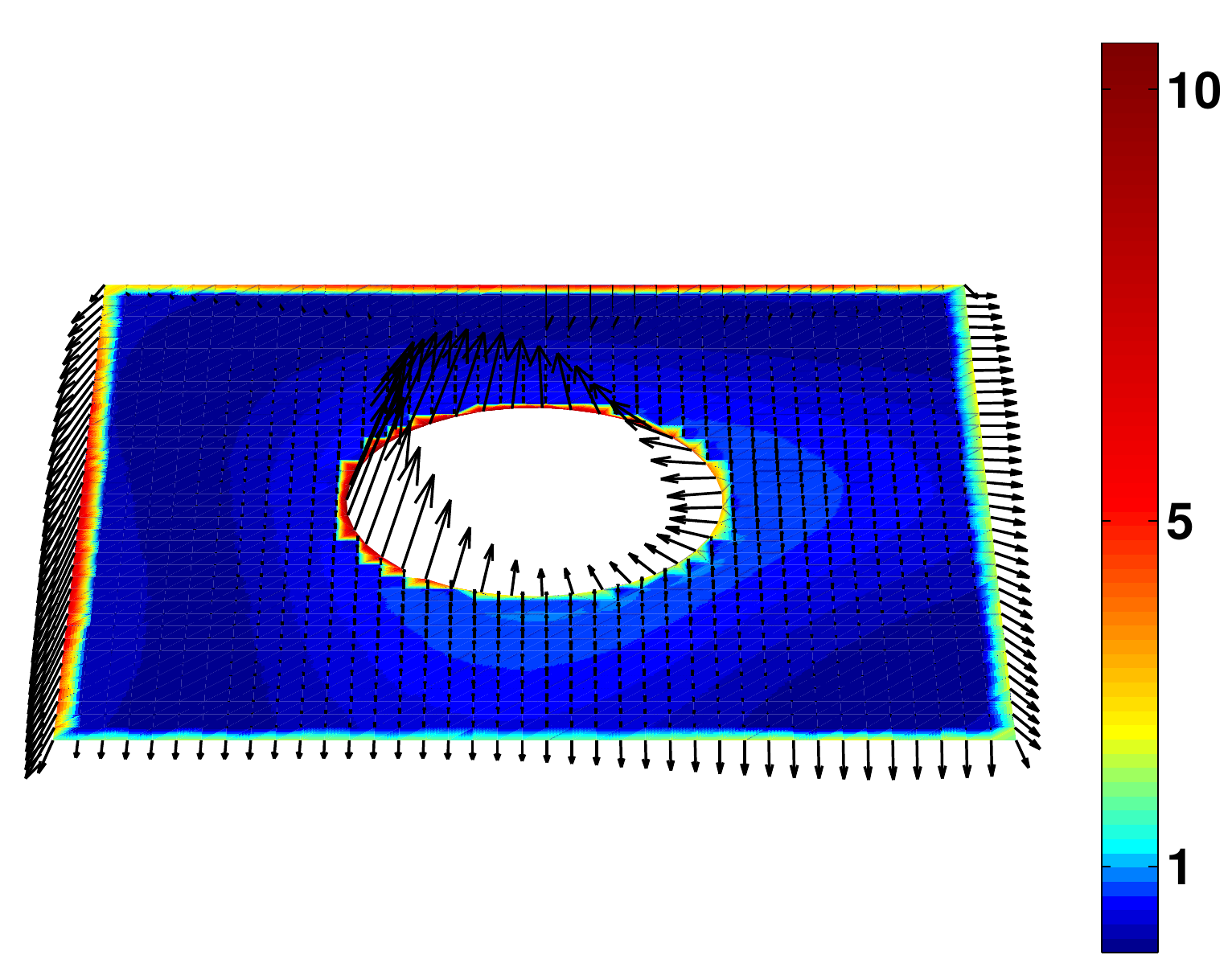}}\hspace{5em}
	\subcaptionbox[.45\textwidth]{Output of the gradient descent after few steps (has to be compared with \ref{fig:target}).\label{fig:targetAA}}{\includegraphics[width=5cm]{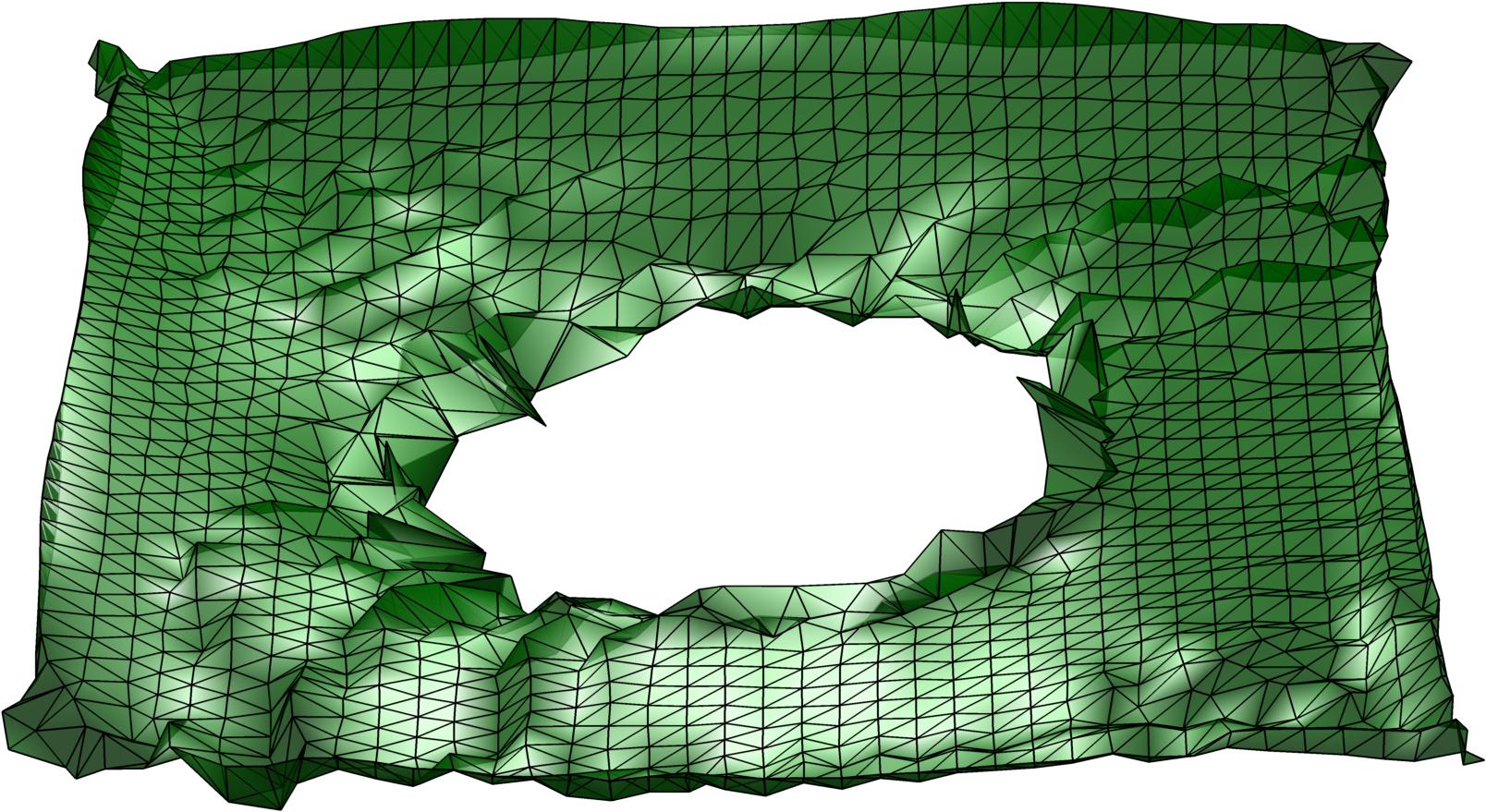}}
	\caption{Gradient descent on $\bJ_{geo}$ of equation \eqref{eq:minJg} starting from $(\x_{init},0)$ of Figure \ref{fig:source}.\label{fig:BPb2}}
\end{figure}

In order to explain the singularities of $\nabla_{\x} \bJ_{geo}$ at the boundary points, notice first that this is sufficient to focus on the gradient of the fvarifold inner product as we have $\nabla_{\x} \snorm{\mu_{(\x,\f)} - \mu_{(\y,\g)} }_{W'}^2 = \nabla_{\x} \prs{ \mu_{(\x,\f)},\mu_{(\x,\f)} }_{W'} - 2\nabla_{\x}\prs{\mu_{(\x,\f)},\mu_{(\y,\g)}}_{W'}$. Nevertheless, it is hard  to clearly understand this phenomenon by looking at the discrete formulas directly. It is then convenient to use the results of Section \ref{part:varForm} whose aim is to analyse the variations of the fvarifold inner product in the continuous setting. Thence, we may consider that $(\x,0)$ (resp. $(\y,0)$) is sampled from a continuous fshape $(X,0)$ (resp. $(Y,0)$) and that $\prs{\mu_{(\x,0)},\mu_{(\y,0}}_{W'}$ approximates its continuous counterpart $\prs{\mu_{(X,0)},\mu_{(Y,0)}}_{W'}$. We further assume that $(\x,0)$ is regularly discretized, so that the typical size of the interior volume elements $r_\ell$ is of order $P^{-1}$ and the typical size of the boundary volume elements is of order $\nP^{-\frac{\dT -1}{\dT}}$.

Formula \eqref{eq:variation_formula} gives an expression of the variations of $\prs{\mu_{(X,0)},\mu_{(Y,0)}}_{W'}$ when $(X,0)$ moves. The right hand side of \eqref{eq:variation_formula} is composed by two separated integral terms whose domains of integration are $X$ and its boundary $\partial X$ respectively. This formula then makes clear the differences between interior points and boundary points. A discrete approximation of the integral over the interior of $X$ should then involve terms of order $\nP^{-1}$ (\textit{i.e.} the volume of the surface element), whereas a discrete approximation of the integral over the boundary $\partial X$ should involve terms of order $\nP^{-\frac{\dT -1}{\dT}}$ (corresponding to the size of the volume element of the boundary). This heuristic explains why there exists a multiplicative factor between the size of the gradient at interior points and boundary points. We also know that this factor is of order $\nP^{\frac{1}{\dT}}$ meaning that the difference increases when the resolution of the discrete fshape $(\x,0)$ increases. 

\begin{figure}[h]
	\centering
	\subcaptionbox[.45\textwidth]{The arrows represent the initial regularized gradient $\tilde \nabla_{\x} \bJ_{geo}(\x_{init})$ on $(\x_{init},0)$. The color represent the norm of the gradient.\label{fig:RegGrad}}{\includegraphics[width=6cm]{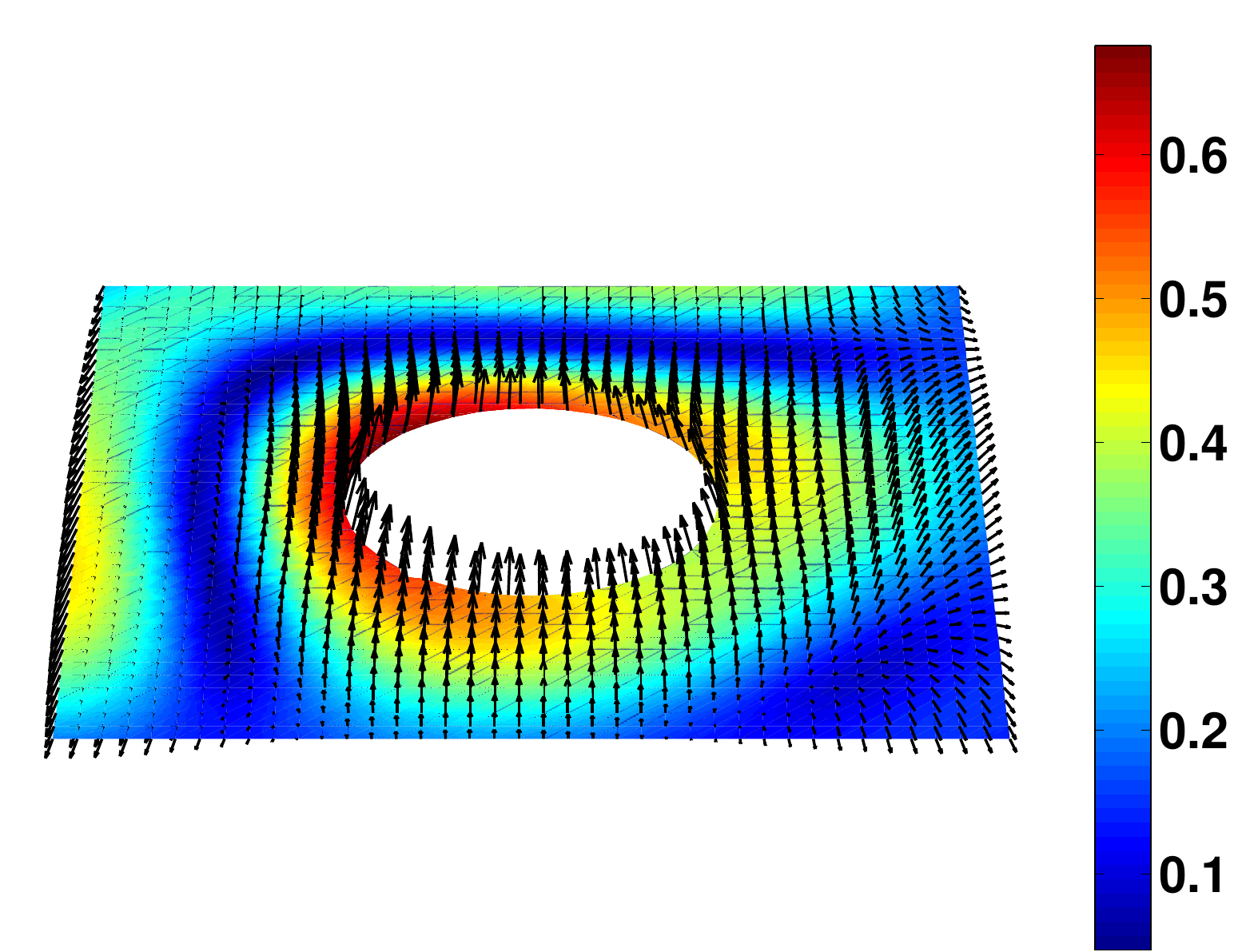}} \hspace{5em} 
	\subcaptionbox[45\textwidth]{Output of the regularized gradient descent after few steps (has to be compared with \ref{fig:target}).\label{fig:BPb3ee}}{\includegraphics[width=5cm]{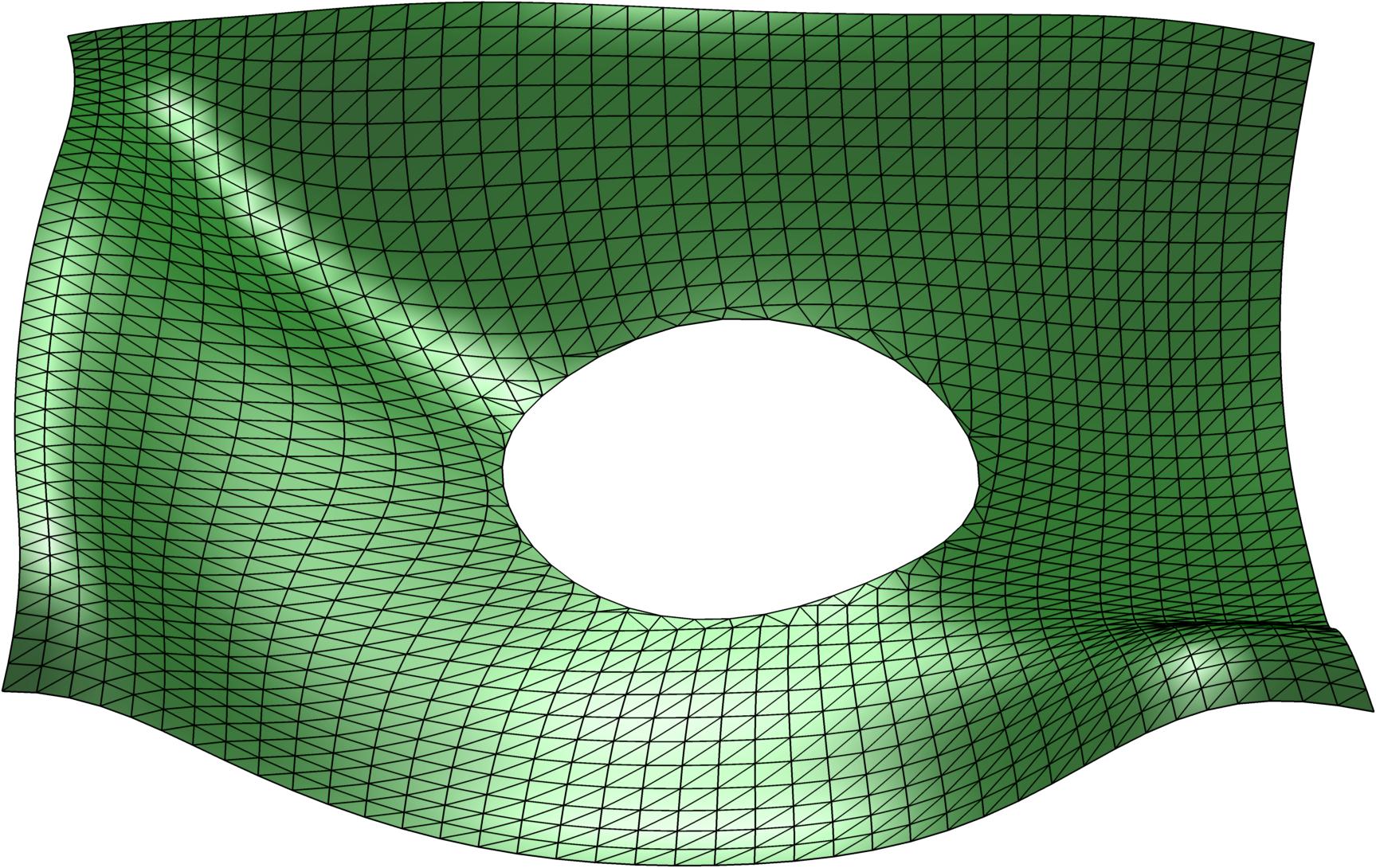}}
	\caption{Regularized gradient descent on $\bJ_{geo}$ of equation \eqref{eq:minJg} starting from $(\x_{init},0)$ of Figure \ref{fig:source}.\label{fig:BPb3}}
\end{figure}

Thus the gradient $\nabla_{\x} \bJ_{geo}$ computed thanks to the Euclidean metric involves singular terms on the boundary. A convenient way to address this problem is to compute a gradient \textit{via} a more regular metric. Let us introduce a (conveniently normalized) Gaussian kernel $K_{reg} : \R^{\dP} \times \R^{\dP} \to \R$ of scale $\sigma_{reg}$. Let $K_{\x,\x}^{reg} \doteq \big(K_{reg}(x_k,x_k)\big)_{1\leq k,k'\leq \nP} \in \R^{\nP\times\nP}$ and define
\[
	\tilde \nabla_{\x} \bJ_{geo}(\x) \doteq K_{\x,\x}^{reg} \nabla_{\x} \bJ_{geo}(\x)
\] 
This new gradient is, in some sense, a convolution of the singular gradient with a Gaussian kernel, see Figure \ref{fig:RegGrad}. As the kernel $K_{reg}$ is positive definite, $-\tilde \nabla \bJ_{geo}$ still defines a descent direction for $\bJ_{geo}$. Figure \ref{fig:BPb3ee} shows the output of a regularized gradient descent. This new gradient can be seen as the discretization of a dense vector field belonging to the RKHS $V_{reg}$ of kernel $K_{reg}$. The displacement of the template from its initial position in the regularized gradient descent can be then considered as the integration of a sequence of vector fields living in $V_{reg}$. It is shown in \cite{Charon_thesis} (Property 4.2.1) that the regularized gradient flow is well defined. Thence, provided that the steps along the gradient descent are sufficiently small, the evolution of the points of $\x$ may be considered as a diffeomorphic evolution (belonging to the orbit $G_{V_{reg}}.(\x_{init},0)$). Note that, in general, $\sigma_{reg}$ can be chosen to be different from the scale $\sigma_V$ of the deformations. This allows more flexibility in the template evolution during the gradient descent and justifies the methodology used in the algorithm described in Section \ref{part:AlgoFree}.

\subsection{Tuning the parameters} \label{part.invariances}

The procedure of atlas estimation for fshapes contains numbers of parameters which have to be tuned by the user. Although there exists some heuristics to choose the values of these parameters, it may be  hard to find a good range of settings for a particular dataset. One of our goals is to provide an algorithm behaving in a similar way whatever the scale and the resolution of the data are. In particular, the user should be allowed to tune the parameters at coarse resolution (when computation times are low) and perform a matching or an atlas estimation at high resolution with only few updates for the parameters values. Even if these refinements are not fundamental in theory, it makes the algorithm usable in practice and it has the great advantage of normalizing the choice of the balance parameters $\gamma_{0},\gamma_f^0,\gamma_V,\gamma_f,\gamma_W$ between the various penalization terms and the data attachment term, as well as the initial step sizes in Algorithm \ref{algo:optim} (these quantities may vary considerably otherwise). 

In this section, $\bJ$ denotes either the functional $\bJ_{0}^{tan}$ of equation \eqref{eq:FB45d} (Algorithms  \ref{algo:j0tan} and \ref{algo:dj0tan}) or $\bJ_{free}^{tan}$ of equation \eqref{eq:free} (Algorithms  \ref{algo:Jfreetan} and \ref{algo:dJfreetan}).

\subsubsection{Scale invariance} \label{part:scale}

We first describe how we normalize the functional $\bJ$ and its gradient in order to provide an algorithm with outputs that are invariant if a scaling is applied to the geometry or to the signal of the fshapes. To do so, we carefully examine the effect of a scaling (in the geometrical or in the signal space) on each term composing $\bJ$. Without this normalization step, the balance between the different terms of $\bJ$ would be modified by scale which is not desirable in practice.

Let us examine the behaviour of the functional when we apply a scaling on the geometrical space $\x\mapsto \x' = \lambda_e \x\in E^{\nP}$ and on the signal $\f \mapsto \f' = \lambda_f \f \in \R^{\nP}$ for some $\lambda_e,\lambda_f>0$. We assume that the scale parameters $\sigma_e$ and $\sigma_f$ of the Gaussian kernels $k_e$ and $k_f$ are also rescaled as they measure the size at which the fshapes are compared. These new kernels are denoted $k_e'$ and $k_t'$ and they verify $k_e'(\x',\y') = k_e(\x,\y)$ and $k_t'(\f',\g') = k_t(\f,\g)$. The scaled local volume elements  $r_\ell'$ now satisfy $r_\ell' = \lambda^{\dT}_e r_\ell$ and the scaled momenta $\p' = \lambda_e \p$ so that the displacement generated by $\p'$ is also scaled. All the terms composing $\bJ$ are normalized in the same way and we just detail the case of the fvarifold data attachment term $g_i$. By formula \eqref{eq.discreteprs}, we easily see that the scaled data attachment term is proportional to $\lambda^{2\dT}_e$ and does not depend on the scale of the signal, namely 
\[
	g_i'((\x',\f'),({\x^i}',{\f^i}'))  = \lambda^{2\dT}_e g_i((\x,\f),(\x^i,\f^i)).
\]
Therefore, we may compute a normalized version $(R_{e})^{-2\dT} g_i$ of $g_i$  with $R_{e} \doteq \max_{i} \left( \sqrt{\trace(\var(\x^i)) }\right)$  where $\var(\x^i)$ is the covariance matrix of the points cloud $\x^i$. From there on $\bJ_n$ denotes the functional containing the normalized terms of $\bJ$ (\ie each term being multiplied by an appropriate power of $R_{e}$ and an appropriate power of $R_{f} \doteq \max_i\sqrt{\var(\f^i)}$). By construction we (formally) have 
\[
\bJ_n(\x,\p,\f) = \bJ_n'(\x',\p',\f')
\]
where $\bJ_n'$ is computed for $\sigma'_e=\lambda_e\sigma_e$ and $\sigma'_f=\lambda_f\sigma_f$. 

We now examine the gradient's behaviour under scaling. The goal is to keep a similar dynamic during the gradient descent whatever the (geometric or functional) scales of the fshapes are. As an illustration, consider the gradient with respect to $\x$: we need $\bJ_n(\x-\delta\nabla_{\x} \bJ_n,\p,\f) = \bJ_n'(\x' - \delta\nabla_{\x'}\bJ_n',\p',\f')$ which  yields to 
\[
	\nabla_{\x'} \bJ_n' = \lambda_e \nabla_{\x} \bJ_n.
\]
As $\nabla_{\x'} g_i' = \lambda_e^{-1}\nabla_{\x} g_i$ we use the normalized version $(R_e)^{2-2\dT} \nabla_{\x} g_i$ of the gradient $\nabla_{\x} g_i$. Each term composing the gradient $\nabla \bJ_n$ is then normalized with an appropriate power of $R_e$ and $R_f$ and we denote $\nabla_n \bJ_n$ the normalized version $\nabla \bJ_n$. 

In our implementation of algorithm \ref{algo:j0tan} and \ref{algo:Jfreetan} (resp. \ref{algo:dj0tan} and \ref{algo:dJfreetan}), we compute $\bJ_n$ (resp. $\nabla_n \bJ_n$) rather than $\bJ$ (resp. $\nabla \bJ$). These normalizations then guarantee that using similar but rescaled (or translated and rotated) data will provide comparable energies and energy decreases during gradient descent.

\subsubsection{Consistency with respect to $\nP$}

We now focus on the behaviour of the template estimation algorithm when the sampling of the discretized template changes. The methodology and notations will be similar to Section \ref{part:scale}. Assume that we are working with surfaces: we wish to understand qualitatively the asymptotic behaviour of the expressions of the functional $\bJ$ and its gradient when the mesh on the template is refined (\ie when $\nP$ and $\nT$ increase). The idea is once again to keep a similar behaviour in the optimization process at coarse and fine resolutions. We illustrate the invariance with respect to the number of points by a numerical experiment presented in Figure \ref{fig:invariancePres} and \ref{fig:invariance} where the same geometrico-functional matching is performed at various resolutions.

Let $(\x^{\nP},\f^{\nP})_{\nP}$ be a sequence of regularly discrete fshapes sampled from a continuous fshape $(X,f)$ and containing respectively $\nP$ points. Let also $(\p^{\nP})_{\nP}$ be a sequence of discretized momenta (sampled from a continuous vector field $p$) attached to the points of $(\x^{\nP},\f^{\nP})$. If the discretization is sufficiently uniform, we may consider that the typical size of a local volume element $r_\ell$ is of order $\nP^{-1}$. When the number $\nP$ of points is large, we then have  
\[
\sabs{\f^{\nP}}^2_{\x^{\nP}} \doteq \sum_\ell \hat f_{\ell}^2 r_\ell \approx \int_X f^2(x) d\mathcal{H}^d(x)
\]
where $\hat f$ is defined by formula \eqref{eq.fvarRepresentation}. Notice that other discretization methods may be used to approximate the $L^2$ norm on $X$ as discussed in Section \ref{Dyn_tangential}. If we now assume that $\|\mu_{(\x^{\nP},\f^{\nP})}-\mu_{(X^{},f^{})}\|_{W'} \to 0 $, then formula \eqref{eq.discreteprs} implies that 
\[
	\snorm{\mu_{(\x^{\nP},\f^{\nP})}}_{W'}^2 \approx \snorm{\mu_{(X,f)}}_{W'}^2
\] 
when $\nP$ is large. Finally, if $\snorm{v^{\p^{\nP}} - v^p}_V \to 0$ 
then the initial velocity fields satisfies 
\begin{equation}\label{eq:momd}
	v^{\p^{\nP}}(\cdot) = \sum_{k=1} K(\cdot,x_k)p^P_k \approx \int_X K(\cdot,x)p(x) d\mathcal{H}^d(x) = v^p(\cdot)
\end{equation}
when $\nP$ is large. This means that the magnitude of the momenta should be proportional to the inverse of the density of points to generate comparable displacements at various resolution.  In particular, the update of the momentums $\p^{\nP}$ should be of order $\nP^{-1}$ as equation \eqref{eq:momd} suggests and $\nabla_{\p^{\nP}} \bJ$ is multiplied, in our code, by $\nP^{-1}$ to be at the right scale. In the same spirit, it can be shown that the gradients $\nabla_{\x^{\nP}} \bJ$ and $\nabla_{\f^{\nP}} \bJ$ are of order $\nP^{-1}$. We then multiply these terms by $\nP$ to ensure the homogeneity of the update and keep a similar dynamics along the optimization procedure even if the resolution of the meshes changes. We provide an example in Figure \ref{fig.enr} where the values of the functional along the gradient descent at various resolutions are plotted.

	\begin{figure}[H]
		\centering
			\subcaptionbox[.52\textwidth]{The source $(\x, \f_{})$ and the target $(\x^1,\f^1)$. Signal $\f=0.5$ is constant and $\f^1$ is equals to 1 on the head and on the tail and 0 elsewhere.\label{fig.pres}}{\includegraphics[width=.40\textwidth]{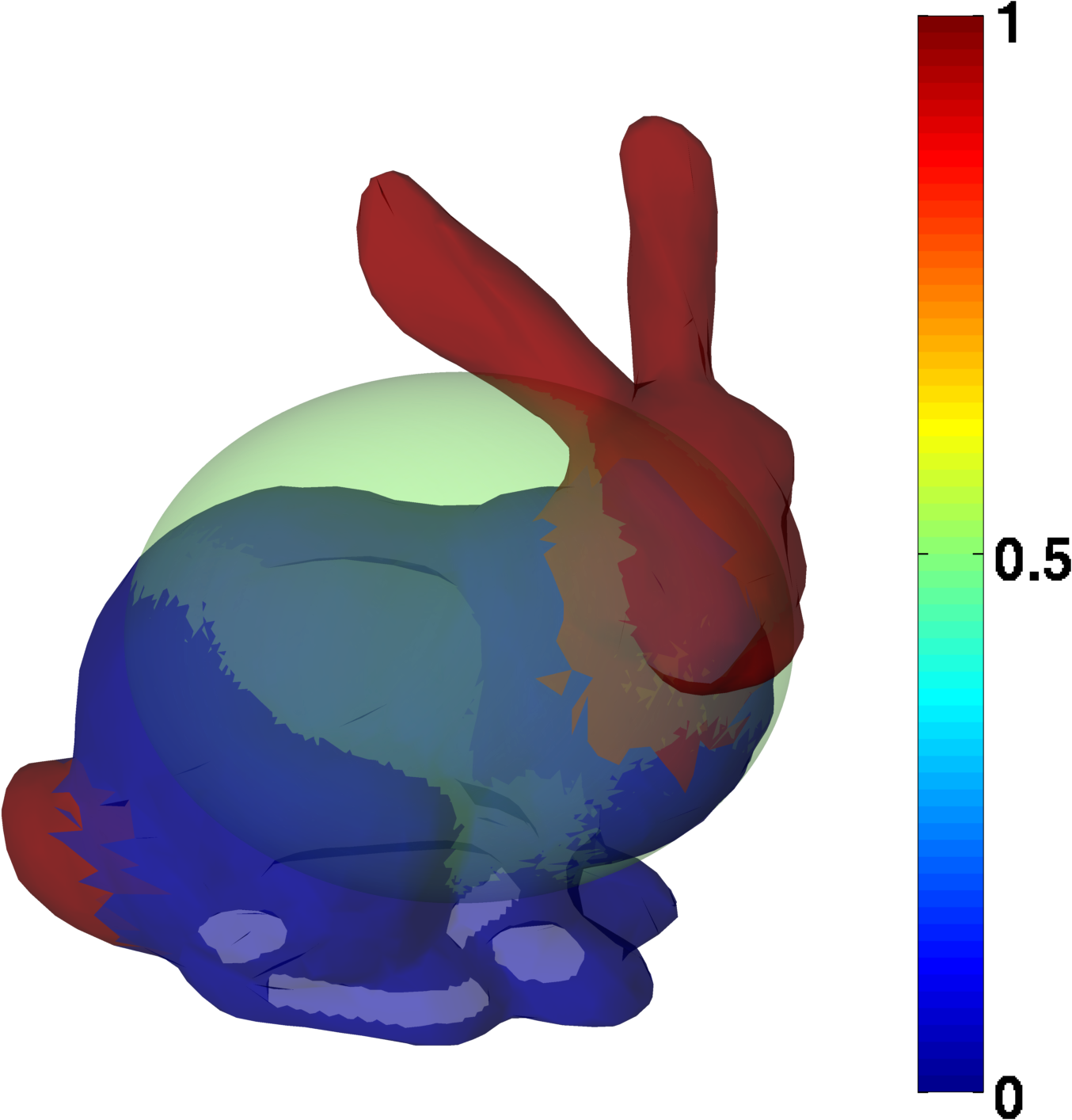}}
		 \hfill
		 \subcaptionbox[.48\textwidth]{Value of $\bJ^{tan}_0$ (y-axis) as a function of the step number (x-axis) at three different resolutions.\label{fig.enr}}{\includegraphics[width=.48\textwidth]{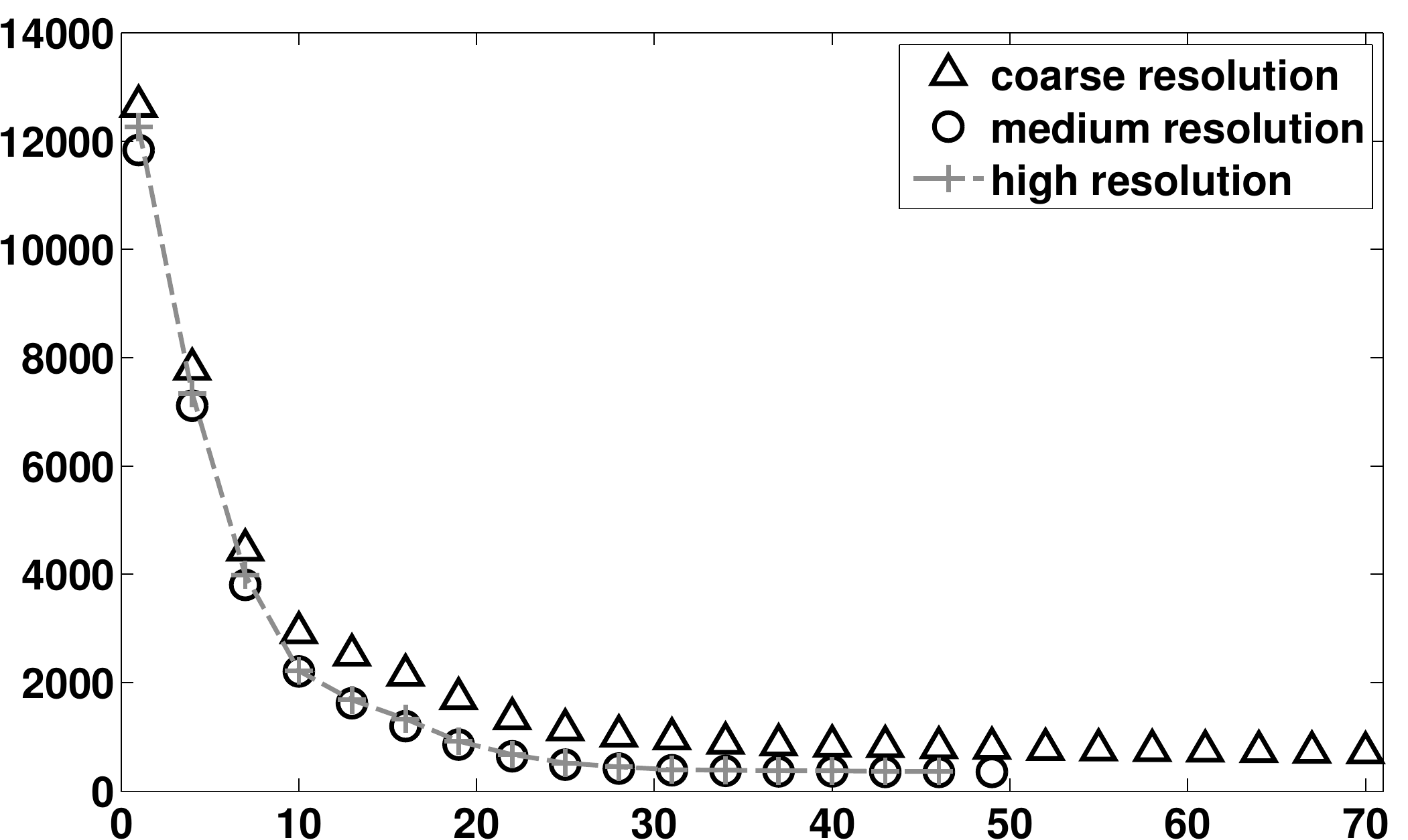}}
			\caption{Geometrico-functional registrations of an ellipsoid onto a functional version of the Stanford's Bunny. The fshapes are depicted in \ref{fig.pres} at medium resolution (around 15000 points). The experiment consists in minimizing the functional $\bJ^{tan}_0$ of equation \eqref{eq:FB45d} in $\p^1$ and $\bzeta^1$ only (here $N=1$). 
			The results are presented Figure \ref{fig:invariance}. \label{fig:invariancePres}}
			\end{figure}
	\begin{figure}[H]
		\centering
		\begin{tabular}{ccc}
			\parbox[t]{.65\textwidth}{\subcaptionbox[.65\textwidth]{ Geometrico-functional registration (at time $t=0$, $t=0.5$ and $t=1$) \label{fig.sourceLR}}{\includegraphics[width=.20\textwidth]{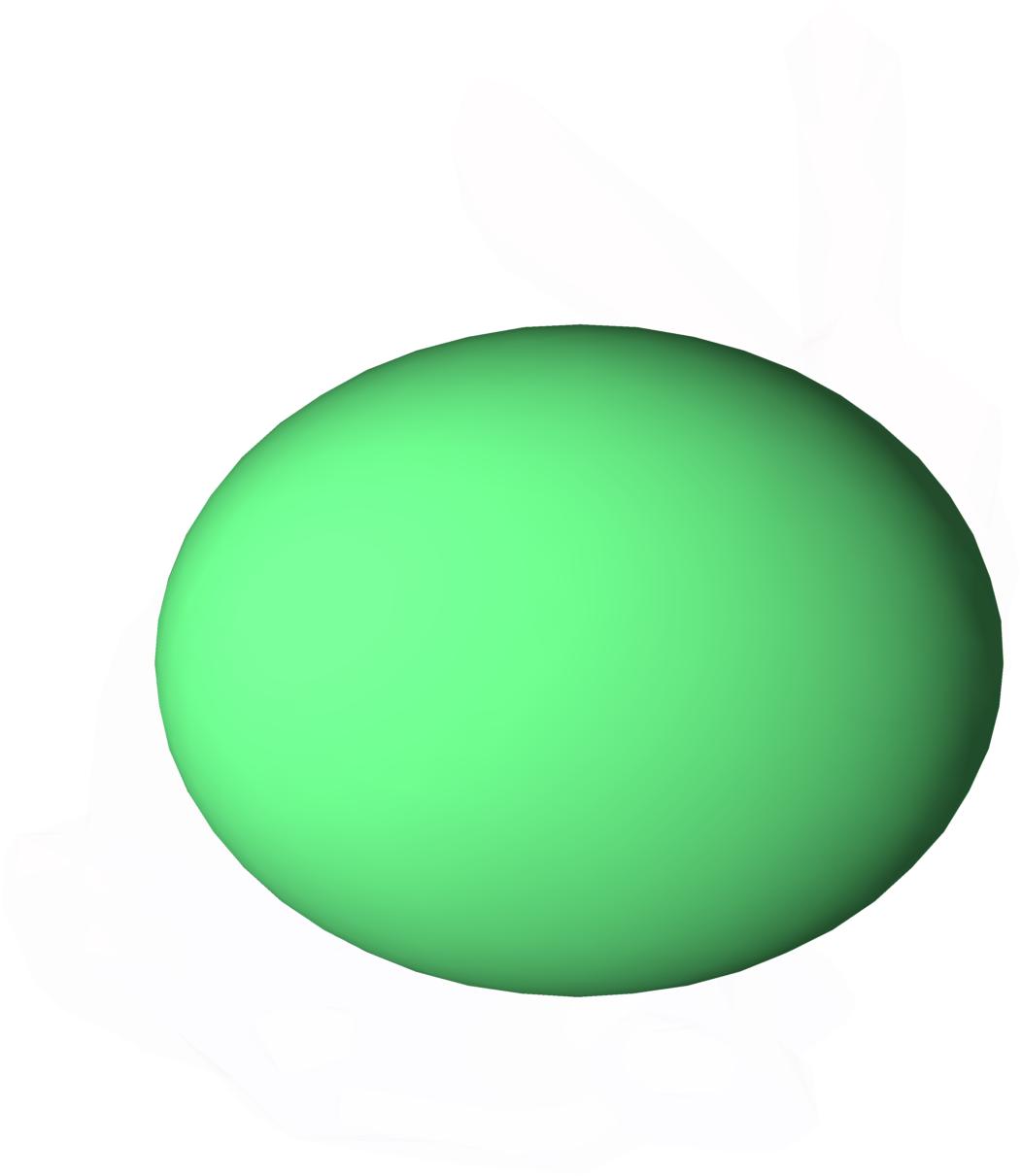}\includegraphics[width=.20\textwidth]{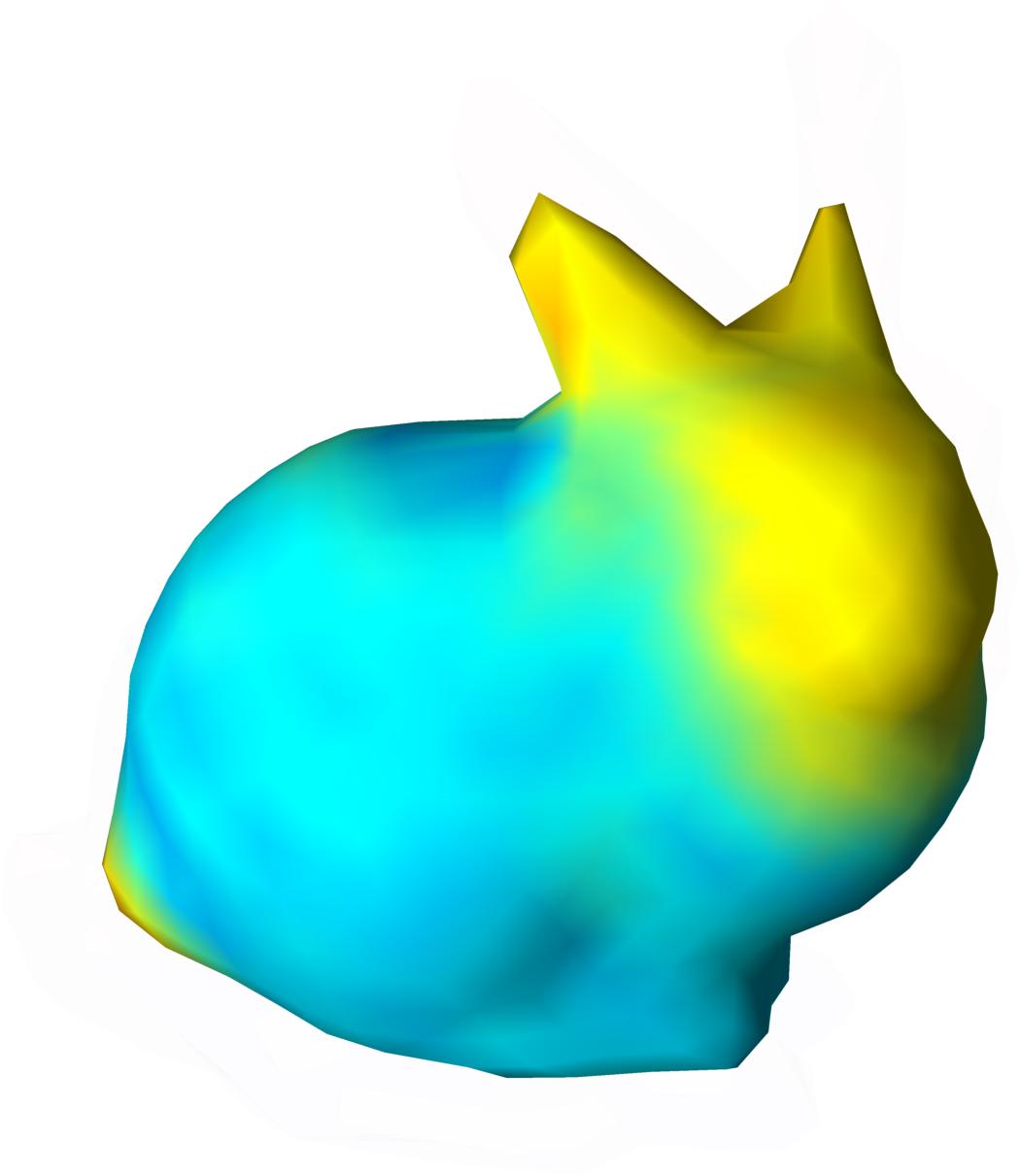}\includegraphics[width=.20\textwidth]{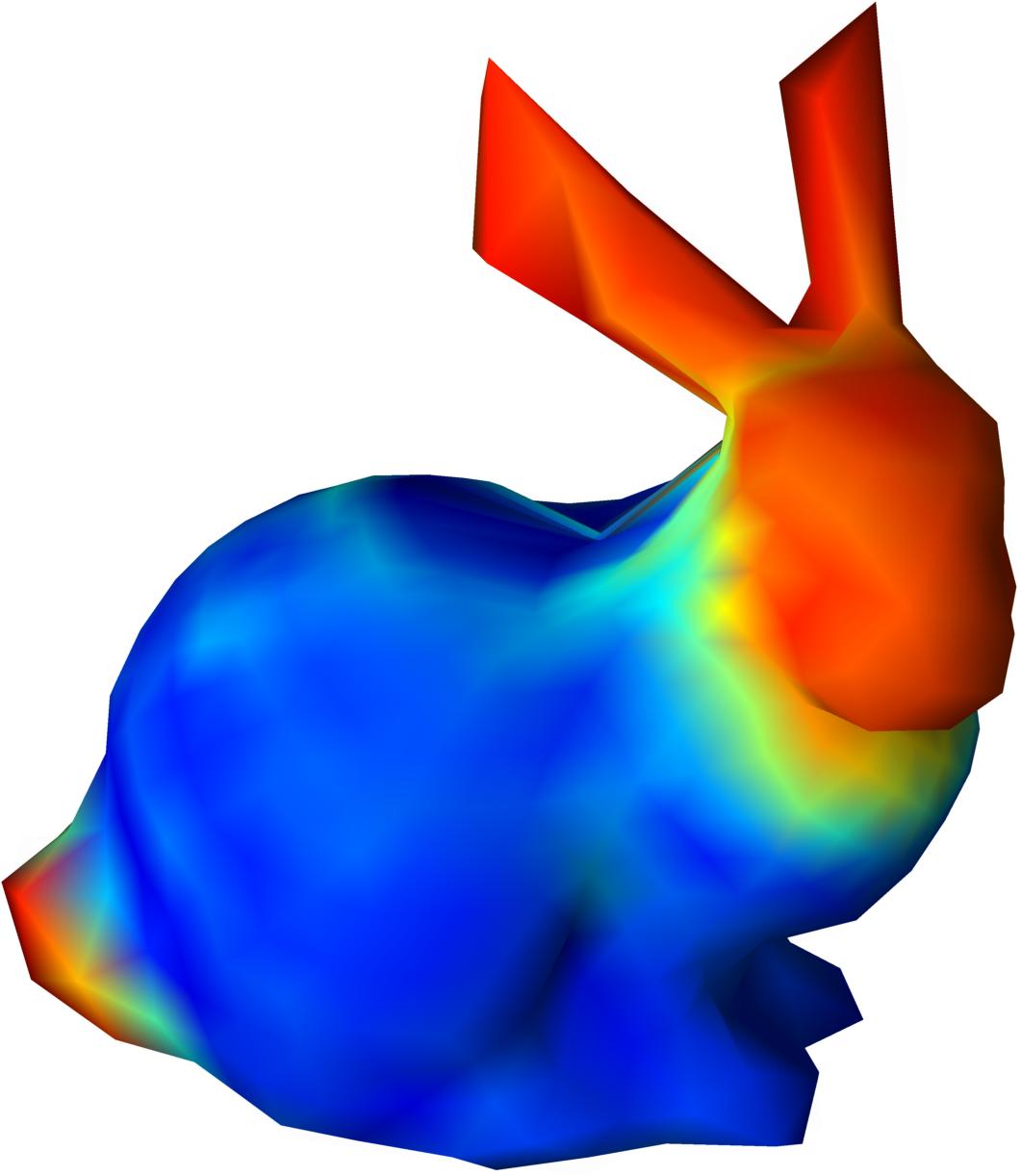}}}
			\parbox[t]{.2\textwidth}{\subcaptionbox[]{Target \label{fig.resultLR}}{\includegraphics[width=.20\textwidth]{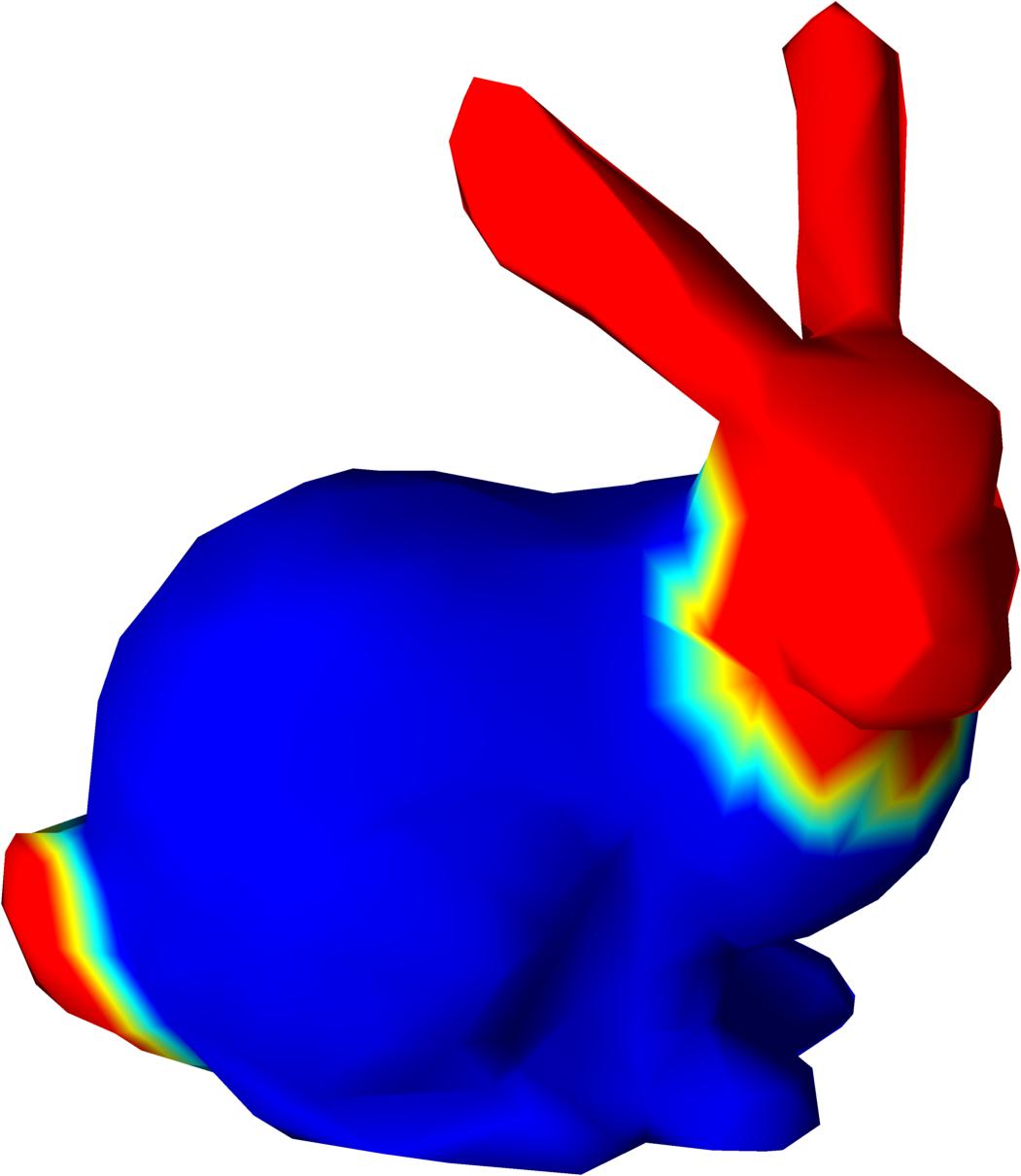}}}
			&\raisebox{-.1\height}{\hspace*{1.5em}\includegraphics[height=4cm]{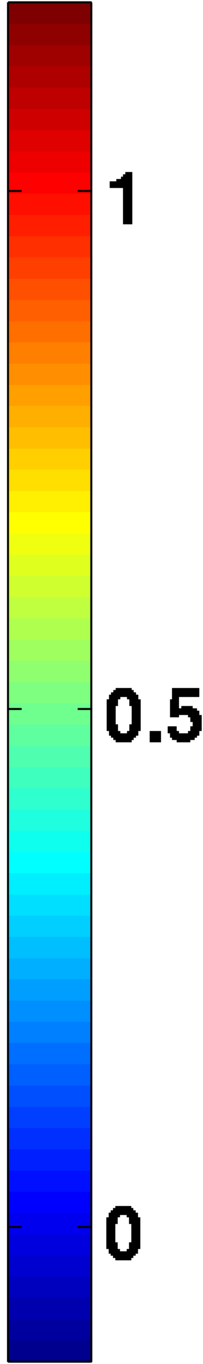}} \\
			\parbox[t]{.65\textwidth}{\subcaptionbox[.65\textwidth]{Geometrico-functional registration (at time $t=0$, $t=0.5$ and $t=1$)\label{fig.sourceHR}}{\includegraphics[width=.20\textwidth]{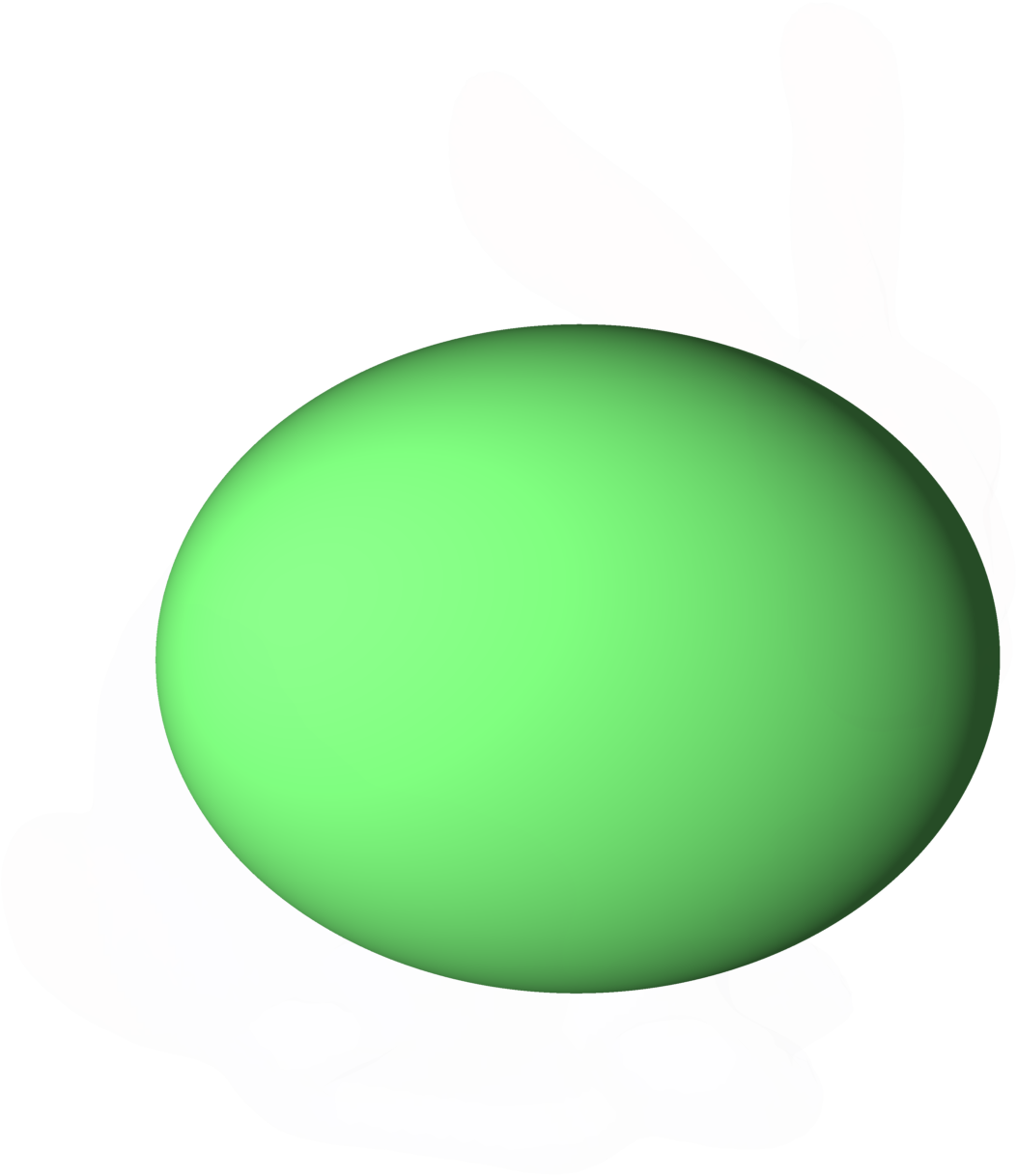}\includegraphics[width=.20\textwidth]{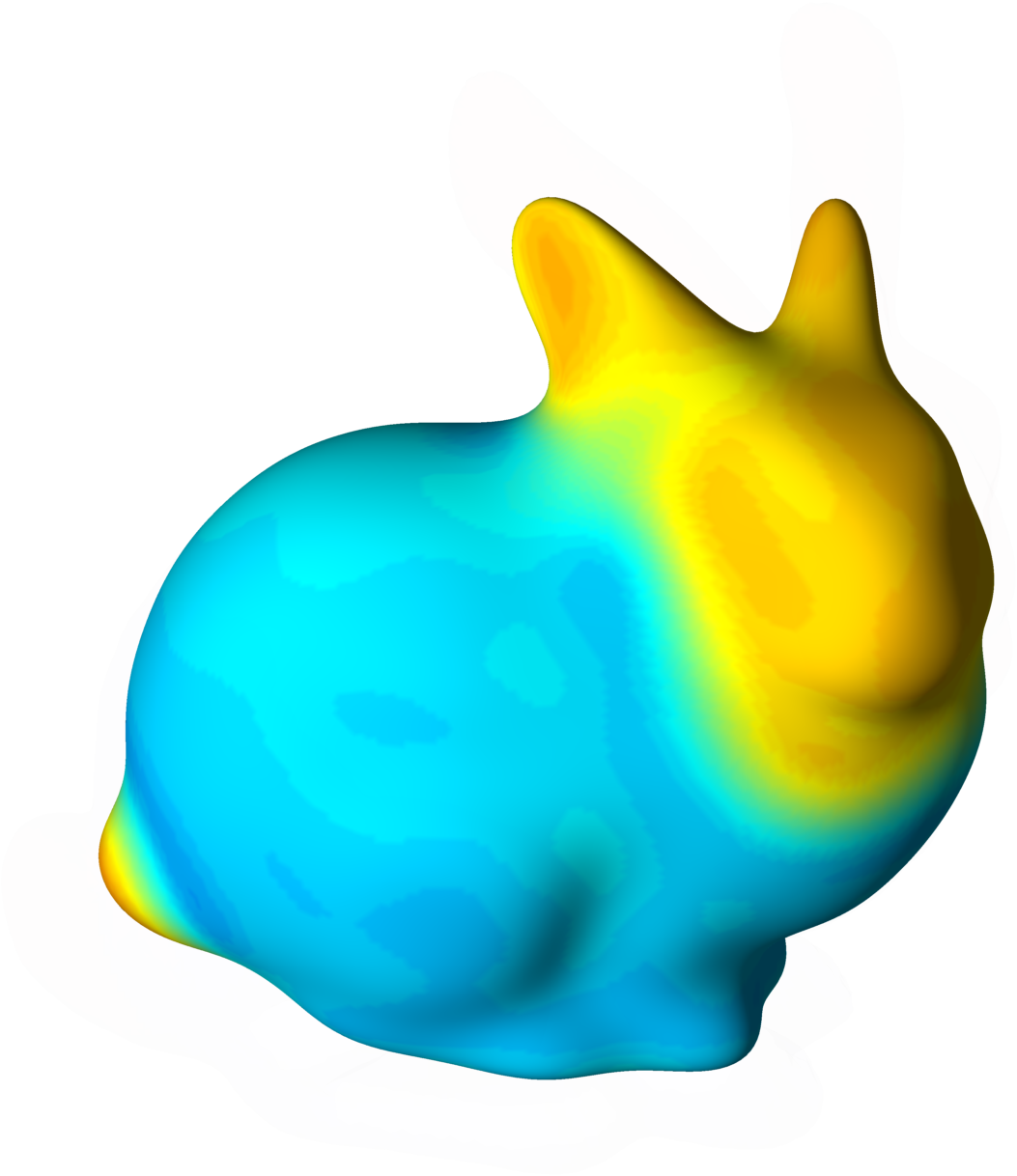}\includegraphics[width=.20\textwidth]{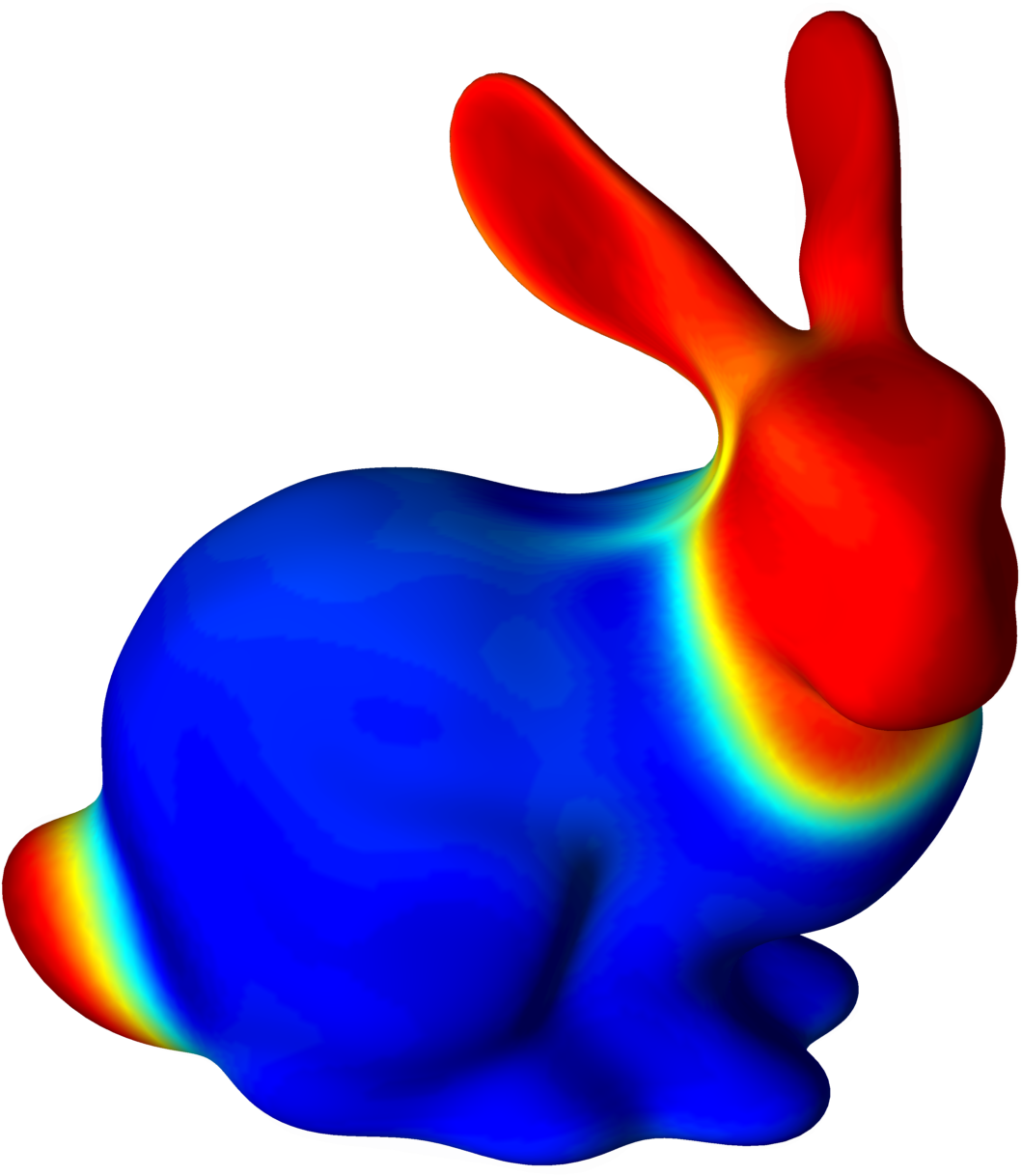}}}
			\parbox[t]{.2\textwidth}{\subcaptionbox[]{Target \label{fig.resultHR}}{\includegraphics[width=.20\textwidth]{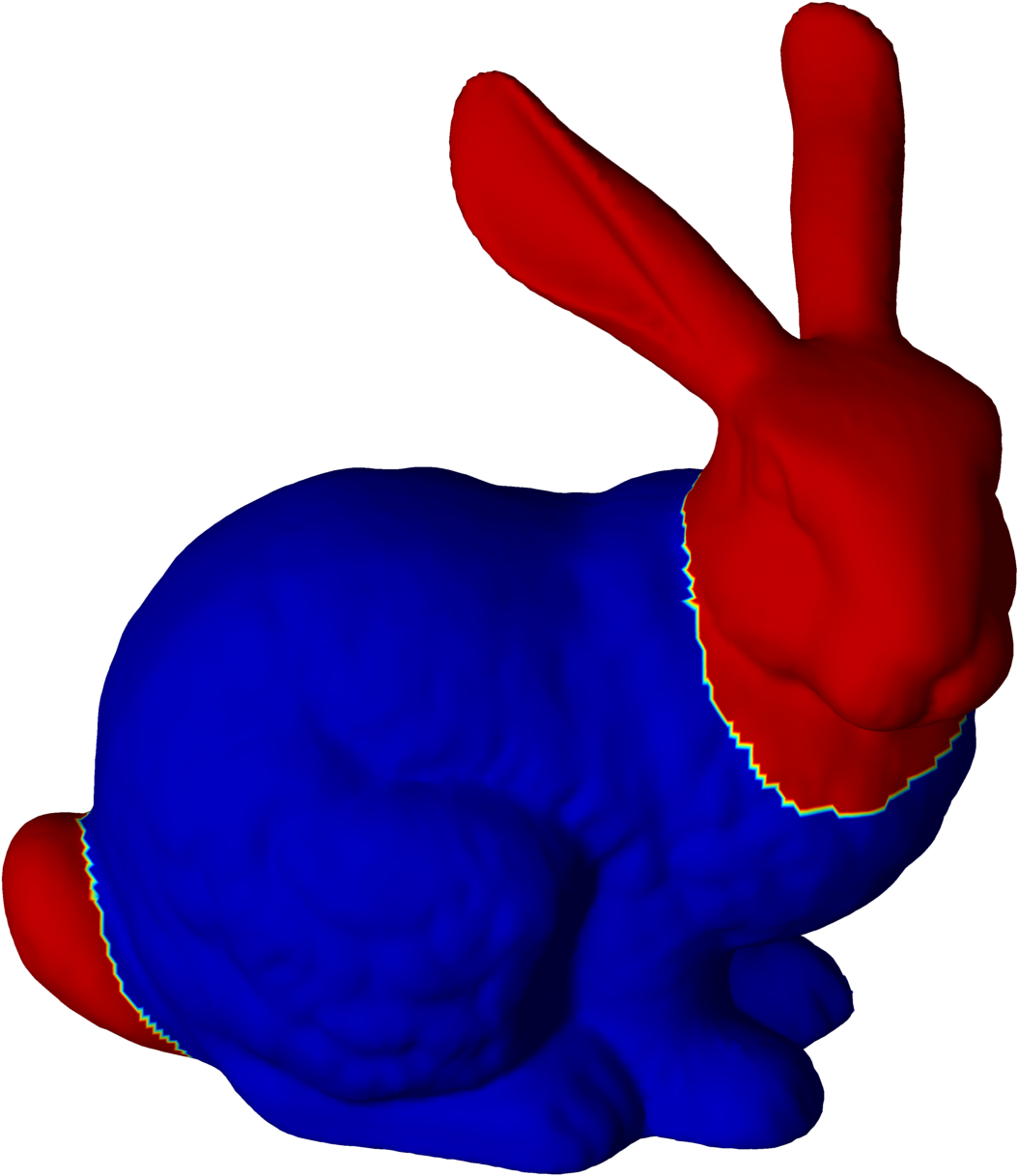}}}
			& \raisebox{-.1\height}{\hspace*{1.5em}\includegraphics[height=4cm]{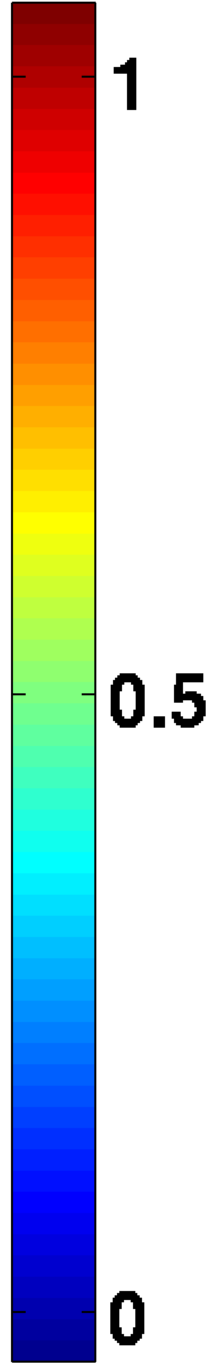}} 
		\end{tabular}
	\caption{Results of registration presented Figure \ref{fig:invariancePres}. First row: coarse resolution (around 1200 points) computations takes 30 seconds. Second row: high resolution (around 80000 points) computations take 2.5 hours.\label{fig:invariance}}
	\end{figure}

\section{Conclusion}
In this article, we have proposed a well-formalized and efficient extension of the ideas of shape spaces for the treatment of functional shapes. To do so, we have introduced the structure of fshapes' vector bundle and the metamorphosis setting to model and quantify geometrico-functional transformations within these bundles. In addition, the concepts of geometric measure theory (varifolds in particular) were generalized to provide dissimilarity metrics between such objects. The combination of these settings enabled us to express atlas estimation on populations of fshapes as a variational problem, for which we were able to prove existence of solutions. 

In the second part of the paper, we addressed the issue of providing practical numerical schemes to efficiently perform the functionals' optimization. In particular, we detailed the discrete expressions corresponding to the fvarifolds' data attachment terms, proposed a gradient descent algorithm for the estimation of all variables in the atlas and carefully examined some of the important numerical issues related to these algorithms. The resulting codes includes atlas estimation on populations of classical curves and surfaces and extends it to fshapes. They shall be made publicly available shortly. This was implemented under the simplified 'tangential model'~: the full metamorphosis setting derived in the theoretical part of the paper is likely to be implemented as well in the near future.    

Extensions of the present framework are possible in several directions worth exploring as future work. One of them is to consider fshape bundles modeled on more regular spaces than $L^2$ and hopefully recover existence of solutions in a more direct way than the proof presented in this paper, at the price of more involved numerical computations. Another track is to generalize such a framework to more general spaces of signals (e.g vector fields or tensor fields) and more general group actions, which has been considered only partially in \cite{Charon_thesis}. Finally, one important follow-up problem to atlas estimation is the one of statistical analysis and classification based on these fshape atlases~: we have deliberately postponed such questions to an upcoming paper.     

\subsection*{Acknowledgment}

We would like to thank Mirza Faisal Beg, Sieun Lee, Evgeniy Lebed, Marinko Sarunic and their collaborators for providing the OCT dataset and for fruitful discussions. The authors also acknowledge the support of the French \textit{Agence Nationale de la Recherche} project HM-TC (number ANR-09-EMER-006).

\appendix
\section{Variation formula for fvarifolds: proof of Theorem \ref{theo:variation_formula}}
\label{appendix:varition_fvar}
The proof follows the same steps as the corresponding result for usual varifolds (cf \cite{Charon2}). Given a $C^{1}$ vector field $v$ on $E$ with compact support, we can consider the $1$-parameter group of diffeomorphisms $\phi_{t}$ with $\phi_{0}=Id$ and $\rst{\partial_{t}}{t=0} \phi_{t} =v$. Then, it follows that:
 \begin{align}
 \label{eq:Lie_der1}
 (\pounds_{(v,h)}\omega)(x,T_{x}X,f(x)) &= \rst{\dfrac{d}{dt}}{t=0} (\psi_{t}^{\ast}\omega)(x,T_{x}X,f(x)) \nonumber \\
  &=\rst{\dfrac{d}{dt}}{t=0} |\rst{d_{x}\phi_{t}}{T_{x}X}|.\omega(\phi_{t}(x),d_{x}\phi_{t}(T_{x}X),f(x)+th(x)) \,.
 \end{align}
 As we see, the previous leads to several terms in the derivative: differentiate the volume change term $J_t\doteq |\rst{d_{x}\phi_{t}}{T_{x}X}|$, the function $\omega$ with respect to the position variable, with respect to the tangent space direction and to the signal part. The derivative with respect to point positions and signal values are easy to obtain and equals respectively, since $\omega$ is assumed to be $C^{1}$, $\left ( \dfrac{\partial \omega}{\partial x} | v \right )$ and $\dfrac{\partial \omega}{\partial m}h$. The two other terms require more attention. \\
\\
\textbf{Derivative of the volume change:} For any vector field $u$ defined on $X$, we shall denote by $u^{\top}$ and $u^{\bot}$ the tangential and normal components of $u$ with respect to the tangent space of $X$ at each point. We also introduce the connection $\nabla_{\cdot}\cdotp$ on the ambient space and an orthonormal frame of tangent vector fields $(e_{i})_{i=1,..,d}$ on $X$. Now $J_{t}=\sqrt{\det([\langle d_{x}\phi_{t}(e_{i}),d_{x}\phi_{t}(e_{j})]_{i,j})}$ so a simple calculation shows that:
\begin{equation*}
	\rst{\dfrac{d}{dt}}{t=0}  J_{t} = \sum_{i=1}^{d} \langle e_{i}, \nabla_{e_{i}}v \rangle
\end{equation*}
Writing $v=v^{\top}+v^{\bot}$ provides a first term $\sum_{i=1}^{d} \langle e_{i}, \nabla_{e_{i}}v^{\top} \rangle$ which is the tangential divergence of the vector field $v^{\top}$ denoted usually $\mdiv_{X}(v^{\top})$. The second term becomes $\sum_{i=1}^{d} \langle e_{i}, \nabla_{e_{i}}v^{\bot} \rangle$. For all $i=1,..,d$, we have $\langle e_{i}, v^{\bot} \rangle = 0$ so that after differentiation we find that $\langle e_{i}, \nabla_{e_{i}}v^{\bot} \rangle = - \langle \nabla_{e_{i}} e_{i}, v^{\bot} \rangle$. Therefore: 
\begin{eqnarray*}
 \sum_{i=1}^{d} \langle e_{i}, \nabla_{e_{i}}v^{\bot} \rangle &=& -\sum_{i=1}^{d} \langle \nabla_{e_{i}} e_{i}, v^{\bot} \rangle \\
 &=& -\left \langle \left (\sum_{i=1}^{d} \nabla_{e_{i}} e_{i} \right )^{\bot}, v^{\bot} \right \rangle \,.
\end{eqnarray*}
In this last expression, we recognize the \textbf{mean curvature vector} to the submanifold $X$, which is the trace of the Weingarten map and is denoted $H_{X}$. As a result, we find that:
\begin{equation*}
	\int_{X} \omega  \rst{\dfrac{d}{dt}}{t=0} J_{t} = \int_{X} \omega  \mdiv_{X}(v^{\top}) - \int_{X} \omega \langle H_{X},v^{\bot} \rangle \,.
\end{equation*}
Now, the first term can be rewritten as a boundary integral by applying the Divergence Theorem. Indeed, if we denote by $\tilde{\omega}$ the function defined on $X$ by $\tilde{\omega}(x)=\omega(x,T_{x}X,f(x))$ which is $C^{1}$, we have $\mdiv_{X}(\tilde{\omega}v^{\top})=\tilde{\omega}\mdiv_{X}(v^{\top})+\langle \nabla \tilde{\omega} | v^{\top} \rangle$. Applying the Divergence Theorem (cf \cite{Simon} Section 7) on the submanifold $X$ gives:
\begin{equation*}
 \int_{X} \omega  \mdiv_{X}(v^{\top}) = -\int_{X} \langle \nabla \tilde{\omega} | v^{\top} \rangle + \int_{\partial X} \omega\langle \nu , v^{\top} \rangle
\end{equation*}
where $\nu$ is the unit outward normal to the boundary. \\
\\
\textbf{Derivative with respect to tangent spaces:} We now come to the derivative term on the tangent space part in equation \eqref{eq:Lie_der1}. To explicitly compute variations with respect to variables in the Grassmann varifold, the most convenient way is to use the embedding of $G_{d}(E)$ into $\mathcal{L}(E)$ that identifies any subspace $V$ with the orthogonal projector $p_{V}$ on $V$. With this identification, one can represent the tangent space at $V$ of $G_{d}(E)$ as $\mathcal{L}(V,V^{\bot})$. Then, as explained with more details in \cite{Charon2}, if we set $V_{t}=d_{x} \phi_{t}(T_{x}X)$, one can show that:
\begin{equation*}
	\rst{\dfrac{d}{dt}}{t=0} V_{t} = p_{T_{x}X^{\bot}} \circ \rst{\nabla v}{T_{x}X} \in \mathcal{L}(T_{x}X,(T_{x}X)^{\bot})\,.
\end{equation*}
We can now introduce $\dfrac{\partial \omega}{\partial V}$ as an element of $\mathcal{L}(T_{x}X,T_{x}X^{\bot})^*\approx(T_{x}X^{\bot})^{*} \otimes T_{x}X$ and which we can write: $\dfrac{\partial \omega}{\partial V} = \sum_{j=d+1}^{n} \eta_{j}^{*}\otimes \alpha_{j}$ for $(\eta_{d+1},..,\eta_{n})$ an orthonormal frame of $T_{x}X^{\bot}$ and $(\alpha_{j})$ vectors of $T_{x}X$ (as usual $\eta^*$ denotes the linear form $\langle \eta,.\rangle$). Then the variation we wish to compute is:
\begin{equation*}
 \left ( \dfrac{\partial \omega}{\partial V} | \nabla v \right ) = \sum_{j=d+1}^{n} \langle \eta_{j}, \nabla_{\alpha_{j}}v \rangle\,.
\end{equation*}
If we introduce $\left( \dfrac{\partial \omega}{\partial V} | v \right )= \sum_{j=d+1}^{n} \eta_{j}^{*}(v)\alpha_{j} = \sum_{j=d+1}^{n} \langle \eta_{j},v \rangle\alpha_{j}$ which is a tangent vector field on $X$, we have:
\begin{equation*}
 \mdiv_{X}\left( \dfrac{\partial \omega}{\partial V} | v \right ) = \sum_{i=1}^{d} \sum_{j=d+1}^{n} \left( \langle e_{i}, \nabla_{e_{i}} \alpha_{j} \rangle\langle \eta_{j},v \rangle + \langle e_{i}, \langle \nabla_{e_{i}}\eta_{j},v \rangle \alpha_{j} \rangle + \langle e_{i}, \langle \eta_{j},\nabla_{e_i}v \rangle \alpha_{j} \rangle \right)
\end{equation*}
The last term in the sum is also $\sum_{j=d+1}^{n} \langle \eta_{j},\nabla_{\alpha_{j}}v \rangle$, which is nothing else than $\left ( \dfrac{\partial \omega}{\partial V} | \nabla v \right )$. As for the two other terms in the sum, it is easy to see that it equals:
\begin{equation*}
 \left ( \sum_{i=1}^{d} \langle e_{i}, \nabla_{e_{i}}\sum_{j=d+1}^{n} \eta_{j}^{*}\otimes \alpha_{j} \rangle | v \right ) = ( \mdiv_{X} \left (\dfrac{\partial \omega}{\partial V} \right ) | v )
\end{equation*}
Hence, it follows that: 
\begin{equation}
\label{eq:Lie_der2}
 \left ( \dfrac{\partial \omega}{\partial V} | \nabla v \right ) = \mdiv_{X}\left( \dfrac{\partial \omega}{\partial V} | v \right ) - ( \mdiv_{X} \left (\dfrac{\partial \omega}{\partial V} \right ) | v )
\end{equation}
Integrating equation (\ref{eq:Lie_der2}) over the submanifold $X$ and using the Divergence Theorem as before, we find that:
\begin{equation}
\label{eq:Lie_der3}
 \int_{X} \left ( \dfrac{\partial \omega}{\partial V} | \nabla v \right ) = \int_{\partial X} \langle \nu , \left( \dfrac{\partial \omega}{\partial V} | v \right ) \rangle - \int_{X} ( \mdiv_{X} \left (\dfrac{\partial \omega}{\partial V} \right ) | v ) 
\end{equation}
\vskip2ex
\noindent \textbf{Synthesis:} Summing all the different terms from eq.\eqref{eq:Lie_der1}, we eventually obtain:
\begin{eqnarray*}
 \int_X (\pounds_{(v,h)}\omega) &=& \int_{X} \left ( \dfrac{\partial \omega}{\partial x} - \mdiv_{X} \left (\dfrac{\partial \omega}{\partial V} \right ) | v \right ) -\int_{X} \langle \nabla \tilde{\omega} | v^{\top} \rangle + \omega\langle H_{X} | v^{\bot} \rangle \\
 & & +\int_{X} \dfrac{\partial \omega}{\partial m}. h  + \int_{\partial X} \langle \nu, \left( \dfrac{\partial \omega}{\partial V} | v \right ) + \omega v^{\top} \rangle \,.
\end{eqnarray*}
Now, we remind that $\tilde{\omega}(x)=\omega(x,T_{x}X,f(x))$ so 
$$(\nabla \tilde{\omega} | v^{\top} ) = \left ( \dfrac{\partial \omega}{\partial x} | v^{\top} \right ) + \left ( \dfrac{\partial \omega}{\partial V} | \nabla v^{\top} \right ) + \dfrac{\partial \omega}{\partial m} \langle \nabla f,v^{\top} \rangle$$
and applying the result of equation (\ref{eq:Lie_der2}) to $v^{\top}$:
\begin{equation*}
\left ( \dfrac{\partial \omega}{\partial V} | \nabla v^{\top} \right ) = \mdiv_{X}\left( \dfrac{\partial \omega}{\partial V} | v^{\top} \right ) - ( \mdiv_{X} \left (\dfrac{\partial \omega}{\partial V} \right ) | v^{\top} )
\end{equation*}
We notice that $\left( \dfrac{\partial \omega}{\partial V} | v^{\top} \right ) = 0$ by the expression of $\dfrac{\partial \omega}{\partial V}$ and using the equality $v=v^\top+v^\bot$ we find eventually that:
\begin{equation*}
\int_X \pounds_{(v,h)}\omega = \int_{X} \left ( \dfrac{\partial \omega}{\partial x} - \mdiv_{X} \left (\dfrac{\partial \omega}{\partial V} \right ) - \omega H_{X} | v^{\bot} \right ) + \dfrac{\partial \omega}{\partial m}.(h-\langle \nabla f,v^{\top}\rangle) + \int_{\partial X} \langle \nu, \left( \dfrac{\partial \omega}{\partial V} | v \right ) + \omega v^{\top} \rangle
\end{equation*}
which proves the result of Theorem \ref{theo:variation_formula}. 

\section{Proof of Proposition \ref{prop:JP1}}
\label{appendix:proof_propJP1}

\subsection{Perturbation}
\label{ssec:perturbation}
We introduce now a perturbation process on any measure $\nu$ on $E \times G_{d}(E) \times \R$ that shall be useful for the following. Let $a>0$ to be fixed later and consider for any $t\in\R$ the function $\rho_t:\R\to\R$ such that
\begin{equation}
	\rho_t(z)=z+t(\sgn(z)a-z)\one_{|z|>a}
	\label{eq:extra9}
\end{equation}
where $\sgn(z)$ is the sign of $z$. We have $\rho_0=\text{Id}_\R$ and $\rho_1$ is a symmetric threshold at level $a$. Now for any $t\in\R$, we denote $\nu_t$
the new measure defined for any $\omega \in C_b(E \times G_{d}(E) \times \R)$ as:
\begin{equation}
  \label{eq:JP4}
  \nu_t(\omega)=\int \omega(x,V,\rho_t(f))d\nu(x,V,f)\,.
\end{equation}
Obviously $\nu_0=\nu$ and $\nu_1$ is such that $\nu_1(|f|>a)=0$ so that $t\mapsto \nu_t$ is an homotopy from $\nu$ to a measure under which the signal is a.e. bounded by $a$.

\subsection{Proof of Lemma \ref{lemma:JP1}}
We show the existence of a fvarifold minimizer in $\MeasX$ (cf \eqref{eq:JP8}) for the extended functional $\eJ$. For any $\nu\in\MeasX$ and $t\in\R$, we denote $J_t\doteq \eJ(\nu_t)$ where $\nu_{t}$ is the previously defined perturbation of $\nu$ (cf \ref{ssec:perturbation}) and we assume that $J_0<\infty$ (which is equivalent to say that $\nu(|f|^2)<\infty$). We recall that $\|\mu_{(X^i,f^i)}-\nu_{t}\|_{W'}^2 = (\mu_{(X^i,f^i)}-\nu_{t})(\omega^{i})$ where $\omega^{i}=K_{W}(\mu_{(X^i,f^i)}-\nu_{t}) \in W$. Then one easily checks that $J_t<\infty$ and, since we assume that $W$ is continuously embedded into $C_0^2(E\times G_{d}(E) \times \R)$, with existing derivative $J'_t$ at any location $t$ given by
\begin{equation}
  \label{eq:JP10}
    J'_t  = \nu\left(\frac{d}{dt}\left(\rho_t(f)\right)\left(\gamma_f\rho_t(f)+\gamma_W\sum_{i=1}^N\frac{\partial \omega^i}{\partial f}(x,V,\rho_{t}(f))\right)\right) 
\end{equation}
Using again the continuous embedding of $W$ into $C^2_0(E\times G_{d}(E) \times \R)$ we get for a constant $C>0$ that
\begin{align}
    \bigg|\frac{\partial \omega^i}{\partial f}(x,V,\rho_{t}(f))\bigg|&\leq C\|\omega^{i}\|_{W} \nonumber \\
    &\leq C \left ( \|\mu_{(X^i,f^i)}\|_{W'} + \|\nu_{t}\|_{W'} \right ) \,.
\label{eq:JP23bis}
\end{align}
Moreover, as we mentioned after Proposition \ref{prop:control_RKHSnorm_totalvar}, $\|\mu_{(X^i,f^i)}\|_{W'} \leq \text{cte}.\Haus^{d}(X^i)$. Similarly, $\|\nu_{t}\|_{W'} \leq \text{cte}.\nu_{t}(E \times G_{d}(E) \times \R)$ and, since $\nu_{t} \in \MeasX$, we have $\nu_{t}(E \times G_{d}(E) \times \R) = \Haus^{d}(X)$ and consequently $\|\nu_{t}\|_{W'} \leq \text{cte}.\Haus^{d}(X)$. Thus, there exists a constant $K>0$ such that: 
\begin{equation}
  \label{eq:JP23}
\bigg|\sum_{i=1}^{N} \frac{\partial \omega^i}{\partial f}(x,V,\rho_{t}(f))\bigg| \leq K \sum_{i=1}^{N} \left(\Haus^{d}(X^i) + \Haus^{d}(X) \right ) 
\end{equation}

Noticing now that $\frac{d}{dt}\left(\rho_t(f)\right)\rho_t(f)\leq 0$, that $|\frac{d}{dt}\left(\rho_t(f)\right)|=0$ for $|f|\leq a$ and that $|\rho_t(f)|\geq a$ for $|f|\geq a$ and $t\in [0,1]$, we get for $t\in [0,1]$
  \begin{equation}
  \label{eq:JP13}
  J'_t\leq \nu\left(-\bigg|\frac{d}{dt}\left(\rho_t(f)\right)\bigg|\one_{|f|>a}\left(\gamma_fa-\gamma_W K \sum_{i=1}^N\left(\Haus^d(X^i)+\Haus^d(X)\right)\right)\right)
\end{equation}
so that
\begin{equation}
  \label{eq:JP14}
  \eJ(\nu_1)\leq \eJ(\nu_0) \text{ if }a\geq K\frac{\gamma_W}{\gamma_f}\sum_{i=1}^N\left(\Haus^d(X^i)+\Haus^d(X)\right)\,.
\end{equation}
An important consequence of (\ref{eq:JP14}) is that one can restrict the search of a minimum for $\eJ$ on fvarifolds $\nu$ such that
\begin{equation}
  \label{eq:JP15}
  \nu(\one_{|f|> a})=0 
\end{equation}
with $a=K\frac{\gamma_W}{\gamma_f}\sum_{i=1}^N\left(\Haus^d(X^i)+\Haus^d(X)\right)$. In particular, since $\nu\in\MeasX$, we will have
\begin{equation}
  \label{eq:JP16}
  x\in X \text{ and } |f|\leq a\text{ $\nu$ a.e.}
\end{equation}
Since $X$ is bounded and $G_{d}(E)$ compact, we can restrict the search of a minimum on a measure supported on a \emph{compact} subset $C\subset E \times G_{d}(E) \times \R$ so that we introduce: 
\begin{equation}
  \label{eq:JP17}
  \MeasXC\doteq \{\ \nu\in \MeasX\ |\ (x,V,f)\in C\ \text{ $\nu$  a.e.}\ \}\,.
\end{equation}

An easy check shows that $\eJ$ is lower semi-continuous on the set $\MeasXC$ for the weak convergence topology. In addition, $\MeasXC$ is sequentially compact. Indeed, if $\nu_{n}$ is a sequence in $\MeasXC$ then all $\nu_{n}$ are supported by the compact $C$ and in particular $(\nu_{n})$ is tight. Also, as already noted, there exists a constant $\text{cte}$ independent of $n$ such that $\nu_{n}(E \times G_{d}(E) \times \mathbb{R}) \leq \text{cte} \Haus^{d}(X)$ and thus the sequence is uniformly bounded for the total variation norm. It results, thanks to the Prokhorov Theorem, that there exists a subsequence of $(\nu_{n})$ converging for the weak topology. These compactness and lower semicontinuity properties guarantee the existence of a minimizer $\nu_*$ of $\eJ$ with $\nu_*\in\MeasXC$ and 
\begin{equation}
  \label{eq:JP18}
  \eJ(\nu_*)\leq \inf_{f\in L^2(X)}J_X(f)\,.
\end{equation}

\subsection{Proof of proposition}
\label{appendix:proof_propJP1:ssec2}
At this point, we do not yet have a minimizer of $J_X$. The problem is that if the marginal on $E\times G_{d}(E)$ of $\nu_*$ is the transport of $\Haus^d_{|X}$ under the application $x\mapsto (x,T_{x}X)$, we cannot guarantee that $\nu_*$ does not weight multiple signal values in the fiber above a location $(x,T_{x}X)$. We will now show that for large enough $\gamma_f/\gamma_W$, there exists $f_*\in L^2(X)$ such that 
$\nu_*=\nu_{X,f_*}$ so that we will deduce 
\begin{equation}
  \label{eq:JP19}
  J_X(f_*)=\eJ(\nu_*)\leq \inf_{f\in L^2(X)}J_X(f)
\end{equation}
and the existence of a minimizer on $L^2(X)$.

Let $\df\in C_b(E\times G_{d}(E) \times \R)$ and for any $t\in \R$ consider the perturbation $\nu_t\in\MeasX$ of any $\nu\in\MeasXC$ such that for any $g\in C_b(E\times G_{d}(E) \times \R)$ we have:
\begin{equation}
  \label{eq:JP20}
  \nu_t(g)\doteq\int g(x,V,f+t\df(x,V,f))d\nu(x,V,f)\,.
\end{equation}
Here again, the function $t\mapsto \eJ(\nu_t)$ is differentiable everywhere and we have for $\omega^i\doteq K_W (\mu_{(X^i,f^i)}-\nu)$
$$\rst{\frac{d}{dt}\eJ(\nu_t)}{t=0}=\nu\left(\left(\gamma_f f+\gamma_W\sum_{i=1}^N\frac{\partial \omega^i}{\partial f}(x,V,f)\right)\df(x,V,f)\right)\,,$$
so that when $\nu=\nu_*$ we get
\begin{equation}
 \left\{ \begin{array}[h]{l}
     \gamma_f f+\gamma_W A(x,V,f)=0\ \nu_*\text{ a.e.}\\
\text{with}\\
A(x,V,f)\doteq \sum_{i=1}^N\frac{\partial \omega^i}{\partial f}(x,V,f)\,.
  \end{array}
 \right.\label{eq:JP25}
\end{equation}
The partial derivative of $f \mapsto \gamma_f f+\gamma_W A(x,V,f)$ with respect to $f$ equals $\gamma_f + \gamma_W \frac{\partial A}{\partial f}(x,V,f)$. As before, using the continuous embedding $W \hookrightarrow C_{0}^{2}(E\times G_{d}(E) \times \R)$, we have once again a certain constant $K$ such that 
\begin{equation}\label{eq:constantOscillation}
	\bigg| \frac{\partial A}{\partial f}(x,V,f) \bigg | \leq K\sum_{i=1}^N\left(\Haus^d(X^i)+\Haus^d(X)\right), \text{ for all } (x,V,f) \in E\times G_d(E)\times \R.
\end{equation}
It results that for $\gamma_f/\gamma_W$ large enough and for all $(x,V) \in E \times G_{d}(E)$, $f \mapsto \gamma_f f+\gamma_W A(x,V,f)$ is a strictly increasing function going from $-\infty$ at  $-\infty$ to  $+\infty$ at $+\infty$ and thus there is a unique solution $\tilde{f}(x,V)$ to (\ref{eq:JP25}). Now, since the application $f \mapsto \gamma_f f+\gamma_W A(x,V,f)$ is also $C^{1}$ on $E \times G_{d}(E) \times \R$, we deduce from the Implicit Function Theorem that $\tilde{f}$ is a $C^1$ function on $E \times G_{d}(E)$. Going back to the solution $\nu_{\ast}$, we know that for $\nu_{\ast}$ almost every $(x,V,f) \in E \times G_{d}(E) \times \R $, we have $(x,V,f) \in C$ and $f=\tilde{f}(x,V)$, so that $|\tilde{f}| \leq a$ a.e. For any continuous and bounded function $\omega$:
$$ \nu_{\ast}(\omega) = \int \omega(x,V,f) d \nu_{\ast} = \int \omega(x,V,\tilde{f}(x,V)) d \nu_{\ast} $$
and if we denote by $\tilde{\omega}(x,V)\doteq \omega(x,V,\tilde{f}(x,V))$ which is a continuous and bounded function on $E \times G_{d}(E)$, we have by definition of the space $\MeasX$ (eq.\eqref{eq:JP8}):
\begin{equation}
\label{eq:JP24}
 \nu_{\ast}(\omega) = \nu_{\ast}(\tilde{\omega}) = \int_{X} \omega(x,T_{x}X,\tilde{f}(x,T_{x}X)) d\Haus^{d}(x) \,.
\end{equation}
Therefore, setting $f_{\ast}(x) = \tilde{f}(x,T_{x}X)$ for $x \in X$, we see that $|f_{\ast}|\leq a$ so that $f \in L^{\infty}(X)$ and with \eqref{eq:JP24}, we deduce that $\nu_{\ast} = \mu_{(X,f_{\ast})}$ which shows that the solution of the optimization is a fvarifold associated to a true fshape $(X,f_{\ast})$. In addition, if $X$ is a $C^{p}$ submanifold then $x\mapsto T_{x}X$ is a $C^{p-1}$ function on $X$ and, if $W\hookrightarrow C_{0}^{m}(E\times G_{d}(E)\times \R)$ with $m\geq 2$ and $m\geq p$, $A$ and $\tilde{f}$ are $C^{p-1}$ functions so $f_{\ast}$ is also $C^{p-1}$, which concludes the proof of Proposition \ref{prop:JP1}.

\section{Proof of Theorem \ref{theo:existence_atlas_fvar_tang}}
\label{appendix:proof_existence_atlas_fvar_tang}
 We shall basically follow the same steps as in the previous simpler cases. First of all, exactly as in \ref{subsection:atlas_existence_Xnonfixed}, existence of a template shape $X$ is guaranteed with the same compacity and lower semicontinuity arguments. Thus we may assume that $X$ is fixed and we only have to show existence of minimizers to the simplified functional:
\begin{equation*}
  \begin{split}
    & J_{X}(f,(\zeta^{i}),(v^i)) \doteq \frac{\gamma_f}{2}\int_X|f(x)|^2d\Haus^d(x)\\
&+\frac{1}{2}\sum_{i=1}^N\left(\|v^i\|^2_{L^{2}([0,1],V)}+\gamma_\zeta\int_X|\zeta^i(x)|^2d\Haus^d(x)+\gamma_W\|\mu_{(X^i,f^i)}-\mu_{(\phi^{v^i}_{1}(X),(f+\zeta^i)\circ(\phi^{v^i}_{1})^{-1})}\|^2_{W'}\right)
  \end{split}
\end{equation*} 
Now, as for $v^{0}$, due to the presence of the penalizations $\|v^i\|^2_{L^{2}([0,1],V)}$, one can assume that all vector fields $v^{i}$ belong to a fixed closed ball $B$ of radius $r > 0$ in $L^{2}([0,1],V)$. As in the proof of Proposition \ref{prop:JP1}, we first show existence of a minimizer in a space of fvarifolds. Namely, extending the definitions of the previous subsections, we introduce the space $\MeasX$ of measures $\nu$ on $E \times G_{d}(E) \times \mathbb{R} \times \mathbb{R}^{N}$ such that for all continuous and bounded function $h$ on $E \times G_{d}(E)$, we have:   
\begin{equation*}
 \nu(h) = \int h(x,V) d\nu(x,V,f,(\zeta^{i})) = \int_{X} h(x,T_{x}X) d \mathcal{H}^{d}(x) \,.
\end{equation*}
For a measure $\nu$ on $E \times G_{d}(E) \times \mathbb{R} \times \mathbb{R}^{N}$ and a diffeomorphism $\phi$, we denote by $\phi.\nu$ the transport of $\nu$ by $\phi$ defined by:
\begin{equation*}
	(\phi.\nu)(g) = \int |\rst{d_{x}\phi}{V}| g(\phi(x),d_{x}\phi(V),f,(\zeta^{i})) d\nu(x,V,f,(\zeta^{i})) \,.
\end{equation*}
We now introduce the extended functional:
\begin{equation*}
    \tilde{J}(\nu,(v^i)) \doteq \frac{\gamma_f}{2} \nu(|f|^{2})+\frac{1}{2}\sum_{i=1}^N\left(\|v^i\|^2_{L^{2}([0,1],V)}+\gamma_\zeta\nu(|\zeta^{i}|^{2})+\gamma_W\|\mu_{(X^i,f^i)}-(\phi^{v^i}_{1}).\nu^{i}\|^2_{W'}\right)
\end{equation*} 
for $\nu \in \MeasX$, $(v^{i}) \in (L^{2}([0,1],V))^{N}$ and for all $i \in \{1,..,N\}$, $\nu^{i}$ being the fvarifold defined for all $\omega \in W$ by:
$$\nu^{i}(\omega) = \int \omega(x,V,f+\zeta^{i}) d\nu(x,V,f,\zeta^{i}) \,.$$
As previously, we can consider the perturbation function $\rho_{t}$ acting on signals and the measures $$\nu_{t}(g)\doteq \int g(x,V,\rho_{t}(f),(\rho_{t}(\zeta^{i}))) d\nu(x,V,f,(\zeta^{i})) \,.$$
Denoting $J_{t}=\tilde{J}(\nu_{t},(v^i))$, we have, for $t\in [0,1]$,
  \begin{align}
	  J'_t &= \nu\left(\frac{d}{dt}\left(\rho_t(f)\right)\left(\gamma_f\rho_t(f)+\gamma_W\sum_{i=1}^N |\rst{d_{x}\phi^{v^i}_{1}}{V}|.\frac{\partial \omega^i}{\partial f}(\phi^{v^{i}}(x),d_{x}\phi^{v^{i}}(V),\rho_{t}(f)+\rho_{t}(\zeta^{i}))\right)\right) \nonumber \\
	  &+ \nu\left(\sum_{i=1}^{N}\frac{d}{dt}\left(\rho_t(\zeta^{i})\right)\left(\gamma_f\rho_t(\zeta^{i})+\gamma_W |\rst{d_{x}\phi^{v^i}_{1}}{V}|.\frac{\partial \omega^i}{\partial f}(\phi^{v^{i}}(x),d_{x}\phi^{v^{i}}(V),\rho_{t}(f)+\rho_{t}(\zeta^{i}))\right)\right)
\end{align}
where, for all $i \in \{1,..,N\}$, $\omega^{i}=K_{W}(\mu_{(X^{i},f^{i})} - (\phi^{v^i}_{1}).\nu_{t}^{i})$. On the first hand, we know that there exists a constant $\text{cte}$ such that for all $i$, $x \in E$ and $V \in G_{d}(E)$, $|\rst{d_{x}\phi^{v^{i}}}{V}|\leq \text{cte} |d\phi^{v^{i}}|_{\infty}$. In addition, it is a classical result on flows of differential equations (cf \cite{Younes}) that there exists a non-decreasing continuous function $\tau: \R^{+}\rightarrow \R^{+}$ independent of $v \in L^{2}([0,1],V)$ such that $|d\phi^{v}_{1}|_{\infty}\leq \tau(\|v\|_{L^{2}([0,1],V)})$. Now, using the same controls as in the previous subsections, we have, on the other hand:
\begin{align*}
  \left| \frac{\partial \omega^i}{\partial f}(\phi^{v^{i}}(x),d_{x}\phi^{v^{i}}(V),\rho_{t}(f)+\rho_{t}(\zeta^{i})) \right | &\leq \left| \frac{\partial \omega^i}{\partial f} \right |_{\infty} \\
  &\leq \text{cte} \|\omega^{i}\|_{W} \\
  &\leq \text{cte} (\|\mu_{(X^{i},f^{i})}\|_{W'}+\|(\phi^{v^i}_{1}).\nu_{t}^{i}\|_{W'}) \\
  &\leq \text{cte} (\mathcal{H}^{d}(X^{i}) + (\phi^{v^i}_{1}).\nu_{t}^{i}(E\times G_{d}(E)\times \R)) \,. 
\end{align*} 
 It is also straightforward that $(\phi^{v^i}_{1}).\nu_{t}^{i}(E\times G_{d}(E)\times \R)) \leq \text{cte} |d\phi^{v^{i}}|_{\infty} \nu_{t}^{i}(E \times G_{d}(E) \times \R)$ and, using the fact that $\nu \in \MeasX$ as already argued in \ref{ssec:atlas_existence_Xfixed}, $\nu_{t}^{i}(E \times G_{d}(E) \times \R)=\mathcal{H}^{d}(X)$. It results, from all the previous inequalities, the existence of a non-decreasing continuous function that we will still call $\tau$ such that for all $i,x,V,f,\zeta^{i}$:
 \begin{equation}
 \label{eq:A2}
 \left| |\rst{d_{x}\phi^{v^i}_{1}}{V}|.\frac{\partial \omega^i}{\partial f}(\phi^{v^{i}}(x),d_{x}\phi^{v^{i}}(V),\rho_{t}(f)+\rho_{t}(\zeta^{i})) \right | \leq \tau(\|v^i\|_{L^{2}([0,1],V)}).(\Haus^{d}(X^{i})+\Haus^{d}(X))
 \end{equation}
 Following the same path that previously lead to \eqref{eq:JP13}
 \begin{align}
  \label{eq:A3}
  J'_t&\leq \nu\left(-\bigg|\frac{d}{dt}\left(\rho_t(f)\right)\bigg|\one_{|f|>a}\left(\gamma_f\, a-\gamma_W  \sum_{i=1}^N \tau(\|v^i\|_{L^{2}([0,1],V)})\left(\Haus^d(X^i)+\Haus^d(X)\right)\right)\right) \nonumber \\
  &+ \sum_{i=1}^{N}\nu\left(-\bigg|\frac{d}{dt}\left(\rho_t(\zeta^{i})\right)\bigg|\one_{|\zeta^{i}|>a}\left(\gamma_\zeta\, a-\gamma_W \tau(\|v^i\|_{L^{2}([0,1],V)}) \left(\Haus^d(X^i)+\Haus^d(X)\right)\right)\right) \,.
\end{align}
Just as in \ref{ssec:atlas_existence_Xfixed}, this implies that $\tilde{J}(\nu_{1},(v^{i})) \leq \tilde{J}(\nu_{0},(v^{i}))$ as soon as:
\begin{equation*}
 \left\{ \begin{array}[h]{l}
  a\geq \frac{\gamma_{W}}{\gamma_{f}} \sum_{i=1}^{N} \tau(\|v^i\|_{L^{2}([0,1],V)})\left(\Haus^d(X^i)+\Haus^d(X)\right) \\
\text{ and }\\
 a\geq \max_{i} \frac{\gamma_{W}}{\gamma_{\zeta}} \tau(\|v^i\|_{L^{2}([0,1],V)}) \left(\Haus^d(X^i)+\Haus^d(X)\right)
\end{array}\right.
\end{equation*}
Therefore, one may restrict the search of a minimum on a set of measures $\nu$ that are supported on a compact subset $C$ of $E \times G_{d} \times \R \times \R^{N}$, which space we shall denote again $\MeasXC$. The rest of the proof is now very close to the one of \ref{ssec:atlas_existence_Xfixed}. Due to lower semi-continuity of the functional and the compacity of $\MeasXC$ and $B$ for the weak convergence topologies (respectively on the space of measures and on $L^{2}([0,1],V)$), we obtain the existence of a minimizer $(\nu_{\ast},(v^{i})_{\ast})$ for the functional $\tilde{J}$. 

The last step is to prove that $\nu_{\ast}$, which belongs a priori to the measure space $\MeasXC$, can be written under the form $\nu_{\ast}= \nu_{X,f_{\ast},(\zeta^{i}_{\ast})}$, i.e that there exists functions $f_{\ast}$ and $\zeta^{i}_{\ast}$ on $X$ such that, for all continuous and bounded function $g$ on $E \times G_{d}(E) \times \R \times \R^{N}$:
\begin{equation}
\label{eq:A4}
 \nu_{\ast}(g) = \int_{X} g(x,T_{x}X,f_{\ast}(x),(\zeta^{i}_{\ast}(x))) d\Haus^{d}(x)
\end{equation}
We then consider variations of the signals $(\delta f, (\delta \zeta^{i}))$ all belonging to the space $C_{b}(E \times G_{d}(E) \times \R \times \R^{N})$ and the path $t\mapsto \nu_{t}$ defined by:
\begin{equation*}
 \nu_{t}(g)= \int g(x,V,f+t\delta f(x,V,f,(\zeta^{i})),(\zeta^{i}+t\delta\zeta^{i}(x,V,f,(\zeta^{i})))) d\nu_{\ast}(x,V,f,(\zeta^{i})) \,.
\end{equation*}
Now, if $J_{t}\doteq \tilde{J}(\nu_{t},v^{i}_{\ast})$, expressing that $\rst{J_{t}'}{t=0}=0$ for all $\delta f$ and $(\delta \zeta^{i})$ gives, similarly to \ref{ssec:atlas_existence_Xfixed}, the following set of equations:
\begin{equation}
\label{eq:A5}
 (\gamma_{f} f, (\gamma_{\zeta} \zeta^{i})) = -A(x,V,f,(\zeta^{i})) \ \nu_{\ast}\text{-a.e}
\end{equation}
with $A(x,V,f,(\zeta^{i}))\doteq \left ( \sum_{i=1}^{N} \frac{\partial \omega^{i}}{\partial f}(\phi^{v^{i}_{\ast}}_{1}(x),d_{x}\phi^{v^{i}_{\ast}}_{1}(V),f+\zeta^{i}), (\frac{\partial \omega^{i}}{\partial f}(\phi^{v^{i}_{\ast}}_{1}(x),d_{x}\phi^{v^{i}_{\ast}}_{1}(V),f+\zeta^{i})) \right )$. The derivatives $\partial_{f} A$ and $\partial_{\zeta^{i}} A$ can be shown again to be uniformly bounded in $x,V,f,\zeta^{i}$, and a previous argument provides the existence of unique solutions $f=\tilde{f}(x,V)$ and $\zeta^{i}=\tilde{\zeta}^{i}(x,V)$ to \eqref{eq:A5}. The rest of the proof is exactly the same as in the end of appendix \ref{appendix:proof_propJP1:ssec2}: we set $f_{\ast}(x)=\tilde{f}(x,T_{x}X)$ and $\zeta^{i}_{\ast}(x)=\tilde{\zeta}^{i}(x,T_{x}X)$, which are again $L^{\infty}$ functions on $X$. In addition, one shows easily that the minimizing measure $\nu_{\ast}$ equals $\nu_{X,f_{\ast},(\zeta_{\ast}^{i})}$ in the sense of \eqref{eq:A4}. Finally, the regularity of $f_{\ast}$ and $\zeta_{\ast}$ when $X$ is a $C^{p}$ submanifold is obtained again by applying the Implicit Function Theorem to \eqref{eq:A5}.

\section{Proof of Theorem \ref{theo:existence_atlas_fvar_meta}}
\label{appendix:proof_existence_atlas_fvar_meta}
 We shall only sketch the essential steps to adapt the content of appendix \ref{appendix:proof_existence_atlas_fvar_tang}. We start by writing \eqref{eq:M1} in an extensive way. This gives:
  \begin{align}
\label{eq:M2}
& J((v^{0},h^{0}),(v^{i},h^{i})) = \frac{\gamma_{V_0}}{2} \|v^{0}\|_{L^{2}([0,1],V_{0})}^{2} + \frac{\gamma_{f_0}}{2} \int_0^1\int_{X_0}|h^0_t|^2|\rst{d_x\phi^{v^0}_1}{T_{x}X}|d\Haus^d(x) \\ 
&+ \sum_{i=1}^N\left( \frac{\gamma_{V}}{2} \|v^{i}\|_{L^{2}([0,1],V)}^{2} +\frac{\gamma_{f}}{2} \int_0^1\int_{X}|h^i_t|^2|\rst{d_x\phi^{v^i}_1}{T_{x}X}|d\Haus^d(x) +\frac{\gamma_W}{2}\|\mu_{(X^i,f^i)}-\mu_{(\phi^{v^i}_{1}(X),(f+\zeta^{h^{i}}_{1})\circ(\phi^{v^i}_{1})^{-1})} \|^2_{W'}\right)
 \end{align} 
Now, with Lemma \ref{lemma_meta_atlas}, we know that the optimal functions $h^{0}_{\ast}$ and $h^{i}_{\ast}$ are given by \eqref{eq:F1} and thus the variational problem of \eqref{eq:M2} can be replaced by the optimization with respect to residual functions $\zeta^0$ and $\zeta^{i}$ living in $L^{2}(X)$ of the functional:
  \begin{align*}
\label{eq:M3}
   & J((v^{0},\zeta^{0}),(v^{i},\zeta^{i})) = \frac{\gamma_{V_0}}{2} \|v^{0}\|_{L^{2}([0,1],V_{0})}^{2} + \frac{\gamma_{f_0}}{2} \int_{X_0} C^0(x).|\zeta^0(x)|^2 d\Haus^d(x) \\ 
    &+ \sum_{i=1}^N\left( \frac{\gamma_{V}}{2} \|v^{i}\|_{L^{2}([0,1],V)}^{2} +\frac{\gamma_{f}}{2} \int_{X}C^{i}(x).|\zeta^i(x)|^2 d\Haus^d(x) +\frac{\gamma_W}{2}\|\mu_{(X^i,f^i)}-\mu_{(\phi^{v^i}_{1}(X),(f+\zeta^{i})\circ(\phi^{v^i}_{1})^{-1})} \|^2_{W'}\right)
  \end{align*}
  where $C^{0}(x)\doteq(\int_0^1\frac{1}{|\rst{d_x[\phi^{v^{0}}_{s}\circ (\phi_1^{v_0})^{-1}]}{T_{x}X}|}ds)^{-1}$ and for all $i\in \{1,...,N\}$, $C^{i}(x)\doteq(\int_0^1\frac{1}{|\rst{d_x\phi^{v^{i}}_{s}}{T_{x}X}|}ds)^{-1}$. But we note that the previous, up to the weights in the $L^2$ metrics given by functions $C^i$, becomes now extremely close to the problem examined in Theorem \ref{theo:existence_atlas_fvar_tang}. In fact, the proof of appendix \ref{appendix:proof_existence_atlas_fvar_tang} can be adapted almost straightforwardly to this situation. As previously, the essential step is to reformulate the optimization problem in a space of measures. Defining the functions: 
\begin{align*}
	\tilde{C}^{0}(x,H)&= \left ( \int_0^1\frac{1}{|\rst{d_x[\phi^{v^{0}}_{s}\circ (\phi_1^{v_0})^{-1}]}{H}|}ds \right )^{-1} \\
	\tilde{C}^{i}(x,H)&= \left ( \int_0^1\frac{1}{|\rst{d_x\phi^{v^{i}}_{s}}{H}|}ds \right )^{-1}
 \end{align*}
for $i\in \{1,...,N\}$ and $(x,H) \in E \times G_{d}(E)$, we can set, with the same definitions as in appendix \ref{appendix:proof_existence_atlas_fvar_tang}:
\begin{equation}
 \tilde{J}(\nu,(v^i))\doteq \frac{\gamma_{f_0}}{2} \nu(\tilde{C}^0.|f|^2) + \dfrac{1}{2} \sum_{i=1}^N\left( \gamma_{V} \|v^{i}\|_{L^{2}([0,1],V)}^{2} +\gamma_{f} \nu(\tilde{C}^{i}.|\zeta^i|^2) + \gamma_W \|\mu_{(X^i,f^i)}-(\phi^{v^i})_{\ast}\nu^i \|^2_{W'}\right)
\end{equation}
for $\nu \in \MeasX$. The rest of the proof follows the same path, relying on the fact that we can assume the vector fields $v^0$ and $v^i$ to be bounded in $L^{2}([0,1],V_{0})$ and $L^{2}([0,1],V)$ as we explained in the beginning of appendix \ref{appendix:proof_existence_atlas_fvar_tang}. This implies, as already argued in the same section, that we have uniform lower and upper bounds for $|\rst{d_x\phi^{v^{i}}_{s}}{H}|$, $s \in [0,1]$ and $i \in \{1,...,N\}$ and for the quantities $|\rst{d_x[\phi^{v^{0}}_{s}\circ (\phi_1^{v_0})^{-1}]}{H}|$. Consequently, we can assume that we have $\alpha,\beta >0$ such that for all $i \in \{0,...,N\}$, $\alpha \leq \|\tilde{C}^i\|_{\infty} \leq \beta$. Using these inequalities, one can check that we get equivalent controls as in the proof of Theorem \ref{theo:existence_atlas_fvar_tang} which allows us to conclude the existence of a measure minimizer for the extended functional $\tilde{J}$ and then go back to a fshape solution for $J$ with a similar implicit functions' argument.

\bibliographystyle{plain}
\bibliography{biblio_thesis} 

\end{document}